\documentclass[a4paper, 10pt, reqno, DIV=calc]{amsart}

\usepackage[utf8]{inputenc}
\usepackage[T1]{fontenc}
\usepackage{lmodern}
\usepackage{microtype}
\usepackage{amsmath, amsthm, amssymb, thmtools}
\usepackage{relsize}
\usepackage{enumitem}		
\usepackage[mathscr]{eucal}

\usepackage{accents}
\newcommand{\ubar}[1]{\underaccent{\bar}{#1}}

\usepackage[dvipsnames]{xcolor}
\usepackage[linktocpage]{hyperref}
\definecolor{colorred}{HTML}{B00000}
\definecolor{colorgreen}{HTML}{258300}
\definecolor{colorblue}{HTML}{2e32fa}
\hypersetup{colorlinks=true, linkcolor=colorred, citecolor=colorgreen, urlcolor=black}								
\usepackage{xifthen}								
\usepackage{url}
\usepackage{todonotes}
\usepackage{caption}
\usepackage{bbm}									
\usepackage{stmaryrd}								
\usepackage{upgreek}

\makeatletter
\newcommand\MyAutoefPhrasecolorGroup[1]{%
  \color@begingroup\color{MyCurrentcolor}#1\endgroup
}%
\def\HyRef@testreftype#1.#2\\{%
 \colorlet{MyCurrentcolor}{.}%
 \ltx@IfUndefined{#1autorefname}{%
   \ltx@IfUndefined{#1name}{%
     \HyRef@StripStar#1\\*\\\@nil{#1}%
     \ltx@IfUndefined{\HyRef@name autorefname}{%
       \ltx@IfUndefined{\HyRef@name name}{%
         \def\HyRef@currentHtag{}%
         \Hy@Warning{No autoref name for `#1'}%
       }{%
         \edef\HyRef@currentHtag{%
           \noexpand\MyAutoefPhrasecolorGroup{%
             \expandafter\noexpand\csname\HyRef@name name\endcsname
           }%
           \noexpand~%
         }%
       }%
     }{%
       \edef\HyRef@currentHtag{%
         \noexpand\MyAutoefPhrasecolorGroup{%
           \expandafter\noexpand
           \csname\HyRef@name autorefname\endcsname
         }%
         \noexpand~%
       }%
     }%
   }{%
     \edef\HyRef@currentHtag{%
       \noexpand\MyAutoefPhrasecolorGroup{%
         \expandafter\noexpand\csname#1name\endcsname
       }%
       \noexpand~%
     }%
   }%
 }{%
   \edef\HyRef@currentHtag{%
     \noexpand\MyAutoefPhrasecolorGroup{%
       \expandafter\noexpand\csname#1autorefname\endcsname
     }%
     \noexpand~%
   }%
 }%
}%
\makeatother

\numberwithin{equation}{section}

\newcommand{\nlb}{{\ensuremath{\textnormal{b}}}}
\newcommand{\nlc}{{\ensuremath{\textnormal{c}}}}

\newcommand{\nlC}{{\ensuremath{\textnormal{C}}}}

\newcommand{\rmd}{{\ensuremath{\mathrm{d}}}}
\newcommand{\rme}{{\ensuremath{\mathrm{e}}}}

\newcommand{\rmg}{{\ensuremath{\mathrm{g}}}}

\newcommand{\rmo}{{\ensuremath{\mathrm{o}}}}

\newcommand{\rmG}{{\ensuremath{\mathrm{G}}}}

\newcommand{\sfc}{{\ensuremath{\mathsf{c}}}}
\newcommand{\sfd}{{\ensuremath{\mathsf{d}}}}
\newcommand{\sfe}{{\ensuremath{\mathsf{e}}}}

\newcommand{\sfg}{{\ensuremath{\mathsf{g}}}}
\newcommand{\sfh}{{\ensuremath{\mathsf{h}}}}

\newcommand{\sfs}{{\ensuremath{\mathsf{s}}}}
\newcommand{\sft}{{\ensuremath{\mathsf{t}}}}

\newcommand{\sfv}{{\ensuremath{\mathsf{v}}}}

\newcommand{\sfB}{{\ensuremath{\mathsf{B}}}}

\newcommand{\sfQ}{{\ensuremath{\mathsf{Q}}}}

\newcommand{\sfS}{{\ensuremath{\mathsf{S}}}}

\newcommand{\scrA}{{\ensuremath{\mathscr{A}}}}

\newcommand{\scrC}{{\ensuremath{\mathscr{C}}}}
\newcommand{\scrD}{{\ensuremath{\mathscr{D}}}}

\newcommand{\scrH}{{\ensuremath{\mathscr{H}}}}

\newcommand{\scrJ}{{\ensuremath{\mathscr{J}}}}
\newcommand{\scrK}{{\ensuremath{\mathscr{K}}}}

\newcommand{\scrM}{{\ensuremath{\mathscr{M}}}}

\newcommand{\scrP}{{\ensuremath{\mathscr{P}}}}

\newcommand{\scrT}{{\ensuremath{\mathscr{T}}}}
\newcommand{\scrU}{{\ensuremath{\mathscr{U}}}}

\newcommand{\scrX}{{\ensuremath{\mathscr{X}}}}
\newcommand{\scrY}{{\ensuremath{\mathscr{Y}}}}
\newcommand{\scrZ}{{\ensuremath{\mathscr{Z}}}}

\newcommand{\bdpi}{{\ensuremath{\boldsymbol{\pi}}}}

\newcommand{\bdsigma}{{\ensuremath{\boldsymbol{\sigma}}}}

\newcommand{\N}{\boldsymbol{\mathrm{N}}}						
\newcommand{\Q}{\boldsymbol{\mathrm{Q}}}						
\newcommand{\R}{\boldsymbol{\mathrm{R}}}						
\renewcommand{\d}{\,\mathrm{d}}				

\let\limsup\undefined
\let\liminf\undefined

\DeclareMathOperator*{\limsup}{limsup}		
\DeclareMathOperator*{\liminf}{liminf}		
\DeclareMathOperator*{\esssup}{esssup}		
\DeclareMathOperator*{\essinf}{essinf}		
\DeclareMathOperator{\supp}{spt}			
\DeclareMathOperator{\tr}{tr}				

\theoremstyle{definition}
\newtheorem{bump}{Bump}[section]

\theoremstyle{plain}
\newtheorem{theorem}[bump]{Theorem}
\newtheorem{proposition}[bump]{Proposition}
\newtheorem{definition}[bump]{Definition}
\newtheorem{lemma}[bump]{Lemma}
\newtheorem{corollary}[bump]{Corollary}
\newtheorem{assumption}[bump]{Hypothesis}

\theoremstyle{remark}
\newtheorem{remark}[bump]{Remark}
\newtheorem{example}[bump]{Example}

\makeatletter
\def\nonumberfootnote{\xdef\@thefnmark{}\@footnotetext}			
\makeatother

\newcommand{\mms}{\mathsf{M}}				
\newcommand{\met}{\sfd}						

\newcommand{\Meas}{\mathfrak{M}}
\newcommand{\Rmet}{g}						
\newcommand{\meas}{\mathfrak{m}}				
\newcommand{\mmeas}{\mathfrak{n}}
\newcommand{\Leb}{\mathscr{L}}				
\newcommand{\vol}{\mathrm{vol}}				
\newcommand{\Prob}{\mathscr{P}}		        
\newcommand{\OptTGeo}{\mathrm{OptTGeo}}
\newcommand{\Id}{\mathrm{Id}}				
\newcommand{\ac}{{\mathrm{ac}}}

\newcommand{\TCD}{\mathrm{TCD}}
\newcommand{\wTCD}{\mathrm{wTCD}}
\newcommand{\TMCP}{\mathrm{TMCP}}
\newcommand{\CD}{\mathrm{CD}}

\newcommand{\bounded}{\nlb}					
\newcommand{\comp}{\nlc}					
\newcommand{\loc}{\mathrm{loc}}				
\newcommand{\pr}{\mathrm{pr}}				
\newcommand{\Ric}{\mathrm{Ric}}				
\newcommand{\Cont}{\nlC}					
\newcommand{\Ell}{\mathit{L}}				

\newcommand{\End}{\mathrm{bp}}

\newcommand{\Dom}{\scrD}					

\DeclareMathOperator{\Ent}{Ent}				
\DeclareMathOperator{\Hess}{Hess}			
\DeclareMathOperator{\diam}{diam}			
\newcommand{\eval}{\sfe}					

\newcommand{\Restr}{\mathrm{restr}}			
\newcommand{\push}{\sharp}					

\newcommand{\cl}{\mathrm{cl}}				
\newcommand{\Len}{\mathrm{Len}}
\newcommand{\TGeo}{\mathrm{TGeo}}
\newcommand{\tsep}{\tau}
\newcommand{\One}{1}
\newcommand{\ray}{\sfg}
\newcommand{\Sec}{S}
\DeclareMathOperator{\graph}{graph}
\newcommand{\Quot}{\mathsf{Q}}
\newcommand{\Inter}{\upiota}
\newcommand{\Alex}{\upalpha}

\usepackage{blindtext}

\allowdisplaybreaks

\providecommand{\bysame}{\leavevmode\hbox to3em{\hrulefill}\thinspace}

\let\oldtocsection=\tocsection

\let\oldtocsubsection=\tocsubsection

\let\oldtocsubsubsection=\tocsubsubsection

\renewcommand{\tocsection}[2]{\hspace{0em}\oldtocsection{#1}{#2}}
\renewcommand{\tocsubsection}[2]{\hspace{1em}\oldtocsubsection{#1}{#2}}
\renewcommand{\tocsubsubsection}[2]{\hspace{2em}\oldtocsubsubsection{#1}{#2}}

\newcommand{\nocontentsline}[3]{}
\newcommand{\tocless}[2]{\bgroup\let\addcontentsline=\nocontentsline#1{#2}\egroup}

\newcommand{\mres}{\mathbin{\vrule height 1.6ex depth 0pt width
0.13ex\vrule height 0.13ex depth 0pt width 1.3ex}}

\DeclareFontFamily{U}{mathb}{\hyphenchar\font45}
\DeclareFontShape{U}{mathb}{m}{n}{
<-6> mathb5 <6-7> mathb6 <7-8> mathb7
<8-9> mathb8  <9-11> mathb9
<11-12> mathb10 <12-> mathb12
}{}
\DeclareSymbolFont{mathb}{U}{mathb}{m}{n}
\DeclareMathSymbol{\llcurly}{\mathrel}{mathb}{"CE}
\DeclareMathSymbol{\ggcurly}{\mathrel}{mathb}{"CF}

\newcommand{\scal}{\mathrm{scal}}
\newcommand{\cost}{\mathfrak{c}}
\newcommand{\SIN}{\sin}

\newcommand{\LSC}{\mathrm{LSC}}
\DeclareMathOperator{\diag}{diag}
\newcommand{\TCut}{\mathrm{TCut}}
\newcommand{\Tr}{\scrT}
\newcommand{\Top}{\uptau}
\newcommand{\q}{\mathfrak{q}}
\newcommand{\hh}{\mathfrak{h}}
\newcommand{\MCP}{\mathrm{MCP}}
\newcommand{\Reg}{\mathrm{Reach}}

\allowdisplaybreaks

\makeatletter
\@namedef{subjclassname@2020}{\textup{2020} Mathematics Subject Classification}
\makeatother

\begin{document}

\title[Variable timelike Ricci curvature bounds]{Causal convergence conditions through variable timelike Ricci curvature bounds}
\author{Mathias Braun}
\author{Robert J. McCann}
\address{Department of Mathematics, University of Toronto, 40 St. George Street Room 6290, Toronto, Ontario M5S 2E4, Canada}
\email{braun@math.toronto.edu}
\email{mccann@math.toronto.edu}
\thanks{MB's research was supported in part by the Fields Institute for Research in Mathematical Sciences.  RM's research is supported in part by the Canada Research Chairs program CRC-2020-00289, the Simons Foundation, Natural Sciences and Engineering Research Council of Canada Discovery Grant RGPIN-2020--04162. We thank Fabio Cavalletti and Andrea Mondino for communicating their constant $k$ versions of \autoref{Th:Localization TCD} and \autoref{Th:Mean zero} to us, and for pointing out their related work \cite{cm++}.}
\thanks{\copyright \today\ by the authors.}

\begin{abstract} We describe a nonsmooth notion of globally hyperbolic, regular length metric spacetimes $(\mms,l)$. It is based on ideas of Kunzinger--Sämann, but does not require Lipschitz continuity of causal curves. We study geodesics on $\mms$ and the space of probability measures over $\mms$ in detail.

Furthermore, for such a spacetime endowed with a reference measure $\mathfrak{m}$, a lower semicontinuous function $k\colon \mms \to \textbf{R}$, and  constants $0<p<1$ and $N\geq 1$, we introduce and study the entropic timelike curvature dimension condition $\smash{\mathrm{TCD}_p^e(k,N)}$ with variable Ricci curvature bound $k$. This provides a unified synthetic approach to general relativistic energy conditions, including
\begin{itemize}
\item the Hawking--Penrose strong energy condition $\mathrm{Ric}\geq 0$, or more generally $\Ric\geq K$ for constant $K\in\textbf{R}$, in all timelike directions,
\item the weak energy condition $\mathrm{Ric} \geq \mathrm{scal} - \Lambda$ in all timelike directions, and
\item the null energy condition $\smash{\mathrm{Ric} \geq 0}$ in all null directions.
\end{itemize}
Our approach also allows for the synthetic quantification of asymptotic conditions or integral controls on the timelike Ricci curvature. For example, we give a nonsmooth generalization of a timelike diameter estimate of Frankel--Galloway (and Schneider), and of a Hawking-type singularity theorem which requires only that the negative Ricci curvature have small enough integral relative to the maximal mean curvature of an achronal slice.

As further applications, we discuss the stability of our notion and provide timelike geometric inequalities. To obtain sharp constants  in the latter, we develop the localization paradigm in the variable $k$ framework.
\end{abstract}

\thispagestyle{empty}

\maketitle

\thispagestyle{empty}

\tableofcontents

\addtocontents{toc}{\protect\setcounter{tocdepth}{1}}

\section{Introduction}\label{Ch:Intro}

\subsection*{Outline} This work  contributes to the recent fascinating developments in mathematical general relativity from the angle of metric measure geometry \cite{gromov1999, villani2009}. Such a nonsmooth approach to  gravity seems highly desirable, as expressed e.g.~in the introductions of \cite{cavalletti2020, kunzinger2018} and  acknowledged in various special issues \cite{cavalletti2022, steinbauer}. We consider a lower semicontinuous function $k\colon \mms\to\R$ on an abstract  space $\mms$ to be specified and a ``dimensional'' parameter $N\geq 1$. By using tools from optimal transport in Lorentzian signature  \cite{cavalletti2020,eckstein2017,kell2018,mccann2020,
mondinosuhr2022, suhr2018}, we introduce and study the (\textit{entropic}) \textit{timelike curvature-dimension condition} $\smash{\TCD_p^e(k,N)}$ for $\mms$, where $0<p<1$. This  synthetically quantifies the ``smooth'' conditions
\begin{itemize}
\item $\Ric \geq k$ in all timelike directions, and
\item $\dim \mms \leq N$.
\end{itemize}
We also survey its dimensionless counterpart $\smash{\TCD_p(k,\infty)}$. To reduce the level of technicality, we leave the precise definition of these conditions as a black box for now, and refer to the ``Motivation of \autoref{Def:TCDe}''  paragraph below for more details. The reader unfamiliar with optimal transport or metric measure geo\-metry is invited to think of the special case of $\mms$  being a smooth, globally hyperbolic Lo\-rentzian spacetime, endowed with its volume form (or a weighted version  thereof); some of our results are new even in this situation.

Our main motivation for this work is the crucial role played by 
\begin{align}\label{Eq:RICmmmm}
\Ric \geq k\quad\textnormal{in all timelike directions}
\end{align}
in classical general relativity as an energy condition. We  provide a unified synthetic treatment of the following three special cases; cf.~\cite{carroll,hellis,wald1984} for background.
\begin{itemize}
\item \textbf{Timelike convergence condition.} This refers to \eqref{Eq:RICmmmm} for constant $k$; for vanishing $k$ and assuming general relativity, this property is also termed  \emph{strong energy condition}. Independently, McCann \cite{mccann2020} and Mondino--Suhr \cite{mondinosuhr2022} pioneered its optimal transport characterization. This lead to the timelike curvature-dimension condition of Cavalletti--Mondino \cite{cavalletti2020}. 
An alternative notion was later proposed by Braun \cite{braun2022}.
\item \textbf{Weak energy condition.} Assuming general relativity, this special property asserts lower boundedness of $\Ric$ by the function $\scal/2 - \Lambda$ in every timelike direction. Here, $\Lambda \in \R$ is the cosmological constant. It is believed to hold in most physically reasonable circumstances. 
\item \textbf{Null convergence condition.} A far less evident  scenario fitting into \eqref{Eq:RICmmmm} is the nonnegativity of the Ricci curvature in all \emph{null} directions. Its equivalence to locally uniform lower bounds on the tensor $\Ric$ in all \emph{timelike} directions --- thus to \eqref{Eq:RICmmmm} --- has recently been shown by McCann \cite{mccann+}, thereby addressing an open question of \cite{cavalletti2022}. An alternative approach  is suggested by work of Ketterer \cite{ketterer2023}. The null energy condition  describes the stress-energy experienced by light rays and is 
expected to hold throughout the entire  non-quantum 
universe. Besides this physical  relevance, it is the basis for Penrose's singularity theorem \cite{penrose1965}, Hawking's area monotonicity \cite{haw72}, Galloway's null splitting theorem \cite{gallo}, topological censorship \cite{eich}, etc.
\end{itemize}
Moreover, our abstract formulation of \eqref{Eq:RICmmmm}  accommodates two further important situations on a synthetic level.
\begin{itemize}
\item \textbf{Asymptotic conditions.} Several spacetimes  that   describe cosmological phenomena fall into the category of Friedmann--Lemaître--Robertson--Walker geometries  \cite{beem1996,minwarped,oneill1983, wald1984}, or related models such as those of \cite{treude}. Such warped products are the mathematically best understood models of the universe. In general, their time\-like Ricci curvature is not uniformly bounded from below, but instead behaves as a function of the time coordinate; asymptotic conditions on it yield  singularity theorems \cite{fg,galloway1982}. Moreover, \eqref{Eq:RICmmmm}  allows one to phrase averaged energy conditions \cite{ford, tipler} (in terms of integral assumptions on $k$ along geodesics, so-called ``worldline bounds''), which lead to smooth singularity theorems as well \cite{borde, roman}. 
\item \textbf{Integral bounds.} Integrated controls on the negative part of the Ricci tensor play a fundamental role in Riemannian geometry \cite{cc}. Analogously to semiclassical phase space estimates for Schr\"odinger operators, 
$\Ell^p$-bounds for $2\,p > \dim\mms$ are particularly frequent, e.g.~\cite{aubry2007,pw1,pw2}; on metric measure spaces, their stability has been explored in \cite{kstab}. In general relativity, these are called ``worldvolume bounds''. Several relevant models exhibit  curvature behavior which is tame in an integrated way despite being badly behaved pointwise \cite{brown, fewster,kontou}. Activities in this direction have also triggered interest in the parallel null case \cite{fliss, freivogel}.
\end{itemize}
We will justify the synthetic treatment of asymptotic  conditions and integral bounds by establishing one singularity theorem in each framework, see \autoref{Th:FGALL} and \autoref{Th:HAW} below for details.

\clearpage

\subsection*{Results} 

\subsubsection*{Metric measure spacetimes} We work in the abstract setting of a \emph{metric measure spacetime}. Roughly speaking, this is a triple $\scrM:= (\mms,l,\meas)$ consisting of
\begin{itemize}
\item a space $\mms$ endowed with a first-countable topology,
\item a \emph{signed time separation function} $l\colon \mms^2\to [0,\infty] \cup\{-\infty\}$ satisfying the reverse triangle inequality, and
\item a Radon measure $\meas$ on $\mms$ (assumed with full support in this introduction to simplify the presentation).
\end{itemize}
Its definition is tailored towards three further key  properties, namely 
\begin{itemize}
\item \textbf{global hyperbolicity}, i.e.~the causality relation induced by $l$ is a closed order and causal diamonds are compact, following the analogous definition for topological preordered spaces \cite{nachbin} suggested by Minguzzi \cite{minguzzi2023},
\item a \textbf{length} property, i.e.~for every $x,y\in\mms$ the quantity $l(x,y)$ is arbitrarily well approximated by the ($l$-)length of a causal curve from $x$ to $y$, and
\item \textbf{regularity}, i.e.~if $l(x,y) > 0$, every nowhere constant causal curve with length $l(x,y)$ is already timelike.
\end{itemize}
Such a nonsmooth approach to ``spacetime structures'' is not new. First proposals by Busemann \cite{busemann} and Kronheimer--Penrose \cite{kronheimer} already date back to the sixties. Causality theory on continuous spacetimes has been studied order-theoretically by Sorkin--Woolgar \cite{sorkin} (see also \cite{chrusciel2012,ghe}). A recent systematic theory capturing various  causality phenomena are Kunzinger--Sämann's \emph{Lorentzian length spaces} \cite{kunzinger2018}, motivated by the well-developed theory of length metric spaces in positive signature. Their article also introduces timelike sectional curvature bounds à la Alexandrov  \cite{abis,AnderssonHoward98,bgp}; moreover, the timelike curvature-dimension conditions of \cite{braun2022, cavalletti2020} have been set up in this framework. The theory of Kunzinger--Sämann has been streamlined in \cite{mccann+} in the  globally hyperbolic, regular case. Related to  \cite{BombelliNoldus04,kunzinger2018,Noldus04a}, and 
tailor-made for Lorentz--Gromov--Hausdorff convergence,  are Minguzzi--Suhr's bounded Lorentzian metric spaces \cite{minguzzi2022}. (For another approach to Lorentz--Gromov--Hausdorff convergence see \cite{Mueller22+}.)  In place of a  time separation function, Sormani--Vega encode the ``metric'' properties of a continuous spacetime by their null distance \cite{sormanivega}. This has been adapted to Lorentzian length spaces in \cite{kste}.

Our first contribution further relaxes the approach of \cite{mccann+}. In fact, many of our results here  follow known proofs distributed over various works, e.g.~\cite{ake2020,burtscher,mintop,minguzzi2019,minguzzi2023,
minguzzi2022}; we provide a self-contained and streamlined exposition on these. Our innovation lies in the generality of our setting. We completely waive rectifiability assumptions on causal curves and any other property which depends on the choice of a metrizable topology, cf.~\autoref{Re:Nonambig}. In particular, our  compactness results solely rely on causality arguments.  This is useful since in the strict sense of \cite{kunzinger2018}, causality of curves is \emph{not} preserved under affine or proper-time reparametrizations. Moreover, Lipschitz continuity of causal curves per se imposes an additional requirement to deal with in the construction of limits. 

Under quite weak assumptions, cf.~\autoref{Ass:GHLLS} (length property and compact diamonds), we establish key properties like global hyperbolicity (\autoref{Th:GHy}), finiteness of $l_+$ (\autoref{Cor:Finiteness l}), topological facts pertaining    to  Hausdorffness, local compactness, and Polishness  (\autoref{Cor:Hausdorff} and \autoref{Th:Compact Polish}), the existence of generalized time functions (\autoref{Th:Ex time function}), a limit curve theorem (\autoref{Th:Limit curve theorem}), and more. The analog of Avez--Seifert's theorem stated in \autoref{Th:Ex geo} holds as well: any two points  $x,y\in\mms$ with $l(x,y) > 0$ are connected by an $l$-geodesic. The latter means a map $\gamma\colon[0,1]\to \mms$ such that for every $0\leq s< t\leq 1$,
\begin{align*}
l(\gamma_s,\gamma_t) = (t-s)\,l(\gamma_0,\gamma_1) >0.
\end{align*}
The additional assumption of regularity, cf.~\autoref{Ass:REG}, will ensure every such curve is  continuous. This fact observed in \cite{mccann+}  is  recalled in \autoref{Le:Continuity geos}.

Second, we study the ``Lorentzian'' geometry of the space of Borel probability measures $\Prob(\mms)$.  Following \cite{eckstein2017}, a natural causality relation $\preceq$ on $\scrP(\mms)$ is defined by setting $\mu\preceq\nu$ if there is a coupling $\pi\in\Pi(\mu,\nu)$ with support in $\{l\geq 0\}$. 

\begin{theorem}[Global hyperbolicity, cf.~\autoref{Th:GHy prob meas}]\label{Th:1111} \autoref{Ass:GHLLS} implies the space $\Prob(\mms)$ is globally hyperbolic in the following sense.
\begin{enumerate}[label=\textnormal{(\roman*)}]
\item The relation $\preceq$ is a narrowly closed order.
\item If $\scrC\subset\Prob(\mms)$ is narrowly compact, so is its  causal emerald $\scrJ(\scrC,\scrC)$.
\end{enumerate}
\end{theorem}

This result is new even in the smooth setting of \cite{eckstein2017}. In other words, global hyperbolicity ``lifts'' from $\mms$ to $\Prob(\mms)$. This adds an entire class of new topological ordered spaces to nonsmooth general relativity. It motivates a systematic study of causality (and chronology) on $\Prob(\mms)$. An interesting question --- beyond the scope of our work --- is to which extent  other phenomena lift from $\mms$ to $\Prob(\mms)$ or ``project down'' in the reverse direction. This would parallel general known metric correspondences between Wasser\-stein spaces and their bases \cite{villani2009}. In addition, we thoroughly study the analog of $l$-geodesics for probability measures \cite{mccann2020}, set up in \autoref{Def:lp geo} in terms of the total transport cost $\smash{\ell_p}$  defined in  \eqref{lp causal}; here $0<p<1$.  Notably, we establish criteria for their narrow continuity using regularity (\autoref{Cor:Equiv notions lp geo}), a lifting theorem (\autoref{Le:Geodesics plan}), and a limit curve theorem (\autoref{Pr:Compactness lp geos}). This relates the narrow topology to properties merely defined in terms of  $\smash{\ell_p}$.

\subsubsection*{Timelike curvature-dimension conditions} In the second part of this work, we define the conditions $\smash{\TCD_p^e(k,N)}$ and $\smash{\TCD_p(k,\infty)}$ for $\scrM$ and derive several properties. Our approach is inspired by the theory of metric measure spaces with variable Ricci curvature lower bounds without \cite{sturm2015} and with  \cite{ketterer2015,ketterer2017} dimensional restrictions (see also \cite{braun2021,sturm2020}). In rough terms, these properties are phrased by convexity properties of the Boltzmann entropy with respect to $\meas$ along $\smash{\ell_p}$-geodesics interpolating any two appropriate  elements of $\Prob(\mms)$. We discuss some of the motivating examples indicated in the ``Outline'' paragraph in \autoref{Ch:Examples}. We also state a stability result in \autoref{Th:Stab}\footnote{As explained in the introduction of \cite{braun2022} (and already contained in the axioms of   \cite{cavalletti2020}), this necessitates the definition of weak versions of $\smash{\TCD_p^e(k,N)}$ and $\smash{\TCD_p(k,\infty)}$.}. We prove geometric inequalities à la 
\begin{itemize}
\item Brunn--Minkowski (\autoref{Th:Timelike BM}), 
\item Bonnet--Myers (\autoref{Th:BonnetMyers}), and 
\item Bishop--Gromov (\autoref{Th:Bishop Gromov}). 
\end{itemize}
Moreover, we extend a singularity theorem by Frankel--Galloway from the smooth \cite{fg,galloway1982} to the synthetic setting; see \autoref{Th:FGALL} below for its sharp form.

Finer results are established under the hypothesis of timelike $p$-essential nonbranching, cf.~\autoref{Def:TL nonbranching}. Roughly speaking, the latter means that $\smash{\ell_p}$-optimal transports never charge sets of forward or backward branching $l$-geodesics.  This notion introduced in \cite{braun2022} following \cite{rajala2014} relaxes timelike nonbranching \cite{cavalletti2020}. First, by local lower boundedness of $k$, all qualitative properties from  \cite{braun2022,braun2023,cavalletti2020} transfer to our setting: uniqueness of chronological $\smash{\ell_p}$-optimal couplings (\autoref{Th:Optimal maps}) and $\smash{\ell_p}$-geodesics (\autoref{Th:Uniqueness geos}), and the existence of good geodesics (\autoref{Th:Good}). Our contribution lies in proving the foregoing results by only assuming $\smash{\TCD_p^e(k,N)}$ locally, cf.~\autoref{Def:Local TCD}. This gives a more streamlined proof of the equivalence of $\smash{\TCD_{p,\loc}^e(k,N)}$ and $\smash{\TCD_p^e(k,N)}$ (\autoref{Th:Local to global}), pioneered in \cite{braun2022} for constant $k$. Lastly, via a pathwise characterization of $\smash{\TCD_p^e(k,N)}$ we identify the latter with its above mentioned weak version (\autoref{Th:Pathwise}).

From a technical point of view, one may ask if $\smash{\TCD_p^e(k,N)}$ is independent of the transport exponent $p$ under suitable nonbranching hypotheses. This is unknown even for constant $k$ \cite{cavalletti2022}, yet expected by analogous facts for $\CD$ metric measure spaces \cite{akdemir,braunc2} based on \cite{cavalletti2021}. 

\subsubsection*{Localization} In the third (rather technical)   part, we develop the Lorentz\-ian localization paradigm in the variable case. The broad term ``localization'' from convex geometry refers to the reduction of a high-dimensional problem to an often easier one-dimensional task.  It originates in Payne--Weinberger's treatment of the optimal Euclidean Poincaré inequality \cite{payne} and was further developed by Gromov--Milman~\cite{gmilman}, \smash{Lovász}--Simonovits \cite{losim}, and Kannan--\smash{Lovász}--Simonovits \cite{kannan}.  Related ideas were used to study optimal transport maps in 
$\R^n$ by Sudakov \cite{Sudakov76}, Ambrosio~\cite{Ambrosio03}, Caffarelli--Feldman--McCann \cite{caffarelli2001}, and Trudinger--Wang \cite{TrudingerWang01}, and in more general geometries by  
 Feldman--McCann \cite{feldman}, Bianchini--Cavalletti \cite{bianchini2013}, and Cavalletti~\cite{cavalletti2014}.   As a technique for deriving inequalities on smooth Riemannian manifolds, 
this method was pioneered by Klartag \cite{klartag2017}. This inspired Cavalletti--Mondino's localization paradigm for $\CD$ metric measure spaces \cite{cavalletti2017} and  their Lorentzian entropic TCD spaces \cite{cavalletti2020}. 

We consider a Borel function $u\colon E \to \R$ on an $l$-geodesically convex Borel set $E\subset\mms$. Assume every $x,y\in E$ satisfies
\begin{align}\label{Eq:uneedle}
u(y) - u(x) \geq l(x,y).
\end{align}
Then an appropriate subset $\smash{\Tr_u\subset E}$ admits a partition into $l$-geodesic ``rays''   $\mms_\alpha$ --- also termed  ``needles'' --- passing through $\alpha\in Q$, where $Q\subset E$ is an index set, with the following property. The restriction $\smash{\meas' := \meas\mres\Tr_u}$  admits a disintegration $\smash{\rmd \meas'(x) = \rmd\meas_\alpha(x)\d\q(\alpha)}$ into ``conditional measures'' $\meas_\alpha$ concentrated on $\mms_\alpha$, where $\q$ is a Borel probability measure on $Q$ (\autoref{Th:Needle}). In fact, for $\q$-a.e.~$\alpha\in Q$ we show $\meas_\alpha$ is absolutely continuous with respect to the one-dimensional Hausdorff measure on $\mms_\alpha$ (\autoref{Pr:Abs cont needles}) and the induced metric measure space obeys a variable curvature-dimension condition in the sense of Ketterer \cite{ketterer2015}  (\autoref{Th:Localization TCD}). In addition, equality holds in \eqref{Eq:uneedle} along $\q$-a.e.~needle. Viewing these  observations from the perspective of the null convergence  condition,  it is interesting to note that nontrivial curvature information can be extracted for $\q$-a.e.~ray from this ``codimension one'' hypothesis.

Besides the variability of $k$, our  localization complements \cite{cavalletti2020} as follows. First, by starting from a stronger assumption, instead of the measure contraction property we establish a stronger curvature property of $\q$-a.e.~ray. For constant $k$, such a result has been achieved by Cavalletti--Mondino in connection with their work on Lorentzian isoperimetry \cite{cm++}, as they first  
communicated to us in 2022. Second, we assume mere  timelike $p$-essential nonbranching. This does not exclude branching $l$-geodesics,  and some effort is needed to show branching points are $\meas$-negligible (\autoref{Th:All negligible}); our arguments are inspired by \cite{cavalletti2014}. Third, in \cite{cavalletti2020} $u$ was chosen as the signed distance function from an achronal, future timelike complete (FTC) Borel set. The advantage in allowing for more general $u$ comes from $\smash{\ell_1}$-optimal transport theory. In typical applications, $u$ will be a Kantorovich potential for an $\smash{\ell_1}$-optimal transport problem. The technical background, \autoref{Th:Mean zero}, seems new in Lorentzian signature, as  does the accompanying Kantorovich duality  formula  for $\smash{\ell_1}$ (\autoref{Th:Rubinstein}).  For  constant $k$, similar results appear in \cite{cm++}, which we did not receive until after the present work had been circulated.  Related developments are also pursued in \cite{Akdemir24+}.

\subsubsection*{Applications} To simplify the presentation, let $N>1$. Using \autoref{Th:Mean zero}, inspired by \cite{cmgeometric2017} we improve all geometric inequalities from \autoref{Sec:Geometric inequ} to their sharp versions; we refer to \cite[Rem.~5.11]{cavalletti2020} and \autoref{Ex:WP} for elaborations on sharpness. This notably pertains to the timelike Brunn--Minkowski inequality (\autoref{Th:Sharp BMink}), whose sharp form under $\smash{\TCD_p^e(k,N)}$ was unknown even for constant $k$. 

A particular application is the subsequent sharp singularity theorem, valid under  an \emph{asymptotic} condition on the timelike Ricci curvature. It extends a smooth result by Frankel--Galloway \cite{fg,galloway1982}.

\begin{theorem}[Synthetic Frankel--Galloway singularity theorem, cf.~\autoref{Cor:Sharp Schneider}]\label{Th:FGALL} Assume $\scrM$ is a timelike $p$-essentially nonbranching $\smash{\TCD_p^e(k,N)}$ space, and suppose that some $o\in \mms$, $\beta > 0$, and $R>0$ satisfy
\begin{align*}
k\geq (N-1)\,\Big[\frac{1}{4} + \beta^2\Big]\,l(o,\cdot)^{-2}\quad\textnormal{\textit{on} }\{l(o,\cdot)>R\}.
\end{align*}
Then the $l$-distance of any two points in $\smash{I^+(o)}$ does not exceed $\smash{R\,\rme^{\pi/\beta}}$.
\end{theorem}

An analogous past version is also available. The importance of \autoref{Th:FGALL} is that $k$ is not hypothesized to be everywhere nonnegative. It already predicts timelike geodesic incompleteness provided the timelike Ricci curvature does not decay too quickly in the far future of $o$. This result has predecessors in smooth \cite{aubry2007,galloway1982,schneider1972} and nonsmooth \cite{ketterer2017} Riemannian geometry. Since Schneider's work \cite{schneider1972} was the first to state the explicit asymptotic condition of \autoref{Th:FGALL}, in the sequel we call its conclusion the \emph{timelike Schneider inequality}.

Lastly, we show the following Hawking-type singularity theorem under an \emph{integral bound} on the negative part of the timelike Ricci curvature. Its smooth predecessor has recently been shown by Graf--Kontou--Ohanyan--Schinnerl \cite{kontou} under an additional uniform lower boundedness assumption on the timelike Ricci curvature. For simplicity, we state the synthetic counterpart to their result here; an extension is provided in  \autoref{Th:Hawking}. Two quantities require a quick introduction: $\hh_0$ is the ``surface measure'' of $V$, and $\smash{V^{0,T}}$ is the radial ``future development'' of $V$ up to temporal height $T$; see Subsections \ref{Sub:1} and \ref{Sub:2}.

\begin{theorem}[Synthetic Hawking-type singularity theorem, cf.~\autoref{Th:Hawking}]\label{Th:HAW} Let $\scrM$ be a timelike $p$-essentially nonbranching proper $\smash{\TCD_p^e(k,N)}$ space with $k \geq K$ for some constant $K<0$. Let $V\subset\mms$ be an achronal FTC Borel set with $\meas[V]=0$ and $\hh_0[V] < \infty$ whose forward mean curvature is bounded from above by $-\beta$ in the sense of \autoref{Def:Mean curv}, where $\beta > 0$. Assume  there exist $T>0$ and $\delta > 0$ with
\begin{align*}
\hh_0[V]^{-1}\int_{V^{0,T}} k_-\d\meas < \Big[\frac{\sinh(\sqrt{-K/(N-1)}\,\delta)}{\sinh(\sqrt{-K/(N-1)}\,(T+\delta))}\Big]^{N-1}\,\Big[\beta - \frac{N-1}{T}\Big].
\end{align*}
Then there is a radial $l$-geodesic starting in $V$ whose maximal domain of definition in $\smash{I^+(V)\cup V}$ is strictly contained in $[0,T+\delta)$.
\end{theorem}

The smooth result of \cite{kontou} has been inspired by Cheeger--Colding \cite{cc} and Sprouse \cite{sprouse}. A different synthetic version of Hawking's singularity theorem under uniform timelike Ricci curvature bounds has been shown in \cite{cavalletti2020}. A counterpart in positive signature can be found in \cite{burt}. An interesting ingredient of independent interest in our proof is a new nonsmooth form of   Cheeger--Colding's \emph{segment inequality} (\autoref{Pr:Segment}). It was used in \cite{cc} in the context of the almost splitting theorem. Roughly speaking, it estimates the line integral of a function along a single needle by its integral over a normal tube.

We point out that \autoref{Th:FGALL} and \autoref{Th:HAW} are new even in space\-times of low regularity, notably for Lorentzian metrics of class $\smash{\Cont^{1,1}}$ (or $\smash{\Cont^1}$ plus timelike $p$-essential nonbranching) thanks to \autoref{Th:Low reg blub} and \cite{braunc1}.   Under the strong energy condition however, a more traditional analog of the Hawking theorem was obtained 
for $\smash{\Cont^1}$-metrics by Graf; see \cite{Graf20} and its references for related results.

\subsection*{Motivation of \autoref{Def:TCDe}} The setup of $\smash{\TCD_p^e(k,N)}$ for variable $k$ is a bit technical. Loosely speaking, it is based  on (Jacques Charles Fran\c{c}ois)  Sturm's comparison theorem for  certain ODEs   involving $k$. Below, we  illustrate its  definition by outlining the smooth proof of ``$\Ric\geq k$ in all timelike directions implies $\smash{\TCD_p^e(k,N)}$''; a similar line of thought motivates the definition of $\smash{\TCD_p(k,\infty)}$, cf.~e.g.~the introduction of \cite{sturm2006b}. We temporarily 
neglect regularity issues and refer to \autoref{Th:Strong energ} for details. 

Let $\mu_0,\mu_1\in\Prob(\mms)$ be  appropriate mass distributions with $\meas$-densities $\rho_0$ and $\rho_1$, respectively. The results of \cite{mccann2020} imply the unique $\smash{\ell_p}$-geodesic $(\mu_t)_{t\in [0,1]}$ from $\mu_0$ to $\mu_1$ is given by $\mu_t = (T_t)_\push\mu_0$. Here $(T_t(x))_{t\in[0,1]}$ is an $l$-geodesic for $\mu_0$-a.e.~$x\in\mms$. Furthermore, the $\meas$-density $\rho_t$ of $\mu_t$ solves the Monge--Ampère-type  equation
\begin{align*}
\rho_0 = (\rho_t\circ T_t)\,\vert\!\det\rmd T_t\vert\quad\mu_0\textnormal{-a.e.}
\end{align*}
The negative logarithmic Jacobian $-\log \jmath$, where $\jmath := \vert\!\det\rmd T_\cdot\vert$, solves a Ric\-cati-type equation. The hypothesis ``$\Ric\geq k$ in all timelike directions'' and several manipulations thus imply the following ODI for every $0<t<1$:
\begin{align*}
(\jmath_t^{1/N})'' \leq -\frac{(k\circ T_t)\,\vert T_t'\vert^2}{N}\,\jmath_t^{1/N}\quad\mu_0\textnormal{-a.e.}
\end{align*}
Sturm's comparison theorem relates $\smash{\jmath^{1/N}}$ $\mu_0$-a.e.~to the solution $v$ of the ODE
\begin{align*}
v''(t) + \frac{(k\circ T_t)\,\theta^2}{N}\,v(t) =0
\end{align*}
satisfying $\smash{\smash{v(0) = \jmath^{1/N}_0}}$ and $\smash{\smash{v(1) = \jmath^{1/N}_1}}$, where $\smash{\theta := \vert T_t'\vert = l(\cdot,T_1)}$, viz.
\begin{align}\label{Eq:Blubz}
\jmath_t^{1/N} \geq v(t)\quad\mu_0\textnormal{-a.e.}
\end{align}
On the other hand, $v$ is explicitly given by
\begin{align*}
v(t) = \sigma_{\kappa^-/N}^{(1-t)}(\theta)\,\jmath^{1/N}_0 + \sigma_{\kappa^+/N}^{(t)}(\theta)\,\jmath_1^{1/N}.
\end{align*}
Here, for $\mu_0$-a.e.~$x\in\mms$ the functions $\smash{\kappa^\pm\colon[0,\theta]\to \R}$ designate the forward and backward potentials along the given $l$-geodesic $(T_t(x))_{t\in[0,1]}$ defined by the relations $\smash{\kappa^+(t\,\theta) := k\circ T_t(x)}$ and $\smash{\kappa^-(t\,\theta) := k\circ T_{1-t}(x)}$. The inherent \emph{distortion coeffi\-cients} $\smash{\sigma_{\kappa^+/N}^{(t)}(\theta)}$, discussed in detail in  \autoref{Def:Dist coeff} and \autoref{Def:dist coeff general k},  arise from the unique solution $\smash{u(t) := \sigma_{\kappa^+/N}^{(t)}(\theta)}$ of the ODE
\begin{align*}
0 = u''(t) + \frac{\kappa^+(t\,\theta)\,\theta^2}{N}\,u(t) = u''(t) + \frac{(k\circ T_t)\,\theta^2}{N}\,u(t)
\end{align*} 
with boundary data $u(0)=0$ and $u(1)=1$; analogously for $\kappa^-$. An integrated, thus more robust version of \eqref{Eq:Blubz} is precisely the estimate we use to phrase $\smash{\TCD_p^e(k,N)}$ in \autoref{Def:TCDe}. Compare also  with the proof of \autoref{Th:Pathwise}.

In the previous discussion, we have implicitly regarded $k$ as continuous. The lower semicontinuous case will be handled by approximation, recalling that every such function is the locally monotone limit of continuous functions, cf.~\eqref{Eq:kappan}.

\subsection*{Organization} In \autoref{Ch:Metric}, we develop our notion of globally hyperbolic, regular length metric measure spacetimes. The proofs of some results therein 
 which require not only $l$-geodesics but their merely causal analogs are outsourced to \autoref{Ch:Loose}. \autoref{Ch:TCD} introduces our variable timelike curvature-dimension condition, which is  preceded by a thorough discussion of the relevant distortion coefficients and succeeded by first elementary properties. In \autoref{Ch:Examples}, we provide examples of variable TCD spaces  obeying  the timelike convergence condition, the weak energy condition, and the null convergence condition, respectively.  In \autoref{Sec:Geometric inequ}, we show geometric inequalities. \autoref{Ch:Localization} develops the localization paradigm with respect to a general $1$-steep potential. In \autoref{Ch:Applications}, these insights are applied to improve all geometric inequalities from  \autoref{Sec:Geometric inequ} to their sharp form, provided a timelike essential nonbranching condition holds; the  Kantorovich duality formula for $\smash{\ell_1}$ required for the inherent localization by optimal transport is shown in \autoref{Ch:Rubinstein}.  \autoref{Ch:Applications} also proves \autoref{Th:FGALL} and \autoref{Th:HAW}. Finally, \autoref{Ch:TMCP} outlines the variable timelike measure-contraction property and its consequences.

\addtocontents{toc}{\protect\setcounter{tocdepth}{2}}

\section{Metric spacetimes}\label{Ch:Metric} We now describe our setting of \emph{metric spacetimes}. It is a version of the approach to Lorentzian (pre-)length spaces by Kunzinger--S\"amann \cite{kunzinger2018}, but differs from it in two major aspects. 
\begin{itemize}
\item It simplifies the axiomatization by taking into account the recent contributions \cite{ake2020,mccann+,minguzzi2019,minguzzi2023,minguzzi2022}. (We refer   to \cite{mccann+} for a similar simplification, and to \cite{beran2023} for a slightly different forthcoming approach.) 
\item It is tailored towards  the requirements of  \emph{global hyperbolicity} and a \emph{length} condition, which we assume almost from the outset, deriving all other desired properties from these; the convenient assumption of \emph{regularity} will be added a bit later.
\end{itemize}

Although an attempt to impose as few assumptions as possible has been made, some machinery needs to be developed, and many results from \cite{ake2020,burtscher,mintop,minguzzi2019,minguzzi2023,minguzzi2022} have to be reproven. If interested in a shortcut, we refer the reader to \autoref{Th:Equiv LLS} for an equivalent definition of globally hyperbolic, regular length metric spacetimes as proposed by \autoref{Def:LLS} and \autoref{Th:GHy}.

\subsection{Basic notions} Let $\mms$ be a space endowed with a  first-countable topology $\Top$. Our discussion will be accompanied by several more sophisticated properties of $\Top$ implied by our standing \autoref{Ass:GHLLS} (length metric spacetime with compact diamonds). Most notably, in \autoref{Cor:Hausdorff} it is shown to be locally compact and Hausdorff; in particular, $\Top$ is Polish if and only if it is second-countable. This will be assumed from \autoref{Sub:LOT}, but our framework can be partly  generalized if we only work with  probability measures living on  compact subsets of $\mms$; cf. \autoref{Th:Compact Polish}.

\subsubsection{Signed time separation function} We adopt the conventions
\begin{align}\label{Eq:Convention infty}
\infty + (-\infty) := -\infty + \infty := -\infty + (-\infty) := -\infty.
\end{align}

\begin{definition}[Signed time separation]\label{Def:signed time sep} A function $\smash{l\colon \mms^2 \to [0,\infty]\cup\{-\infty\}}$ is a  \emph{signed time separation function} if it has the following properties.
\begin{enumerate}[label=\textnormal{\alph*.}]
\item \textnormal{\textbf{Causality \textnormal{\cite[Def.~2.35]{kunzinger2018}}.}} It is nonnegative on the diagonal of $\mms^2$, and whenever two points $x,y\in \mms$ satisfy
\begin{align*}
\min\{l(x,y), l(y,x)\} > -\infty,
\end{align*}
we already have $x=y$.
\item \textnormal{\textbf{Continuity.}} It is upper  semicontinuous, and its positive part $\smash{l_+}$ is lower  semicontinuous.
\item \textnormal{\textbf{Reverse triangle inequality.}} For every $x,y,z\in\mms$,
\begin{align}\label{Eq:Reverse tau}
l(x,z) \geq l(x,y) + l(y,z).
\end{align}
\end{enumerate}
\end{definition}

The second property of causality is naturally called \emph{antisymmetry}. By \eqref{Eq:Reverse tau} and causality, we have $l(x,x) \in\{0,\infty\}$ for every $x\in\mms$.

The continuity of $\smash{l_+}$ has to be understood with respect to the extended  topology of the half-line $[0,\infty]$; we do not exclude the case $l(x,x) = \infty$ for some $x\in \mms$ a priori. Until \autoref{Cor:Finiteness l} --- which is  only used around \autoref{Le:BLMS} --- upper semicontinuity of $\smash{l_+}$ will not play any role.

The following remark can be skipped at first reading.

\begin{remark}[Possible generalizations] In two aspects, it is interesting to study whether our assumptions made thus far  can be weakened. To make the exposition less cumbersome, though, we leave these possible generalizations to future work.
\begin{itemize}
\item \textbf{Sequential topology.} One could merely assume $\Top$ to be \emph{sequential}, i.e.~it is completely determined by convergence of sequences; see \cite{engelking} for details. Every (quotient of a) first-countable space is sequential, and unlike the latter property, first-countability is not stable under quotients in general. That way, one could always guarantee antisymmetry by taking a quotient \cite[Lem.~1]{mccann+}. On the other hand, sequential topologies come with various subtleties to be addressed with care. For instance, the product of two  sequential topologies is not necessarily sequential; neither does compactness imply sequential compactness in  this generality, nor vice versa.
\item \textbf{Upper semicontinuity near the diagonal.} 
It seems merely assuming upper semicontinuity of $l$ near the diagonal of $\smash{\mms^2}$ suffices. For smooth spacetimes, this  follows from strong causality \cite[Cor.~3.5]{beem1979}, in turn implied by compact diamonds in all physically relevant cases \cite[Thm.~2.8]{hounnonkpe2019}. In our approach, it appears possible to prove the affected results \autoref{Le:BLMS}, \autoref{Th:Limit curve theorem}, and \autoref{Le:Upper semiconti} first in small chronological diamonds, then get upper semicontinuity of $l$ as in \cite[Thm.~3.28]{kunzinger2018}, and finally redo the proofs of the previous results globally. \autoref{Cor:Finiteness l} and \autoref{Th:Compact Polish}  
only need continuity of $\smash{l_+}$ near the diagonal. 
\end{itemize}
\end{remark}

Let such a signed time separation function $l$ be fixed in the sequel. For brevity, every property of the tuple $(\mms,l)$ will be formulated for $\mms$ only.

\subsubsection{Causality and chronology} We define a relation $\leq$ on $\mms$ by
\begin{align*}
x\leq y \quad :\Longleftrightarrow\quad l(x,y)\geq 0
\end{align*}
and call it \emph{causality}. \autoref{Def:signed time sep} ensures the reflexivity of $\leq$ and its antisymmetry, i.e.~if $x,y\in\mms$ satisfy $x\leq y$ and $y\leq x$, then necessarily $x=y$. The reverse triangle inequality \eqref{Eq:Reverse tau} implies transitivity of $\leq$. These observations are summarized by saying the triple $(\mms,\Top,\leq)$ is a topological ordered space \cite{minguzzi2023,nachbin}.  If $x\leq y$ we term $x$ to lie in the \emph{causal past} of $y$, and $y$ to lie in the \emph{causal future} of $x$; a future  orientation is implicit in this and subsequent terminology. We define
\begin{itemize}
\item the \emph{causal future} of $x$ by $J^+(x) := \{l(x,\cdot)\geq 0\}$,
\item the \emph{causal past} of $y$ by $J^-(y) := \{l(\cdot,y)\geq 0\}$, and
\item the \emph{causal diamond} between $x$ and $y$ by $J(x,y) := J^+(x)\cap J^-(y)$. 
\end{itemize}
Given any set $X\subset \mms$, we define
\begin{align*}
J^+(X) := \bigcup_{x\in X} J^+(x),
\end{align*}
and analogously $J^-(Y)$ for a given set $Y\subset \mms$. We also define
\begin{align*}
J(X,Y) := J^+(X) \cap J^-(Y)
\end{align*}
and call this  a \emph{causal emerald} provided $X$ and $Y$ are compact.

We will write $x<y$ if $x\leq y$ yet $x\neq y$.

Furthermore, we define a relation $\ll$ on $\mms$ by
\begin{align*}
x\ll y\quad :\Longleftrightarrow\quad l(x,y)>0.
\end{align*}
The induced  relation $\ll$ is called \emph{chrono\-lo\-gy}. By \eqref{Eq:Reverse tau}, it is transitive. All notions from the preceding  paragraph are defined analogously for $\ll$ instead of $\leq$, where every terminological occurrence of ``causal'' is replaced by ``chronological''. Define
\begin{itemize}
\item the \emph{chronological future} of $x$ by $I^+(x) := \{l(x,\cdot)> 0\}$,
\item the \emph{chronological past} of $y$ by $I^-(y) := \{l(\cdot,y)> 0\}$, and
\item the \emph{chronological diamond} between $x$ and $y$ by $I(x,y) := I^+(x)\cap I^-(y)$. 
\end{itemize}
Given any sets $X$ and $Y$ as above, the open sets $I^+(X)$, $I^-(Y)$, and $I(X,Y)$ are defined analogously.

These properties turn the triple $(\mms,\ll,\leq)$ into a \emph{causal space} \cite[Def.~2.1]{kunzinger2018}, as originally studied in a slightly different form by Kronheimer--Penrose \cite{kronheimer}.

Let us collect the following elementary properties, cf.~\cite[Lems.~2.10, 2.12, Prop. 2.13]{kunzinger2018}. The push-up property follows from \eqref{Eq:Reverse tau}, while openness is implied by lower semicontinuity of $\smash{l_+}$.

\begin{lemma}[Push-up and openness of chronology]\label{Le:Pushup openness} The two relations $\leq$ and $\ll$ have the following properties for every $x,y,z\in\mms$ and every $X\subset\mms$.
\begin{enumerate}[label=\textnormal{\textcolor{black}{(}\roman*\textcolor{black}{)}}]
\item \textnormal{\textbf{Push-up.}} If $x\leq y\ll z$ or $x\ll y\leq z$, then $x\ll z$.
\item \textnormal{\textbf{Openness.}} The set $\smash{I^\pm(X)}$ is open. Moreover, $\ll$ is an open relation in the sense that $\smash{\{l>0\}}$ is an open subset of $\smash{\mms^2}$.
\end{enumerate}
\end{lemma}

Here, the set $X$ does not need to be open, or to have any other regularity.

\subsection{Length metric spacetimes} We now set up a condensed and more flexible notion of length metric spacetimes akin to Lorentzian length spaces \cite[Sec.~3.4]{kunzinger2018}. When coupled with compact diamonds,  as we explore below it implies almost all desired causal properties of $\mms$ except regularity, cf.~\autoref{Sec:Regularity}.

In the following two subsections, let $a<b$ and $c<d$ be real numbers.

\subsubsection{Causal curves and loops} 

\begin{definition}[Causal character]
\label{Def:Curves} A curve\footnote{We use the terms \emph{curve} or \emph{path} synonymously for a map whose domain is a convex subset of $\R$. In particular, we do not require its continuity a priori. If relevant, (possible dis)continuity is always clear from our terminology or stressed explicitly.} $\gamma\colon [a,b]\to \mms$ is called
\begin{enumerate}[label=\textnormal{\alph*.}]
\item \emph{causal} if it is continuous, and $\gamma_s\leq \gamma_t$ for every $a\leq s< t \leq b$,
\item \emph{strictly causal} if it is causal yet nowhere constant, i.e.~$\gamma_s < \gamma_t$ for every $a\leq s< t\leq b$, 
\item \emph{timelike} if it is continuous, and $\gamma_s\ll\gamma_t$ for every $a\leq s<t\leq b$, and
\item \emph{null} if it is causal yet $\gamma_s \not\ll\gamma_t$ for every $a\leq s<t\leq b$.
\end{enumerate}
\end{definition}

Any of these causality properties may be referred to as the \emph{causal character} of $\gamma$. If a curve $\gamma\colon [a,b]\to \mms$ has one of the mentioned causal characters yet is not required  to be continuous, we term it \emph{rough causal}, \emph{rough strictly causal}, \emph{rough time\-like}, or \emph{rough null}, respectively. Moreover, any 
nonconstant  path $\gamma\colon[a,b]\to\mms$ with $\gamma_a=\gamma_b$ is called  a \emph{loop}.

Since we have not excluded the possibility $l(x,x) > 0$ for some $x\in\mms$ yet, it is not clear at this point that every timelike curve is strictly causal.

Lastly, by reverse triangle inequality \eqref{Eq:Reverse tau} a curve $\gamma\colon[a,b]\to\mms$ is rough null if and only if it is rough causal and $\gamma_a\not\ll \gamma_b$.

\subsubsection{Length} 

\begin{definition}[Length functional] The \emph{length} of a curve $\gamma\colon [a,b]\to\mms$ is 
\begin{align*}
\Len_l(\gamma) := \inf \sum_{i=1}^m l(\gamma_{t_{i-1}}, \gamma_{t_i}),
\end{align*}
where the infimum is taken over all $m\in\N$ and all partitions $a = t_0 < t_1 < \dots < t_{m-1} < t_m = b$ of $[a,b]$ of cardinality $m+1$.
\end{definition}

These notions can be extended to intervals of the form $(a,b)$, $[a,b)$, and $(a,b]$, including unbounded ones, using limits. 

Note that $\smash{\Len_l(\gamma) \geq 0}$ if and only if $\gamma\colon[a,b]\to\mms$ is rough causal by \eqref{Eq:Convention infty}.

The following properties are proven in exactly the same way as \cite[Lems.~2.25, 2.27, 2.28]{kunzinger2018}. As usual, by a \emph{reparametrization} of a curve $\gamma\colon [a,b] \to \mms$ we mean a map $\sigma\colon [c,d]\to \mms$ such that $\sigma = \gamma\circ\phi$, where $\phi\colon [c,d] \to [a,b]$ is \emph{continuous} and strictly increasing. Reparametrization does not change the causal character of a rough curve \cite[Lem.~2.27]{kunzinger2018}.

\begin{lemma}[Concatenation and reparameterization invariance]
\label{Le:Additivity} The following properties hold.
\begin{enumerate}[label=\textnormal{(\roman*)}]
\item \textnormal{\textbf{Additivity.}} If $\gamma\colon[a,b]\to \R$ is a rough causal curve and $a<c<b$, then
\begin{align*}
\Len_l(\gamma) = \Len_l(\gamma\big\vert_{[a,c]}) + \Len_l(\gamma\big\vert_{[c,b]}).
\end{align*}
\item \textnormal{\textbf{Reparametrization invariance.}} The functional $\Len_l$ is invariant under reparametrizations of rough causal curves. 
\end{enumerate}
\end{lemma}

We will thus mostly restrict our attention to rough curves defined on $[0,1]$.

\subsubsection{Imposing a length property} Every rough causal curve $\gamma\colon[0,1]\to \mms$ satisfies 
\begin{align*}
\Len_l(\gamma) \leq l(\gamma_0,\gamma_1).
\end{align*}
In general, this inequality is strict, but one can at least ask these quantities to be arbitrarily close to each other once the endpoints are fixed. This leads to the Lo\-rentz\-ian analog of length metric spaces \cite{burago2003}, see also \cite{kunzinger2018, minguzzi2022}.

\begin{definition}[Length metric spacetime]\label{Def:LLS} 
A set $\mms$ equipped with a first countable topology $\Top$  and a signed time separation function $l$ as above is termed  \emph{length metric spacetime} if the following statements hold.
\begin{enumerate}[label=\textnormal{\alph*.}]
\item \textnormal{\textbf{Nontrivial chronology.}} Every point $x\in \mms$ is the midpoint of a timelike curve $\gamma\colon[0,1]\to\mms$, i.e.~$x = \gamma_{1/2}$.
\item \textnormal{\textbf{Length property.}} If $x,y\in\mms$ satisfy $x\ll y$, then
\begin{align*}
l(x,y) = \sup \Len_l(\gamma),
\end{align*}
where the supremum is taken over all strictly causal curves $\gamma\colon[0,1]\to \mms$ such that $\gamma_0=x$ and $\gamma_1=y$. 
\end{enumerate}
\end{definition}

The subtle properties of continuity and nowhere constancy of  admissible curves in the length property will be  crucial  for us. Inter alia, chronology will be implied, i.e.~the vanishing of $l$ on the diagonal of $\mms^2$; cf.~\autoref{Cor:Causality}. 

The nontrivial chronology condition implies $\smash{I^\pm(x) \neq \emptyset}$ for every $x\in\mms$. This is one requirement for localizability in \cite[Def.~3.16]{kunzinger2018}. Unlike \cite[Def.~3.1]{kunzinger2018}, we  merely assume  chronologically related points to be connected by strictly causal (as opposed to timelike) curves, and we do not assume any a priori connectedness between merely causally related points. These will be made up for by our assumption of regularity in \autoref{Sec:Regularity}, as well as \autoref{Th:Ex geo}, nontrivial chronology, and continuity of $\smash{l_+}$, respectively.

\begin{remark}[First countability of the chronological topology]\label{Def:Inter} 
In general the topology $\Top$ is not uniquely determined by the signed time separation $l$.
A choice of $\Top$ which is determined by $l$ is the \emph{chronological topology}  
\cite[Def.~2.4
]{kunzinger2018}, also called the {\em order topology} associated with $\ll$, 
which is the coarsest topology $\Inter$ on $\mms$ containing the sets $\smash{\{I^\pm(x) : x\in \mms\}}$.
The nontrivial chronology of \autoref{Def:LLS} ensures first countability of $\Inter$,
since for each $x \in X$ (say, the midpoint $\gamma_{1/2}$ of a timelike curve), a local base for the topology at $x$ is then given by
$\smash{\{I(\gamma_{1/2 - 1/2k},
\gamma_{1/2+1/2k}) :k \in\N\}}$.
\end{remark}

We will exclusively deal with length metric spacetimes satisfying the following.

\begin{assumption}[First standing assumption]\label{Ass:GHLLS} From now on, we assume  $\mms$ to be a length metric spacetime with compact diamonds.
\end{assumption}

Here and in the sequel, we say $\mms$ has
\begin{itemize}
\item \emph{compact diamonds} if $J(x,y)$ is compact for every $x,y\in\mms$, and
\item \emph{compact emeralds} if $J(X,Y)$ is compact for every compact $X,Y\subset\mms$.
\end{itemize}

As in the smooth setting \cite[Lem.~14.10]{oneill1983}, \autoref{Ass:GHLLS} is not well-behaved under restriction to subsets, especially causal diamonds. However, this will not be an issue for us since we can always locally use all relevant properties that hold globally, cf.~\autoref{Re:Restriction to subsets}.

\subsection{Causality theory for length metric spacetimes} Now we study numerous consequences of \autoref{Ass:GHLLS}. We address
\begin{itemize}
\item causality and chronology properties along with
\item global hyperbolicity and
\item refined properties of the topology $\Top$,
\item existence of generalized time functions,
\item the causal ladder, especially with regards to strong causality, and
\item a limit curve theorem.
\end{itemize}
The existence of maximizers in the length property is postponed to  \autoref{Sec:Geodesy}.

\subsubsection{Causality, chronology, and global hyperbolicity} We will use the following notations for  ``monotone'' convergence. For a sequence $(x_n)_{n\in\N}$ in $\mms$ and $x\in\mms$, we write $x_n \downarrow x$ as $n\to\infty$ or $x_n\uparrow x$ as $n\to\infty$, respectively, if  $(x_n)_{n\in\N}$ converges to $x$ and, respectively, $x_{n+1} \ll x_n$ or $x_n \ll x_{n+1}$ for every $n\in\N$.

A cornerstone for our discussion is the following result from \cite[Lem.~2.18]{ake2020},  called the ``sequence lemma'' therein.

\begin{lemma}[Timelike monotone approximation]
\label{Le:Sequ lemma} For every $x\in \mms$, there exist sequences $(a_i)_{i\in\N}$ and $\smash{(b_i)_{i\in\N}}$ such that $a_i\uparrow x$ as $i\to\infty$ and $b_i\downarrow x$ as $i\to\infty$.
\end{lemma}

\begin{proof} Owing to the nontrivial chronology from \autoref{Def:LLS}, let $\gamma\colon[0,1]\to\mms$ be a timelike curve with midpoint $x$. Then setting $\smash{a_i := \gamma_{1/2-1/2i}}$ and $\smash{b_i := \gamma_{1/2+1/2i}}$ for $i\in\N$ leads to the desired sequences by continuity.
\end{proof}

Note again that the sequences constructed above might be constant if $l(x,x) > 0$ for the given point $x$.

The following two consequences of \autoref{Le:Sequ lemma} are variants of \cite[Thm.~4.12]{minguzzi2019} and \cite[Prop.~2.3]{hounnonkpe2019}, whose proofs we adopt with no essential change. The following causal closedness is clear by upper semicontinuity of $l$, but we prefer to give an independent proof which does not rely on this property.

\begin{proposition}[Causal closedness {\cite[Def.~3.4]{kunzinger2018}}]\label{Pr:Closed} The relation $\leq$ is closed in the product topology $\smash{\Top^2}$ of $\smash{\mms^2}$.
\end{proposition}

\begin{proof} We first claim closedness of $\smash{J^\pm(x)}$ for every $x\in\mms$. We focus on $\smash{J^+(x)}$, the proof for closedness of $\smash{J^-(x)}$ is analogous. Given any  $y$
in the closure $\smash{\bar{J}^+(x)}$ of $\smash{{J}^+(x)}$, there exists a directed net $\smash{(y_\alpha)_{\alpha\in A}}$ in $\smash{J^+(x)}$ converging to $y$. Let $y'\in I^+(y)$ be as provided by nontrivial chronology. By compactness of diamonds, $\smash{J(x,y')}$ closed. By the openness asserted by \autoref{Le:Pushup openness}, we have $\smash{y_\alpha\in J^+(x)\cap I^-(y')}$ for sufficiently large $\alpha\in A$. Since the latter is a subset of $J(x,y')$, we obtain $\smash{y\in J^+(x)}$.

To prove closedness of $\leq$, let $(x_\alpha,y_\alpha)_{\alpha\in A}$ be a directed net in $\smash{\mms^2}$ converging to $\smash{(x,y)\in\mms^2}$ such that $x_\alpha\leq y_\alpha$ for every $\alpha\in A$. Let $\smash{(b_i)_{i\in\N}}$ be a sequence with $b_i\downarrow y$ as $i\to\infty$ as given by \autoref{Le:Sequ lemma}. Given any $i\in\N$, we obtain $\smash{y_\alpha\in I^-(b_i)}$ for sufficiently large $\alpha\in A$ by \autoref{Le:Pushup openness}. The same result implies $x_\alpha \in I^-(b_i)$ for sufficiently large $\alpha\in A$, and therefore $\smash{x\in J^-(b_i)}$ by the preceding paragraph; in other words, $\smash{b_i\in J^+(x)}$. Letting $i\to\infty$ and using closedness of $\smash{J^+(x)}$ yields $\smash{y\in J^+(x)}$, or equivalently $x\leq y$.
\end{proof}


\begin{corollary}[Compactness of emeralds]\label{Cor:K-GH} 
The space $\mms$ has compact emeralds. 
\end{corollary}

\begin{proof} By \autoref{Pr:Closed} and \cite[Thm.~3.3]{minguzzi2023} it suffices to prove compactness of $J(C,C)$ if $C\subset\mms$ is compact. Closedness of $J(C,C)$ follows from \autoref{Pr:Closed} and the theory of closed ordered spaces \cite{nachbin}, cf.~\cite[Thm.~4.12]{minguzzi2019}. 

The rest follows as for  \cite[Prop.~2.3]{hounnonkpe2019}. Given any $x\in C$, by nontrivial chronology there exist $a,b\in \mms$ with $c\in I(a,b)$. By \autoref{Le:Pushup openness}, these diamonds are an open cover of $C$, hence $\smash{C\subset\bigcup_{i=1}^n I(a_n,b_n)}$ for some $n\in\N$ and some $a_1,b_1,\dots,a_n,b_n\in\mms$. Taking causal pasts and causal futures of both sides of the preceding inclusion yields $J(C,C) \subset \smash{\bigcup_{i=1}^n J(a_n,b_n)}$. Therefore, as a closed subset of a compact subset, the causal diamond $J(C,C)$ is itself compact.
\end{proof}

As a remarkable consequence of \autoref{Pr:Closed} and \autoref{Cor:K-GH}, all length metric spacetimes $\mms$ having compact diamonds are \emph{globally hyperbolic} (after the modification of the classical definition of this property suggested in \cite{mintop}). 

\begin{theorem}[Global hyperbolicity]\label{Th:GHy} \autoref{Ass:GHLLS} implies $\mms$ is globally hyperbolic, i.e.~it satisfies the following two properties. 
\begin{enumerate}[label=\textnormal{(\roman*)}]
\item The relation $\leq$ is a closed order.
\item The causally convex hull $J(C,C)$ of every compact set $C\subset\mms$ is compact.
\end{enumerate}
\end{theorem}

\begin{remark}[About various notions of global hyperbolicity]\label{Re:Nonambig} There are several slightly different definitions of global hyperbolicity in the literature. We shortly elaborate on the connections to the definition we use.

Traditionally \cite{beem1996,oneill1983}, for smooth spacetimes global hyperbolicity means nonexistence of almost closed strictly causal curves and compact diamonds. This property will follow from \autoref{Le:Str caus equiv} and \autoref{Cor:Strong causality} below. In particular, $\mms$ satisfies the Bernal--Sánchez definition of global hyperbolicity from the smooth case \cite{bernal}.  More generally, \autoref{Th:GHy} echoes the results of Hounnonkpe--Minguzzi \cite{hounnonkpe2019, minguzzi2023} which inspire it: for smooth noncompact spacetimes of dimension three or higher, endowed with a $\smash{\Cont^{1,1}}$-Lorentzian metric, \cite[Thm.~2.8]{hounnonkpe2019} asserts the causality requirements in the indicated traditional definition of global hyperbolicity are superfluous. Our result does not extend \cite{hounnonkpe2019}, though, since we \emph{assume} causality.

For various characterizations of global hyperbolicity for spacetimes with continuous metrics, we refer to \cite{samann2016}.

In the more recent approach \cite{kunzinger2018}, global hyperbolicity for Lorentzian pre-length spaces is defined by non-total imprisonment \cite[Def.~2.35]{kunzinger2018} and  compact diamonds. By \cite[Thm.~3.7]{minguzzi2023}, for Lorentzian length spaces \cite[Def.~3.22]{kunzinger2018} this  is equivalent to the definition of global hyperbolicity from \autoref{Th:GHy}. (\cite{burtscher} shows further equivalent characterizations of global hyperbolicity in this abstract setting.) On the other hand, non-total imprisonment is a metric property, hence becomes less comparable to our setting in which $\Top$ is not required to be metrizable a priori. Yet, in fact for Lorentzian length spaces global hyperbolicity depends on the reference metric only through the topology it induces \cite[Rem.~14]{mccann+}.

The assumption of compact emeralds for Lorentzian length spaces was termed ``$\scrK$-global hyperbolicity'' in \cite{braun2022,braun2023,cavalletti2020}.
\end{remark}

We go on with further properties of $\leq$ and $\ll$. Recall that a loop refers to any nonconstant path $\gamma\colon[0,1]\to\mms$ with $\gamma_0=\gamma_1$.

\begin{corollary}[Enhanced properties]\label{Cor:Causality} The following properties  hold.
\begin{enumerate}[label=\textnormal{\textcolor{black}{(}\roman*\textcolor{black}{)}}]
\item\label{La:Dwa} \textnormal{\textbf{No rough causal loops.}} If a rough causal curve $\gamma\colon [0,1]\to \mms$ satisfies $\gamma_s = \gamma_t$ for some $0\leq s<t\leq 1$, then $\smash{\gamma\big \vert_{[s,t]}}$ is constant.
\item\label{La:Tri} \textnormal{\textbf{Chronology} \cite[Def.~2.35]{kunzinger2018}\textbf{.}} The relation $\ll$ is irreflexive. In other words, every $x\in\mms$ obeys $x\not\ll x$, and in particular $l(x,x) = 0$.
\end{enumerate}
\end{corollary}

\begin{proof} 
The existence of rough causal loops would contradict  causality assumed in \autoref{Def:signed time sep}. Indeed, if $\gamma\colon[0,1]\to\mms$ is a rough causal loop (hence nonconstant), then we get $\gamma_0 \leq \gamma_t \leq \gamma_1$ yet $\gamma_t \neq \gamma_0 = \gamma_1$ for some $0<t<1$. This gives 
\ref{La:Dwa}.

To show \ref{La:Tri}, assume $x\ll x$ for some $x\in \mms$.  As  $l(x,x)>0$, the length property would imply the existence of a causal loop $\gamma\colon[0,1]\to\mms$ starting and ending at $x$, a contradiction to \ref{La:Dwa}. This shows $l(x,x)\leq 0$, while the inequality $l(x,x) \geq 0$ follows from our assumption on $l$ from \autoref{Def:signed time sep}.
\end{proof}

\begin{remark}[Restriction caveat]\label{Re:Restriction to subsets} \autoref{Cor:Causality} implies that causal diamonds $J(x,y)$, where $x,y\in\mms$, endowed with the restriction $\smash{l\big\vert_{J(x,y)^2}}$ will never obey \autoref{Ass:GHLLS} if $\mms$ does. Indeed, otherwise $x$ would be the midpoint of a timelike curve through it, and since $l(x,x)=0$ by \autoref{Cor:Causality} we would find $z\in J(x,y) \setminus \{x\}$ with $z\ll x$. On the other hand $x\leq z$, a contradiction to said corollary.

However, this will not affect our ``local'' discussions, since all relevant properties we use (compactness, limit curve theorems,  basic optimal transport theory, etc.) are inherited by the considered subsets from the ambient space.
\end{remark}

Lastly, we show everywhere finiteness of $\smash{l_+}$. Note that we only use continuity of $\smash{l_+}$ near the diagonal of $\smash{\mms^2}$.

\begin{corollary}[Finiteness of time separation]\label{Cor:Finiteness l} The function $\smash{l_+}$ is real-valued.
\end{corollary}

\begin{proof} By \autoref{Cor:Causality}, it suffices to consider $x,y\in\mms$ with $l(x,y)>0$. Let $\gamma\colon[0,1]\to\mms$ be a strictly causal curve connecting $x$ and $y$, as provided through the length property. For every $0\leq t\leq 1$, by the continuity of $\smash{l_+}$ near the diagonal of $\smash{\mms^2}$ as well as \autoref{Cor:Causality} there exists $\smash{\delta > 0}$ with the property  $\smash{l_+(\gamma_s,\gamma_{s'})\leq 1}$ for every $\smash{s,s'\in  (t-\delta,t+\delta)\cap [0,1]}$. By  compactness, finitely many such  intervals $(t_i-\delta_i,t_i+\delta_i) \cap [0,1]$ cover $[0,1]$,
where $i = 0,\ldots, m$ and $m\in\N$. Without loss of generality, we assume $0 = t_0 < t_1 < \dots < t_{m-1} < t_m = 1$. Invoking \autoref{Cor:Causality}, a telescoping sum, and  the reverse triangle inequality \eqref{Eq:Reverse tau}, we obtain the desired conclusion:
\begin{align*}
l_+(x,y) &= \sum_{i=1}^m \big[l_+(\gamma_{t_0}, \gamma_{t_i}) - l_+(\gamma_{t_0},\gamma_{t_{i-1}})\big] \leq \sum_{i=1}^m l_+(\gamma_{t_{i-1}}, \gamma_{t_i}) \leq m.\qedhere
\end{align*}
\end{proof}

\subsubsection{Topological consequences}

\begin{corollary}[Topological properties and metrizability]\label{Cor:Hausdorff} The topology $\Top$ is Hausdorff,  locally compact yet noncompact, and does not have isolated points.

In particular, $\Top$ is Polish if and only if it is second-countable.
\end{corollary}

\begin{proof}  The global hyperbolicity of \autoref{Th:GHy} combines with \cite[Thm.~3.3]{minguzzi2023} to imply the Hausdorff property of $\Top$.

Local compactness follows from nontrivial chronology and  compact diamonds. Indeed, given any $x\in\mms$, taking appropriate points on a timelike segment passing through $x$ and considering the spanned chronological diamond gives an open neighborhood of $x$ with compact closure.

Compactness of $\Top$ is contradicted in the standard way, cf.~\cite[Lem.~14.10]{oneill1983}.

To show nonexistence of isolated points, assume to the contrary that $x\in\mms$ is such a point. By continuity, every timelike curve with midpoint $x$ as given by nontrivial chronology is  constant around $1/2$, which contradicts \autoref{Cor:Causality}.

The last statement follows straightforwardly from these considerations. Every Polish topology is second-countable. Conversely, every second-countable, locally compact Hausdorff topology is Polish.
\end{proof}

More refined topological properties can be deduced for all compact subsets of $\mms$. The proof relies on linking the theory of bounded Lorentzian metric spaces from \cite{minguzzi2022} to our setting. To streamline the presentation, we defer it to \autoref{Ch:Loose}, but  recall  that it only requires $\smash{l_+}$ to be continuous near the diagonal of $\smash{\mms^2}$.

\begin{theorem}[Polish compact sets]\label{Th:Compact Polish} \autoref{Ass:GHLLS} implies the relative topology of $\Top$ on every compact subset $C\subset\mms$ is Polish.
\end{theorem}

Next, we pass to the distinction property of $\mms$ shown in \autoref{Cor:Distinction}. To this aim, the closedness of causal pasts and causal futures implicitly proven for \autoref{Pr:Closed} leads to the identities in the following \autoref{Le:Closure}. They should be read as ``non-bubbling'' under \autoref{Ass:GHLLS} (in the terminology of Chru\'sciel--Grant \cite{chrusciel2012}). The proof is performed as for \cite[Prop. 2.19, Lem.~3.11]{ake2020}.

\begin{lemma}[No causal bubbles]\label{Le:Closure} For every $x\in\mms$, 
\begin{align*}
\bar{I}^\pm(x) &= \{y\in\mms : I^\pm(y) \subset I^\pm(x)\} = J^\pm(x),
\end{align*}
where $\smash{\bar{I}^{\pm}(x)}$ denotes the closure of $\smash{I^\pm(x)}$.
\end{lemma}

\begin{proof} Since $I^\pm(x) \subset J^\pm(x)$, taking closures of both sides yields the inclusion ``$\subset$''   by \autoref{Pr:Closed}. 

On the other hand, \autoref{Le:Sequ lemma} (timelike monotone approximation) as well as \autoref{Le:Pushup openness} (push-up lemma) show each $y \in J^+(x)$ to be the decreasing limit of a sequence $(y_i)_{i\in\N}$ in $I^+(x)$ to establish the opposite inclusion ``$\supset$''. 

The inclusion $\smash{\bar I^-(x) \supset J^-(x)}$ is argued in an analogous way (or by time reversal as in \autoref{Ex:Causal reversal}).
\end{proof}

\begin{corollary}[Future and past distinction {\cite[Def.~3.3]{ake2020}}]\label{Cor:Distinction}  If two points $x,y\in\mms$ satisfy $I^+(x) = I^+(y)$ or $I^-(x) = I^-(y)$, then $x=y$. In particular, the space $\mms$ is distinguishing, i.e.~if $I^+(x) = I^+(y)$ \emph{and} $I^-(x) = I^-(y)$, then $x=y$.
\end{corollary}

\begin{proof} In the first case, \autoref{Le:Closure} yields $J^+(x) = J^+(y)$ by taking closures on both sides. In turn, since $\leq$ is reflexive this implies $x\in J^+(y)$ and $y\in  J^+(x)$. The claim follows from \autoref{Cor:Causality}
 or the antisymmetry from \autoref{Def:signed time sep}.

The second case is argued analogously.
\end{proof}

\subsubsection{Generalized time functions}

Following \cite{burtscher}, this corollary will be used to establish the existence of \emph{generalized time functions} through averaged volumes.

\begin{theorem}[Existence of generalized time functions]\label{Th:Ex time function} Every compact subset $C\subset \mms$ admits a generalized time function, i.e.~there exists a function $\sft\colon C\to \R$ with the property $\sft(x) < \sft(y)$ whenever $x,y\in C$ satisfy $x<y$.
\end{theorem}

\begin{proof} By covering $C$ with finitely many chronological diamonds and taking the closure, without loss of generality we may and will assume $C$ lies in the interior of a larger compact set $C'\subset\mms$. By \autoref{Th:Compact Polish}, $C'$ is Polish --- with topology induced by, say, the metric $\met$ --- and hence there is  a Borel probability measure $\mu\in\scrP(\mms)$ with $\supp \mu = C'$ \cite[Prop.~4.1]{burtscher} (see \autoref{Sub:Notation prob} for relevant notation). For $x\in\mms$, define the \emph{past averaged volume function} \cite[Def. 4.2]{burtscher} $\sft^-\colon \mms\to [0,1]$ by
\begin{align*}
\sft^-(x) := \int_0^1 \mu\big[\met(\cdot,I^-(x)) < r\big]\d r.
\end{align*}
By \autoref{Le:Pushup openness}, if $x,y\in\mms$ satisfy $x\leq y$ then $\sft^-(x) \leq \sft^-(y)$. 

Following the proof of \cite[Prop.~4.8]{burtscher} and applying \autoref{Cor:Distinction}, in fact we get $x<y$ implies $\sft^-(x) < \sft^-(y)$ for every $x,y\in C$. Here, we used the construction of $C'$, which ensures all sequences converging to points in $C$ to eventually lie in $C'$, and hence the applicability of the ``past-approximation property'' from the proof of \cite[Prop.~4.8]{burtscher} yielded by nontrivial chronology.
\end{proof}

\begin{remark}[Alternative construction] In the previous proof, one could have also taken the \emph{future averaged volume function} $\sft^+\colon \mms \to [-1,0]$ defined by
\begin{align*}
\sft^+(x) := -\int_0^1 \mu\big[\met(\cdot,I^+(x)) < r\big]\d r.
\end{align*}
\end{remark}

\begin{remark}[Global existence of generalized time functions] If $\Top$ is induced by a separable metric,  \autoref{Th:Ex time function} even holds for $C$ replaced by the noncompact space $\mms$. The proof is similar, in fact simpler in this case. 
\end{remark}

\subsubsection{Causal ladder} Now we ``descend'' the  causal ladder from \cite[Thm.~3.16]{ake2020} and outline how   global hyperbolicity implies further standard causality properties. The proof follows  \cite[Props.~3.12, 3.13]{ake2020} by using Lemmas \ref{Le:Sequ lemma} and \ref{Le:Closure}. We point out that by \autoref{Pr:Closed}, there is no need to  discuss \emph{stable causality} \cite[Def.~3.9]{ake2020}.

\begin{theorem}[Causal ladder]\label{Th:Ladder} The following  hold in decreasing strength if we assume \autoref{Ass:GHLLS}.
\begin{enumerate}[label=\textnormal{(\roman*)}]
\item \textnormal{\textbf{Causal simplicity} \cite[Def.~3.7]{ake2020}\textbf{.}} The space $\mms$ is causal in the sense of \autoref{Cor:Causality}, and $\smash{J^\pm(x)}$ is closed for every $x\in\mms$.
\item \textnormal{\textbf{Causal continuity} \cite[Def.~3.7]{ake2020}\textbf{.}} The space $\mms$ is distinguishing in the sense of \autoref{Cor:Distinction}, and \emph{reflective}, i.e.~every $x,y\in \mms$ satisfy $\smash{I^+(x) \subset I^+(y)}$ if and only if $\smash{I^-(y)\subset I^-(x)}$.
\end{enumerate}
\end{theorem}

The property of strong causality \cite[Def.~2.35]{kunzinger2018}, which is another usual part of the causal ladders \cite[Thm.~3.16]{ake2020} and \cite[Thm.~3.26]{kunzinger2018} is investigated now.

\begin{definition}[Strong causality] Let $\Alex$ be the \emph{Alexandrov topology} on $\mms$, i.e. the topology on $\mms$ having $\{I(x,y) : x,y\in \mms\}$ as a subbase. 

We call $\mms$ \emph{strongly causal} if the  $\Alex$ coincides with the original topology $\Top$. 
\end{definition} 

\begin{remark}[Equivalence of topologies] Under strong causality, both $\Alex$ and $\Top$  agree with the {chronological topology} $\Inter$ of \autoref{Def:Inter}.
\end{remark}

We obtain the following characterization of strong causality by the traditional definition through nonexistence of almost closed strictly causal curves \cite{beem1996,oneill1983}. As always, topological notions such as neighborhoods are understood with respect to $\Top$.

\begin{lemma}[Characterizing strong causality]\label{Le:Str caus equiv} Strong causality is equivalent to the following. For every $x\in\mms$ and every neighborhood $U\subset\mms$ of $x$, there exists a neighborhood $V\subset U$ of $x$ such that every strictly causal curve $\gamma\colon[0,1]\to\mms$ with endpoints in $V$ does not leave $U$.
\end{lemma}

\begin{proof} The implication from strong causality to the second part is argued as in \cite[Lem.~2.38]{kunzinger2018} by using \autoref{Le:Pushup openness}.

The converse is similar to \cite[Thm.~3.26]{kunzinger2018}; we give a proof since we do not use the notion of localizability. Assume to the contrary that $\mms$ is not strongly causal. This means the existence of $x\in \mms$ and a neighborhood $U\subset\mms$ of $x$ such that no chronological diamond is contained in $U$. By assumption, there exists an open set $\smash{V\subset U}$ containing $x$ such that all strictly causal curves with endpoints in $V$ do not leave $\smash{U}$. Nontrivial chronology implies existence of a timelike curve $\gamma\colon[0,1]\to\mms$ passing through $x$. For sufficiently small $\delta>0$, by continuity we have $\smash{\gamma_{1/2-\delta}, \gamma_{1/2+\delta}\in V}$. By \autoref{Cor:Causality}, $\gamma$ is strictly causal, and therefore $\smash{x\in I(\gamma_{1/2-\delta}, \gamma_{1/2+\delta}) \subset U}$. This  contradicts our assumption.
\end{proof}

\begin{proposition}[Strong causality]\label{Cor:Strong causality} The space $\mms$ is strongly causal.
\end{proposition}

\begin{proof} Assume to the contrary that there exist $x\in \mms$ and a neighborhood $U\subset\mms$ of $x$ such that for every neighborhood $V\subset U$ of $x$, there exists a  strictly causal curve $\gamma\colon[0,1]\to\mms$ with endpoints in $V$, but which leaves $U$. 

We claim the existence of a nonincreasing sequence $(W_n)_{n\in\N}$ of neighborhoods of $x$ such that
\begin{itemize}
\item $\smash{\bigcap_{n\in\N} W_n = \{x\}}$, and
\item every neighborhood of $x$ contains $W_n$ for every sufficiently large $n\in\N$.
\end{itemize}
Let $(N_n)_{n\in\N}$ be a sequence of neighborhoods of $x$ provided by first-countability of $\Top$, and set $W_n := N_1 \cap \dots \cap N_n$, where $n\in\N$. Then $\smash{x\in \bigcap_{n\in\N} W_n}$ by construction. Now suppose that the latter contains another point $z\in \mms\setminus \{x\}$. Let $W\subset\mms$ be a neighborhood of $x$ which does not contain $z$, as provided by the Hausdorff property of $\Top$ shown in \autoref{Cor:Hausdorff}. First-countability of $\Top$, however, yields $W_n \subset W$ for some $n\in\N$, a contradiction to $z\in W_n$.

By nontrivial chronology as well as \autoref{Cor:Causality}, there exist $a,b\in U$ with $a\ll x\ll b$, and these points are mutually distinct. Since $I(a,b)$ is open, again by first-countability we have $W_n\subset I(a,b)$ for sufficiently large $n\in\N$. There is a strictly causal curve $\smash{\gamma^n\colon [0,1] \to \mms}$ with $\smash{\gamma_0^n,\gamma_1^n\in W_n}$ leaving $U$. Since $\smash{\bigcap_{n\in\N} W_n = \{x\}}$ and employing first-countability of $\Top$ once more, $\smash{(\gamma_0^n)_{n\in\N}}$ and $\smash{(\gamma_1^n)_{n\in\N}}$ converge to $x$. On the other hand, since $\smash{\gamma^n_{[0,1]}\subset J(a,b)}$ yet $\smash{\gamma^n}$ leaves $U$, we find a sequence $(z_n)_{n\in\N}$ in $J(a,b)\setminus U$ such that $\smash{\gamma_0^n\leq z_n\leq \gamma_1^n}$. By compactness of $J(a,b)\setminus U$ and \autoref{Pr:Closed}, there is a point $z\in J(a,b)\setminus U$ with $x\leq z\leq x$. Since $z\notin U$, necessarily $x\neq z$,  which contradicts \autoref{Cor:Causality}.
\end{proof}

\begin{remark}[Exhaustion  by emeralds]\label{Re:Compact Polish} The foregoing discussions yield that if $\Top$ is induced by a separable metric or is $\sigma$-compact, then it can be exhausted by compact emeralds. In the first situation, simply cover $\mms$ with countably many precompact chronological diamonds, and successively take their causally convex hulls.

In particular, in this case $\mms$ can be exhausted by Polish  spaces by  \autoref{Th:Compact Polish}.
\end{remark}

\subsubsection{Limit curve theorem} Since our limit curve theorems involve uniform convergence of curves, we need to clarify the meaning of this notion for the reference topology $\Top$ being merely first-countable. Also, we have to make sense of uniform convergence of curves with varying domains; cf.~\cite[Def.~5.10]{minguzzi2022}.

We employ the notation $\N_\infty := \N\cup\{\infty\}$.

\begin{definition}[Uniform convergence of curves]\label{Def:Unif cvg} Given any $n\in\N_\infty$, let $\alpha_n \leq \beta_n$ be given real numbers.  Set $\alpha:=\alpha_\infty$ and $\beta:=\beta_\infty$.  
Let $\smash{(\gamma^n)_{n\in \N_\infty}}$ be a family of curves $\gamma^n\colon[\alpha_n,\beta_n] \to \mms$. 
We define the extension $\smash{\hat{\gamma}^n\colon \R\to\mms}$ of $\smash{\gamma^n}$ by
\begin{align}\label{Eq:Extension}
\begin{split}
\hat{\gamma}_s^n := \begin{cases} \gamma^n_{\alpha_n} & \textnormal{\textit{if} } s < \alpha_n,\\
\gamma^n_s & \textnormal{\textit{if} } \alpha_n\leq s\leq \beta_n,\\
\gamma^n_{\beta_n}& \textnormal{\textit{otherwise}}.
\end{cases}
\end{split}
\end{align}
We term $(\gamma^n)_{n\in\N}$ to \emph{converge uniformly} to $\gamma^\infty$ if 
\begin{enumerate}[label=\textnormal{\alph*.}]
\item $(\alpha_n)_{n\in\N}$ converges to $\alpha$, 
\item $(\beta_n)_{n\in\N}$ converges to $\beta$, and 
\item there exists a metric $\met$ inducing the relative topology of $\Top$ on a neighborhood of $J(\gamma_\alpha^\infty,\gamma_\beta^\infty)$ such that
\begin{align*}
\lim_{n\to\infty} \sup_{s\in\R} \met(\hat\gamma_s^n,\hat\gamma^\infty_s) = 0.
\end{align*}
\end{enumerate}
\end{definition}

\begin{remark}[Disambiguation]\label{Cor:Polish} When dealing with uniform convergence, we always work on compact subsets of $\mms$. These are Polish by \autoref{Th:Compact Polish}, whence they always admit a uniform structure. In other words, a metric $\met$ as in \autoref{Def:Unif cvg} will always be given on compact subsets of $\mms$.



In fact, since our notion of uniform convergence works with compact metrizable spaces, the uniform structure will always be  \emph{unique} \cite[Thm.~36.19]{willard}. In particular, uniform convergence with respect to the metric from the abstract \autoref{Th:Compact Polish} will hold even if uniform convergence is already known with respect to some other metric (e.g.~the one from a Polish $\Top$).
\end{remark}




We now  turn to our limit curve theorem. For now, we have to allow for uniform convergence up to reparametrization. See, however, \autoref{Cor:Cpt TGeo} below. Anyway this is no issue for the next results since the length functional is invariant under reparametrizations by \autoref{Le:Additivity}.

\begin{theorem}[Limit curve theorem]\label{Th:Limit curve theorem} Assume \autoref{Ass:GHLLS}. Let $C_0,C_1\subset\mms$ be compact subsets, and let $\smash{(\gamma^n)_{n\in\N}}$ be a sequence of strictly causal curves on $[0,1]$ such that $\smash{\gamma_0^n \in C_0}$ and $\smash{\gamma_1^n\in C_1}$ for every $n\in\N$. Then $\smash{(\gamma^n)_{n\in\N}}$ has a subsequence which --- after reparametrization --- converges uniformly, in the sense of \autoref{Def:Unif cvg}, 
\begin{enumerate}[label=\textnormal{(\roman*)}]
\item to a strictly causal curve defined on a compact interval $[\alpha,\beta] \subset \R$ of positive length unless the endpoints of the extracted subsequence converge to the same limit, and 
\item to a constant causal curve defined on $[0,1]$ otherwise.
\end{enumerate}
\end{theorem}

As for \autoref{Th:Compact Polish}, the lengthy proof relies on constructions from \cite{minguzzi2022} and is deferred to \autoref{Ch:Loose}.

\subsection{Geodesy}\label{Sec:Geodesy} We now turn to \emph{maximizing curves}, especially \emph{geodesics}. See \cite[Def.~2.33]{kunzinger2018} or \cite[Def.~5.9]{minguzzi2022} for previous Lorentzian analogs. Global hyperbolicity somewhat replaces the local compactness assumption from metric geometry, which yields  existence of curves \emph{minimizing} length in this setting  \cite[Prop.~2.5.19]{burago2003}.

\begin{definition}[Maximizing curves]\label{Def:Max Curv} A curve $\gamma\colon[0,1]\to\mms$ is called 
\begin{enumerate}[label=\textnormal{\alph*.}]
\item \emph{maximizing} if it is causal --- in particular continuous --- and
\begin{align*}
l(\gamma_0,\gamma_1) = \Len_l(\gamma),
\end{align*}
\item \emph{strictly causal maximizing} if it is maximizing and  strictly causal,
\item \emph{timelike maximizing} if it is maximizing and  timelike, and
\item \emph{null maximizing} if it is maximizing and null.
\end{enumerate}
\end{definition}

Analogously, as after \autoref{Def:Curves} we define rough versions of these. Note our convention of rough maximizers $\gamma$ as above always satisfying $\Len_l(\gamma)\geq 0$.

The proof of the following is standard, cf.~\cite[Prop.~2.34]{kunzinger2018} and \cite[Prop.~6.2]{minguzzi2022}.

\begin{lemma}[Characterizing rough maximizers]\label{Le:ADD GEO} The following hold for every rough causal curve $\gamma\colon[0,1]\to\mms$.
\begin{enumerate}[label=\textnormal{(\roman*)}]
\item The curve $\gamma$ is a rough maximizer if and only if for every $0\leq r<s<t\leq 1$,
\begin{align*}
l(\gamma_r, \gamma_s) + l(\gamma_s,\gamma_t) = l(\gamma_r,\gamma_t).
\end{align*}
\item If $\gamma$ is rough null, it is a rough null maximizer.
\item If $\gamma$ is a rough causal maximizer, then any of its subsegments are rough causal maximizing. The same holds by replacing every occurrence of ``rough causal'' by ``rough strictly causal'', ``rough timelike'', or ``rough null''.
\end{enumerate}
\end{lemma}

To prove existence of maximizers, we need the following lemma, whose proof is analogous to \cite[Thm.~6.3]{minguzzi2022}.

\begin{lemma}[Upper semicontinuity of the length functional]\label{Le:Upper semiconti} Let $\smash{(\gamma^n)_{n\in\N}}$ be a sequence of rough causal curves $\smash{\gamma^n\colon[0,1]\to\mms}$ converging pointwise to a rough causal curve $\gamma\colon [0,1]\to\mms$. Then
\begin{align*}
\limsup_{n\to\infty} \Len_l(\gamma^n) \leq \Len_l(\gamma).
\end{align*}
\end{lemma}

\begin{proof}  Recall by \autoref{Cor:Finiteness l}  that  $\smash{l_+}$ is finite, hence $\smash{\Len_l(\gamma)<\infty}$. By definition of $\smash{\Len_l(\gamma)}$, given $\varepsilon > 0$ there exist $m\in\N$ and a partition  $0 =t_0 < t_1 < \dots < t_{m-1} < t_m=1$ of $[0,1]$ such that
\begin{align*}
\sum_{i=1}^m l(\gamma_{t_{i-1}}, \gamma_{t_i}) \leq \Len_l(\gamma) + \frac{\varepsilon}{2}.
\end{align*}
By point\-wise convergence and causality we have
\begin{align*}
l(\gamma^n_{t_{i-1}}, \gamma^n_{t_i}) \leq l(\gamma_{t_{i-1}}, \gamma_{t_i}) + \frac{\varepsilon}{2m}
\end{align*}
for every $i = 1,\dots,m$ and every sufficiently large $n\in\N$. Combining these two inequalities with  \autoref{Le:Additivity} then leads to
\begin{align*}
\Len_l(\gamma^n) &= \sum_{i=1}^m \Len_l(\gamma^n\big\vert_{[t_{i-1}, t_i]}) \leq \sum_{i=1}^m l(\gamma^n_{t_{i-1}},\gamma^n_{t_i}) \leq \Len_l(\gamma) + \varepsilon.\qedhere
\end{align*}
\end{proof}

\begin{theorem}[Nonsmooth Avez--Seifert theorem]\label{Th:Ex geo}  Assume \autoref{Ass:GHLLS}. For every $x,y\in\mms$ satisfying $x<y$, there exists a strictly causal maximizing curve   connecting $x$ to $y$. In particular, the length property from \autoref{Def:LLS} holds for $x$ and $y$ merely obeying $x<y$. 

Moreover, every rough causal curve connecting $x$ to itself is a constant, and in particular a continuous maximizer.
\end{theorem}

\begin{proof} The last statement follows from \autoref{Cor:Causality}.

To show the remaining claim, first we assume $x\ll y$, and in particular $x\neq y$. Recall that $l(x,y)< \infty$ by \autoref{Cor:Finiteness l}. Thus, owing to the length property  let $(\gamma^n)_{n\in\N}$ be a sequence of strictly causal curves from $x$ to $y$ with $\smash{\Len_l(\gamma^n)\to l(x,y)}$ as $n\to\infty$. \autoref{Th:Limit curve theorem} thus yields a uniform limit of reparametrizations of $\smash{(\gamma^n)_{n\in\N}}$ which is a strictly causal curve connecting $x$ to $y$. Since $x\neq y$,  we easily see these curves can be rescaled to curves $\smash{\sigma^n\colon[0,1]\to\mms}$, where $n\in\N$, and $\sigma\colon[0,1]\to\mms$ defined on $[0,1]$, respectively, in such a way that $\smash{(\sigma^n)_{n\in\N}}$ converges pointwise to $\sigma$. Since the length functional is invariant under reparametrizations by \autoref{Le:Additivity}, from \autoref{Le:Upper semiconti} we obtain
\begin{align*}
l(x,y) \geq \Len_l(\sigma)\geq \limsup_{n\to\infty}\Len_l(\sigma^n) = l(x,y).
\end{align*} 
This forces equality to hold throughout.

If only $x<y$, we first use \autoref{Le:Sequ lemma} to find sequences $(x_n)_{n\in\N}$ and $(y_n)_{n\in\N}$ with $x_n\uparrow x$  as $n\to\infty$ and $y_n\downarrow y$ as $n\to\infty$, respectively. Then for every $n\in\N$, the points $x_n$ and $y_n$ are connected by a strictly causal maximizer $\smash{\gamma^n\colon[0,1]\to\mms}$ by the preceding part. As  above, by also taking continuity of $\smash{l_+}$ and the fact $x\neq y$ into account, this leads to the desired strictly causal maximizer from $x$ to $y$.
\end{proof}

At this point, we do not know in general whether chronologically related points $x$ and $y$ can be connected by a \emph{timelike} maximizer. It might happen that all such curves have null subsegments. This will be ruled out later by the assumption of regularity in \autoref{Sec:Regularity}, which will be very useful for our purposes.

A topic we first address is the question of reparametrization of curves or maximizers by $l$-arclength. Again the procedure is quite standard \cite{burago2003}, and we refer the reader to \cite[Prop.~3.34, Cor.~3.35]{kunzinger2018} for details; the inherent assumption $\smash{l_+(x,x)=0}$ for every $x\in\mms$ is guaranteed by \autoref{Cor:Causality}.

Let us call a rough causal curve $\gamma\colon[0,1]\to\mms$  \emph{$l$-rectifiable} if $\smash{\Len_l(\gamma\big\vert_{[s,t]}) > 0}$ for every $0 \leq s<t\leq 1$ \cite[Def.~2.29]{kunzinger2018}. By definition of the length functional, every $l$-rectifiable curve is timelike \cite[Lem.~2.30]{kunzinger2018}.

\begin{proposition}[Reparameterizations of maximizers]\label{Pr:Reparam}  Every given causal curve $\gamma\colon[0,1]\to\mms$ which is $l$-rectifiable has a reparametrization $\bar{\gamma}\colon[0,\Len_l(\gamma)]\to \mms$ by $l$-arclength, i.e.~for every $0 \leq t\leq \Len_l(\gamma)$,
\begin{align*}
\Len_l(\bar\gamma\big\vert_{[0,t]}) = t.
\end{align*}

In particular, every timelike maximizer $\gamma\colon[0,1]\to\mms$ has a proper-time repara\-metrization $\bar\gamma\colon[0,l(\gamma_0,\gamma_1)]\to\mms$, i.e.~for every $0 \leq s<t \leq l(\gamma_0,\gamma_1)$,
\begin{align*}
l(\bar\gamma_s,\bar\gamma_t) = t-s.
\end{align*}
\end{proposition}

This naturally motivates the following definition. 

\begin{definition}[Geodesics] A curve $\gamma\colon[0,1]\to\mms$ will be called \emph{$l$-geodesic} if every $0\leq s<t\leq 1$ satisfies
\begin{align}\label{Eq:Aff}
l(\gamma_s,\gamma_t) = (t-s)\,l(\gamma_0,\gamma_1) > 0.
\end{align}
\end{definition}

\begin{remark}[About terminology] We point out the slight change of terminology of curves satisfying \eqref{Eq:Aff} to e.g.~\cite{braun2022, cavalletti2020} from ``timelike geodesics'' to ``$l$-geodesics''. This is in line with the  nomenclature for probability measures in \autoref{Def:lp geo}. Moreover, it is debatable to call a rough null curve obeying  the equality in \eqref{Eq:Aff} ``geodesic'' --- every reparametrization again has the same property, and it is a priori unclear which parametrization of null maximizers is the ``correct'' one;
 see \cite[Rem.~9]{mccann+} however. For this reason and since we mainly deal with curves satisfying \eqref{Eq:Aff}, we have decided  on this short-hand terminology.
\end{remark}

 For emphasis, we may term a curve satisfying \eqref{Eq:Aff} \emph{affinely parametrized}. It is clear from the definition of $\smash{\Len_l}$ that $l$-geodesics are maximizing. 

For later use, we define the (a priori possibly empty) class of $l$-geodesics by
\begin{align}\label{Eq:TGeo def}
\begin{split}
\TGeo(\mms) &:= \{\gamma \colon[0,1]\to\mms : l(\gamma_s,\gamma_t) = (t-s)\,l(\gamma_0,\gamma_1) > 0 \\
&\qquad\qquad\textnormal{for every }0\leq s<t\leq 1\}.
\end{split}
\end{align}

\subsection{Regularity}\label{Sec:Regularity} It is a priori unclear whether $l$-geodesics are continuous. Even worse, as  remarked even their general existence may a priori fail. Both issues will be ruled out by our assumption of regularity introduced in this section.

\subsubsection{Definition and a characterization} 

\begin{definition}[Regularity]\label{Def:Regularity} We term $\mms$  \emph{regular} if every strictly causal maximizer \textnormal{[sic]} $\gamma\colon[0,1]\to\mms$ with $\gamma_0\ll \gamma_1$ is already timelike.
\end{definition}

Note the implicit continuity hypothesis in \autoref{Def:Regularity}  which is never void by \autoref{Th:Ex geo}.

In other words, \autoref{Def:Regularity} asks every \emph{maximizer} to be either timelike or null. By \autoref{Le:Pushup openness}, the space $\mms$ is regular if and only if every maximizer $\gamma\colon[0,1]\to\mms$ with the property $\gamma_s\ll \gamma_t$ for some $0\leq s<t\leq 1$ is already timelike. This  notion is motivated by the main consequence of a similar concept from \cite{kunzinger2018} called \emph{regular localizability} \cite[Def.~3.16]{kunzinger2018}, namely that every maximizer in a regularly locali\-zable Lorentzian pre-length space à la \cite{kunzinger2018} is either timelike or null \cite[Thm.~3.18]{kunzinger2018}. 

\begin{assumption}[Second standing assumption]\label{Ass:REG} From now on, in addition to \autoref{Ass:GHLLS} \textnormal{(}length metric spacetime with compact diamonds\textnormal{)}, we assume $\mms$ to be regular according to  \autoref{Def:Regularity}.
\end{assumption}

As \autoref{Ass:REG} is standing, for terminological convenience we henceforth mostly drop the addenda ``globally hyperbolic'', ``regular'', and ``length'', and call the space under consideration simply a \emph{metric spacetime}.

\autoref{Pr:Reparam} provides the desired correspondence between geodesics and maximizers: every geodesic is clearly a maximizer, while every maximizer admits a reparametrization belonging to $\TGeo(\mms)$.

For the convenience of the reader, we now provide an equivalent definition of our setting obeying \autoref{Ass:REG} which collects the properties of metric spacetimes we most frequently use. 

\begin{theorem}[Characterization of reinforcement]\label{Th:Equiv LLS} Let $\smash{l\colon \mms^2\to [0,\infty]\cup\{-\infty\}}$ be a signed time separation function. The tuple $(\mms,l)$ is a globally hyperbolic, regular length metric spacetime if and only if the following conditions hold.
\begin{enumerate}[label=\textnormal{(\roman*)}]
\item \textnormal{\textbf{Nonempty chronology.}} For every $x\in\mms$, the sets $\smash{I^\pm(x)}$ are nonempty.
\item \textnormal{\textbf{Global hyperbolicity.}} The space $\mms$ has is globally hyperbolic according to the conclusion of \autoref{Th:GHy}.
\item \textnormal{\textbf{Regular timelike geodesy.}} Every points $x,y\in\mms$ with $x\ll y$ can be connected by a continuous element of $\TGeo(\mms)$\footnote{By \autoref{Le:Continuity geos} below, every rough strictly causal maximizer from $x$ to $y$ is timelike and continuous.}.
\end{enumerate}
\end{theorem}

We wrap up our discussion with examples obeying \autoref{Ass:REG}.

\begin{example}[Lipschitz spacetimes]\label{Ex:Lipschitz} Let $(\mms,\Rmet)$ be a Lipschitz spacetime, i.e.~$\mms$ is a connected smooth topological manifold endowed with a Lorentzian metric  tensor $\Rmet$ of Lipschitz regularity. Since $\mms$ is causally plain \cite[Def.~1.16, Cor.~1.17]{chrusciel2012}, it naturally induces a length metric spacetime \cite[Thm.~5.12]{kunzinger2018}. By \cite[Thm.~1.1]{graf2018}, we know regularity of such spaces (and here Lipschitz regularity plays a crucial role). Consequently, the metric spacetime induced by a globally hyperbolic Lipschitz spacetime satisfies \autoref{Ass:REG}.
\end{example}

\begin{example}[Convex neighborhoods] If $\mms$ admits convex neighborhoods in the sense of \cite[Def.~6.10]{minguzzi2022}, then $\mms$ is regular \cite[Prop.~6.11]{minguzzi2022}.
\end{example}

\begin{example}[Causal reversal]\label{Ex:Causal reversal} We define the function $\smash{l^\leftarrow\colon \mms^2\to [0,\infty]\cup\{-\infty\}}$ by $\smash{l^\leftarrow(x,y) := l(y,x)}$. Evidently, this is a signed time separation function as in \autoref{Def:signed time sep}, and the induced structure $\smash{(\mms,l^\leftarrow)}$ on $\mms$ is called \emph{causal reversal} of $(\mms,l)$. One readily verifies the former satisfies \autoref{Ass:REG} if and only if the latter does. In particular, a path $\gamma\colon[0,1]\to\mms$ belongs to $\TGeo(\mms)$ according to \eqref{Eq:TGeo def} if and only if its time reversal $\smash{\gamma^\leftarrow\colon[0,1]\to\mms}$ with $\smash{\gamma^\leftarrow_t := \gamma_{1-t}}$ belongs to 
\begin{align*}
\TGeo^\leftarrow(\mms) &:= \{\sigma\colon[0,1]\to\mms : l^\leftarrow(\sigma_s,\sigma_t) = (t-s)\,l^\leftarrow(\sigma_0,\sigma_1) > 0\\
&\qquad\qquad \textnormal{for every }0\leq s<t\leq 1\}.
\end{align*}
\end{example}

\subsubsection{Continuity of $l$-geodesics} The additional assumption of regularity will be useful to ``lift'' the geometry of the base space $\mms$ to the space of Borel probability measures over it, cf.~\autoref{Le:Geodesics plan} and  \cite{braun2022,braun2023, mccann+}. As a first step, we note the following fact  established in \cite{mccann+}, whose proof we repeat for convenience. 

Recall \eqref{Eq:TGeo def} for the definition of $\TGeo(\mms)$.

\begin{lemma}[Continuity of $l$-geodesics]\label{Le:Continuity geos} All elements of $\TGeo(\mms)$ are continuous.
\end{lemma}

\begin{proof} We first show right-continuity of a given $\gamma\in\TGeo(\mms)$; left-continuity is argued analogously. 

First, let $0<s<1$, suppose to the contrary that $\gamma$ is not right-continuous at $s$. By compactness of $J(\gamma_0,\gamma_1)$, this means the existence of a sequence $(\varepsilon_n)_{n\in\N}$ of numbers $0<\varepsilon_n < 1-s$ such that $(\gamma_{s+\varepsilon_n})_{n\in\N}$ converges to some $\smash{\gamma_{s+}\in\mms}$ with $\smash{\gamma_{s+}\neq \gamma_s}$. \autoref{Pr:Closed} shows $\smash{\gamma_s < \gamma_{s+}}$. Since $\gamma_0 \ll \gamma_s$ and $\smash{\gamma_{s+} \ll \gamma_1}$ by the definition of $\TGeo(\mms)$ and continuity of $\smash{l_+}$, concatenating
\begin{itemize}
\item a strictly causal maximizer $\smash{\sigma^1 \colon[0,1/3]\to \mms}$ from $\smash{\gamma_0}$ to $\gamma_s$,
\item a strictly causal maximizer $\smash{\sigma^2\colon[1/3,2/3] \to \mms}$ from $\gamma_s$ to $\smash{\gamma_{s+}}$, and
\item a strictly causal maximizer $\smash{\sigma^3\colon[2/3,1]\to \mms}$ from $\gamma_{s+}$ to $\gamma_1$
\end{itemize}
yields a strictly causal curve $\sigma\colon[0,1]\to\mms$ from $\gamma_0$ to $\gamma_1$. Then
\begin{align*}
l(\gamma_0,\gamma_1) &\geq \Len_l(\sigma)\\
&= \Len_l(\sigma^1) + \Len_l(\sigma^2) + \Len_l(\sigma^3)\\
&= l(\gamma_0,\gamma_s) + l(\gamma_s,\gamma_{s+}) + l(\gamma_{s+},\gamma_t)\\
&= l(\gamma_0,\gamma_1);
\end{align*}
here, in the second line, we have used \autoref{Le:Additivity}, while the last line employs the definition of $\TGeo(\mms)$. In particular, $\sigma$ is a strictly causal maximizer from $\gamma_0$ to $\gamma_1$. On the other hand, affine parametrization implies $\smash{l(\gamma_s,\gamma_{s+})=0}$. Thus, $\sigma$  has a null subsegment, in contradiction with regularity.

Right-continuity at $s=0$ is argued similarly. Here we directly start with a null strictly causal maximizer from $\smash{\gamma_0}$ to $\smash{\gamma_{0+}}$ and continue with a strictly causal maximizer connecting the chronologically related points $\smash{\gamma_{0+}}$ and $\gamma_1$.
\end{proof}

Recall \autoref{Th:Compact Polish} and \autoref{Cor:Polish} for  how $\smash{\Cont([0,1];\mms)}$ appearing below can be  naturally and unambiguously endowed with a topology of uniform convergence on compact subsets of $\mms$, even if $\Top$ is not globally metrizable. 

The next result can be interpreted as a limit curve theorem for $l$-geodesics. By \autoref{Le:Continuity geos}, its (independent) proof is somewhat more straightforward than that of  the limit curve theorem (\autoref{Th:Limit curve theorem}) for general strictly causal curves.

\begin{corollary}[Compactness of uniformly $l$-geodesics]
\label{Cor:Cpt TGeo} Given any $r>0$ and any compact sets $C_0,C_1\in \mms$, the set
\begin{align*}
G_r := \{\gamma\in \TGeo(\mms) : \gamma_0\in C_0,\, \gamma_1\in C_1,\, l(\gamma_0,\gamma_1)\geq r\}
\end{align*}
is compact with respect to uniform convergence in $\Cont([0,1];\mms)$.
\end{corollary}

\begin{proof} Since the set $G_r$ is closed by continuity of $\smash{l_+}$, it suffices to prove its precompactness. By \autoref{Le:Continuity geos}, we have $G_r\subset\Cont([0,1];\mms)$, so that the statement makes sense. Let $\smash{(\gamma^n)_{n\in\N}}$ be a sequence in $G_r$. 

Following an idea developed in \cite{beran2023}, we construct a candidate $\gamma\colon[0,1]\to\mms$ for a uniform limit of $(\gamma^n)_{n\in\N}$ along a nonrelabeled subsequence. By compactness, a diagonal procedure --- recalling that $J(C_0,C_1)$ is Polish by \autoref{Th:Compact Polish} --- and  \autoref{Pr:Closed} there exists a nonrelabeled subsequence of $\smash{(\gamma^n)_{n\in\N}}$ such that for every $t\in [0,1]\cap \Q$, there exists $\gamma_t\in J(C_0,C_1)$ such that $\smash{\gamma_t^n \to \gamma_t}$ as $n\to\infty$. The rest of the candidate $\gamma$ is constructed by approximation from the right. For $t\in [0,1]\setminus \Q$, let $\smash{(t_i)_{i\in \N}}$ and $\smash{(s_j)_{j\in\N}}$ be sequences in $(t,1]\cap \Q$ which converge to $t$. Up to passing to subsequences, we may and will assume the existence of the limits $\smash{x_t := \lim_{i\to \infty} \gamma_{t_i}}$ and $\smash{y_s := \lim_{j\to\infty} \gamma_{s_j}}$, which evidently lie in $J(C_0,C_1)$. We claim $x_t = y_s$. Indeed, given any $i\in\N$ we have $\smash{\gamma_{s_j} \leq \gamma_{t_i}}$ for every large enough $j\in\N$, and thus $y_s \leq x_t$ by first letting $j\to\infty$ and then $i\to\infty$. Analogously, we obtain $x_t\leq y_s$, and the claim follows from causality. In particular, this argument shows the limit $\gamma_t$ can be unambiguously defined as $\smash{\lim_{i\to\infty}\gamma_{t_i}}$, where $(t_i)_{i\in\N}$ is \emph{any} sequence as above. This finishes the construction of $\gamma$.

Since $\smash{l(\gamma_0^n,\gamma_1^n) \geq r}$ for every $n\in\N$, clearly $\gamma\in\TGeo(\mms)$ --- in particular, $\gamma$ is continuous by \autoref{Le:Continuity geos}.

Lastly, we claim uniform convergence of $\smash{(\gamma^n)_{n\in\N}}$ to $\gamma$. Let $\met$ metrize the relative topology on $J(C_0,C_1)$. Given any $\varepsilon > 0$, for every $0\leq t\leq 1$, by \autoref{Cor:Strong causality} (strong causality) there is an open neighborhood $V_t\subset\mms$ of $\gamma_t$ such that every strictly causal curve with endpoints in $V_t$ does not leave  the metric ball $\smash{\sfB^\met(\gamma_t,\varepsilon)}$. Up to shrinking $V_t$, we may and will assume it to be an open $\met$-ball of radius $0< \delta_t < \varepsilon$. We cover the compact set $\gamma_{[0,1]}$ by finitely many of these sets $V_{t_1},\dots,V_{t_m}$, where $m\in\N$, and $t_1,\dots,t_m$ may and will be chosen to be rational. By the compactness of $[0,1]$ and uniform continuity of $\gamma$, we may and will further assume that $\vert t_i - t_{i+1}\vert \leq \delta$ for every $i=1,\dots,m-1$, and that $0 = t_1 < t_2 < \dots < t_{m-1} < t_m = 1$ is an increasing $\delta$-net in $[0,1]$, where $\delta > 0$ is such that every $0\leq s,t\leq 1$ with $\vert s-t\vert \leq \delta$ satisfies $\met(\gamma_s,\gamma_t) \leq \varepsilon$.  By our choice of $t_1,\dots,t_m$, we have
\begin{align}\label{Eq:Comb1}
\sup_{t\in[t_i,t_{i+1}]} \met(\gamma_{t_i},\gamma_t) \leq 2\,\varepsilon.
\end{align}
Moreover, by pointwise convergence of $\smash{(\gamma^n)_{n\in\N}}$ to $\gamma$ on the rationals, there exists $n_i \in \N$ such that $\smash{\gamma_{t_i^n} \in V_{t_i}}$ and $\smash{\gamma_{t_{i+1}^n} \in V_{t_{i+1}}}$ for every $n\in\N$ with $n\geq n_i$. Since $t_1 < t_2 < \dots < t_{m-1} < t_m$, we have $\smash{V_{t_i} \cap V_{t_{i+1}}\neq \emptyset}$. Therefore, for every $n\geq n_i$ the restriction $\smash{\gamma^n\big\vert_{[t_i,t_{i+1}]}}$ is a strictly causal curve which starts in $\smash{V_{t_i}}$, passes through $\smash{V_{t_i} \cap V_{t_{i+1}}\neq \emptyset}$, and ends in $\smash{V_{t_{i+1}}}$; strong causality then easily implies
\begin{align}\label{Eq:Comb2}
\sup_{t\in[t_i,t_{i+1}]}\met(\gamma_t^n, \gamma_{t_i}) \leq 2\,\varepsilon.
\end{align}
Whenever $\smash{n\geq \max\{n_1, \dots,n_{m-1}\}}$, combining  \eqref{Eq:Comb1} and \eqref{Eq:Comb2} finally yields
\begin{align*}
\sup_{t\in[0,1]} \met(\gamma_t^n,\gamma_t) \leq 4\,\varepsilon.\tag*{\qedhere}
\end{align*}
\end{proof}



\subsubsection{Intermediate points} In order to study geodesics of probability measures later, given any $x,y\subset\mms$ and any $0\leq t\leq 1$, and following \cite[Ch.~2]{mccann2020}, we now investigate the regularity of the (possibly empty) \emph{$t$-intermediate point set}
\begin{align*}
Z_t(x,y) &:= \{z\in J(x,y) : l(x,z) = t\,l(x,y),\, l(z,y) = (1-t)\,l(x,y)\}.
\end{align*}
Moreover, for $C\subset \mms^2$ we define
\begin{align}\label{Eq:Zt(C)}
\begin{split}
Z_t(C) &:= \bigcup_{(x,y)\in C} Z_t(x,y),\\
Z(C) &:= \bigcup_{t\in[0,1]} Z_t(C).
\end{split}
\end{align}

\begin{remark}[Initial and final points] If $x,y\in\mms$ obey $x\ll y$ then $Z_0(x,y) = \{x\}$ \cite[Rem.~2.6]{mccann2020}. Indeed, if $x'\in Z_0(x,y)\setminus \{x\}$ then $l(x,x') = 0$, and together with $l(x',y) = l(x,y)>0$ this implies the concatenation of a null maximizer from $x$ to $x'$ and an  $l$-geodesic from $x'$ to $y$ to be a maximizer from $x$ to $y$ which intermediately changes its causal character from null to timelike. This contradicts regularity of $\mms$. Analogously, we show $Z_1(x,y) = \{y\}$ if $l(x,y)>0$.
\end{remark}

The following is a nonsmooth analog to \cite[Lem.~2.5]{mccann2020}.

\begin{lemma}[Intermediate point sets inherit compactness I]
\label{Le:Zt lemma} Let $0 \leq t\leq 1$. If $C\subset \mms^2$ is precompact, so are $Z_t(C)$ and $Z(C)$. Moreover, if $C$ is compact, so are $Z_t(C)$ and $Z(C)$.
\end{lemma}

\begin{proof} Let  $C_1, C_2\subset\mms$ be the projections of $C$ to the first and second coordinate, respectively. We first assume precompactness of $C$. By its definition,  $Z(C)$ is a subset of $J(\bar{C}_1,\bar{C}_2)$ which is compact by \autoref{Cor:K-GH}, and hence $Z(C)$ is precompact. In turn, this implies precompactness of $Z_t(C) \subset Z(C)$.

The same argument combined with the continuity of $\smash{l_+}$  and \autoref{Cor:Causality} gives the second statement provided $C$ is compact.
\end{proof}

\subsection{Lorentzian optimal transport}\label{Sub:LOT} Next, we review the optimal transport problem in our Lorentzian setting. We refer to \cite{cavalletti2020,eckstein2017,kell2018,mccann2020, mondinosuhr2022,suhr2018} for details. Their results carry over with no essential change to our slightly different setting.

From now on, we insist on $\Top$ being Polish. As outlined in \autoref{Re:Compact Polish}, under mild conditions this can be relaxed by working on a compact exhaustion, but to streamline the presentation we directly assume Polishness of $\Top$. Whenever needed, $\met$ will denote a complete and separable metric inducing $\Top$.

Fix a Radon measure $\meas$ on $\mms$. To avoid pathologies, we assume $\meas$ is nontrivial, but we do \emph{not} assume $\meas$ has full support in general. By \autoref{Re:Compact Polish}, $\meas$ is $\sigma$-finite (which has to be assumed if $\Top$ is merely first-countable). We often write
\begin{align}\label{Eq:Space}
\scrM := (\mms,l,\meas).
\end{align}

\begin{definition}[Metric measure spacetime]\label{Def:Metric measure spacetime} The  triple $\scrM$ in \eqref{Eq:Space}, i.e.~a regular  Polish length metric spacetime $(\mms,l)$ with compact diamonds according to our second standing \autoref{Ass:REG} endowed with a nontrivial Radon measure $\meas$, is called  a \emph{metric measure spacetime}.
\end{definition}

Also, we write
\begin{align*}
\scrM^\leftarrow := (\mms, l^\leftarrow,\meas)
\end{align*}
for the structure induced by the causal reversal from \autoref{Ex:Causal reversal}.

\subsubsection{Suslin measurability}\label{Sub:Suslin} In \autoref{Ch:Localization} we need to discuss Suslin measurability, basic notions of which are briefly summarized here. See \cite{srivastava1998} for details.

A subset $E$ of a Polish space $X$ is called \emph{Suslin measurable}, or briefly \emph{Suslin}, if it is the continuous image of a Borel subset of a Polish space (possibly different from $X$). In the last clause, one can  equivalently remove  ``Borel subset of a''. Another characterization of Suslin measurability is that $E$ is the projection of a Borel subset of a Polish product space, i.e.~there are a Polish space $Y$ and a Borel set $F\subset X\times Y$ such that $E=\pr_1(F)$, where $\pr_1\colon X\times Y\to X$ is defined as $\pr_1(x,y) := x$.

The crucial property we use later is \emph{universal measurability} of all Suslin sets. Here, a subset $E$ of a Polish space $X$  is universally measurable if, for every Borel probability measure $\mu$ on $X$, it is measurable with respect to the completion of the Borel $\sigma$-field of $X$ with respect to $\mu$.

\subsubsection{Notation from probability theory}\label{Sub:Notation prob} Let $\scrP(\mms)$ denote the space of all Borel probability measures on $\mms$, endowed with the narrow topology, as detailed below in \autoref{Sub:Weak convergence}. The Dirac measure at $x\in\mms$ is denoted $\delta_x$. Let $\scrP_\comp(\mms)$ and $\scrP^\ac(\mms,\meas)$ denote the spaces of  compactly supported and $\meas$-absolutely continuous measures in $\scrP(\mms)$, respectively. We set $\smash{\scrP_\comp^\ac(\mms,\meas) := \scrP_\comp(\mms)\cap\scrP^\ac(\mms,\meas)}$. When writing $\mu = \rho\,\meas + \mu_\perp$ for a given $\mu\in\scrP(\mms)$, we mean the Lebesgue decomposition of $\mu$ with respect to $\meas$, where $\rho$ is the Radon--Nikodým density of its $\meas$-absolutely continuous part, and $\smash{\mu_\perp}$ is its $\meas$-singular part.

The following terminology will be useful in the sequel. It has an evident analog for sequences of probability measures.

\begin{definition}[Uniformly compact support]\label{Def:Unif cpt supp} A family $\scrC\subset\scrP(\mms)$ is called \emph{uniformly compactly supported} if there exists a compact subset $C\subset\mms$ such that $\supp\mu \subset C$ for every $\mu\in \scrC$.
\end{definition}

The \emph{push-forward} of a measure $\mu\in\Prob(\mms)$ under a Borel map $F\colon \mms\to\mms'$, where $\mms'$ is any given topological space, is the probability measure $F_\sharp\mu \in\scrP(\mms')$ defined by setting, for every Borel set $B\subset\mms'$,
\begin{align*}
F_\sharp\mu[B] := \mu\big[F^{-1}(B)\big].
\end{align*}

A measure $\pi\in\scrP(\mms^2)$ is called \emph{coupling} of given $\mu,\nu\in\scrP(\mms)$ if 
\begin{align*}
(\pr_1)_\push\pi &= \mu,\\ 
(\pr_2)_\push\pi &=\nu. 
\end{align*}
Here and in the sequel, given any $i=1,2$, $\smash{\pr_i\colon \mms^2\to \mms}$ denotes the projection $\pr_i(x) := x_i$; the projection $\pr_i\colon \mms^n\to \mms$ is defined analogously for arbitrary $n\in \N$ and $i=1,\dots,n$. Let $\Pi(\mu,\nu)$ be the set of all couplings of $\mu$ and $\nu$. It is never empty as it contains at least the product measure $\smash{\mu\otimes\nu\in\scrP(\mms^2)}$. Moreover, by \emph{diagonal coupling} of $\mu$ we mean the probability measure $\pi := D_\push\mu$, where $D\colon \mms\to\mms^2$ is defined by $D(x) := (x,x)$.

We recall the following well-known result, cf.~e.g.~\cite[Lem.~2.1]{ambrosio2011}.

\begin{lemma}[Gluing]\label{Le:Gluing lemma} Given any $n\in\N$, let $\mu_0,\dots,\mu_n\in\scrP(\mms)$, and let $\pi_{i-1,i}\in\Pi(\mu_{i-1},\mu_i)$ for every $i=1,\dots,n$. Then there exists a measure $\omega\in \Prob(\mms^{n+1})$ such that every $i=1,\dots,n$ satisfies
\begin{align*}
(\pr_{i-1},\pr_i)_\push\omega = \pi_{i-1,i}.
\end{align*}
\end{lemma}

We call a set $\scrC\subset\scrP(\mms)$ \emph{tight} if for every $\varepsilon > 0$ there exists a compact set $C\subset\mms$ such that $\mu[C^\sfc]\leq \varepsilon$ for every $\mu\in\scrC$, where $C^\sfc:=\mms\setminus C$ denotes set complementation; we also agree on the evident meaning of a tight sequence. Tightness is a  useful property in view of Prokhorov's theorem  (\autoref{Th:Prokhorov} below). Recall that tightness transfers to couplings \cite[Lem.~4.4]{villani2009}: if $\scrC_0,\scrC_1\subset\scrP(\mms)$ are tight, so is the set $\Pi(\scrC_0\times\scrC_1)$ of all couplings of elements of $\scrC_0$ and $\scrC_1$. This argument does not require $\Top$ to be Polish. The converse implication holds as well by the continuity of the projection maps and Prokhorov's theorem.

Lastly, let $\Cont([0,1];\mms)$ be the space of continuous maps $\gamma\colon[0,1]\to\mms$, endowed with the topology of uniform convergence. For $0\leq t\leq 1$, define the continuous \emph{evaluation map} $\eval_t\colon \Cont([0,1];\mms)\to \mms$ by
\begin{align}\label{Eq:Evaluation map}
\eval_t(\gamma) := \gamma_t.
\end{align}
The continuous \emph{restriction map} $\smash{\Restr_s^t\colon \Cont([0,1];\mms) \to \Cont([0,1];\mms)}$, where $0\leq s<t\leq 1$, is the map which restricts $\smash{\gamma\in\Cont([0,1];\mms)}$ to the segment on $[s,t]$ and then ``stretches'' this segment 
 affinely to $[0,1]$, i.e.
\begin{align}\label{Eq:Restr}
\Restr_s^t(\gamma)_r := \gamma_{(1-r)s + rt}.
\end{align}

\subsubsection{Narrow convergence}\label{Sub:Weak convergence} A sequence $(\mu_n)_{n\in\N}$ in $\scrP(\mms)$ converges \emph{narrowly}\footnote{Some authors prefer to call this \emph{weak} convergence; note that for noncompact $\mms$, this usage of ``weak'' is inconsistent with its meaning in Banach space theory.} to $\mu\in\scrP(\mms)$ \cite[p.~13]{billingsley} if every continuous bounded test function $\varphi\in\Cont_\bounded(\mms)$ satisfies
\begin{align*}
\lim_{n\to\infty}\int_\mms \varphi\d\mu_n = \int_\mms\varphi\d\mu.
\end{align*}
When $\Top$ is Polish, narrow convergence is induced by a metric \cite[Rem.~5.1.1]{ambrosio2008}. In the more general case, narrow convergence of directed nets is defined in the evident way.

By continuity of the projection, it is clear that the property of being a coupling is narrowly closed with respect to narrow convergence of the marginals.

Given an interval $I\subset\R$, by narrow continuity of a collection $(\mu_t)_{t\in I}$ in $\scrP(\mms)$ we mean continuity of the map $t\mapsto \mu_t$ on $I$ with respect to narrow convergence.

We recall two results on convergence and compactness in the narrow topology we occasionally employ. The first is Alexandrov's theorem\footnote{This theorem seems to have been proven first in \cite{alexandrov1943}. A more common name for it is \emph{Portmanteau's theorem}, a terminology originating in \cite{billingsley}. The references of \cite{billingsley} quote the article ``Espoir pour l'ensemble vide?'', viz. ``Hope for the empty set?'', of Jean-Pierre Portmanteau in the \textit{Annales de l'Université de Felletin}. This appears to be a hoax, since we find no historical record of mathematician of that name, nor of the University of Felletin (let alone the mentioned paper). To acknowledge Alexandrov's contribution, we have decided to call \autoref{Th:Alexandrovs theorem} ``Alexandrov's theorem'', as suggested by Bogachev \cite[p.~453]{bogachevII}.} \cite[Thm.~2.1]{billingsley}, which relates narrow convergence to convergence of set evaluations. The second is Prokhorov's theorem \cite[Thms.~5.1, 5.2]{billingsley}, which is a useful characterization of narrow precompactness in $\scrP(\mms)$. 

\begin{theorem}[Alexandrov's theorem]\label{Th:Alexandrovs theorem} Narrow convergence of a sequence $(\mu_n)_{n\in\N}$ in $\scrP(\mms)$ to $\mu\in\scrP(\mms)$ is equivalent to any of the following.
\begin{enumerate}[label=\textnormal{(\roman*)}]
\item \textnormal{\textbf{Lower semicontinuity on open sets.}} For every open set $U\subset\mms$,
\begin{align*}
\mu[U] \leq \liminf_{n\to\infty} \mu_n[U].
\end{align*}
\item \textnormal{\textbf{Upper semicontinuity on closed sets.}} For every closed set $C\subset\mms$,
\begin{align*}
\mu[C]\geq \limsup_{n\to\infty}\mu_n[C].
\end{align*}
\item \textnormal{\textbf{Continuity sets.}} For every Borel set $E\subset\mms$ with $\mu[\partial E]=0$,
\begin{align*}
\mu[E]=\lim_{n\to\infty} \mu_n[E].
\end{align*}
\end{enumerate}
\end{theorem}

\begin{theorem}[Prokhorov's theorem]\label{Th:Prokhorov} A given set $\scrC\subset\scrP(\mms)$ is narrowly precompact if and only if it is tight.
\end{theorem}

\begin{remark}[Extension beyond Polish spaces]  Even if $\Top$ is not Polish, \autoref{Th:Alexandrovs theorem} and \autoref{Th:Prokhorov} extend to merely first-countable $\Top$ in our setting, and at this point there is no need to restrict ourselves to uniformly compactly supported measures. The crucial property is \emph{complete regularity} of $\Top$, i.e.~for every closed set $C\subset \mms$ and every $\smash{x\in C^\sfc}$ in its complement $C^\sfc$, there exists a continuous function $f\colon \mms\to \R$ vanishing on $C$ and $f(x) =1$. Every locally compact Hausdorff topology is completely regular, hence $\Top$ is by \autoref{Cor:Hausdorff}.

\autoref{Th:Alexandrovs theorem} then follows from \cite[Thm.~8.2.3, Cor.~8.2.4]{bogachevII}. 

\autoref{Th:Prokhorov} is found in \cite[Thm.~8.6.7]{bogachevII}. Indeed, \autoref{Th:Compact Polish} implies every compact subset of $\mms$ is metrizable.
\end{remark}

\subsubsection{Causal and chronological couplings} Given any $\mu,\nu\in\scrP(\mms)$, the (possibly empty) set $\smash{\Pi_\leq(\mu,\nu)}$ of \emph{causal couplings} of $\mu$ and $\nu$ is the set of all $\pi\in\Pi(\mu,\nu)$ with $\pi[l\geq 0] =1$. Since $\{l\geq 0\}$ is closed, the latter condition holds if and only if $\supp\pi\subset\{l\geq 0\}$. By replacing ``$\leq$'' and ``$\geq$'' by ``$\ll$'' and ``$>$'', respectively, in an analogous way we define the set $\Pi_\ll(\mu,\nu)$ of \emph{chronological couplings} of $\mu$ and $\nu$. However, note that the property $\pi\in \Pi_\ll(\mu,\nu)$ does not imply $\supp\pi\subset\{l>0\}$ in general, while the converse is trivially true.

\subsubsection{The $\smash{\ell_p}$-optimal transport problem} Given an exponent $0<p\leq 1$, besides \eqref{Eq:Convention infty} we will adopt the conventions
\begin{align*}
(-\infty)^p := (-\infty)^{1/p} := -\infty.
\end{align*}

We define the function $\smash{\ell_p\colon \scrP(\mms)^2\to [0,\infty]\cup\{-\infty\}}$ by
\begin{align}
\label{lp causal}
\ell_p(\mu,\nu) &:= \sup_{\pi\in \Pi_\leq(\mu,\nu)} \big\Vert l_+\big\Vert_{\Ell^p(\mms^2,\pi)},
\end{align}
and call it total transport cost. We agree on the usual convention $\sup\emptyset := -\infty$, in which case we can equivalently use an unconditional supremum to infer
\begin{align}
\label{lp conventional}
\ell_p(\mu,\nu) = \sup_{\pi\in\Pi(\mu,\nu)} \big\Vert l\big\Vert_{\Ell^p(\mms^2,\pi)}.
\end{align}
The sets of maximizers of \eqref{lp causal} and \eqref{lp conventional} coincide \cite[Rem.~2.2]{cavalletti2020}. A coupling $\pi\in\Pi(\mu,\nu)$ is called \emph{$\ell_p$-optimal} if it is causal and it attains the supremum in the  definition of $\smash{\ell_p}$; if $\smash{\ell_p(\mu,\nu) = -\infty}$ no coupling of $\mu$ and $\nu$ is $\smash{\ell_p}$-optimal according to this terminology, even though $\smash{\ell_p(\mu,\nu) = \big\Vert l\big\Vert_{\Ell^p(\mms^2,\pi)}}$ holds for every $\pi\in\Pi(\mu,\nu)$. On the other hand, by our definition every $\smash{\pi\in\scrP(\mms^2)}$ with $\smash{\big\Vert l\big\Vert_{\Ell^p(\mms^2,\pi)} = \infty}$ is an $\smash{\ell_p}$-optimal coupling of its marginals.

Standard arguments of calculus of variations, combined with the upper boundedness of $\smash{\ell_+}$ on compact subsets of $\smash{\mms^2}$, ensure the  existence of $\smash{\ell_p}$-optimal couplings in all cases relevant for our purposes, cf.~\cite[Prop.~2.3]{cavalletti2020} and \cite[Thm.~4.1]{villani2009}.

\begin{lemma}[Existence of $\smash{\ell_p}$-optimal couplings]\label{Le:Existence} If two measures $\mu,\nu\in\scrP_\comp(\mms)$ obey $\smash{\Pi_\leq(\mu,\nu)\neq\emptyset}$, then $\mu$ and $\nu$ admit an $\smash{\ell_p}$-optimal coupling.
\end{lemma}

Another feature is the reverse triangle inequality for $\smash{\ell_p}$ akin to its counterpart \eqref{Eq:Reverse tau} for $l$ \cite[Prop.~2.5]{cavalletti2020}, which is a standard application of the gluing lemma: for every $\mu,\nu,\sigma\in\scrP(\mms)$, we have
\begin{align}\label{Eq:Reverse lp}
\ell_p(\mu,\sigma)\geq \ell_p(\mu,\nu) + \ell_p(\nu,\sigma).
\end{align}

These observations yield interesting facts about the causal structure of the space of probability measures on $\mms$, which generalize those from \cite{eckstein2017} to our nonsmooth setting. We defer this discussion to \autoref{Ch:Loose}.

\subsubsection{Timelike $p$-dualizability} Even if $\smash{\ell_p(\mu,\nu) > 0}$ for two given $\mu,\nu\in\scrP(\mms)$, in general $\smash{\ell_p}$-optimal couplings need not vanish outside set $\{l>0\}$, though they cannot vanish inside it  (e.g.~transport a Dirac mass halfway to a point in its null future and halfway to a point in its chronological future in Minkowski spacetime). Various notions for $\mu$ and $\nu$ to have ``entirely chronological'' $\smash{\ell_p}$-optimizers have been studied in the literature. Starting from McCann's \emph{$p$-separation} \cite[Def.~4.1]{mccann2020}, Cavalletti--Mondino \cite[Defs.~2.18, 2.27]{cavalletti2020} introduced the notions of  \emph{\textnormal{(}strong\textnormal{)} timelike $p$-dualizability} 
recapitulated in \autoref{Def:Str tl dual}
as nonsmooth analogs for the more relaxed criteria of \cite[Thms.~7.1, 7.4]{mccann2020}.
They allow for a good duality theory \cite[Props.~2.19, 2.21, Thm.~2.26]{cavalletti2020}, see also \cite[Thm.~4.3]{mccann2020}. 

As usual, given two functions $a,b\colon\mms\to\R$ we define $a\oplus b \colon \mms^2\to\R$ by 
\begin{align*}
(a\oplus b)(x,y) := a(x) + b(y).
\end{align*}

\begin{definition}[Timelike $p$-dualizability]\label{Def:Str tl dual} Let $\mu,\nu\in\scrP(\mms)$ be given.
\begin{enumerate}[label=\textnormal{\alph*\textcolor{black}{.}}]
\item We call the pair $(\mu,\nu)$ \emph{timelike $p$-dualizable} if $\smash{\ell_p(\mu,\nu) > 0}$, some $\smash{\ell_p}$-optimal coupling of $\mu$ and $\nu$ belongs to $\smash{\Pi_\ll(\mu,\nu)}$, and there are lower semi\-continuous functions $a,b\colon \mms\to [0,\infty)$ with $\smash{a\oplus b\in \Ell^1(\mms^2,\mu\otimes\nu)}$ and
\begin{align}\label{Eq:Moment condition}
l^p \leq a\oplus b\quad\textnormal{\textit{on} } \supp\mu\times\supp\nu.
\end{align}
If $\smash{\pi\in\Pi_\ll(\mu,\nu)}$ is such a coupling, we call $(\mu,\nu)$ \emph{timelike $p$-dualizable by $\pi$}, and the coupling $\pi$ \emph{timelike $p$-dualizing}.
\item\label{La:Be} Furthermore, we call the pair $(\mu,\nu)$ \emph{strongly timelike $p$-dualizable} if it is timelike $p$-dualizable, and there exists an $\smash{l^p}$-cyclically monotone Borel set $\Gamma\subset \{l>0\}\cap (\supp\mu\times\supp\nu)$ such that every  $\smash{\ell_p}$-optimal coupling $\pi\in\Pi_\leq(\mu,\nu)$ is concentrated on $\Gamma$.
\end{enumerate}
\end{definition}

Here, we call a set $\smash{\Gamma\subset \{l\geq 0\}}$ \emph{$l^p$-cyclically monotone} \cite[Def.~2.6]{cavalletti2020} if for every $n\in\N$ and every $(x_1,y_1),\dots,(x_n,y_n)\in\Gamma$, we have
\begin{align*}
\sum_{i=1}^n l^p(x_i,y_i) \geq \sum_{i=1}^n l^p(x_{i+1},y_i),
\end{align*}
where $x_{n+1} := x_1$. Accordingly, an element of $\smash{\scrP(\mms^2)}$ is \emph{$l^p$-cyclically monotone} if it is concentrated on an $l^p$-cyclically monotone subset of $\{l\geq 0\}$.

Note that \eqref{Eq:Moment condition} forces $\ell_p(\mu,\nu)<\infty$ if $(\mu,\nu)$ is timelike $p$-dualizable.

Clearly, strong timelike $p$-dualizability implies timelike $p$-dualizability. Moreover, any set $\Gamma$ as in the preceding item \ref{La:Be} actually characterizes $\smash{\ell_p}$-optimality. Indeed, if $\smash{\pi\in\Pi(\mu,\nu)}$ satisfies $\pi[\Gamma]=1$, it is chronological and $\smash{l^p}$-cyclically monotone, hence $\smash{\ell_p}$-optimal by the hypothesized moment conditions \eqref{Eq:Moment condition} \cite[Prop.~2.8]{cavalletti2020}.

We will mainly need these notions for compactly supported measures, in which case the definitions become easier to handle.

\begin{remark}[Compact support simplifies timelike $p$-dualizability] Assume $\mu,\nu\in\scrP_\comp(\mms)$. Then  $(\mu,\nu)$ is timelike $p$-dualizable if and only if they admit a timelike $p$-dualizing coupling $\pi\in\Pi_\ll(\mu,\nu)$, i.e.~a \emph{chronological} $\smash{\ell_p}$-optimal coupling. 
\end{remark}

\begin{example}[Chronological product measures]\label{Ex:Str tl dual} If $\mu,\nu\in\scrP_\comp(\mms)$ satisfy 
\begin{align}\label{Eq:Support cond}
\supp\mu\times\supp\nu\subset\{l>0\},
\end{align}
the pair $(\mu,\nu)$ is strongly timelike $p$-dualizable for every $0<p\leq 1$ \cite[Cor.~2.29]{cavalletti2020}. However, unlike strong timelike $p$-dualizability, cf.~\autoref{Le:Propagates} below, the condition \eqref{Eq:Support cond} does not propagate to the interior of $\smash{\ell_p}$-geodesics.
\end{example}

\begin{remark}[Timelike $p$-dualizabilty in the smooth framework] The $p$-separation of $(\mu,\nu)\in\scrP_\comp(\mms)^2$ \cite[Def.~4.1]{mccann2020} implies strong timelike $p$-dualizability of $(\mu,\nu)$. Indeed, by definition the set $S$ therein is contained in $\{l>0\}\cap (\supp\mu\times\supp\nu)$, and it is evidently $l^p$-cyclically monotone. Moreover, standard Kantorovich duality \cite[Thm.~5.10]{villani2009} ensures every $\smash{\ell_p}$-optimal coupling of $\mu$ and $\nu$ to be concentrated on $S$. In the smooth setting,  (strong) timelike $p$-dualizability is encoded in the hypotheses of \cite[Thm.~7.1]{mccann2020}, which does not assume $p$-separation.
\end{remark}

\subsection{Geodesics of probability measures}\label{Sub:Geo prob meas}  
Following \cite[Def.~1.1]{mccann2020} we now introduce various notions of $\smash{\ell_p}$-geodesics, and study their properties in detail. Key parts of our analysis are sufficient conditions for
\begin{itemize}
\item their narrow continuity, and
\item their admittance of a ``lifting'' to the space of Borel probability measures concentrated on affinely parametrized geodesics on $\mms$.
\end{itemize}
We refer to \cite{beran2023,braun2022, braun2023, cavalletti2020,mccann+,mondinosuhr2022} for further treatments of related notions; in a certain sense, our following discussion continues the study of the Lorentzian structure of $\scrP(\mms)$ initiated in \cite{eckstein2017} for smooth spacetimes. We point out that the notion of $\smash{\ell_p}$-geodesics in \cite{braun2022, braun2023, cavalletti2020} is  slightly more restrictive than that of \cite{mccann2020}. Roughly speaking, in \cite{braun2022,braun2023, cavalletti2020} $\smash{\ell_p}$-geodesics are represented by  probability measures concentrated on appropriate curves, and hence given by liftings \emph{by definition}. As we do not have Lipschitz causal curves at our disposal, some desired properties need to be reworked; especially compactness arguments become trickier. A byproduct of our results is that the restrictions from \cite{braun2022, braun2023, cavalletti2020} are not really necessary, and a much weaker notion is sufficient to work with under \autoref{Ass:REG}. Also, \autoref{Def:lp geo} is more flexible with respect to convergence questions, cf.~e.g.~\autoref{Pr:Compactness lp geos}.

\subsubsection{Various notions of geodesics}\label{Subsub:Various} Recall the definition of the evaluation map $\eval_t$ from \eqref{Eq:Evaluation map}, where $0\leq t\leq 1$. We say a collection $(\mu_t)_{t\in[0,1]}$ in $\scrP(\mms)$ is \emph{represented} by $\smash{\bdpi\in\Prob(\Cont([0,1];\mms))}$ if every $0\leq t\leq 1$ satisfies 
\begin{align*}
\mu_t = (\eval_t)_\push\bdpi.
\end{align*}

We endow the path space $\Prob(\Cont([0,1];\mms)$ with its natural narrow topology.

Given any $\mu_0,\mu_1\in\scrP(\mms)$, we define the (possibly empty) class
\begin{align}\label{OptTGeop}
\begin{split}
\OptTGeo_p(\mu_0,\mu_1) &:= \{\bdpi\in \Prob(\Cont([0,1];\mms)) : (\eval_0,\eval_1)_\push\bdpi \in \Pi(\mu_0,\mu_1)\\
&\qquad\qquad \textnormal{is }\ell_p\textnormal{-optimal},\, \bdpi[\TGeo(\mms)]=1\}.
\end{split}
\end{align}
Its elements will be called \emph{timelike $\ell_p$-optimal dynamical plans}, or simply \emph{plans};
a plan can also be thought of as a traffic plan or a routing.  This definition makes sense since $\TGeo(\mms)$ is a Borel subset of $\Cont([0,1];\mms)$: it is a $\sigma$-compact subset of $\Cont([0,1];\mms)$ by 
\autoref{Cor:Cpt TGeo}.



\begin{definition}[Geodesics of probability measures]
\label{Def:lp geo} We call a curve $(\mu_t)_{t\in[0,1]}$ of elements in $\scrP(\mms)$
\begin{enumerate}[label=\textnormal{\alph*.}]
\item \emph{rough $\smash{\ell_p}$-geodesic} if for every $0\leq s<t\leq 1$,
\begin{align}\label{Eq:lp geo path}
0< \ell_p(\mu_s,\mu_t) = (t-s)\,\ell_p(\mu_0,\mu_1) <\infty,
\end{align}
\item \emph{$\smash{\ell_p}$-geodesic} if it is a narrowly continuous rough $\smash{\ell_p}$-geodesic, and
\item \emph{displacement $\smash{\ell_p}$-geodesic} if it is represented by a timelike $\smash{\ell_p}$-optimal dynamical plan $\smash{\bdpi \in \OptTGeo_p(\mu_0,\mu_1)}$.
\end{enumerate}
\end{definition}

\begin{remark}[Displacement $\smash{\ell_p}$-geodesics are $\smash{\ell_p}$-geodesics]\label{Re:TL Geo to geo} Every displacement $\smash{\ell_p}$-geodesic is  an $\smash{\ell_p}$-geodesic. Indeed, let a curve $(\mu_t)_{t\in[0,1]}$ be represented by a plan $\smash{\bdpi\in\OptTGeo_p(\mu_0,\mu_1)}$. By Lebesgue's theorem and \autoref{Le:Continuity geos}, it is clear that $(\mu_t)_{t\in[0,1]}$ is narrowly continuous. Moreover, $(\eval_0,\eval_1)_\push\bdpi$ is an $\smash{\ell_p}$-optimal coupling of $\mu_0$ and $\mu_1$. To show \eqref{Eq:lp geo path}, let $0 \leq s<t\leq 1$. Note that $(\eval_s,\eval_t)_\push\bdpi$ is a causal coupling of $\mu_s$ and $\mu_t$, and thus
\begin{align}\label{Eq:lp inequ s t}
\begin{split}
\ell_p(\mu_s,\mu_t) &\geq \big\Vert l\circ(\eval_s,\eval_t)\big\Vert_{\Ell^p(\TGeo(\mms),\bdpi)}\\
&= (t-s)\,\big\Vert l\circ(\eval_0,\eval_1)\big\Vert_{\Ell^p(\TGeo(\mms),\bdpi)}\\
&= (t-s)\,\ell_p(\mu_0,\mu_1).
\end{split}
\end{align}
This together with  the reverse triangle inequality  \eqref{Eq:Reverse tau} implies
\begin{align*}
\ell_p(\mu_0,\mu_1) &\geq \ell_p(\mu_0,\mu_s) + \ell_p(\mu_s,\mu_t) + \ell_p(\mu_s,\mu_t)\\
&\geq \big[s + (t-s) + (1-t)\big]\,\ell_p(\mu_0,\mu_1)\\
&= \ell_p(\mu_0,\mu_1).
\end{align*}
Consequently, equality holds throughout \eqref{Eq:lp inequ s t}, and $(\eval_s,\eval_t)_\push\bdpi$ couples $\mu_s$ and $\mu_t$ optimally with respect to $\smash{\ell_p}$.
\end{remark}

\subsubsection{Narrow continuity and equivalence} The previous \autoref{Re:TL Geo to geo} establishes a hierarchy between the various types of $\ell_p$-geodesic (\autoref{Def:lp geo}). Under natural hypotheses, these are in fact equivalent, which is the content of \autoref{Cor:Equiv notions lp geo}. The latter is partly a measure version of \autoref{Le:Continuity geos}.

This result requires  the technical \autoref{Le:Intermediate pts geodesics}. Like \cite[Cor.~2.10]{mccann2020}, it follows from the characterization of equality cases in the reverse triangle inequality \eqref{Eq:Reverse lp} from \cite[Prop.~2.9]{mccann2020} as well as Lemmas  \ref{Le:Existence} and \ref{Le:Zt lemma}.

\begin{lemma}[Interpolants inherit compact support]\label{Le:Intermediate pts geodesics} Let $(\mu_t)_{t\in[0,1]}$ be a rough $\smash{\ell_p}$-geodesic, and let $0<t<1$. If $\mu_0$ and $\mu_1$ are compactly supported, so is $\mu_t$. More precisely,  its support $\supp\mu_t \subset Z_t(\supp\mu_0\times\supp\mu_1)$ is compact by \autoref{Le:Zt lemma}.
\end{lemma}

\begin{corollary}[Chronology makes rough $\ell_p$-geodesics narrowly continuous]\label{Cor:Equiv notions lp geo} Let $(\mu_t)_{t\in[0,1]}$ be a rough $\smash{\ell_p}$-geodesic  with compactly supported endpoints $\mu_0$ and $\mu_1$, where $0<p<1$. Suppose every $\smash{\ell_p}$-optimal coupling of $\mu_0$ and $\mu_1$ to be chronological. Then $(\mu_t)_{t\in[0,1]}$ is narrowly continuous. In particular, if
\begin{align*}
\supp\mu_0 \times\supp\mu_1 \subset \{l>0\},
\end{align*}
then $(\mu_t)_{t\in[0,1]}$ is a displacement $\smash{\ell_p}$-geodesic by \autoref{Le:Geodesics plan} below.
\end{corollary}

\begin{proof} Our defininition of  geodesy (or the  compactness of $\supp\mu_0\times\supp\mu_1$) ensures $\smash{\ell_p(\mu_0,\mu_1)<\infty}$. By \autoref{Le:Intermediate pts geodesics} and compactness of $\supp\mu_0$ and $\supp\mu_1$, the measure $\mu_t$ has support in $Z(\supp\mu_0\times \supp\mu_1)$ for every $0\leq t\leq 1$. The latter is compact by \autoref{Le:Zt lemma}, and notably independent of $t$.  \autoref{Le:Existence}  establishes existence of $\smash{\ell_p}$-optimal couplings of causally related, compactly supported marginals. The fact $\smash{\ell_p(\mu_0,\mu_1)>0}$ implied by our last hypothesis and our above observations will thus yield existence of $\smash{\ell_p}$-optimal couplings of all measures considered below.

We first claim narrow right-continuity of $(\mu_t)_{t\in(0,1)}$.  Let $(t_n)_{n\in\N}$ be any sequence decreasing to a given $0< t < 1$. By compactness of $Z(\supp\mu_0\times\supp\mu_1)$, $\smash{(\mu_{t_n})_{n\in\N}}$ has a narrow  limit $\nu\in\Prob_\comp(\mms)$ along a nonrelabeled  subsequence. We claim $\nu = \mu_t$. Suppose to the contrary $\nu\neq \mu_t$. By  \autoref{Th:GHy prob meas}, we have $\smash{\ell_p(\mu_t,\nu)\geq 0}$, 
and by \autoref{Le:Existence} an $\smash{\ell_p}$-optimal coupling $\pi_2$ of $\mu_t$ and $\nu$ exists. Now we construct a coupling $\pi$ of $\mu_0$ and $\mu_1$ as follows. Let $\pi_1$ as well as $\pi_3$ be $\smash{\ell_p}$-optimal couplings of $\mu_0$ and $\mu_t$ as well as $\nu$ and $\mu_1$, respectively. We glue these couplings $\pi_1$, $\pi_2$, and $\pi_3$ together, cf.~e.g.~\cite[Prop.~2.5]{cavalletti2020} or \cite[Prop.~2.9]{mccann2020} for this standard construction. We thus obtain a measure $\omega\in \Prob(\mms^4)$ with the properties
\begin{align*}
(\pr_1,\pr_2)_\push\omega &= \pi_1,\\
(\pr_2,\pr_3)_\push\omega &= \pi_2,\\
(\pr_3,\pr_4)_\push\omega &= \pi_3.
\end{align*}
Finally, we define $\pi := (\pr_1,\pr_4)_\push\omega$. By construction, $\pi$ is causal and obeys
\begin{align}\label{Eq:Long comp}
\begin{split}
\ell_p(\mu_0,\mu_1) &\geq \big\Vert l\big\Vert_{\Ell^p(\mms^2,\pi)}\\
&\geq \big\Vert l\circ (\pr_1,\pr_2) + l\circ(\pr_2,\pr_3) + l\circ(\pr_3,\pr_4)\big\Vert_{\Ell^p(\mms^4,\omega)}\\
&\geq \big\Vert l\big\Vert_{\Ell^p(\mms^2,\pi_1)} + \big\Vert l\big\Vert_{\Ell^p(\mms^2,\pi_2)} + \big\Vert l\big\Vert_{\Ell^p(\mms^2,\pi_3)}\\
&= \ell_p(\mu_0,\mu_t) + \ell_p(\mu_t,\nu) + \ell_p(\nu,\mu_1)\phantom{\big\Vert_{L^p}}\\
&\geq t\,\ell_p(\mu_0,\mu_1) + (1-t)\,\ell_p(\mu_0,\mu_1)\phantom{\big\Vert_{L^p}}\\
&= \ell_p(\mu_0,\mu_1);
\end{split}
\end{align}
here we used \eqref{Eq:Reverse tau} in the second, the reverse Hölder inequality in the third, and \eqref{Eq:lp geo path} (plus the inequality $\smash{\ell_p(\mu_t,\nu)\geq 0}$)
in the second last line. Therefore, all the inequalities in \eqref{Eq:Long comp} must be equalities. In particular, $\pi$ is an $\smash{\ell_p}$-optimal coupling of $\mu_0$ and $\mu_1$ and thus chronological by assumption. For $\omega$-a.e.~$(x_1,x_2,x_3,x_4)\in\mms^4$ we thus have $x_1 \ll x_4$; moreover, since equality holds throughout \eqref{Eq:Long comp}, by strict convexity (implied by $p<1$) again we obtain 
\begin{align}\label{Eq:Blabla}
\begin{split}
l(x_1,x_2) &= t\,l(x_1,x_4),\\
l(x_2,x_3) &= 0,\\
l(x_3,x_4) &= (1-t)\,l(x_1,x_4),
\end{split}
\end{align}
and therefore $x_1\ll x_2$ and $x_3 \ll x_4$. On the other hand, $\pi_2$ necessarily places mass outside the diagonal of $\smash{\mms^2}$ since $\nu\neq \mu_t$. In particular $\omega[E_\varepsilon] >0$ for some $\varepsilon > 0$ by monotone convergence, where $E_\varepsilon$ is the set of points $x\in \mms^4$ with $\met(x_2,x_3)\geq \varepsilon$; cf.~\autoref{Th:Compact Polish}. For every $x\in E_\varepsilon$, by \eqref{Eq:Blabla} following an $l$--geodesic from $x_1$ to $x_2$, jumping from $x_2$ to $x_3$, and going from there to $x_4$ by a further $l$-geodesic gives a discontinuous element of $\TGeo(\mms)$ from $x_1$ to $x_4$, a contradiction to \autoref{Le:Continuity geos}.
 
The case $t=0$ is dealt with similarly. Here the $l$-geodesics constructed above directly start with a jump. Thus $(\mu_t)_{t\in[0,1)}$ is right-continuous.

Analogously, we show narrow left-continuity of $(\mu_t)_{t\in(0,1]}$.

The last claim about displacement $\smash{\ell_p}$-geodesy  follows from \autoref{Le:Geodesics plan} once we have narrow  continuity at our disposal.
\end{proof}

\begin{remark}[Propagation  of moments]\label{Re:Do not have} If $\mu_0$ or $\mu_1$ do not have compact support, it is less clear how the moment condition  \eqref{Eq:Moment condition}, ensuring existence of $\smash{\ell_p}$-optimal couplings \cite[Prop.~2.3]{cavalletti2020}, propagates throughout rough $\smash{\ell_p}$-geodesics; such a property  might follow from the construction of a nonsmooth Hopf-Lax semigroup. 
\end{remark}

\subsubsection{Lifting the geometry on $\mms$ to $\scrP(\mms)$} We turn to the lifting procedure for $\smash{\ell_p}$-geodesics stated in \autoref{Le:Geodesics plan}. Some of the statements hold in higher generality, but we will not need these extensions. On the other hand, unlike \cite{braun2022}  where \ref{La:05} below was available for strongly timelike $p$-dualizable pairs $(\mu_0,\mu_1)$ --- based on uniform Lipschitz continuity of causal curves and reparametrization  --- our general compactness result from  \autoref{Cor:Cpt TGeo} only guarantees the necessary tightness properties away from the lightcone. 

Under curvature and nonbranching assumptions, this lifting result holds for more general marginal measures compared to \eqref{Eq:Supp}, cf.~\autoref{Cor:Lifting nonbr}.

Recall the definition \eqref{Eq:Restr} of the restriction map $\smash{\Restr_s^t}$ for $0\leq s<t\leq 1$.
 
\begin{proposition}
[
Criteria for an $\ell_p$-geodesic to be represented by a plan]
\label{Le:Geodesics plan} Let $\smash{\mu_0,\mu_1\in\scrP(\mms)}$. Then the following hold.
\begin{enumerate}[label=\textnormal{\textcolor{black}{(}\roman*\textcolor{black}{)}}]
\item\label{La:01} For every chronological $\ell_p$-optimal coupling $\smash{\pi\in\Pi_\ll(\mu_0,\mu_1)}$, there exists a timelike $\smash{\ell_p}$-optimal dynamical plan $\smash{\bdpi\in\OptTGeo_p(\mu_0,\mu_1)}$ with
\begin{align*}
\pi = (\eval_0,\eval_1)_\push\bdpi.
\end{align*}
\item\label{La:02} If $(\mu_0,\mu_1)$ is timelike $p$-dualizable, there exists at least one displacement $\smash{\ell_p}$-geo\-desic from $\mu_0$ to $\mu_1$.
\item\label{La:03} Given any $0\leq s<t\leq 1$, the inclusion $\smash{\bdpi\in\OptTGeo_p(\mu_0,\mu_1)}$ implies the restriction $\smash{(\Restr_s^t)_\push\bdpi}$ to belong to $\smash{\OptTGeo_p(\mu_s,\mu_t)}$, where
\begin{align*}
\mu_s &:= (\eval_s)_\push\bdpi,\\
\mu_t &:= (\eval_t)_\push\bdpi.
\end{align*}
\item\label{La:04} Let $\smash{\bdpi\in\OptTGeo_p(\mu_0,\mu_1)}$, and let $\bdsigma\in\scrP(\Cont([0,1];\mms))$ be a nonzero measure with $\bdsigma\leq \bdpi$. Define the probability measures
\begin{align*}
\sigma_0 &:= \bdsigma\big[\Cont([0,1];\mms)\big]^{-1}\,(\eval_0)_\push\bdsigma,\\
\sigma_1 &:= \bdsigma\big[\Cont([0,1];\mms)\big]^{-1}\,(\eval_1)_\push\bdsigma.
\end{align*}
Then $\smash{\bdsigma\big[\Cont([0,1];\mms)\big]^{-1}\,\bdsigma\in\OptTGeo_p(\sigma_0,\sigma_1)}$.
\end{enumerate}

Now we additionally assume compactness of $\supp\mu_0$ and $\supp\mu_1$ as well as
\begin{align}\label{Eq:Supp}
\supp\mu_0\times\supp\mu_1 \subset \{l>0\},
\end{align}
and that $0<p<1$. Then the following hold.
\begin{enumerate}[label=\textnormal{\textcolor{black}{(}\roman*\textcolor{black}{)}}]\setcounter{enumi}{4}
\item\label{La:05} Every $\smash{\ell_p}$-geodesic from $\mu_0$ to $\mu_1$ is in fact a displacement $\smash{\ell_p}$-geodesic. 
\item\label{La:06} Let $\smash{(\mu_0^n)_{n\in\N}}$ and $\smash{(\mu_1^n)_{n\in\N}}$ be sequences in $\scrP_\comp(\mms)$   converging narrowly to $\mu_0$ and $\mu_1$, respectively. Then every sequence $\smash{(\bdpi^n)_{n\in\N}}$ of timelike $\smash{\ell_p}$-optimal dy\-namical plans $\smash{\bdpi^n\in\OptTGeo_p(\mu_0^n,\mu_1^n)}$ has a narrow accumulation point. 
\item\label{La:07} If, in addition,  there are compact sets $C_0,C_1\subset\mms$ such that $\supp\mu_0^n\subset C_0$ and $\supp\mu_1^n\subset C_1$ for every $n\in\N$ --- in other words, $\smash{(\mu_0^n)_{n\in\N}}$ and $\smash{(\mu_1^n)_{n\in\N}}$ are uniformly compactly supported --- and
\begin{align*}
\ell_p(\mu_0,\mu_1) \leq \liminf_{n\to\infty} \ell_p(\mu_0^n,\mu_1^n),
\end{align*}
then $\smash{\ell_p(\mu_0^n,\mu_1^n)\to\ell_p(\mu_0,\mu_1)}$ as $n\to \infty$, and any accumulation point from \ref{La:06} belongs to $\smash{\OptTGeo_p(\mu_0,\mu_1)}$. This holds in particular if the endpoint marginal sequences are constant, in which case the previous statement implies narrow compactness of $\smash{\OptTGeo_p(\mu_0,\mu_1)}$.
\end{enumerate}
\end{proposition}

\begin{proof} We first establish \ref{La:01}. Using causal closedness (\autoref{Pr:Closed}), we deduce the multivalued map sending a chronologically related pair $\smash{(x,y)\in\mms^2}$ to the set of all $\gamma\in\TGeo(\mms)$ connecting $x$ to $y$ to have closed graph. The statement then follows from standard measurable selection arguments, cf.~e.g.~\cite{ambrosio2011, villani2009}.

To see \ref{La:02}, apply \ref{La:01} to a timelike $p$-dualizing coupling of $\mu_0$ and $\mu_1$.

Item \ref{La:03} easily follows from \autoref{Re:TL Geo to geo}.

We turn to \ref{La:04}. The inequality $\bdsigma\leq \bdpi$ is stable under push-forwards, which shows the sub-probability measure $(\eval_0,\eval_1)_\push\bdsigma$ to be chronological. By the restriction property of $\smash{\ell_p}$-optimal couplings \cite[Lem.~2.10]{cavalletti2020}, $(\eval_0,\eval_1)_\push\bdsigma$ is thus an $\smash{\ell_p}$-optimal coupling of its marginals $\sigma_0$ and $\sigma_1$. Since $\bdsigma\leq \bdpi$, the measure $\bdsigma$ is concentrated on $\TGeo(\mms)$, and so is its normalized version.

We pass over to \ref{La:06}. We claim tightness of $\smash{(\bdpi^n)_{n\in\N}}$. Given any $\varepsilon > 0$, by narrow convergence of $(\mu_0^n)_{n\in\N}$ and $(\mu_1^n)_{n\in\N}$ there exist compact sets $\smash{C_0,C_1\subset\mms}$ with
\begin{align*}
\sup_{n\in\N} \mu_0^n\big[C_0^\sfc\big] &\leq \varepsilon,\\
\sup_{n\in\N} \mu_1^n\big[C_1^\sfc\big] &\leq \varepsilon.
\end{align*}
Recall the set $G:=\TGeo(\mms) \cap \eval_0^{-1}(C_0)\cap \eval_1^{-1}(C_1)$ is precompact thanks to  \autoref{Cor:Cpt TGeo} and \eqref{Eq:Supp}. Observe that
\begin{align*}
\sup_{n\in\N} \bdpi^n\big[\Cont([0,1];\mms)\setminus G\big] &\leq \sup_{n\in\N} \bdpi^n\big[\eval_0^{-1}(C_0^\sfc)\big] + \sup_{n\in\N} \bdpi^n\big[\eval_1^{-1}(C_1^\sfc)\big]\\
&= \sup_{n\in\N} \mu_0\big[C_0^\sfc\big] + \sup_{n\in\N} \mu_1\big[C_1^\sfc\big]\\
&\leq 2\,\varepsilon.
\end{align*}
This consideration shows the desired tightness by \autoref{Cor:Cpt TGeo}. 
By Prokhorov's theorem, $\smash{(\bdpi^n)_{n\in\N}}$ has an accumulation point $\smash{\bdpi\in\Prob(\Cont([0,1];\mms)}$. 

Next, we prove $\smash{\bdpi\in \OptTGeo_p(\mu_0,\mu_1)}$ under the assumptions of \ref{La:07}. It suffices to show $\smash{\ell_p}$-optimality of $(\eval_0,\eval_1)_\push\bdpi$; indeed, $\bdpi$ is concentrated on the uniform closure of $\smash{\TGeo(\mms)}$ by Alexandrov's theorem, and $\smash{\ell_p}$-optimality of $(\eval_0,\eval_1)_\push\bdpi$ would imply chronology of the latter by assumption, and thus $\smash{\bdpi[\TGeo(\mms)]=1}$. To this aim, uniform boundedness of $\smash{l_+}$ on $J(C_0,C_1)$ combined with \cite[Lem.~4.3]{villani2009} yield
\begin{align*}
\ell_p(\mu_0,\mu_1) &\geq \big\Vert l \big\Vert_{\Ell^p(\mms^2,(\eval_0,\eval_1)_\push\bdpi)}\\
&= \big\Vert l_+ \big\Vert_{\Ell^p(\mms^2,(\eval_0,\eval_1)_\push\bdpi)}\\
&\geq \limsup_{n\to\infty}\big\Vert l_+ \big\Vert_{\Ell^p(\mms^2,(\eval_0,\eval_1)_\push\bdpi^n)}
\\
&= \limsup_{n\to\infty} \ell_p(\mu_0^n,\mu_1^n)\\
&\geq \liminf_{n\to\infty} \ell_p(\mu_0^n,\mu_1^n)\\
&\geq \ell_p(\mu_0,\mu_1).
\end{align*}
This confirms the limit asserted and proves the desired $\smash{\ell_p}$-optimality of $(\eval_0,\eval_1)_\push\bdpi$.

Lastly, item \ref{La:05} requires most work, but is performed in analogy with the usual dyadic argument, cf.~e.g.~\cite[Thm.~2.10]{ambrosio2011} or \cite[Thm.~7.21, Cor.~7.22]{villani2009}. We build a sequence $\smash{(\bdpi^n)_{n\in\N}}$ in $\smash{\OptTGeo_p(\mu_0,\mu_1)}$ as follows. Given any $n\in\N$ and any $k=1,\dots,2^n$, we first observe that $\mu_{(k-1)\,2^{-n}}$ and $\mu_{k\,2^{-n}}$ admit a causal coupling since they lie on an $\smash{\ell_p}$-geodesic. These measures  also have support in a compact set which is independent of $k$ or $n$ by \autoref{Le:Intermediate pts geodesics}. Thus, by \cite[Prop.~2.3]{cavalletti2020} these admit an $\smash{\ell_p}$-optimal coupling $\smash{\pi_k^n\in\Pi_{\leq}(\mu_{(k-1)\,2^{-n}}, \mu_{k\,2^{-n}})}$.

Using \autoref{Le:Gluing lemma} we find $\omega^n\in \Prob(\mms^{2^n+1})$ such that, for every $k=1,\dots,2^n$,
\begin{align}\label{Eq:Pikn}
(\pr_{k-1},\pr_k)_\push \omega^n = \pi_k^n.
\end{align}
The coupling $\smash{\pi^n := (\pr_0,\pr_{2^n})_\push\omega^n}$ of $\mu_0$ and $\mu_1$ is causal and obeys 
\begin{align}\label{Eq:Coupling causal}
\begin{split}
\ell_p(\mu_0,\mu_1) &\geq \big\Vert l\big\Vert_{\Ell^p(\mms^2,\pi^n)}\\
&\geq \Big\Vert\! \sum_{k=1}^{2^n} l\circ (\pr_{k-1},\pr_k)\Big\Vert_{\Ell^p(\mms^{2^n+1},\omega^n)} \\
&\geq \sum_{k=1}^{2^n} \big\Vert l\big\Vert_{\Ell^p(\mms^2,\pi_k^n)}\\
&= \sum_{k=1}^{2^n} \ell_p(\mu_{(k-1)\,2^{-n}}, \mu_{k\,2^{-n}})\\
&= \ell_p(\mu_0,\mu_1);
\end{split}
\end{align}
here we used \eqref{Eq:Reverse tau} in the second line, the reverse Hölder inequality in the third line, and $\smash{\ell_p}$-geodesy in the last line. Therefore, equalities must hold throughout the above chain of inequalities. This implies $\smash{\pi^n}$ is an $\smash{\ell_p}$-optimal coupling of $\mu_0$ and $\mu_1$, hence chronological. Moreover, by \cite[Prop.~2.9]{mccann2020} $\omega^n$-a.e.~$(x_0,x_1,\dots,x_{2^n})\in\mms^{2^n+1}$ satisfies the following identity for every $k=1,\dots,2^n$:
\begin{align}\label{Eq:2^n inequ}
l(x_{k-1}, x_k) = 2^{-n}\,l(x_0,x_{2^n}).
\end{align}
In turn, this yields chronology of $\smash{\pi_k^n}$ for every such $k$; by item \ref{La:01}, we consequently find a plan $\smash{\bdpi_k^n\in\OptTGeo_p(\mu_{(k-1)\,2^{-n}}, \mu_{k\,2^{-n}})}$ with
\begin{align*}
\pi_k^n = (\eval_0,\eval_1)_\push\bdpi_k^n.
\end{align*}
Owing to \eqref{Eq:Pikn} and \eqref{Eq:2^n inequ} we thus construct $\smash{\bdpi^n\in\Prob(\Cont([0,1];\mms))}$ such that 
\begin{itemize}
\item for every $k=1,\dots,2^n$, we have $\smash{(\Restr_{(k-1)\,2^{-n}}^{k\,2^{-n}})_\push\bdpi^n = \bdpi_k^n}$,
\item $\smash{(\eval_0,\eval_1)_\push\bdpi}$ is a chronological $\smash{\ell_p}$-optimal coupling of $\mu_0$ and $\mu_1$, and
\item $\bdpi^n$ is concentrated on $\TGeo(\mms)$.
\end{itemize}
In other words, we have $\smash{\bdpi^n\in\OptTGeo_p(\mu_0,\mu_1)}$.

For every $n\in\N$, by construction the timelike $\smash{\ell_p}$-optimal dynamical plan $\smash{\bdpi^n}$ is concentrated on the set of all $\smash{\gamma\in \TGeo(\mms)}$ with $\gamma_0\in\supp\mu_0$ and $\gamma_1\in\supp\mu_1$. An argument as for \ref{La:06}  above (note that singletons of probability measures are automatically tight) shows tightness of the sequence $\smash{(\bdpi^n)_{n\in\N}}$, and thus a nonrelabeled subsequence of it converges narrowly to some  $\smash{\bdpi\in\Prob(\Cont([0,1];\mms))}$. In particular, the endpoint couplings $((\eval_0,\eval_1)_\push\bdpi^n)_{n\in\N}$ converge narrowly to $(\eval_0,\eval_1)_\push\bdpi$, which is a coupling of $\mu_0$ and $\mu_1$. Alexandrov's theorem yields $\bdpi$ to be concentrated on causal curves. Thus, stability of $\smash{\ell_p}$-optimal couplings in the form \cite[Thm.~2.14]{cavalletti2020} in fact  saturates $\smash{\ell_p}$-optimality and hence chronology of $(\eval_0,\eval_1)_\push\bdpi$.

We claim that $\bdpi$ represents $(\mu_t)_{t\in[0,1]}$. By construction, we have $(\eval_0)_\push\bdpi = \mu_0$ and $(\eval_1)_\push\bdpi = \mu_1$. Given any dyadic number $0<t<1$, passing to the limit in the identity $\smash{\mu_t = (\eval_t)_\push\bdpi^n}$, valid for every sufficiently large $n\in\N$, yields $\mu_t= (\eval_t)_\push\bdpi$ for every such $t$. Since $(\mu_t)_{t\in[0,1]}$ is narrowly continuous, and since $\bdpi$ is concentrated on continuous curves, both sides of the previous identity depend continuously on $t$ with respect to narrow convergence. This yields $\mu_t = (\eval_t)_\push\bdpi$ for every $0< t< 1$.

Since $\smash{\bdpi[\TGeo(\mms)]=1}$ by \ref{La:07}, this terminates the proof.
\end{proof}

By using the first part of \autoref{Le:Geodesics plan}, the proof of the following result is analogous to \cite[Lem.~3.1]{braun2023}, see also \cite[Prop.~5.5]{mccann2020}. The difference to point out here compared to \cite{braun2023} is that with our notion of $\smash{\ell_p}$-geodesics, propagation of bare timelike $p$-dualizability is unclear, but not expected to hold in general.

\begin{lemma}[Criteria for propagation of (strong) timelike $p$-dualizabilty]
\label{Le:Propagates} Let $(\mu_t)_{t\in[0,1]}$ be an $\smash{\ell_p}$-geodesic. 
\begin{enumerate}[label=\textnormal{(\roman*)}]
\item If the pair $(\mu_0,\mu_1)$ is strongly timelike $p$-dualizable, so is $(\mu_s,\mu_t)$ for every $0\leq s < t\leq 1$. 
\item If $(\mu_t)_{t\in[0,1]}$ constitutes a displacement $\smash{\ell_p}$-geodesic, then timelike $p$-dualizabi\-li\-ty of $(\mu_0,\mu_1)$ implies that  of $(\mu_s,\mu_t)$ for every $0\leq s<t\leq 1$.
\end{enumerate}
\end{lemma}

\subsubsection{Compactness}

\begin{theorem}[Accumulation of rough $\ell_p$-geodesics at rational times]\label{Pr:Compactness lp geos} Assume \autoref{Ass:REG}. For every $n\in\N$, let $\smash{(\mu_t^n)_{t\in \N}}$ be a given rough $\smash{\ell_p}$-geodesic, such that the sequences $\smash{(\mu_0^n)_{n\in\N}}$ and $\smash{(\mu_1^n)_{n\in\N}}$ are uniformly compactly supported, say in compact sets $C_0,C_1\subset\mms$, respectively. Assume  accumulation points $\mu_0,\mu_1\in\scrP_\comp(\mms)$ of the two sequences satisfy 
\begin{align}\label{Eq:lsc crit}
0 <\ell_p(\mu_0,\mu_1) &\leq \liminf_{n\to\infty}\ell_p(\mu_0^n,\mu_1^n).
\end{align}
Then there is a rough $\smash{\ell_p}$-geodesic $(\mu_t)_{t\in[0,1]}$ from $\mu_0$ to $\mu_1$ such that up to a nonrelabeled subsequence, $\smash{(\mu_t^n)_{n\in\N}}$ converges narrowly to $\mu_t$ for every rational $0\leq t\leq 1$.
\end{theorem}

As in \autoref{Le:Geodesics plan}, the lower semicontinuity  \eqref{Eq:lsc crit} in \autoref{Pr:Compactness lp geos} implies the limit exists and equality holds \cite[Lem.~2.11, Thm.~2.14]{cavalletti2020}.

\begin{proof}[Proof of \autoref{Pr:Compactness lp geos}] The construction is similar to the first part of the proof of \autoref{Cor:Cpt TGeo}, thus we only outline it. By Prokhorov's theorem and \autoref{Le:Zt lemma}, we may and will assume the existence of a family $(\mu_t)_{t\in [0,1]\cap \Q}$ in $\scrP(\mms)$ with support in the common compact subset $Z(C_0\times C_1)\subset\mms$ such that up to a nonrelabeled subsequence, the sequence $(\mu_t^n)_{n\in\N}$ converges narrowly to $\mu_t$ for every rational $0\leq t\leq 1$. Given any irrational $0\leq t\leq 1$, let $(t_i)_{i\in\N}$ be defined by $\smash{t_i := \lceil t\,2^i\rceil\,2^{-i}}$. As above, we can thus define $\mu_t\in\scrP_\comp(\mms)$ as a subsequential narrow limit of $\smash{(\mu_{t_i})_{i\in\N}}$ --- in fact, the definition of $\mu_t$ does not depend on the choice of monotone approximating sequence by antisymmetry, cf.~\autoref{Pr:Order}; compare with the proof of \autoref{Cor:Cpt TGeo}. 

We claim that the collection $(\mu_t)_{t\in[0,1]}$ is a rough $\smash{\ell_p}$-geodesic from $\mu_0$ to $\mu_1$. To this aim, we first show every fixed rational $0\leq s <t\leq 1$ to satisfy
\begin{align}\label{Eq:Shows}
\ell_p(\mu_s,\mu_t) \geq (t-s)\,\ell_p(\mu_0,\mu_1).
\end{align}
The construction of $\mu_s$ and $\mu_t$, tightness of couplings, Alexandrov's theorem, and closedness of $\{l\geq 0\}$ ensure   $\smash{\Pi_\leq(\mu_s,\mu_t)\neq \emptyset}$. Since $\smash{l_+}$ is bounded on $Z(C_0\times C_1)$, again by Alexandrov's theorem and tightness of couplings some $\pi_{s,t}\in\Pi_\leq(\mu_s,\mu_t)$ is the narrow limit of a sequence $(\pi_{s,t}^n)_{n\in\N}$ of $\smash{\ell_p}$-optimal couplings $\smash{\pi_{s,t}^n\in \Pi_\leq(\mu_s^n,\mu_t^n)}$, up to a nonrelabeled subsequence. This already yields \eqref{Eq:Shows} because
\begin{align*}
\ell_p(\mu_s,\mu_t) &\geq \big\Vert l\big\Vert_{\Ell^p(\mms^2,\pi_{s,t})}\\
&= \lim_{n\to\infty} \big\Vert l\big\Vert_{\Ell^p(\mms^2,\pi_{s,t}^n)} \\
&\geq (t-s)\,\liminf_{n\to\infty}\ell_p(\mu_0^n,\mu_1^n)\\
&\geq (t-s)\,\ell_p(\mu_0,\mu_1).
\end{align*}

Finally, we show \eqref{Eq:Shows} to hold for every $0\leq s < t\leq 1$ --- as in \autoref{Re:TL Geo to geo}, this will imply $\smash{\ell_p}$-geodesy of $(\mu_t)_{t\in[0,1]}$. But by ``causally monotone continuity of $\smash{\ell_p}$'' \cite[Thm.~2.14]{cavalletti2020}, it is clear that \eqref{Eq:Shows} extends to every $0\leq s < t\leq 1$. 
\end{proof}

\begin{remark}[Narrow continuity] If every $\smash{\ell_p}$-optimal coupling of the endpoint marginals $\mu_0$ and $\mu_1$ in \autoref{Pr:Compactness lp geos} is chronological, by \autoref{Cor:Equiv notions lp geo} the limit rough $\smash{\ell_p}$-geodesic is an $\smash{\ell_p}$-geodesic, and it is uniquely determined by its values on rational $0\leq t\leq 1$.

More generally, if a limit curve $(\mu_t)_{t\in[0,1]}$ of rough $\smash{\ell_p}$-geodesics $(\mu_t^n)_{t\in[0,1]}$, where $n\in\N$, according to \autoref{Pr:Compactness lp geos} above is narrowly continuous, then  $\smash{(\mu_t^n)_{n\in\N}}$ converges narrowly to $\mu_t$ for \emph{every} $0\leq t\leq 1$. This follows from antisymmetry of the causal relation $\preceq$ on $\Prob(\mms)$ discussed in \autoref{Sub:Causality pm}; compare with the proof of \autoref{Cor:Cpt TGeo}. Indeed, every subsequence of $(\mu_t^n)_{n\in\N}$ has a further nonrelabeled subsequence converging to some $\nu\in \Prob_\comp(\mms)$ by Prokhorov's theorem. For $s<t<r$ with $s,r\in\Q$, passing to the limit as $n\to\infty$ in the chain $\smash{\mu_s^n\preceq\mu_t^n\preceq\mu_r^n}$ yields $\smash{\mu_s \preceq\nu\preceq\mu_r}$. Narrow continuity of $(\mu_t)_{t\in[0,1]}$ gives $\smash{\mu_t \preceq\nu\preceq\mu_t}$, and thus $\mu_t=\nu$ by antisymmetry of $\preceq$ (\autoref{Pr:Order}).
\end{remark}

\subsection{Taking the reference measure $\meas$ into account}

\subsubsection{Timelike $p$-essential nonbranching} 
The following condition, 
which relaxes timelike nonbranching \cite[Def.~1.10]{cavalletti2020}  
to an ``a.e.~requirement'', has been introduced in \cite[Def.~2.21, Rem.~2.22]{braun2022} following \cite{rajala2014}. It involves the reference measure $\meas$ from the metric measure spacetime $\scrM$ discussed at \eqref{Eq:Space}.

A subset $G\subset\TGeo(\mms)$ is \emph{timelike nonbranching} if for every $0\leq s < t \leq 1$, the map $\smash{\Restr_s^t\big\vert_G}$ is injective, where $\smash{\Restr_s^t}$ is from \eqref{Eq:Restr}.

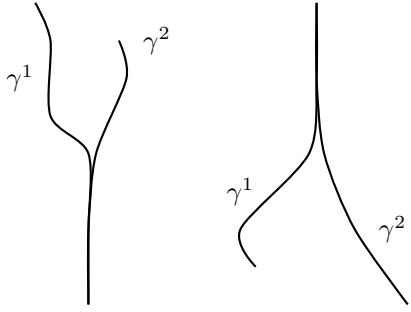
\begin{figure}
\centering
\begin{tikzpicture}
\draw[thick] plot [smooth] coordinates {(0,0) (0,1)  (0,2) (-0.5,2.5) (-0.5,3.5) (-0.7,4)};
\draw[thick] plot [smooth] coordinates {(0,0) (0,1)  (0.1,2) (0.5,3) (0.4,3.5)};
\draw[thick] plot [smooth] coordinates {(3,4) (3,3) (2.9,2) (2,1) (2.2,0.5)};
\draw[thick] plot [smooth] coordinates {(3,4) (3,3) (3.1,2) (3.5,1) (4.2,0)};
\node at (-0.9,3) {$\gamma^1$};
\node at (0.9,3.5) {$\gamma^2$};
\node at (2,1.5) {$\gamma^1$};
\node at (4,1) {$\gamma^2$};
\end{tikzpicture}
\caption{Timelike (forward and backward) branching.}\label{Fig:TLBranching}
\end{figure}

\begin{definition}[Timelike essential nonbranching]\label{Def:TL nonbranching} We term the space  $\scrM$ \emph{timelike $p$-essentially nonbranching} if for every  $\bdpi\in\smash{\OptTGeo_p(\scrP_\comp^\ac(\mms,\meas)^2)}$ there exists a timelike nonbranching Borel set $G\subset\TGeo(\mms)$ with $\bdpi[G]=1$.
\end{definition}

Note that timelike $p$-essential nonbranching is a property of 
$\meas$ rather than $\mms$
(and that $\mms = \supp\meas$ is not assumed).

Despite our more general approaches to geodesics on $\mms$ and $\smash{\scrP_\comp(\mms)}$, respectively, the proof of  \cite[Lem.~2.23]{braun2023} carries over with no essential change.

\begin{lemma}[Mutually singular interpolants]
\label{Le:Mutually singular} Assume $\scrM$ is timelike $p$-essentially nonbranching. Decompose $\smash{\bdpi\in\OptTGeo_p(\scrP_\comp^\ac(\mms,\meas)^2)}$ as $\smash{\bdpi = \lambda_1\,\bdpi^1 + \dots + \lambda_n\,\bdpi^n}$ for certain $0<\lambda_1,\dots,\lambda_n<1$ and $n\in\N$. Lastly, assume $\smash{(\eval_t)_\push\bdpi^1,\dots,(\eval_t)_\push\bdpi^n}$ to be $\meas$-absolutely continuous for some $0<t<1$, and mutual singularity of $\smash{\bdpi^1,\dots,\bdpi^n}$. Then the measures $(\eval_t)_\push\bdpi^1,\dots,(\eval_t)_\push\bdpi^n$ are mutually singular as well.
\end{lemma}

\subsubsection{Boltzmann entropy}  We define the \emph{Boltzmann} {(or \emph{Boltzmann--Shannon}) entropy} $\Ent_\meas\colon \scrP(\mms)\to [-\infty,\infty]$ by decomposing $\mu =\rho\,\meas + \mu_\perp$ into its absolutely continuous and singular parts with respect to $\meas$ and setting 
\begin{align*}
\Ent_\meas(\mu) &:= \begin{cases}
\displaystyle\int_\mms \rho\log\rho\d\meas & \textnormal{if 
$(\rho\log\rho)_-\in\Ell^1(\mms,\meas)$ and $\mu_\perp(\mms)=0$},
\\ +\infty & \textnormal{if 
$(\rho\log\rho)_-\in\Ell^1(\mms,\meas)$ and $\mu_\perp(\mms)>0$},
\\-\infty & \textnormal{if } (\rho\log\rho)_-\not\in\Ell^1(\mms,\meas).
\end{cases}
\end{align*}
We denote by $\Dom(\Ent_\meas)$ the finiteness domain of $\Ent_\meas$. We will mostly consider $\Ent_\meas$ on compactly supported measures. In particular, on $\scrP_\comp(\mms)$ the Boltzmann entropy does not attain the value $-\infty$, since
\begin{align*}
\Ent_\meas(\mu) \geq -\log\meas\big[\!\supp\mu\big] > -\infty
\end{align*}
for every $\smash{\mu\in\scrP_\comp(\mms)}$ by Jensen's inequality and the Radon property of $\meas$.

We occasionally use the following narrow lower semicontinuity property of $\Ent_\meas$, cf.~e.g.~\cite[Lem.~9.4.3]{ambrosio2008} or \cite[Lem.~4.1]{sturm2006a}. If $(\mu_n)_{n\in\N}$ is a given uniformly compactly supported sequence in $\scrP(\mms)$ which converges narrowly to $\mu\in\scrP(\mms)$,  then
\begin{align*}
\Ent_\meas(\mu) \leq \liminf_{n\to\infty}\Ent_\meas(\mu_n).
\end{align*}

Given any $N>0$, we also define the functional $\scrU_N \colon \scrP(\mms) \to [0,\infty]$ by
\begin{align}\label{Eq:UN DEF}
\scrU_N := \rme^{-\Ent_\meas(\mu)/N}.
\end{align}
We employ the conventions $\rme^{-\infty}:=0$ and $\rme^\infty := \infty$. Note that $\scrU_N$ is narrowly upper semicontinuous in the same sense as $\Ent_\meas$, and for every $\mu\in\Prob_\comp(\mms)$,
\begin{align}\label{Eq:UN Jensen}
\scrU_N(\mu) \leq \meas\big[\!\supp\mu\big]^{1/N}.
\end{align}

\section{Timelike curvature-dimension condition}\label{Ch:TCD}

From now on, unless explicitly stated otherwise we fix an exponent $0<p<1$, a lower semi\-continuous  function $k\colon\mms\to\R$, and a real number $N >1$\footnote{The entropic condition we will introduce actually only requires $N>0$. For other timelike curvature-dimension conditions such as in \cite{braun2022} the situation $N=1$ is technically slightly different from $N>1$, but can be understood as the limiting case of the latter; cf.~e.g.~\autoref{Pr:Consistency}.}. The lower boundedness of $k$ is \emph{not} assumed.

\subsection{Distortion coefficients along $l$-geodesics} Now we describe how to generalize the distortion coefficients employed to extend the timelike-curvature dimension condition \cite{braun2022, cavalletti2020} to the variable case. We refer to Ketterer \cite{ketterer2015,ketterer2017} for more  details.

\subsubsection{Continuous potentials on real intervals}\label{Sub:Cts potentials} Given any $L>0$, let $\kappa\colon[0,L]\to\R$ be a  continuous function. The \emph{generalized sine function} $\SIN_\kappa\colon[0,L]\to \R$ is the unique solution to the ODE $f'' + \kappa\,f = 0$ satisfying $f(0) = 0$ and $f'(0) = 1$. This solution depends uniformly continuously on $\kappa$. 

\begin{definition}[Distortion coefficients on the half-line $[0,\infty)$]
\label{Def:Dist coeff} For $0\leq t\leq 1$, the \emph{generalized distortion coefficient} $\smash{\sigma_\kappa^{(t)}}\colon[0,L]\to [0,\infty]$ is defined by
\begin{align}
\sigma_{\kappa}^{(t)}(\theta) := \begin{cases} \displaystyle\frac{\SIN_\kappa(t\,\theta)}{\SIN_\kappa(\theta)} & \textnormal{if } \SIN_\kappa(r) > 0 \textnormal{ for every } 0<r\leq \theta,\\
\infty & \textnormal{otherwise}.
\end{cases}
\label{distortion on R}
\end{align}
\end{definition}


\begin{example}[Constant potentials]\label{Ex:Constant pot} If $\kappa$ is constant, the distortion coefficients from \autoref{Def:Dist coeff} take the following customary form:
\begin{align*}
\sigma_{\kappa}^{(t)}(\theta) &= \begin{cases} \displaystyle\frac{\sin(\sqrt{\kappa}\,t\,\theta)}{\sin(\sqrt{\kappa}\,\theta)} & \textnormal{if } 0 < \kappa\,\theta^2 <\pi^2,\\
\infty & \textnormal{if } \kappa\,\theta^2 \geq \pi^2,\\
t &  \textnormal{if }\kappa\,\theta^2 = 0,\\
\displaystyle\frac{\sinh(\sqrt{-\kappa}\,t\,\theta)}{\sinh(\sqrt{-\kappa}\,\theta)} & \textnormal{if }\kappa < 0.
\end{cases}
\end{align*}
\end{example}

\begin{remark}[Scaling]\label{Re:Scaling prop} Given any $0 \leq \theta\leq L$, we have  either $\smash{\sigma_\kappa^{(t)}(\theta)= \infty}$ for every $0\leq t\leq 1$, or  $\smash{\sigma_\kappa^{(t)}(\theta)< \infty}$ for every $0\leq t\leq 1$, cf.~\autoref{Le:Properties} below. 
 
 In the latter case, the function $u\colon [0,1]\to [0,\infty)$ given by $\smash{u(t) := \sigma_\kappa^{(t)}(\theta)}$  obeys $u(0) = 0$, $u(1) = 1$, and the following ODE for every $0<t<1$:
\begin{align}\label{Eq:ODE u}
u''(t) + \kappa(t\,\theta)\,\theta^2\,u(t) = 0.
\end{align}
Uniqueness implies the useful scaling property  
\begin{align}\label{Eq:Scaling prop sigma}
\sigma_{\kappa}^{(t)}(\theta) = \sigma_{\kappa \theta^2}^{(t)}(1),
\end{align}
where  $\kappa\,\theta^2$ is understood as the function $\tilde{\kappa}\colon[0,1]\to \R$ with
\begin{align*}
\tilde{\kappa}(t) := \kappa(t\,\theta)\,\theta^2.
\end{align*}
\end{remark}

\subsubsection{Continuous potentials along $l$-geodesics} Now we adapt these notions to a continuous function $k\colon \mms\to\R$. The natural approach is to consider $k$ along an appropriate timelike curve. 

Let $\gamma\in\TGeo(\mms)$ be fixed, and define
\begin{align}\label{Eq:speed}
\vert\dot\gamma\vert := l(\gamma_0,\gamma_1).
\end{align}
We use this notation only for formal reasons. However, it should be compared to the \emph{metric speed} of absolutely continuous curves in metric spaces, cf.~e.g.~\cite[Thm. 1.1.2]{ambrosio2008}. In fact,   for causal curves $\gamma$ which need not be geodesic, we define a notion $|\dot \gamma|:[0,1]\to\R$ of 
\emph{causal speed} in \cite{beran2023} which reduces to \eqref{Eq:speed} if $\gamma\in\TGeo(\mms)$; in this more general setting, $\vert\dot\gamma\vert$ need not be constant; only its average over $[0,1]$ is dominated by $l(\gamma_0,\gamma_1)$, and we can then consider the $l$-arclength  reparametrization $\smash{\bar\gamma\colon[0,\vert\dot\gamma\vert]\to\R}$ of a rectifiable curve $\gamma$ given by \autoref{Pr:Reparam}. For $\gamma\in\TGeo(\mms)$ --- the only relevant case here --- it has the simple form
\begin{align*}
\bar\gamma_r := \gamma_{r/\vert\dot\gamma\vert}.
\end{align*}

In view of \autoref{Re:Scaling prop},  define the continuous  forward and backward functions  $\smash{k_\gamma^\pm\colon[0,\vert\dot\gamma\vert]\to\R}$ by 
\begin{align}\label{Eq:kgamma defn}
\begin{split}
k_\gamma^-(t\,\vert\dot\gamma\vert) &:= k(\gamma_{1-t}),\\
k_\gamma^+(t\,\vert\dot\gamma\vert) &:= k(\gamma_t),
\end{split}
\end{align}
which depend continuously, with respect to uniform convergence, on $\gamma \in\TGeo(\mms)$ for fixed $0 \leq t\leq 1$. Equivalently, we can express \eqref{Eq:kgamma defn} as
\begin{align*}
k_\gamma^-(r) &:= k(\bar\gamma_{\vert\dot\gamma\vert-r}),\\
k_\gamma^+(r) &:= k(\bar\gamma_r).
\end{align*}

\subsubsection{Lower semicontinuous potentials along $l$-geodesics}\label{Sub:LSC pot tl geo} Now we assume $k\colon M\to \R$ to be only lower semicontinuous (and we still fix $\gamma\in\TGeo(\mms)$). In this case, we use the previous constructions and argue by approximation. The scheme can be found in \cite[Sec.~5.1]{ambrosio2008}. It has proven useful in the general  study of metric measure spaces with variable Ricci curvature bounds \cite{braun2021, ketterer2015, ketterer2017, sturm2015, sturm2020}. 

Let $E\subset \mms$ be any set containing $\gamma_{[0,1]}$ such that $\smash{\inf k(E) > -\infty}$. Typically, we only consider $E := \gamma_{[0,1]}$ when $\gamma$ is fixed, but when $\gamma$ varies it will be convenient to choose $E$ larger; the precise choice of $E$ does not matter, cf.~\autoref{Re:Indep E}. We approximate $k$ pointwise by a nondecreasing sequence of functions as follows.  For $n\in\N$, define $k_n \colon E \to \R$ by
\begin{align}\label{Eq:kappan}
k_n(x) := \inf\{\min\{k(y),n\} + n\,\met(x,y) : y\in E\}.
\end{align}
It is not difficult to check that
\begin{itemize}
\item $k_n$ is a bounded Lipschitz function that satisfies $\inf k(E) \leq k_n \leq k$ on $E$ for every $n\in\N$, and 
\item $(k_n)_{n\in\N}$ is nondecreasing and converges pointwise to $k$ on $E$.
\end{itemize}

For the next definition, we consider the function $\smash{k_{n,\gamma}^\pm\colon [0,\vert\dot\gamma\vert] \to \R}$ defined by \eqref{Eq:kgamma defn} with $k$ replaced by $k_n$ from \eqref{Eq:kappan}.

\begin{definition}[Distortion coefficients along $l$-geodesics]\label{Def:dist coeff general k} Given any $0\leq \theta\leq\vert\dot\gamma\vert$, we define
\begin{enumerate}[label=\textnormal{\alph*.}]
\item the \emph{generalized sine functions} $\smash{\SIN_{k_\gamma^\pm}}$ associated with $\smash{k_\gamma^\pm}$ by
\begin{align*}
\SIN_{k_\gamma^\pm}(\theta) &:= \lim_{n\to\infty} \SIN_{k_{n,\gamma}^\pm}(\theta),
\end{align*}
\item the \emph{distortion coefficients} $\smash{\sigma_{k_\gamma^\pm}^{(t)}}$, where $0\leq t\leq 1$, associated with $\smash{k_\gamma^\pm}$ by
\begin{align*}
\sigma_{k_\gamma^\pm}^{(t)}(\theta) &:= \lim_{n\to\infty} \sigma_{k_{n,\gamma}^\pm}^{(t)}(\theta).
\end{align*}
\end{enumerate}
\end{definition}

\begin{remark}[Well-definedness]\label{Re:Indep E} \autoref{Def:dist coeff general k} is  reasonable  since the involved sequences are monotone by the nonde\-creasingness of $(k_n)_{n\in\N}$ and (Jacques Charles Fran\c{c}ois) 
Sturm's comparison theorem \cite[Thm.~3.1]{ketterer2017}, and since $\smash{(k_{n,\gamma}^\pm)_{n\in\N}}$ converges pointwise to $\smash{k_\gamma^\pm}$. By a similar argument, these definitions are independent of the choice of the approximating sequence \cite[Rem.~3.6]{ketterer2015}, and in particular of $E$. 

The ODE \eqref{Eq:ODE u} still holds for $\smash{\kappa := k_\gamma^\pm}$ in the distributional sense.
\end{remark}

If $k$ is constant(ly equal to $\kappa\in\R$) along $\gamma$ in \autoref{Def:dist coeff general k}, then
\begin{align*}
\sigma_{k_\gamma^\pm}^{(t)}(\theta) = \sigma_{\kappa}^{(t)}(\theta),
\end{align*}
where $\smash{\sigma_{\kappa}^{(t)}(\theta)}$ is as in \eqref{Eq:ODE u} with constant potential $\kappa$.

We refer to \cite[Lem.~3.9]{ketterer2015} and \cite[Props.~3.4, 3.15, Lems.~3.10, 3.14]{ketterer2017} for proofs of the subsequent facts. Let $\LSC(\mms)$ denote the convex space of real-valued lower semicontinuous functions on $\mms$.

\begin{lemma}[Some properties of distortion coefficients]
\label{Le:Properties} Let $E\subset\mms$ be as above. The generalized distortion coefficients from \autoref{Def:dist coeff general k} satisfy the following for every $\smash{k_\gamma\in \{k_\gamma^-,k_\gamma^+\}}$ and every $0\leq \theta\leq\vert\dot\gamma\vert$.
\begin{enumerate}[label=\textnormal{(\roman*)}]
\item \textnormal{\textbf{Nondegeneracy.}} Either $\smash{\sigma_{k_\gamma}^{(t)}(\theta) = \infty}$ for every $0<t<1$ or $\smash{\sigma_{k_\gamma}^{(t)}(\theta) < \infty}$ for every $0<t<1$. In the latter case, 
\begin{align*}
\sigma_{k_\gamma}^{(t)}(\theta) = \frac{\SIN_{k_\gamma}(t\,\theta)}{\SIN_{k_\gamma}(\theta)}.
\end{align*}
\item \textnormal{\textbf{Monotonicity.}} If a lower semicontinuous function $k'\colon E \to \R$ satisfies $k'\leq k$ on $E$, then $\smash{\sigma_{k'_\gamma}^{(t)}(\theta) \leq \sigma_{k_\gamma}^{(t)}(\theta)}$ for every $0 \leq t\leq 1$.
\item \textnormal{\textbf{Lower semicontinuity.}} Let $(k_i)_{i\in\N}$ be a sequence of lower semicontinuous functions $k_i \colon E\to\R$ such that for every $x\in E$,
\begin{align*}
k(x) \leq \liminf_{i\to\infty}k_i(x).
\end{align*}
Define $\smash{k_{i,\gamma}}$ by \eqref{Eq:kgamma defn} with $k$ replaced by $k_i$. Then for every $0\leq t\leq 1$,
\begin{align*}
\sigma_{k_\gamma}^{(t)}(\theta)\leq \liminf_{i\to\infty} \sigma_{k_{i,\gamma}}^{(t)}(\theta).
\end{align*}
In particular, $\smash{\sigma_{k_\gamma}^{(t)}(\vert\dot\gamma\vert)}$ is lower semicontinuous in $\gamma\in\TGeo(\mms)$ with respect to uniform convergence.
\item \textnormal{\textbf{Convexity.}} For every $0\leq t\leq 1$, $\smash{\sigma_{k_\gamma}^{(t)}(1)}$ is logarithmically con\-vex, hence convex in $k\in\LSC(\mms)$. Moreover, the function $\rmG_t\colon \R^2 \times \LSC(\mms)\to \R\cup\{\infty\}$ defined by
\begin{align*}
\rmG_t(r,s,k) := \log\!\big[\sigma_{k_\gamma^-}^{(1-t)}(1)\,\rme^r + \sigma_{k_\gamma^+}^{(t)}(1)\,\rme^s\big]
\end{align*}
is jointly convex.
\end{enumerate}
\end{lemma}

\subsubsection{Timelike $(k,N)$-convexity}\label{Sub:Timelike convexity} We continue to assume lower semicontinuity of $k\colon\mms\to\R$. Based on the distortion coefficients from \autoref{Def:dist coeff general k}, we introduce \emph{timelike $(k,N)$-convexity} for functionals on metric spacetimes  (without a reference measure); see \cite[Def.~6.5]{mccann2020} for a related definition for functionals defined on $\scrP(\mms)$ in the smooth case.  We follow the approach initiated for metric measure spaces with constant $k$ in \cite{erbar2015}, and   further developed in \cite{ketterer2015} with variable $k$. 

Recall our conventions $\smash{\rme^{-\infty}:=0}$ and $\rme^\infty := \infty$. Let $\smash{\rmg\colon [0,1]^2\to\R}$ denote the usual Green's function 
\begin{align}\label{Eq:1d Green's function}
\rmg(s,t) := \min\{s\,(1-t),t\,(1-s)\}
\end{align}
on $[0,1]$ with Dirichlet boundary conditions. 
As usual, the finiteness domain of a functional $S\colon\mms\to \R\cup\{\infty\}$ is denoted by $\Dom(S)$.

\begin{definition}[Timelike $k$-convexity] We say that $S$ is
\begin{enumerate}[label=\textnormal{\alph*.}]
\item \emph{weakly timelike $k$-convex} if for every $x_0,x_1\in\Dom(S)$ there is $\gamma\in \TGeo(\mms)$ from $x_0$ to $x_1$ passing through $\Dom(S)$ such that for every $0\leq t\leq 1$,
\begin{align}\label{Eq:k convexity def inequ}
S(\gamma_t) \leq (1-t)\,S(x_0) + t\,S(x_1) - \int_0^1 \rmg(s,t)\,k(\gamma_s)\,\vert\dot\gamma\vert^2\d s,
\end{align}
\item \emph{strongly timelike $k$-convex} if it is weakly timelike $k$-convex, and for every $x_0,x_1\in \Dom(S)$, \emph{every} $l$-geodesic $\gamma\in\TGeo(\mms)$ from $x_0$ to $x_1$ passing through $\Dom(S)$ satisfies \eqref{Eq:k convexity def inequ} for every $0\leq t\leq 1$.
\end{enumerate}
\end{definition}

The following modifies the previous definition by taking a  further ``dimensional parameter'' $N$ into account. Indeed, it is a stronger property than mere timelike $k$-convexity, cf.~\autoref{Cor:N<infty to N=infty}. Given a functional $S$ as above, analogously to \eqref{Eq:UN DEF} we define $U_N \colon \mms \to [0,\infty)$ by
\begin{align*}
U_N(x) := \rme^{-S(x)/N}.
\end{align*}

\begin{definition}[Timelike $(k,N)$-convexity]
\label{Def:(k,N) conv} We say that $S$ is
\begin{enumerate}[label=\textnormal{\alph*.}]
\item \emph{weakly timelike $(k,N)$-convex} if for every $x_0,x_1\in \Dom(S)$, there exists an $l$-geodesic $\gamma\in\TGeo(\mms)$ from $x_0$ to $x_1$ passing through $\Dom(S)$ such that for every $0\leq t\leq 1$,
\begin{align}\label{Eq:(k,N) convexity def inequ}
U_N(\gamma_t) \geq \sigma_{k_\gamma^-}^{(1-t)}(\vert\dot\gamma\vert)\,U_N(x_0) + \sigma_{k_\gamma^+}^{(t)}(\vert\dot\gamma\vert)\,U_N(x_1),
\end{align}
\item \emph{strongly timelike $(k,N)$-convex} if it is weakly timelike $(k,N)$-convex, and for every $x_0,x_1\in\Dom(S)$,  \emph{every}  $\gamma\in\TGeo(\mms)$ from $x_0$ to $x_1$ passing through $\Dom(S)$ satisfies \eqref{Eq:(k,N) convexity def inequ} for every $0\leq t\leq 1$.
\end{enumerate}
\end{definition}

The following properties are elementary. The proofs from the metric case carry over with no noteworthy changes with the help of \autoref{Le:Properties}, cf.~\cite[Lems.~2.9, 2.10,  2.12]{erbar2015} and \cite[Lem.~3.15, Cor.~3.16]{ketterer2015}. Moreover, to get the idea for \ref{La:Due} and \ref{La:Tre}, the reader is invited to consult the proofs of \autoref{Pr:Potential} and \autoref{Cor:N<infty to N=infty}  below, respectively.

\begin{lemma}[Dependence on parameters]
\label{L:dependence on parameters}
Suppose $k'\colon\mms\to \R$ is lower semicontinuous, $N'>0$, and let $S,S'\colon \mms\to \R \cup\{\infty\}$ be  functionals. Then the following properties hold.
\begin{enumerate}[label=\textnormal{\textcolor{black}{(}\roman*\textcolor{black}{)}}]
\item \textnormal{\textbf{Scaling.}} If $S$ is weakly timelike $(k,N)$-convex, then $\lambda\,S$ is weakly timelike $(\lambda\,k,\lambda\,N)$-convex for every $\lambda>0$.
\item\label{La:Due} \textnormal{\textbf{Addition.}} If $S$ is weakly timelike $(k,N)$-convex and $S'$ is strongly timelike $(k',N')$-convex, then $\smash{S + S'}$ is weakly timelike $(k+k',N+N')$-convex relative to the base space $\Dom(S) \cap \Dom(S')$.
\item\label{La:Tre} \textnormal{\textbf{Monotonicity.}} If $S$ is weakly timelike $(k,N)$-convex, $k'\leq k$ on $\mms$, and $N'\geq N$, then $S$ is weakly timelike $(k',N')$-convex. In particular, $S$ is weakly timelike $k$-convex.
\end{enumerate}

These claims also hold by replacing every occurrence of  ``weakly'' by ``strongly''.
\end{lemma}

\subsection{Distortion coefficients along plans} In order to formulate \autoref{Def:TCDe} below, we need to extend \autoref{Def:dist coeff general k} from functions along $l$-geodesics to functions along the dynamical plans $\smash{\OptTGeo_p(\scrP_\comp(\mms)^2)}$ of \eqref{OptTGeop} (for more general plans concentrated on causal curves; adaptations of the definitions given below are straightforward). To this aim, we require some further notation, reflecting the fact that unless $k$ is nonnegative,  distorted convexity of the entropy along $\ell_q$-geodesics needs to be measured in an $L^2$-way rather than by $\ell_q$-arclength, analogous to \cite{cavalletti2020}.

\subsubsection{$L^2$-cost}  The $\Ell^2$-cost of a dynamical plan $\bdpi$ is defined by
\begin{align*}
\cost_\bdpi := \big\Vert l\circ(\eval_0,\eval_1)\big\Vert_{\Ell^2(\TGeo(\mms),\bdpi)}.
\end{align*}
Denoting
the endpoint marginals of $\bdpi$ by $\mu_0 := (\eval_0)_\push\bdpi$ and $\mu_1 := (\eval_1)_\push\bdpi$, Jensen's inequality and $\smash{\ell_p}$-optimality of $(\eval_0,\eval_1)_\push\bdpi\in \Pi(\mu_0,\mu_1)$ imply
\begin{align}\label{Eq:By Jensen}
\cost_\bdpi \geq \ell_p(\mu_0,\mu_1).
\end{align}

\subsubsection{Lower semicontinuous potentials along  plans} Let $(k_n)_{n\in\N}$ be the monotone sequence from \eqref{Eq:kappan} approximating $k$ on $E := J(\supp\mu_0,\supp\mu_1)$. For $n\in\N$, define the continuous forward and backward ``superpositions''  $\smash{k_{n,\bdpi}^\pm\colon[0,\cost_\bdpi]\to \R}$ by 
\begin{align}\label{Eq:Def kpi+-}
\begin{split}
k_{n,\bdpi}^-(t\,\cost_\bdpi)\,\cost_\bdpi^2 &:= \int k_{n,\gamma}^-(t\,\vert\dot\gamma\vert)\,\vert\dot\gamma\vert^2\d\bdpi(\gamma)\\
&\phantom{:}= \int k_n(\gamma_{1-t})\,\vert\dot\gamma\vert^2\d\bdpi(\gamma),\\
k_{n,\bdpi}^+(t\,\cost_\bdpi)\,\cost_\bdpi^2 &:= \int k_{n,\gamma}^+(t\,\vert\dot\gamma\vert)\,\vert\dot\gamma\vert^2\d\bdpi(\gamma)\\ &\phantom{:}= \int k_n(\gamma_t)\,\vert\dot\gamma\vert^2\d\bdpi(\gamma);
\end{split}
\end{align}
recall \eqref{Eq:kgamma defn} and compare with \autoref{Re:Scaling prop}. In other words,
\begin{align*}
k_{n,\bdpi}^-(r) &= \frac{1}{\cost_\bdpi^2}\int k_n(\bar{\gamma}_{\cost_\bdpi-r})\,\vert\dot\gamma\vert^2\d\bdpi(\gamma),\\
k_{n,\bdpi}^+(r) &= \frac{1}{\cost_\bdpi^2}\int k_n(\bar{\gamma}_r)\,\vert\dot\gamma\vert^2\d\bdpi(\gamma),
\end{align*}
where $\smash{\bar\gamma\colon [0,\cost_\bdpi]\to \mms}$ designates the proper-time reparametrization $\smash{\bar\gamma_r := \gamma_{r/\cost_\bdpi}}$ with respect to $\cost_\bdpi$. Since the involved integrands are all bounded, \eqref{Eq:Def kpi+-} is well-defined. The sequence $\smash{(k_{n,\bdpi}^\pm)_{n\in\N}}$ converges pointwise to a lower semicontinuous function $\smash{\smash{k_\bdpi^\pm}}\colon [0,\cost_\bdpi] \to \R\cup\{\infty\}$. By Levi's monotone convergence theorem, the latter is related to the pointwise limit $\smash{k_\gamma^\pm}$ of $\smash{(k_{n,\gamma}^\pm)_{n\in\N}}$ by the obvious analog of \eqref{Eq:Def kpi+-}. 

Similar arguments as those leading to \autoref{Def:dist coeff general k} yield the following notions and their well-definedness.

\begin{definition}[Distortion coefficients along  plans]\label{Def:dist coeff bdpi} Given $0 \leq \theta\leq \cost_\bdpi$ for
$\bdpi \in \smash{\OptTGeo_p(\scrP_\comp(\mms)^2)}$, as in \eqref{distortion on R} we define
\begin{enumerate}[label=\textnormal{\alph*.}]
\item the \emph{generalized sine functions} $\smash{\SIN_{k_\bdpi^\pm}}$ associated with $\smash{k_\bdpi^\pm}$ by
\begin{align*}
\SIN_{k_\bdpi^\pm}(\theta) &:= \lim_{n\to\infty} \SIN_{k_{n,\bdpi}^\pm}(\theta),
\end{align*}
\item the \emph{distortion coefficients} $\smash{\sigma_{k_\bdpi^\pm}^{(t)}}$, where $0\leq t\leq 1$, associated with $\smash{k_\bdpi^\pm}$ by
\begin{align*}
\sigma_{k_\bdpi^\pm}^{(t)}(\theta) &:= \lim_{n\to\infty}\sigma_{k_{n,\bdpi}^\pm}^{(t)}(\theta).
\end{align*}
\end{enumerate}
\end{definition}

It is not difficult to generalize the following lemma beyond $\smash{\OptTGeo_p(\scrP_\comp(\mms)^2)}$. We will only need it in this form. Furthermore, the implicit lower boundedness assumption on $k$ is technical and will hold in all relevant cases since we mostly work with compactly supported measures.

\begin{lemma}[Narrow lower semicontinuous dependence on plans]\label{Le:LSC sigma} Let $(\bdpi^i)_{i\in\N}$ be a sequence in $\smash{\OptTGeo_p(\scrP_\comp(\mms)^2)}$ converging narrowly to $\smash{\bdpi\in \scrP(\Cont([0,1];\mms))}$. Furthermore, assume the endpoint marginal sequences $\smash{((\eval_0)_\push\bdpi^i)_{i\in\N}}$ and $\smash{((\eval_1)_\push\bdpi^i)_{i\in\N}}$ are uniformly compactly supported, say on the compact sets  $C_0,C_1\subset\mms$, respectively. Then for every $0\leq t\leq 1$,
\begin{align*}
\sigma_{k_\bdpi^\pm}^{(t)}(\cost_\bdpi) \leq \liminf_{i\to \infty} \sigma_{k_{\bdpi^i}^\pm}^{(t)}(\cost_{\bdpi^i}).
\end{align*}
\end{lemma}

\begin{proof} By Alexandrov's theorem and causal closedness (\autoref{Pr:Closed}), $\bdpi$ is concentrated on causal curves contained in the compact set $J(C_0,C_1)$. Let $(k_n)_{n\in\N}$ be the sequence from  \eqref{Eq:kappan} approximating the function $k$ on $E := J(C_0,C_1)$. Given any $n\in \N$ and any $0\leq t\leq 1$, continuity of $k_n$ and $\smash{l_+}$, respectively, imply the lower semicontinuous dependence of the  quantity $k_n(\gamma_t)\,\vert\dot\gamma\vert^2$ on $\gamma\in G$. Here, $G$ is the uniform closure of the set of $\gamma\in\TGeo(\mms)$ having $\gamma_0\in C_0$ and $\gamma_1\in C_1$. Hence, \cite[Lem.~4.3]{villani2009} --- here we use $G$ being Polish, and uniform lower boundedness --- and the definition \eqref{Eq:Def kpi+-} entail
\begin{align}\label{Eq:Well-def}
k_{n,\bdpi}^+(t\,\cost_{\bdpi})\,\cost_{\bdpi}^2 \leq \liminf_{i\to \infty} k_{n,\bdpi^i}^+(t\,\cost_{\bdpi^i})\,\cost_{\bdpi^i}^2. 
\end{align}
An analogous argument yields the same inequality for $\smash{k_{n,\bdpi}^+}$ replaced by $\smash{k_{n,\bdpi}^-}$.

Therefore, in view of \autoref{Re:Scaling prop} we obtain
\begin{align*}
\sigma_{k_\bdpi^\pm\cost_\bdpi^2}^{(t)}(1) &= \sup_{n\in\N}\sigma_{k_{n,\bdpi}^\pm\cost_\bdpi^2}^{(t)}(1)\\ 
&\leq \sup_{n\in\N} \liminf_{i\to\infty} \sigma_{k_{n,\bdpi^i}^\pm\cost_{\bdpi^i}^2}^{(t)}(1) \\ &\leq \liminf_{i\to\infty}\sup_{n\in\N} \sigma_{k_{n,\bdpi^i}^\pm\cost_{\bdpi^i}^2}^{(t)}(1)\\ &\leq \liminf_{i\to\infty} \sigma_{k_{\bdpi^i}^\pm\cost_{\bdpi^i}^2}^{(t)}(1).
\end{align*}
Here we have used analogs of \autoref{Le:Properties} for general potentials, namely lower semicontinuity in the second inequality and monotonicity in the last one, as stated in \cite[Lem.~3.9]{ketterer2015} and \cite[Props.~3.4, 3.5]{ketterer2017}. The claim follows from \eqref{Eq:Scaling prop sigma}.
\end{proof}

The other parts of 
\autoref{Le:Properties} carry over with appropriate modifications, taking into account that, unlike \autoref{Sub:LSC pot tl geo}, here the potential can be infinite.

\subsection{Definition and basic properties} We start with the  infinite-dimensional analog of the variable curvature-dimension condition for metric measure spaces \cite{sturm2015} à la Lott--Sturm--Villani \cite{lott2009,sturm2006a}. For constant $k$, this has been adapted to the Lorentzian context in  \cite{braun2023}. Let $\smash{\rmg\colon[0,1]^2\to \R}$ be the Green's function with Dirichlet boundary conditions, cf.~\eqref{Eq:1d Green's function} of \autoref{Sub:Timelike convexity}.

\begin{definition}[Dimensionless variable timelike curvature-dimension condition]
\label{Def:TCD infty} We term $\scrM$ to satisfy the \emph{timelike curvature-dimension condition} $\smash{\TCD_p(k,\infty)}$ if for every timelike $p$-dualizable pair  $(\mu_0,\mu_1) \in\smash{\Prob_\comp^\ac(\mms,\meas)^2}$, there exist
\begin{itemize}
\item an $\smash{\ell_p}$-geodesic $(\mu_t)_{t\in[0,1]}$ connecting $\mu_0$ to $\mu_1$, and
\item a plan $\bdpi\in\OptTGeo_p(\mu_0,\mu_1)$
\end{itemize}
such that for every $0\leq t\leq 1$, 
\begin{align*}
\Ent_\meas(\mu_t) \leq (1-t)\Ent_\meas(\mu_0) + t\Ent_\meas(\mu_1) -\int_0^1\!\!\int \rmg(s,t)\,k(\gamma_s)\,\vert\dot\gamma\vert^2\d\bdpi(\gamma)\d s.
\end{align*}

If the previous statement holds for merely every \emph{strongly} timelike $p$-dualizable pair $(\mu_0,\mu_1)$ as above, the space $\scrM$  is said to satisfy the \emph{weak timelike curvature-dimension condition} $\wTCD_p(k,\infty)$.
\end{definition}

Our next definition, taking into account an ``upper dimension bound'' $N$, is based on an adaptation of the abstract convexity from of the previous sections to the functionals $\scrU_N$ and $\Ent_\meas$ in place of $U_N$ and $S$, respectively.

\begin{definition}[Variable entropic timelike curvature-dimension condition] 
\label{Def:TCDe} We term $\scrM$ to satisfy the \emph{entropic timelike curvature-dimension condition} $\smash{\TCD_p^e(k,N)}$ if for every timelike $p$-dualizable pair $(\mu_0,\mu_1)\in \smash{\scrP_\comp^\ac(\mms,\meas)^2}$, there exist 
\begin{itemize}
\item an $\smash{\ell_p}$-geodesic $(\mu_t)_{t\in[0,1]}$\footnote{One could alternately start with \emph{rough} $\smash{\ell_p}$-geodesics here. For $\smash{\wTCD_p^e(k,N)}$, this makes no difference by \autoref{Cor:Equiv notions lp geo}. For $\smash{\TCD_p^e(k,N)}$, by \autoref{Th:Pathwise}, a posteriori one can always \emph{choose} a narrowly continuous collection $\smash{(\mu_t)_{t\in[0,1]}}$ certifying \eqref{Eq:Entropic displacement conv inequ} if $\scrM$ is timelike $p$-essentially nonbranching.} connecting $\mu_0$ to $\mu_1$, and
\item a plan $\smash{\bdpi\in\OptTGeo_p(\mu_0,\mu_1)}$
\end{itemize}
such that for every $0\leq t\leq 1$, the functional $\scrU_N$ from 
\eqref{Eq:UN DEF} satisfies
\begin{align}\label{Eq:Entropic displacement conv inequ}
\scrU_N(\mu_t) &\geq \sigma_{k_\bdpi^-/N}^{(1-t)}(\cost_\bdpi)\,\scrU_N(\mu_0) + \sigma_{k_\bdpi^+/N}^{(t)}(\cost_\bdpi)\,\scrU_N(\mu_1).
\end{align}

If the previous statement holds for merely every \emph{strongly} timelike $p$-dualizable pair $(\mu_0,\mu_1)$ as above, the space $\scrM$ is said to satisfy the \emph{weak entropic timelike  curvature-dimension condition} $\wTCD_p^e(k,N)$.
\end{definition}

\begin{remark}[Nonnegative potentials] If $k$ is nonnegative, by \eqref{Eq:Scaling prop sigma}, \eqref{Eq:By Jensen}, and the monotonicity property stated in  \autoref{Le:Properties}, the inequality   \eqref{Eq:Entropic displacement conv inequ} yields
\begin{align*}
\scrU_N(\mu_t) &\geq \sigma_{k_\bdpi^-/N}^{(1-t)}(\ell_p(\mu_0,\mu_1))\,\scrU_N(\mu_0) + \sigma_{k_\bdpi^+/N}^{(t)}(\ell_p(\mu_0,\mu_1))\,\scrU_N(\mu_1).
\end{align*}
In other words, $\Ent_\meas$ is \emph{weakly $(k,N,p)$-convex} on $\scrP_\comp(\mms)$ with respect to the obvious adaptation of \cite[Def.~6.5]{mccann2020} for variable curvature bounds.
\end{remark}

\begin{remark}[Time reversal invariance]\label{Re:Causal reversal} By time reversal in \eqref{Eq:Entropic displacement conv inequ}, $\scrM$ satisfies $\smash{\TCD_p^e(k,N)}$ if and only if its causal reversal does; analogously for $\smash{\wTCD_p^e(k,N)}$. 
\end{remark}

\begin{remark}[Only the support matters] The support of $\meas$ might fail to be a length metric spacetime in the strict sense of \autoref{Ass:GHLLS} by the possible lack of nontrivial chronology. In fact, as shown in \autoref{Th:Timelike BM} below,  if $\meas$ has an atom at $x\in \mms$ then the   chronological future and the chronological past of $x$ relative to $\supp\meas$ are empty. In addition, regularity of $\supp\meas$ does not imply regularity of $\mms$.

On the other hand, $\supp\meas$ is $l$-geodesic under a $\smash{\wTCD_p^e(k,N)}$ condition, in the following sense. For every $x,y\in \supp\meas$ with $x\ll y$ there is $\gamma\in \TGeo(\mms)$ with $\gamma_0 = x$, $\gamma_1 = y$, and $\gamma_t\in \supp\meas$ for every $0\leq t\leq 1$; compare also with \cite[Rem.~3.10]{cavalletti2020}. Indeed, choose sufficiently small neighborhoods $U_x,U_y\subset\mms$ around $x$ and $y$, respectively, such that $U_x \times U_y \Subset \{l>0\}$, apply \autoref{Le:Geodesics plan} to the $\smash{\ell_p}$-geodesic between the uniform distributions of $U_x$ and $U_y$ certifying \eqref{Eq:Entropic displacement conv inequ}, and then argue as in \cite[Rem.~6.11]{cavalletti2021} for the obtained lifting.

Moreover, it still makes sense to say the (weak) entropic timelike curvature-dimension condition is satisfied by $\mms$ if and only if it holds for $\supp \meas$ --- indeed, every $(\mu_t)_{t\in[0,1]}$ as in \autoref{Def:TCDe} obeys $\supp\mu_t \subset \supp\meas$ for every $0\leq t\leq 1$. In this sense, $\smash{\TCD_p^e(k,N)}$ and $\smash{\wTCD_p^e(k,N)}$ are properties of $\supp\meas$ rather than $M$.
\end{remark}

We now show further elementary properties of the notion from  \autoref{Def:TCDe}. They are only stated for the weak timelike entropic curvature-dimension condition, but are valid analogously if every occurrence of ``$\wTCD^e$'' is replaced by ``$\TCD^e$''.

We start by stating the following natural Lorentzian adaptation of isomorphisms of metric measure spacetimes, cf.~e.g.~\cite{giglimondino2015,greven2009,gromov1999,sturm2006a} for the metric measure case. This will be studied in detail in \cite{braun+}; in particular, in some cases the inherent ---  quite strong --- topological hypotheses can be dropped.

\begin{definition}[Metric measure spacetime isomorphy]\label{Def:Isomorphy} Let $\scrM$ and $\scrM'$ be two metric measure spacetimes. A map $\iota\colon\supp\meas \to \supp\meas'$ is called \emph{isomorphism} if it is a homeomorphism, it is measure-preserving --- i.e.~$\iota_\push\meas = \meas'$ ---, and for every $x,y\in\supp\meas$ we have
\begin{align*}
l_+'(\iota(x),\iota(y)) = l_+(x,y).
\end{align*}

If such an isomorphism exists, we call $\scrM$ and $\scrM'$ \emph{isomorphic}.
\end{definition}

\begin{proposition}[Basic properties]
\label{Pr:Consistency} Let $k'\colon \mms\to\R$ be lower semicontinuous, and let $N'> 0$ as well as $\alpha,\beta > 0$. Then the following properties hold.
\begin{enumerate}[label=\textnormal{\textcolor{black}{(}\roman*\textcolor{black}{)}}]
\item\label{La:Eins} \textnormal{\textbf{Isomorphism.}}  Assume $\scrM$ and $\scrM'$ to be isomorphic through $\iota$ according to \autoref{Def:Isomorphy}. If $\scrM$ obeys $\smash{\wTCD_p^e(k,N)}$, then $\scrM'$ obeys $\smash{\wTCD_p^e(\iota_\push k,N)}$, where $\iota_\push k := k\circ\iota$.
\item\label{La:Zwei} \textnormal{\textbf{Consistency.}} If $k'\geq k$ on $\mms$ and $N' \leq N$, and $\scrM$ obeys $\smash{\wTCD_p^e(k',N')}$, then it also satisfies $\smash{\wTCD_p^e(k,N)}$.
\item\label{La:Drei} \textnormal{\textbf{Scaling.}} If $\scrM$ obeys $\smash{\wTCD_p^e(k,N)}$, then its rescaling $\scrM_{\alpha,\beta}:=(\mms, \alpha\,l,\beta\,\meas)$ satisfies $\smash{\wTCD_p^e(k/\alpha^2,N)}$.
\end{enumerate}

These statements hold analogously by setting $N=\infty$.
\end{proposition}

\begin{proof} The claim \ref{La:Eins} follows from the subsequent properties of isomorphy whose straightforward proofs are left to the reader. It is clear from \autoref{Def:Isomorphy} that $\iota_\push k$ is still lower semicontinuous on $\smash{\supp\meas'}$.
\begin{itemize}
\item The inverse $\iota^{-1}$ pushes forward every $\meas'$-absolutely continuous measures $\smash{\mu_0',\mu_1'\in\scrP_\comp^\ac(\mms',\meas')}$ to $\meas$-absolutely continuous measures $\mu_0$ and $\mu_1$ on $J$, and it preserves compactness of  supports.
\item If $\smash{\pi'}$ is a chronological $\ell_p'$-optimal coupling of $\mu_0'$ and $\mu_1'$, then $\smash{(\iota^{-1},\iota^{-1})_\push\pi}$ is a chronological $l^p$-cyclically monotone coupling of $\mu_0$ and $\mu_1$, and thus $\smash{\ell_p}$-optimal \cite[Prop.~2.8]{cavalletti2020}.
\item If the pair $(\mu_0',\mu_1')$ is strongly timelike $p$-dualizable relative to an $(l')^p$-cyclically monotone set $\smash{\Gamma'\subset \mms'}$ as in \autoref{Def:Str tl dual}, then $(\mu_0,\mu_1)$ is stongly timelike $p$-dualizable relative to the $l^p$-cyclically monotone set $\Gamma := (\iota^{-1},\iota^{-1})(\Gamma')$.
\item If $\smash{\bdpi\in\OptTGeo_p(\mu_0,\mu_1)}$ then $\smash{\bdpi' := (\iota_\push)_\push\bdpi\in \OptTGeo_p'(\mu_0',\mu_1')}$, where we set $\iota_\push(\gamma) := \iota\circ\gamma$. The respective $\Ell^2$-costs coincide,
\begin{align*}
\big\Vert l\circ(\eval_0,\eval_1)\big\Vert_{\Ell^2(\TGeo(\mms),\bdpi)} = \big\Vert l'\circ(\eval_0,\eval_1)\big\Vert_{\Ell^2(\TGeo(\mms'),\bdpi')}.
\end{align*}
In particular, recalling that $\iota_*k := k\circ\iota$, for every $0\leq t\leq 1$ we have
\begin{align*}
\sigma_{k_\bdpi^\pm/N}^{(t)}(\cost_\bdpi) = \sigma_{(\iota_*k)_{\bdpi'}^\pm/N}^{(t)}(\cost'_{\bdpi'}).
\end{align*}
\item If $(\mu_t)_{t\in[0,1]}$ is an $\smash{\ell_p}$-geodesic consisting of $\meas$-absolutely continuous measures, so is $(\mu_t')_{t\in[0,1]}$ relative to $\smash{\ell_p'}$,  where $\mu_t' := \iota_\push\mu_t$, and
\begin{align*}
\Ent_\meas(\mu_t) = \Ent_{\meas'}(\mu_t').
\end{align*}
\end{itemize}

Item \ref{La:Zwei} follows  from the monotonicity properties of the distortion coefficients asserted in  \autoref{Le:Properties}.

Finally, \ref{La:Drei} is a consequence of the identities $\Ent_{\beta\meas}(\mu) = \Ent_\meas(\mu) - \log \beta$ for every $\mu\in\Prob_\comp(\mms)$, $1$-homogeneity of the $\Ell^2$-norm, and \eqref{Eq:Scaling prop sigma}.
\end{proof}

\begin{corollary}[Dimensional consistency]
\label{Cor:N<infty to N=infty} The condition $\smash{\wTCD_p^e(k,N)}$ implies $\smash{\wTCD_p(k,\infty)}$.
\end{corollary}

\begin{proof} By \autoref{Pr:Consistency},  \autoref{Def:dist coeff bdpi}, as well as Levi's monotone convergence theorem we may and will assume that $k$ is continuous and bounded.

Let $\smash{\mu_0,\mu_1\in\scrP_\comp^\ac(\mms,\meas)}$ be as in \autoref{Def:TCD infty}. Assume  $\smash{\mu_0,\mu_1\in\Dom(\Ent_\meas)}$;  otherwise the claim is clear. Again by    \autoref{Pr:Consistency}, let $(\mu_t)_{t\in[0,1]}$ and $\smash{\pi\in\Pi_\ll(\mu_0,\mu_1)}$ certify the displacement semiconvexity of $\smash{\scrU_{N'}}$ according to  \autoref{Def:TCDe} for every $N'\geq N$ (noting \autoref{L:dependence on parameters}). Jensen's inequality and \autoref{Le:Intermediate pts geodesics} imply $\Ent_\meas(\mu_t) > -\infty$ uniformly in $0\leq t\leq 1$,  while $\smash{\wTCD_p^e(k,N)}$, finiteness of $\meas$ on compact sets, boundedness of $k$, and \autoref{Le:Properties} yield  $\Ent_\meas(\mu_t) <\infty$ uniformly in $0\leq t\leq 1$. Therefore,
\begin{align*}
&(1-t)\Ent_\meas(\mu_0) + t\Ent_\meas(\mu_1) - \Ent_\meas(\mu_t)\\
&\qquad\qquad = (1-t) \lim_{N\to\infty}N\,\big[1-\scrU_N(\mu_0)\big] + t\lim_{N\to\infty} N\,\big[1-\scrU_N(\mu_1)]\\
&\qquad\qquad\qquad\qquad + \lim_{N\to\infty} N\,\big[\scrU_N(\mu_t) - 1 \big]\\
&\qquad\qquad \geq \liminf_{N\to\infty} N\,\big[\sigma_{k_\bdpi^-/N}^{(1-t)}(\cost_\bdpi) + \sigma_{k_\bdpi^+/N}^{(t)}(\cost_\bdpi) - 1\big].
\end{align*}
In the last line, the entropy terms have been dropped since they converge to $1$.

The function $u\colon[0,1]\to \R$ defined by 
\begin{align*}
u(t) := N\,\big[\sigma_{k_\bdpi^-\cost_\bdpi^2/N}^{(1-t)} + \sigma_{k_\bdpi^+\cost_\bdpi^2/N}^{(t)}-  1\big]
\end{align*}
satisfies Dirichlet boundary conditions $u(0) = u(1) = 0$. Moreover, by \autoref{Re:Scaling prop}, the hypotheses on $k$, and \eqref{Eq:Def kpi+-}, for every $t\in (0,1)$ it satisfies
\begin{align*}
u''(t) &= -k_\bdpi^+(t\,\cost_\bdpi) \,\cost_\bdpi^2\,\big[\sigma_{k_\bdpi^-\cost_\bdpi^2/N}^{(1-t)} + \sigma_{k_\bdpi^+\cost_\bdpi^2/N}^{(t)}\big]\\ 
&= -\int k(\gamma_t)\,\vert\dot\gamma\vert^2\,\big[\sigma_{k_\bdpi^-\cost_\bdpi^2/N}^{(1-t)} + \sigma_{k_\bdpi^+\cost_\bdpi^2/N}^{(t)}\big]\d\bdpi(\gamma),
\end{align*}
and therefore, e.g.~by \cite[Prop.~3.8]{ketterer2015},
\begin{align*}
u(t) =\int_0^1\!\!\int \rmg(s,t)\, k(\gamma_s)\,\vert\dot\gamma\vert^2\,\big[\sigma_{k_\bdpi^-\cost_\bdpi^2/N}^{(1-s)} + \sigma_{k_\bdpi^+\cost_\bdpi^2/N}^{(s)}\big]\d\bdpi(\gamma)\d s.
\end{align*}
By the discussion from \autoref{Sub:Cts potentials} we finally have
\begin{align*}
\lim_{N\to\infty} \sup_{s\in[0,1]} \big\vert \sigma_{k_\bdpi^-\cost_\bdpi^2/N}^{(1-s)} + \sigma_{k_\bdpi^+\cost_\bdpi^2/N}^{(s)} - 1 \big\vert =0,
\end{align*}
from which the claim follows.
\end{proof}

Lastly, we study measure perturbations. Unfortunately, since the arguments of the distortion coefficients in \autoref{Def:TCDe} depend on the entire plan $\bdpi$ and not only on its endpoint marginals, the more restrictive additional  assumption of timelike $p$-essential nonbranching seems needed; see also \autoref{Cor:Lifting nonbr}. 

\begin{proposition}[Variable timelike curvature dimension conditions with densities]
\label{Pr:Potential} Assume $\smash{\wTCD_p^e(k,N)}$, and assume $\scrM$ is timelike $p$-essentially nonbranching. Moreover, let $V\colon\supp\meas\to \R$ be a Borel function which is locally bounded and $(k',N')$-convex in the sense of \autoref{Def:(k,N) conv} for a lower semicontinuous function $k'\colon\mms\to \R$ and $N'>0$. Lastly,  define the measure $\meas_V := \rme^{-V}\,\meas$. Then the weighted globally hyperbolic, regular length metric measure space\-time $\scrM_V := (\mms,l,\meas_V)$ satisfies $\smash{\wTCD_p^e(k+k',N+N')}$.
\end{proposition}

\begin{proof} We first prove a $(k,N)$-convexity-type inequality similar to \autoref{Def:(k,N) conv} and  \autoref{Def:TCDe} for the functional $V^*\colon \smash{\scrP_\comp^\ac(\mms,\meas)\to \R}$ given by
\begin{align*}
V^*(\mu) := \int_\mms V\d\mu.
\end{align*}
We note that $\smash{\Dom(V^*) = \scrP_\comp^\ac(\mms,\meas)}$. Given any $0\leq t\leq 1$, recall the convex function $\rmG_t$ from \autoref{Le:Properties}. Let $(\mu_t)_{t\in[0,1]}$ be any   $\smash{\ell_p}$-geodesic interpolating two given measures $\smash{\mu_0,\mu_1\in\Prob_\comp^\ac(\mms,\meas)}$, passing through $\Dom(V^*)$, and represented by a plan $\smash{\bdpi\in\OptTGeo_p(\mu_0,\mu_1)}$, cf.~\autoref{Cor:Lifting nonbr} below. Thanks to our hypothesis on $V$ and Jensen's inequality, 
\begin{align*}
-\frac{1}{N'}\,V^*(\mu_t) &= -\frac{1}{N'}\int V(\gamma_t)\d\bdpi(\gamma)\\
&\geq \int \rmG_t\Big[\!-\!\frac{1}{N'}\,V(\gamma_0),-\frac{1}{N'}\,V(\gamma_1), \frac{1}{N'}\, k'_\gamma\,\vert\dot\gamma\vert^2\Big]\d\bdpi(\gamma)\\
&\geq \rmG_t\Big[\!-\!\frac{1}{N'}\,V^*(\mu_0),-\frac{1}{N'}\,V^*(\mu_1), \frac{1}{N'}\,k'_\bdpi\,\cost_\bdpi^2\Big].
\end{align*}

Now we turn to the actual statement. Assume $\mu_0$ and $\mu_1$ to be strongly timelike $p$-dualizable in addition. Let $(\mu_t)_{t\in [0,1]}$ be the $\smash{\ell_p}$-geodesic connecting $\mu_0$ to $\mu_1$ represented by $\smash{\bdpi'\in \OptTGeo_p(\mu_0,\mu_1)}$, shown to exist and be unique in  \autoref{Cor:Lifting nonbr} below, and also let $\smash{\bdpi\in\OptTGeo_p(\mu_0,\mu_1)}$ certify the displacement semiconvexity of $\scrU_N$ defining $\smash{\wTCD_p^e(k,N)}$.  Recall that 
\begin{align}\label{Eq:Ent mv formula}
\Ent_{\meas_V}(\mu) = \Ent_\meas(\mu) + V^*(\mu)
\end{align}
for every $\mu\in\scrP_\comp(\mms)$, and define $\smash{\scrU^V_{N+N'} \colon \Prob_\comp(\mms) \to [0,\infty]}$ by
\begin{align*}
\scrU^V_{N+N'}(\mu) := \rme^{-\Ent_{\meas_V}(\mu)/(N+N')}.
\end{align*}
Given any $0\leq t\leq 1$, \eqref{Eq:Ent mv formula} and the previous considerations yield
\begin{align*}
\log\scrU^V_{N+N'}(\mu_t) &= - \frac{N}{N+N'}\,\frac{1}{N}\Ent_\meas(\mu_t) - \frac{N'}{N+N'}\,\frac{1}{N'}\,V^*(\mu_t)\\
&\geq \frac{N}{N+N'}\,\rmG_t\Big[\!-\!\frac{1}{N}\Ent_\meas(\mu_0), -\frac{1}{N}\Ent_\meas(\mu_1), \frac{1}{N}\,k_\bdpi\,\cost_\bdpi^2\Big]\\
&\qquad\qquad + \frac{N'}{N+N'}\,\rmG_t\Big[\!-\!\frac{1}{N'}\,V^*(\mu_0),-\frac{1}{N'}\,V^*(\mu_1), \frac{1}{N'}\,k_{\bdpi'}\,\cost_{\bdpi'}^2\Big]\\
&\geq \rmG_t\Big[\!-\!\frac{1}{N+N'}\Ent_{\meas_V}(\mu_0), -\frac{1}{N+N'}\Ent_{\meas_V}(\mu_1),\\
&\qquad\qquad \frac{1}{N+N'}\,(k+k')_\bdpi\,\cost_\bdpi^2\Big].
\end{align*}
In the last estimate, we have used Jensen's inequality, the linearity \eqref{Eq:Def kpi+-} of
\begin{align}\label{Eq:SUM k k'}
k_\bdpi^\pm\,\cost_\bdpi^2 + k'^\pm_{\bdpi'}\,\cost_{\bdpi'}^2 = (k+k')^\pm_\bdpi\,\cost_\bdpi^2
\end{align}
in $k$,
and the fact that $\bdpi = \bdpi'$ ensured by \autoref{Th:Uniqueness geos} below.
\end{proof}

\subsection{Qualitative properties from local(ly uniform) curvature bounds} Unlike balls in metric geometry,  causal diamonds are always (causally) convex in a metric space\-time. Hence, by local lower boundedness of $k$ and \autoref{Pr:Consistency}, every $\smash{\wTCD_p^e(k,N)}$ space is locally a $\smash{\wTCD_p^e(K,N)}$ space for some $K\in\R$;
 the same is true with both w's removed.

\begin{lemma}[From variable to constant bounds]
\label{Le:var to const} Assume $\smash{\wTCD_p^e(k,N)}$. Let $C_0,C_1\subset\mms$ be two compact sets. Then $K := \inf k(J(C_0,C_1))\in\R$, and the metric measure spacetime structure on $J(C_0,C_1)$ induced by $\scrM$ obeys $\smash{\wTCD_p^e(K,N)}$.
\end{lemma}

An analogous claim holds clearly for the nonweak entropic timelike curvature-dimension condition, using \autoref{Re:Restriction to subsets} to make sense of both weak and nonweak entropic timelike curvature-dimension condition for subsets of $\mms$,
in spite of the fact that
\autoref{Def:TCDe} is stated only for length metric   spacetimes.

This result allows us to transfer qualitative properties from the variable to the constant case. The relevant ones for our purposes treated now are 
\begin{itemize}
\item the uniqueness of chronological $\smash{\ell_p}$-optimal couplings,
\item the uniqueness of $\smash{\ell_p}$-geodesics, and 
\item the existence of $\smash{\ell_p}$-geodesics with uniformly bounded densities.
\end{itemize}
Although $k$ might not be uniformly bounded from below on $\mms$, the proofs of all results below from \cite{braun2022,braun2023,cavalletti2020} are local in nature and can thus be proven with a local uniform lower bound on $k$, as in \autoref{Le:var to const}. We will extend the above mentioned results in two ways.
\begin{itemize}
\item First, recall  our $\smash{\ell_p}$-geodesics are defined differently than in \cite{braun2022,braun2023}. Since \cite{braun2022, cavalletti2020} essentially establish uniqueness of \emph{displacement} $\smash{\ell_p}$-geodesics, and in our case not every $\smash{\ell_p}$-geo\-desic is lifted by a plan, this requires additional care (although the proofs will be similar to \cite{braun2022,cavalletti2020}).
\item Second, slightly generalizing \cite{braun2023,cavalletti2020} we show the above mentioned  results to hold under a mere local entropic timelike curvature-dimension condition, as stated in \autoref{Def:Local TCD} below.
\end{itemize}


\subsubsection{Uniqueness of chronological $\smash{\ell_p}$-optimal couplings from local curvature bounds}

\begin{definition}[Local timelike curvature-dimension bounds]
\label{Def:Local TCD} We say $\scrM$ obeys the \emph{local entropic timelike curvature-dimension condition} $\smash{\TCD_{p,\loc}^e(k,N)}$ if every point $o\in \supp\meas$ has a neighborhood $U\subset\mms $ with the sub\-sequent property. For every $(\mu_0,\mu_1) \in \scrP_\comp^\ac(\mms,\meas)^2$ with $\supp\mu_0,\supp\mu_1\subset U$ and
\begin{align}\label{Eq:Subset condition}
\supp\mu_0 \times\supp\mu_1\subset \{l>0\},
\end{align}
there exist
\begin{itemize}
\item an $\smash{\ell_p}$-geodesic $(\mu_t)_{t\in[0,1]}$ connecting $\mu_0$ to $\mu_1$, and 
\item a plan $\smash{\bdpi\in\OptTGeo_p(\mu_0,\mu_1)}$
\end{itemize}
such that for every $0\leq t\leq 1$, the functional $\scrU_N$ from \eqref{Eq:UN DEF} satisfies
\begin{align*}
\scrU_N(\mu_t) \geq \sigma_{k_\bdpi^-/N}^{(1-t)}(\cost_\bdpi)\,\scrU_N(\mu_0) + \sigma_{k_\bdpi^+/N}^{(t)}(\cost_\bdpi)\,\scrU_N(\mu_1).
\end{align*}
\end{definition}

Note that the above $\smash{\ell_p}$-geodesic is defined with respect to the ambient structure from $\scrM$ and thus, in particular, allowed to leave $U$.

By \autoref{Ex:Str tl dual}, $\smash{\TCD_{p,\loc}^e(k,N)}$ is clearly implied by $\smash{\wTCD_p^e(k,N)}$.
 Moreover, time reversal (\autoref{Re:Causal reversal}) implies the
the $\meas$-absolute continuity of $\mu_1$ can be substituted for that of $\mu_0$ in the hypotheses of the following theorem.


\begin{theorem}[Maps and uniqueness for chronological $\smash{\ell_p}$-optimal couplings]\label{Th:Optimal maps} We assume $\scrM$ is a timelike $p$-essentially nonbranching $\smash{\TCD_{p,\loc}^e(k,N)}$ space, and let $(\mu_0,\mu_1)\in\smash{\Prob(\mms)^2}$ be timelike $p$-dualizable by $\smash{\pi\in\Pi_\ll(\mu_0,\mu_1)}$. Lastly, assume $\mu_0$ 
to be $\meas$-absolutely continuous. Then there is a $\mu_0$-measurable map $T\colon \mms\to\mms$ with
\begin{align*}
\pi = (\Id,T)_\push\mu_0.
\end{align*} 

In particular, $\pi$ is the unique \emph{chrono\-logical} $\smash{\ell_p}$-optimal coupling of $\mu_0$ and $\mu_1$.
\end{theorem}

\begin{proof} We assume $\meas$-absolute continuity of $\mu_0$ with density $\rho_0$. By the local compactness of $\Top$, cf.~\autoref{Cor:Hausdorff}, and the usual contradiction argument, cf.~the proofs of \cite[Thm.~4.15]{braun2023} and \cite[Thm.~3.20]{cavalletti2020}, it is not restrictive to assume the compactness of $C_0 := \supp\mu_0$ and $C_1 := \supp\mu_1$. Furthermore, by modifying $\pi$ and  \cite[Lem.~2.10]{cavalletti2020}, we may and will assume without restriction $\smash{\Vert\rho_0\Vert_{\Ell^\infty(\mms,\meas)}<\infty}$, so that $\mu_0\in\Dom(\Ent_\meas)$, as well as the existence of $r>0$ with
\begin{align*}
\supp\mu_0\times\supp\mu_1 \subset \{l\geq r\}.
\end{align*}
In turn, up to replacing $\mms$ with $J(\supp\mu_0,\supp\mu_1)$, for notational convenience we assume compactness of $\mms$ (which does not conflict with \autoref{Re:Restriction to subsets}).

By $\smash{\TCD_{p,\loc}^e(k,N)}$, there exist $n\in\N$ and relatively open sets $U_1,\dots,U_n\subset \mms$ covering $\mms$  such that for every $i=1,\dots,n$, the statement of the theorem is true for every pair of timelike $p$-dualizable measures with support in $U_i$. This can be shown by solely relying on \emph{displacement} $\smash{\ell_p}$-geodesics. Indeed, first note that every $\smash{\ell_p}$-geodesic as in \autoref{Def:Local TCD} is  a  displacement $\smash{\ell_p}$-geodesic by \autoref{Le:Geodesics plan}. \autoref{Le:var to const}, a local form of \autoref{Pr:TCD to TMCP}, and again \autoref{Le:Geodesics plan} allow us to show a local (entropic, i.e.~in terms of $\scrU_N$) version of \cite[Thm.~4.12]{braun2023} within $U_i$ for every $i=1,\dots,n$, where ``timelike proper-time parametrized $\smash{\ell_p}$-geodesic''  is replaced by ``displacement $\smash{\ell_p}$-geodesic'' in our case; cf.~\cite[Thm.~4.11]{braun2022}. Following the arguments for \cite[Lems.~4.13, 4.14, Thm.~4.15]{braun2023} we obtain the claim.

Consequently, by following the proof of \cite[Thm.~4.16]{braun2023}, for every $i=1,\dots,n$ we also get uniqueness of \emph{displacement} $\smash{\ell_p}$-geodesics connecting timelike $p$-dualizable pairs both supported in $U_i$ with $\meas$-absolutely continuous initial distribution.


We now follow the argument for \cite[Cor.~5.4]{cavmon} to get  prove the claim for $\mu_0$ and $\mu_1$ as above. By Lebesgue's number lemma, there exists $\varepsilon > 0$ such that if any subset $E\subset \mms$ has diameter less than $\varepsilon$, then $E\subset U_i$ for some $i=1,\dots,n$. Let $G_r$ denote the compact  set of all $\gamma\in\TGeo(\mms)$ with $\gamma_0\in C_0$ and $\gamma_1\in C_1$, cf.~\autoref{Cor:Cpt TGeo}. By Arzelà--Ascoli's theorem, there exists $0 < \delta < \varepsilon/5$ such that if $0\leq s,t\leq 1$ satisfies $\vert s-t \vert <\delta$, then every $\gamma\in G_r$ obeys $\met(\gamma_s,\gamma_t) < \varepsilon/5$.

Let $\smash{\bdpi\in\OptTGeo_p(\mu_0,\mu_1)}$ satisfy $\pi = (\eval_0,\eval_1)_\push\bdpi$, whose existence is granted by \autoref{Le:Geodesics plan}. Set $\mu_t := (\eval_t)_\push\bdpi$, where $0\leq t\leq 1$.  Given any $0\leq t<1$, let $\smash{\bdpi^t\colon \mms\to \scrP(\Cont([0,1];\mms))}$ denote  the disintegration of $\bdpi$ with respect to $\eval_t$, i.e.
\begin{align*}
\rmd\bdpi(\gamma) = \rmd\bdpi_x^t(\gamma)\d \mu_t(x).
\end{align*}
By the triangle inequality, given any $z\in\mms$ we obtain
\begin{align}\label{Eq:DIAM EST}
\diam\!\Big[\!\supp \mu_t\mres \sfB^\met(z,\delta) \cup \supp (\eval_{t+\delta/2})_\push\Big[\!\int_{\sfB^\met(z,\delta)} \bdpi_x^t \d\mu_t(x)\Big]\Big] < \varepsilon.
\end{align}

Now we argue successively to get the desired map $T\colon \mms\to\mms$ which satisfies $\pi = (\Id,T)_\push\mu_0$. We start the process from $t=0$. Covering $\supp\mu_0$ with finitely many $\delta$-balls with respect to $\met$, by \eqref{Eq:DIAM EST}, the mentioned Lebesgue number lemma, and pasting optimal maps together we construct a $\smash{\mu_0}$-measurable map $T_{\delta/3}\colon \mms \to \mms$ such that $(\eval_0,\eval_{\delta/3})_\push\bdpi = (\Id, T_{\delta/3})_\push\mu_0$.  We repeat this process finitely many times --- say $n\in\N$, chosen to be the smallest integer with $(n+1)\,\delta/3 \geq 1$ --- by starting from $t=i\,\delta/3$ for every $i = 1,\dots,n$, thereby obtaining a $\smash{\mu_{i\delta/3}}$-measurable map $T_{(i+1)\delta/3}\colon \mms \to \mms$ such that $(\eval_{i\delta/3}, \eval_{(i+1)\delta/3})_\push\bdpi = (\Id, T_{(i+1)\delta/3})_\push\mu_{i\delta/3}$. Here, \eqref{Eq:DIAM EST} and   local versions of \autoref{Pr:TCD to TMCP} as well as the qualitative result \cite[Lem.~4.14]{braun2023} --- extended to  timelike $p$-dualizable endpoints as in Step 4 in the proof of \cite[Prop.~3.19]{cavalletti2020} using \autoref{Le:Mutually singular} --- inductively yield that for every $0=1,\dots,n$, there exists at least one \emph{displacement} $\smash{\ell_p}$-geodesic $(\nu_t)_{t\in[0,1]}$ from $\mu_{i\delta/3}$ to $\mu_{(i+3/2)\delta/3}$ with the following properties. 
\begin{itemize}
\item The measure $\nu_t$ is $\meas$-absolutely continuous for every $0\leq t < 1$.
\item For every $\eta >0$ the respective $\meas$-density of $\nu_t$ belongs to $\Ell^\infty(\mms,\meas)$, uniformly in  $0\leq t\leq 1-\eta$. 
\end{itemize}
By the above  uniqueness results, it is necessarily represented by $\smash{(\Restr_{i\delta/3}^{(i+3/2)\delta/3})_\push\bdpi}$. Thus, $\smash{\mu_t}$ is $\meas$-absolutely continuous for every $i \in 0,\dots,n$ and every $i\delta/3 \leq t \leq (i+1)\delta/3$ with $t<1$, and the respective $\meas$-densities are uniformly in $\Ell^\infty(\mms,\meas)$. The maps obtained that way readily yield the claimed $T$ by composition. 
\end{proof}

\subsubsection{Maps and uniqueness of $\smash{\ell_p}$-geodesics from local curvature bounds} With a similar reasoning based upon \cite[Prop.~3.19, Thm.~3.21]{cavalletti2020} and again \cite[Lem.~4.14, Thms.~4.16, 4.19]{braun2022} with \autoref{Le:Intermediate pts geodesics}, the subsequent \autoref{Th:Uniqueness geos} is readily  verified. 

\begin{theorem}[Maps, uniqueness, and regularity of displacement $\smash{\ell_p}$-geodesics]\label{Th:Uniqueness geos} Retain the assumptions and notation from \autoref{Th:Optimal maps}. Then $\smash{\OptTGeo_p(\mu_0,\mu_1)}=\{\bdpi\}$ and there is a $\mu_0$-measurable map $\mathfrak{T}\colon \mms\to \TGeo(\mms)$ with 
\begin{align*}
\bdpi = \mathfrak{T}_\push\mu_0.
\end{align*}

In particular, $\mu_0$ and $\mu_1$ are connected by exactly one \emph{displacement} $\ell_p$-geodesic; the latter is represented by the plan $\bdpi$.

Lastly, if $\mu_0\in\Prob_\comp(\mms)\cap \Dom(\Ent_\meas)$ and $\mu_1\in\scrP_\comp(\mms)$ above, then we have $(\eval_t)_\push\bdpi\in \scrP_\comp(\mms)\cap \Dom(\Ent_\meas)$ for every $0\leq t < 1$.
\end{theorem}

\begin{corollary}[Strongly timelike $p$-dualizable, (compactly supported) densities have unique (rough) $\ell_p$-geodesics]
\label{Cor:Lifting nonbr} Retain the assumptions of \autoref{Th:Optimal maps} for $\scrM$. Moreover, let the pair $\smash{(\mu_0,\mu_1)\in\scrP(\mms)^2}$ be timelike $p$-dualizable, and assume all $\smash{\ell_p}$-optimal couplings of $\mu_0$ and $\mu_1$ to be chronological. Moreover, suppose $\meas$-absolute continuity of $\mu_0$ or $\mu_1$. Then there exists exactly one $\smash{\ell_p}$-geodesic $(\mu_t)_{t\in[0,1]}$ from $\mu_0$ to $\mu_1$; it is in fact a displacement $\smash{\ell_p}$-geodesic.

If in addition $\mu_0$ and $\mu_1$ have compact support, then there exists exactly one \emph{rough} $\smash{\ell_p}$-geodesic joining them.
\end{corollary}

\begin{proof} By the first part of \autoref{Le:Geodesics plan}, there exists at least one displacement $\smash{\ell_p}$-geodesic from $\mu_0$ to $\mu_1$. By \autoref{Re:TL Geo to geo}, this is also an $\smash{\ell_p}$-geodesic.

Let $(\mu_t)_{t\in[0,1]}$ be an $\smash{\ell_p}$-geodesic from $\mu_0$ to $\mu_1$. To show its uniqueness, we revisit the proof of \ref{La:05} in \autoref{Le:Geodesics plan}. Retaining its notation, we note that since every $\smash{\ell_p}$-optimal coupling of $\mu_0$ and $\mu_1$ is chronological and by \autoref{Th:Uniqueness geos}, the constructed sequence $\smash{(\bdpi^n)_{n\in\N}}$ interpolating  $(\mu_t)_{t\in[0,1]}$ at dyadic points in $[0,1]$ is constant, say equal to $\smash{\bdpi\in\OptTGeo_p(\mu_0,\mu_1)}$. By narrow continuity of $\smash{(\mu_t)_{t\in[0,1]}}$, the latter is represented by $\bdpi$. The claim follows from \autoref{Th:Uniqueness geos}.

The claim about uniqueness of rough $\smash{\ell_p}$-geodesics is clear, since \autoref{Cor:Equiv notions lp geo} asserts every rough $\smash{\ell_p}$-geodesic from $\mu_0$ to $\mu_1$ to be an $\smash{\ell_p}$-geodesic.
\end{proof}

We do not know whether uniqueness of \emph{rough} $\smash{\ell_p}$-geodesics holds under the first set of assumptions, i.e.~without compact supports; recall  \autoref{Re:Do not have}.

\subsubsection{Good geodesics} The following adaptation of \cite[Thm.~1.2]{braun2022} is not used in our work, but is a crucial technical ingredient to develop a rich calculus in our setting as initiated in \cite{beran2023}.

\begin{theorem}[Existence of good geodesics]\label{Th:Good} Assume $\smash{\wTCD_p^e(k,N)}$. Then every pair $(\mu_0,\mu_1) = (\rho_0\,\meas,\rho_1\,\meas)\in \smash{\scrP_\comp^\ac(\mms,\meas)^2}$ of strongly timelike $p$-dualizable measures with $\rho_0,\rho_1\in\Ell^\infty(\mms,\meas)$ are joined by an $\smash{\ell_p}$-geodesic $(\mu_t)_{t\in[0,1]}$ such that
\begin{enumerate}[label=\textnormal{(\roman*)}]
\item $\mu_t$ is $\meas$-absolutely continuous with density $\rho_t$ for every $0\leq t\leq 1$, with
\begin{align*}
\sup_{t\in[0,1]} \Vert\rho_t\Vert_{\Ell^\infty(\mms,\meas)} < \infty,
\end{align*}
\item for every $0\leq t\leq 1$, there exists $\smash{\bdpi\in\OptTGeo_p(\mu_0,\mu_1)}$ such that
\begin{align*}
\scrU_N(\mu_t) \geq \sigma_{K/N}^{(1-t)}(\cost_\bdpi)\,\scrU_N(\mu_0) + \sigma_{K/N}^{(t)}(\cost_\bdpi)\,\scrU_N(\mu_1),
\end{align*}
where $K := \inf k(J(\supp\mu_0,\supp\mu_1))$.
\end{enumerate}
\end{theorem}

\begin{proof} We only outline the proof; the strategy of \cite{braun2023} carries over almost verbatim. However, our framework a priori lacks a lifting theorem for general strongly timelike $p$-dualizable measures (unless $\scrM$ is timelike $p$-essentially nonbranching, cf.~\autoref{Cor:Lifting nonbr}) --- the latter has been  relied on in \cite{braun2023} quite extensively. The necessary modification is to consider certain sets of $t$-intermediate points introduced in \autoref{Sub:Causality pm} below. Following \cite{rajala} (which has inspired \cite{braun2023}), the minimization procedure from \cite{braun2023} is then performed on this set instead of sets of plans. Strong timelike $p$-dualizability propagates thanks to an analog of \autoref{Le:Propagates}, and the limit geodesic is constructed from countably many intermediate measures as in the proof of \autoref{Cor:Cpt TGeo} by using antisymmetry  (\autoref{Th:GHy prob meas}).
\end{proof}

\subsection{Pathwise characterization} Our goal now to show that $\smash{\TCD_p^e(k,N)}$ admits a pathwise characterization along $\smash{\ell_p}$-geodesics under timelike $p$-essential nonbranching. In this setting, we also show the equivalence of $\smash{\TCD_p^e(k,N)}$ and $\smash{\wTCD_p^e(k,N)}$. For constant $k$, this is due to \cite[Thm.~3.35]{braun2022}, see also \cite[Thm.~3.12]{erbar2015} and \cite[Thm.~7.1]{mccann2020}. In the smooth framework, the corresponding estimate \eqref{Eq:PW EST} is precisely the one  obtained from Raychaudhuri's equation for the transport Jacobian \cite{braunohta,mccann2020}, see also \autoref{Sub:Smooth Lorentzian}.

\begin{theorem}[Distorted displacement concavity of mean free path]\label{Th:Pathwise} Assume $\scrM$ to be timelike $p$-essentially nonbranching. Then the following are equivalent.
\begin{enumerate}[label=\textnormal{\textcolor{black}{(}\roman*\textcolor{black}{)}}]
\item\label{La:Statement 1} The condition $\smash{\TCD_p^e(k,N)}$ holds.
\item\label{La:Statement 1.5} The condition $\smash{\wTCD_p^e(k,N)}$ holds.
\item\label{La:Statement 2} For every timelike $p$-dualizable pair $\smash{(\mu_0,\mu_1) = (\rho_0\,\meas,\rho_1\,\meas)\in\scrP^\ac(\mms,\meas)^2}$, there is a displacement $\smash{\ell_p}$-geodesic from $\mu_0$ to $\mu_1$ which is represented by $\smash{\bdpi\in\OptTGeo_p(\mu_0,\mu_1)}$ such that for every $0\leq t\leq 1$, we have $(\eval_t)_\push\bdpi = \rho_t\,\meas\in \scrP^\ac(\mms,\meas)$, and $\bdpi$-a.e.~$\gamma\in\TGeo(\mms)$ satisfies
\begin{align}\label{Eq:PW EST}
\rho_t(\gamma_t)^{-1/N} \geq \sigma_{k_\gamma^-/N}^{(1-t)}(\vert\dot\gamma\vert)\,\rho_0(\gamma_0)^{-1/N} + \sigma_{k_\gamma^+/N}^{(t)}(\vert\dot\gamma\vert)\,\rho_1(\gamma_1)^{-1/N}.
\end{align}
\end{enumerate}
\end{theorem}

\begin{proof} It is evident that \ref{La:Statement 1} implies \ref{La:Statement 1.5}.

Next, we show \ref{La:Statement 2} implies \ref{La:Statement 1}. Let $\smash{(\mu_0,\mu_1)=(\rho_0\,\meas,\rho_1\,\meas)\in \Prob_\comp^\ac(\mms,\meas)^2}$, and assume $\mu_0 \in\Dom(\Ent_\meas)$. Furthermore, letting $\smash{\bdpi\in\OptTGeo_p(\mu_0,\mu_1)}$ be as given by \ref{La:Statement 2} as well as $0<t<1$, compactness of $J(\supp\mu_0,\supp\mu_1)$ implies the existence of $K\leq 0$ with $k\geq K$ on $J(\supp\mu_0,\supp\mu_1)$. By \eqref{Eq:PW EST}, $\bdpi$-a.e.~$\gamma\in\TGeo(\mms)$ thus satisfies the estimate
\begin{align*}
\rho_t(\gamma_t) \leq \sup \sigma_{K/N}^{(1-t)}(l_+(x,y))^{-N}\,\rho_0(\gamma_0).
\end{align*}
The supremum is taken over all  $x\in \supp\mu_0$ and $y\in\supp\mu_1$. In fact, it is finite by \autoref{Le:Properties} as well as compactness of $\supp\mu_0$ and $\supp\mu_1$. In particular, we obtain $\mu_t\in\Dom(\Ent_\meas)$. Letting $\rmG_t$ denote the convex function from \autoref{Le:Properties}, from \eqref{Eq:PW EST} again and Jensen's inequality we get
\begin{align*}
-\frac{1}{N}\Ent_\meas(\mu_t) &= -\frac{1}{N}\int \log\rho_t(\gamma_t)\d\bdpi(\gamma)\\
&\geq \int \rmG_t\Big[\!-\!\frac{1}{N}\log\rho_0(\gamma_0),-\frac{1}{N}\log\rho_1(\gamma_1), k_\gamma\,\vert\dot\gamma\vert^2\Big]\d\bdpi(\gamma)\\
&\geq \rmG_t\Big[-\frac{1}{N}\Ent_\meas(\mu_0),-\frac{1}{N}\Ent_\meas(\mu_1), \frac{1}{N}\,k_\bdpi\,\cost_\bdpi^2\Big].
\end{align*}
Exponentiating both sides yields the desired inequality.

The case $\mu_1\in\Dom(\Ent_\meas)$ is argued analogously. If neither $\mu_0$ nor $\mu_1$ belong to $\Dom(\Ent_\meas)$, necessarily $\scrU_N(\mu_0) = \scrU_N(\mu_1)=0$ by \eqref{Eq:UN Jensen}, in which case \eqref{Eq:Entropic displacement conv inequ} is clear.

The implication from \ref{La:Statement 1.5} to \ref{La:Statement 2} if the pair $(\mu_0,\mu_1)$ from  \ref{La:Statement 2} satisfies 
\begin{align*}
\supp\mu_0\times\supp\mu_1 \subset\{l>0\}
\end{align*} 
is shown in the same way as in the proof of \cite[Prop.~3.38]{braun2022} with the modifications described in the proof of \cite[Prop.~3.40]{braun2022}, employing  \autoref{Le:Mutually singular} and \autoref{Le:Properties}. 

The general case is argued as for e.g.~\cite[Prop.~3.38]{braun2022} or \cite[Thm. 7.1]{mccann2020}. The idea is to decompose the support of a timelike $p$-dualizing coupling into countably many rectangles compactly contained in $\{l>0\}$, and to restrict $\pi$ to these, modulo successive removal of overlapping supports. Then we apply \ref{La:Statement 2} to these restrictions. Taking an  appropriate convex combination of the plans constructed in that way yields the desired plan.  This works since  \autoref{Le:Mutually singular} ensures that the individual $t$-slices are mutually singular for every $0\leq t\leq 1$. Thus, the intersection of the supports of any two $\meas$-densities of these restrictions is $\meas$-negligible, and hence the local inequality \eqref{Eq:PW EST} globalizes straightforwardly.
\end{proof}

\begin{remark}[Parameter independence of negligible set]\label{Re:LIPSCHITZ} By imitating the proof of \cite[Cor.~9.5]{cavalletti2021}, we get the following  prima facie stronger, but --- under timelike $p$-essential nonbranching of $\scrM$ --- equivalent version of \ref{La:Statement 2} in \autoref{Th:Pathwise}. For every timelike $p$-dualizable pair $(\mu_0,\mu_1)$ as in \ref{La:Statement 2} there exists $\smash{\bdpi\in \OptTGeo_p(\mu_0,\mu_1)}$ with the following properties. For every $0\leq t\leq 1$, we have $(\eval_t)_\push\bdpi=\rho_t\,\meas\in \scrP^\ac(\mms,\meas)$, and there exist $\meas$-versions of the respective $\meas$-densities such that for $\bdpi$-a.e.~$\gamma\in \TGeo(\mms)$, $\rho_t(\gamma_t)$ depends locally Lipschitz continuously on $0<t<1$, and \eqref{Eq:PW EST} holds for every $0\leq t\leq 1$.

Thus, the exceptional set in \eqref{Eq:PW EST} can be chosen to be independent of $t$.
\end{remark}

\begin{remark}[Variants]\label{Re:Verify?} The proof of \autoref{Th:Pathwise} shows that to verify either of its inherent conditions, if $\scrM$ is timelike $p$-essentially nonbranching it suffices to check any of the following properties for every $\mu_0,\mu_1\in\scrP_\comp^\ac(\mms,\meas)$ with
\begin{align}\label{Eq:No pro}
\supp\mu_0\times\supp\mu_1 \subset\{l>0\}.
\end{align}
\begin{itemize}
\item The existence of $(\mu_t)_{t\in[0,1]}$ and $\bdpi$ as in \autoref{Def:TCDe} certifying \eqref{Eq:Entropic displacement conv inequ}.
\item The existence of $\bdpi$ as in \ref{La:Statement 2} in \autoref{Th:Pathwise}.
\end{itemize}

However, unlike (strong) timelike $p$-dualizability, cf.~\autoref{Le:Propagates}, the condition \eqref{Eq:No pro} does not propagate along $\ell_p$-geodesics.
\end{remark}

\subsection{Local-to-global property} Now we prove the local-to-global property of the entropic timelike curvature-dimension condition. Roughly speaking, it states that to verify synthetic timelike Ricci curvature lower bounds, it suffices to check this condition locally (which is clearly true in the smooth case if such bounds are interpreted in terms of the ordinary Bakry--\smash{Émery}--Ricci tensor).

Recall \autoref{Def:Local TCD} for the condition $\smash{\TCD_{p,\loc}^e(k,N)}$. 


\begin{theorem}[Globalization]\label{Th:Local to global} If $\scrM$ is timelike $p$-essentially nonbranching, the properties $\smash{\TCD_{p,\loc}^e(k,N)}$ and $\smash{\TCD_p^e(k,N)}$ are equivalent.
\end{theorem}

\begin{proof} We only need to prove the forward implication. Let $\mu_0,\mu_1 \in \scrP_\comp^\ac(\mms,\meas)$ satisfy \eqref{Eq:No pro}, i.e.~there exists $r>0$ such that 
\begin{align*}
\supp\mu_0 \times \supp\mu_1 \subset \{l\geq r\}.
\end{align*}
By \autoref{Re:Verify?}, it suffices to verify the existence of a plan $\smash{\bdpi\in\OptTGeo_p(\mu_0,\mu_1)}$ obeying  \ref{La:Statement 2} in \autoref{Th:Pathwise} to obtain $\smash{\TCD_p^e(k,N)}$. By restriction, truncation, and an eventual limit argument, we may and  will assume $\mu_0,\mu_1\in\Dom(\Ent_\meas)$. Recall from Theorems \ref{Th:Optimal maps} and \ref{Th:Uniqueness geos} and \autoref{Cor:Lifting nonbr} that under  $\smash{\TCD_{p,\loc}^e(k,N)}$,  chronological $\smash{\ell_p}$-optimal couplings and $\smash{\ell_p}$-geodesics [sic] between strongly timelike $p$-dualizable pairs with endpoints in $\Dom(\Ent_\meas)$ are indeed  unique, and consist of $\meas$-absolutely continuous measures. This is used several times in this proof without further note.

For notational convenience, up to replacing $\mms$ by $J(\supp\mu_0,\supp\mu_1)$ we may and will assume compactness of $\mms$ (which does not counteract \autoref{Re:Restriction to subsets}). Owing to \eqref{Eq:kappan}, \autoref{Le:Properties}, and Levi's monotone convergence theorem, we  may and will assume continuity of $k$ on $\mms$.  Then there exist $n\in\N$, $0 < \lambda <\diam \mms$, mutually disjoint Borel sets $L_1,\dots,L_n \subset\mms$ with nonzero $\meas$-measure, and relatively open sets $U_1,\dots, U_n\subset \mms$ certifying \autoref{Def:Local TCD} with $\smash{\sfB^\met(L_i,\lambda) \subset U_i}$ for every $i = 1,\dots,n$; cf.~\cite[Thm.~5.1]{bacher2010}.  Finally, let $G_r$ be the set of all $\gamma\in\TGeo(\mms)$ with $\gamma_0\in \supp\mu_0$ and $\gamma_1\in\supp\mu_1$ with $\smash{l(\gamma_0,\gamma_1)\geq r}$ as defined in \autoref{Cor:Cpt TGeo}. The latter states compactness of $G_r$ in $\Cont([0,1];\mms)$. In particular, Arzelà--Ascoli's theorem implies uniform equicontinuity of $G_r$, thus there exists $\delta > 0$ such that if $0 \leq s,t\leq 1$ obey $\vert s-t\vert < \delta$, every $\gamma\in G_r$ obeys $\met(\gamma_s,\gamma_t)< \lambda$.

Let $\smash{\bdpi\in \OptTGeo_p(\mu_0,\mu_1)}$ be the unique timelike $\smash{\ell_p}$-optimal dynamical plan from $\mu_0$ to $\mu_1$. Let $\rho_t$ denote the $\meas$-density of $(\eval_t)_\push\bdpi$ for every $0\leq t\leq 1$. Fix $0< s < t < \delta$. Given any $0\leq r\leq 1$, we define 
\begin{align*}
\bdsigma &:= (\Restr_s^t)_\push\bdpi,\\
\nu_r &:= (\eval_r)_\push\bdsigma = \varrho_r\,\meas,
\end{align*}
and recall that $\bdsigma$ is the only element of $\smash{\OptTGeo_p(\nu_0,\nu_1)}$. Furthermore, given any $i = 1,\dots,n$ and any $0\leq r\leq 1$ we define
\begin{align*}
\bdsigma^i &:= (\eval_0)_\push\bdsigma[L_i]^{-1}\,\bdsigma\mres (\eval_0)^{-1}(L_i)\\
\nu_r^i &:= (\eval_r)_\push\bdsigma^i = \varrho_r^i\,\meas,
\end{align*}
where here and henceforth, we ignore indices $i$ with $(\eval_0)_\push\bdsigma[L_i]=0$. By restriction, cf.~\autoref{Le:Geodesics plan}, $\smash{\bdsigma^i}$ is the only element of $\smash{\OptTGeo_p(\nu_0^i,\nu_1^i)}$. In particular, the pair $\smash{(\nu_0^i,\nu_1^i)}$ is timelike $p$-dualizable, and $\smash{\bdsigma^i}$ is concentrated on $l$-geodesics that do not leave $U_i$ by our choice of $s$ and $t$. Consequently, by $\smash{\TCD_{p,\loc}^e(k,N)}$, a localization argument and relaxation of the condition \eqref{Eq:Subset condition} to mere timelike $p$-dualizability as outlined in the proof of \autoref{Th:Pathwise} --- making  \autoref{Def:Local TCD} applicable to  $\smash{(\nu_0^i,\nu_1^i)}$ by construction ---, for every $i=1,\dots,n$ and every $0\leq r\leq 1$, $\smash{\bdsigma^i}$-a.e. $\alpha\in\TGeo(\mms)$ satisfies
\begin{align}\label{Eq:From}
\varrho_r^i(\alpha_r)^{-1/N} \geq  \sigma_{k_\alpha^-/N}^{(1-r)}(\vert\dot\alpha\vert)\,\varrho_0^i(\alpha_0)^{-1/N} + \sigma_{k_\alpha^+/N}^{(r)}(\vert\dot\alpha\vert)\,\varrho_1^i(\alpha_1)^{-1/N}.
\end{align}
We point out that the exceptional set depends on $i$ and $r$.

Now observe that $\smash{\bdsigma^1,\dots,\bdsigma^n}$ are mutually singular. For every $0 < r < 1$, since $\smash{\nu_t}$ is $\meas$-absolutely continuous,  $\smash{\nu_r^1,\dots,\nu_r^n}$ are thus mutually singular by \autoref{Le:Mutually singular}. Since $\smash{\nu_0^1,\dots,\nu_0^n}$ are also  mutually singular by construction, so are $\smash{\nu_1^1,\dots,\nu_1^n}$ since every chronological $\smash{\ell_p}$-optimal coupling is induced by a map by \autoref{Th:Optimal maps}. In particular, for every $0\leq r\leq 1$ and every $i,j=1,\dots,n$ with $i\neq j$ we have
\begin{align*}
\meas\big[\varrho_r^i > 0,\, \varrho_r^j > 0\big] = 0.
\end{align*}
Hence, by \eqref{Eq:From} for every $0\leq r\leq 1$, $\smash{\bdsigma}$-a.e.~$\alpha\in\TGeo(\mms)$ satisfies
\begin{align*}
\varrho_r(\alpha_r)^{-1/N} \geq \sigma_{k_\alpha^-/N}^{(1-r)}(\vert\dot\alpha\vert)\,\varrho_0(\alpha_0)^{-1/N} + \sigma_{k_\alpha^+/N}^{(r)}(\vert\dot\alpha\vert)\,\varrho_1(\alpha_1)^{-1/N}.
\end{align*}
In other words, for every $0\leq r\leq 1$, $\smash{\bdpi}$-a.e.~$\gamma\in \TGeo(\mms)$ obeys
\begin{align}\label{Eq:From II}
\begin{split}
\rho_{\tau(r)}(\gamma_{\tau(r)})^{1/N} &\geq \sigma_{k_{\gamma\circ\tau}^-/N}^{(1-\tau(r))}((t-s)\,\vert\dot\gamma\vert)\,\rho_s(\gamma_s)^{-1/N}\\
&\qquad\qquad + \sigma_{k_{\gamma\circ\tau}^+/N}^{(\tau(r))}((t-s)\,\vert\dot\gamma\vert)\,\rho_t(\gamma_t)^{-1/N},
\end{split}
\end{align}
where $\tau\colon [0,1]\to [0,\delta]$ is defined by $\tau(r) := (1-r)\,s + r\,t$. We observe that the exceptional set still depends on $r$.

From here, we only outline the argument of obtaining  \eqref{Eq:PW EST}, leaving out technicalities and referring to  \cite{braunc2,cavalletti2021, ketterer2017} for details. From \eqref{Eq:From II}, following the proof of \cite[Cor.~9.5]{cavalletti2021}, see also \autoref{Re:LIPSCHITZ}, we get the existence of nonrelabeled $\meas$-versions of $\rho_{\tau(r)}$, where $0\leq r\leq 1$, such that $\bdpi$-a.e.~$\gamma\in\TGeo(\mms)$ satisfies \eqref{Eq:From II} for every $0\leq r\leq 1$, hence the exceptional set in \eqref{Eq:From II} can be chosen to be independent of $r$. Consequently, for $\bdpi$-a.e.~$\gamma\in\TGeo(\mms)$ the function $h\colon[s,t]\to\R$ defined by $h(u) := \rho_u(\gamma_u)$ is a $\CD(\kappa,N+1)$ density in an evident variable adaptation of \cite[Def.~A.1]{cavalletti2021}, where $\kappa\colon [s,t]\to M$ is given by $\kappa(u) := k(\gamma_u)$. By \autoref{Le:Properties}, it is a $\CD(K,N+1)$ density for some $K\in\R$ independent of $\gamma$, thus twice differentiable at $\smash{\Leb^1}$-a.e.~$u\in[s,t]$. Up to removing another $\smash{\Leb^1}$-negligible set, by a variable version of \cite[Lem.~A.3]{cavalletti2021}, cf.~\cite[Cor.~3.13]{ketterer2017}, for $\smash{\Leb^1}$-a.e.~$u\in[s,t]$ we obtain
\begin{align*}
\frac{\rmd^2}{\rmd u^2} h(u)^{1/N}  \leq - \frac{\vert \dot\gamma\vert^2}{N}\,\kappa(u)\,h(u)^{1/N}.
\end{align*}

This differential inequality globalizes straightforwardly to $\smash{\Leb^1}$-a.e.~$u\in [0,1]$ by covering $[0,1]$ with intervals $[s,t]$, where $s$ and $t$ are as initially chosen above, leading to \eqref{Eq:PW EST} again  by  \cite[Lem.~A.3]{cavalletti2021} and  \cite[Cor.~3.13]{ketterer2017}.
\end{proof}

\begin{remark}[{About chronological $\smash{\ell_p}$-geodesy \cite[Def.~3.44]{braun2023}}] As a byproduct, \autoref{Th:Local to global} shows the condition of chronological $\smash{\ell_p}$-geodesy of $\smash{\scrP_\comp^\ac(\mms,\meas)}$ \cite[Def.~3.44]{braun2023} in the local-to-global property \cite[Thm.~3.45]{braun2023} is in fact redundant.
\end{remark}

\section{Examples}\label{Ch:Examples}

Now we discuss examples satisfying our variable timelike curvature-dimension condition, inspired by corresponding results for constant $k$ \cite{braun2023,braunc1, braunohta,mccann2020, mondinosuhr2022} and the recent observation \cite{mccann+}.

In all cases, it is clear that the induced structure constitutes a metric measure spacetime $\scrM$ satisfying \autoref{Ass:REG}, cf.~e.g.~\cite[Sec.~2.2]{braunc1} or \cite[Thm.~5.16]{kunzinger2018}.

\subsection{Timelike convergence condition and stability}\label{Sub:Smooth Lorentzian} 

\subsubsection{High regularity Lorentzian spacetimes}\label{Sub:High} As a warm-up, let us first discuss the most regular case which is well-established by now for constant $k$.  \autoref{Th:Strong energ} below includes the famous strong energy condition of Hawking and Penrose for solutions to the Einstein equations  --- it is also termed \emph{timelike convergence condition} in the physics literature --- synthetic versions of which have been initiated in   \cite{cavalletti2020} following   \cite{mccann2020,mondinosuhr2022} and studied further in \cite{braun2022, mccann+}. However, even if the function $k$ therein is uniformly bounded from below, \autoref{Th:Strong energ} generally entails a quantitatively better statement.

Let $g$ be a fixed smooth, globally hyperbolic Lorentzian metric with signature $+,-,\dots,-$ on a smooth,  connected,  noncompact topological manifold $\mms$ of dimension $\smash{n\in\N_{\geq 2}}$. We write $\vert v\vert^2 := g(v,v)$ for $v\in T\mms$, but stress the usual danger of confusion that the quantity $\vert v\vert^2$ may be negative despite the square.  Let $V\colon \mms \to \R$ be twice continuously differentiable. Then, given any $N> 0$ with $N\neq n$, consider the \emph{Bakry--Émery--Ricci tensor} 
\begin{align*}
\Ric^{N,V} := \Ric + \Hess V - \frac{1}{N-n}\,\rmd V\otimes\rmd V
\end{align*}
associated with the measure $\smash{\meas := \rme^{-V}\,\vol}$. Here $\vol$ denotes the Lorentzian volume measure associated with $g$. By convention, we set $\smash{\Ric^{n,V} := \Ric}$, i.e.~we  assume $V$ to be  constant in the case $N=n$. 

Given any function $k\colon \mms\to \R$, we will say 
\begin{align}\label{Eq:Inter TL Dir}
\Ric^{N,V}\geq k\quad\textnormal{in all timelike directions}
\end{align}
if $\smash{\Ric^{N,V}(v,v)\geq k(x)\,\vert v\vert^2}$ for every $x\in\mms$ and every $v\in T_x\mms$ with  $\vert v\vert^2>0$.

Although the argument for \autoref{Th:Strong energ} is well-established, we give some more details to familiarize the reader  with our variable setting, the way  that the distortion coefficients in \autoref{Def:TCDe} arise from basic ODE theory, etc.

\begin{theorem}[Equivalence with the timelike convergence condition]\label{Th:Strong energ} Let $k\colon \mms\to \R$ be a continuous function and $N>0$. Then
\begin{align}\label{Eq:TL Ric}
\Ric^{N,V} \geq k\quad\textnormal{\textit{in all timelike directions}}
\end{align}
if and only if $N\geq n$ and the $\smash{\wTCD_p^e(k,N)}$ condition holds for every $0<p<1$.

Furthermore, the preceding equivalence is still valid  if $\smash{\wTCD_p^e(k,N)}$ therein is replaced by $\smash{\TCD_p^e(k,N)}$.
\end{theorem}

\begin{proof} Since $k$ is real-valued, the last statement follows from \autoref{Th:Pathwise}.

The forward direction of the first claim follows from straightforward adaptations of the usual arguments, cf.~  \cite[Prop.~A.3]{braun2022}, \cite[Thm.~6.4, Cor.~6.6]{mccann2020}, \cite[Thm.~25]{mccann+}, and \cite[Thm.~4.3]{mondinosuhr2022}, together with a comparison theorem, cf.~e.g.~\cite[Cor.~3.13]{ketterer2017}. For the convenience of the reader, we include some details here.

Given any $0<p<1$, the Cauchy--Lipschitz theorem for the geodesic equation implies $\scrM$ to be timelike $p$-essentially nonbranching. Consequently, by \autoref{Re:Verify?}  it suffices to verify \eqref{Eq:PW EST} for every $\smash{\mu_0,\mu_1 \in \scrP_\comp^\ac(\mms,\meas)}$ with
\begin{align*}
\supp\mu_0\times\supp\mu_1 \subset \{l>0\}.
\end{align*}
Hence, the pair $(\mu_0,\mu_1)$ is $p$-separated according to \cite[Def.~4.1]{mccann2020} by \cite[Lem.~4.4]{mccann2020}. By \cite[Thm.~5.8, Cor.~5.9]{mccann2020} the measures $\mu_0$ and $\mu_1$ admit a unique chronological $\smash{\ell_p}$-optimal coupling $\smash{\pi\in\Pi_\ll(\mu_0,\mu_1)}$ and are joined  through a unique $\smash{\ell_p}$-geodesic $(\mu_t)_{t\in[0,1]}$ consisting of $\meas$-absolutely continuous measures. Furthermore, the same results from \cite{mccann2020} imply the existence of an $\meas$-a.e.~twice approximately differentiable \cite[Def.~3.8]{mccann2020} Borel function $u\colon \mms\to\mms$ with the subsequent properties. Setting $q := p/(p-1)<0$, the map $T\colon[0,1]\times\mms \to \mms$ given by
\begin{align*}
T_t(x) := \exp_x(t\,\vert\nabla u\vert^{q-2}\,\nabla u)
\end{align*}
satisfies
\begin{align}\label{Eq:mut pi}
\begin{split}
\mu_t &= (T_t)_\push\mu_0,\\
\pi &= (\Id, T_1)_\push\mu_0
\end{split}
\end{align}
for every $0\leq t \leq 1$. Furthermore, for $\meas$-a.e.~$x\in\mms$ the approximate derivative $\smash{\rmd T_t\colon T_x\mms\to T_{T_t(x)}\mms}$ depends smoothly on $0\leq t\leq 1$, and $(A_t)_{t\in[0,1]}$ is a matrix of Jacobi fields along the $l$-geodesic $(T_t(x))_{t\in[0,1]}$, where $A_t := \rmd T_t$. Thus, \cite[Prop.~6.1]{mccann2020} ensures that $\meas$-a.e., for every $0\leq t\leq 1$ the quantities
\begin{align*}
\phi_t &:= -\log\vert\!\det A_t\vert,\\
B_t &:= A_t'\,A_t^{-1}
\end{align*}
satisfy
\begin{align}\label{Eq:Riccati}
\begin{split}
\phi_t' &= -\tr B_t,\\
\phi_t'' &= \Ric(T_t',T_t')  + \tr B_t^2,\\
\vert\!\tr B_t\vert^2 &\leq n\,\tr B_t^2.
\end{split}
\end{align}

Now we come to the proof. First, \cite[Thm.~25]{mccann+} implies $N\geq n$. For simplicity, we then assume $N>n$, since  the case $N=n$ is deduced along the same lines by ignoring $V$. Given any $0\leq t\leq 1$, setting
\begin{align*}
\varphi_t := \phi_t + V\circ T_t,
\end{align*}
by Cauchy's inequality for $\varepsilon := (N-n)/n > 0$, from \eqref{Eq:Riccati} we deduce
\begin{align*}
\frac{1}{N}\,\vert\varphi_t'\vert^2 \leq \frac{1+\varepsilon}{N}\,\vert \!\tr B_t\vert^2 + \frac{1+\varepsilon^{-1}}{N}\,\rmd V(T_t')^2 \leq \tr B_t^2 + \frac{1}{N-n}\,\rmd V(T_t')^2\quad\meas\textnormal{-a.e.}
\end{align*}
Setting $\theta := \vert T_t'\vert = l(\cdot,T_1)$, again by \eqref{Eq:Riccati} this yields
\begin{align*}
\varphi_t'' + \frac{1}{N}\,\varphi_t' \leq -\Ric^{N,V}(T_t,T_t') \leq -(k\circ T_t)\,\theta^2\quad\meas\textnormal{-a.e.}
\end{align*}
A standard comparison theorem, cf.~e.g.~\cite[Prop.~3.8]{ketterer2017}, thus yields the transport Jacobian $\smash{\jmath_t := \rme^{\varphi_t}}$ \cite[Cor.~5.11]{mccann2020} to satisfy
\begin{align*}
\jmath_t^{1/N} \geq \sigma_{\kappa^-/N}^{(1-t)}(\theta)\,\jmath_0^{1/N} + \sigma_{\kappa^+/N}^{(t)}(\theta)\,\jmath_1^{1/N}\quad\meas\textnormal{-a.e.}
\end{align*}
Here, analogously to \eqref{Eq:kgamma defn} the functions $\smash{\kappa^\pm\colon[0,\theta] \to \R}$ are defined by
\begin{align*}
\kappa^-(t\,\theta) &:= k\circ T_{1-t},\\
\kappa^+(t\,\theta) &:= k\circ T_t.
\end{align*}
By \eqref{Eq:mut pi} and the change of variables formula, this yields \eqref{Eq:PW EST}, as desired.

The backward direction is shown along the lines of \cite[Thm.~8.5]{mccann2020}. We shortly outline the necessary modifications in the variable case. Given any $z\in \mms$ and any $\varepsilon > 0$, fix $K\in \R$ and a neighborhood $U\subset\mms$ of $z$ such that $K-\varepsilon \leq k\leq K+\varepsilon$ on $U$. By \autoref{Le:Properties} and the local argument for \cite[Thm.~8.5]{mccann2020}, we get 
\begin{align*}
\Ric^{N,V}(v,v) \geq (K-\varepsilon)\,\vert v\vert^2 \geq (k(x) - 2\,\varepsilon)\,\vert v\vert^2
\end{align*}
for every $x\in U$ and every $\smash{v\in T_x\mms}$ with $\vert v\vert^2>0$. Covering $\mms$ with such sets $U$ and sending $\varepsilon\to 0$, the claim follows.
\end{proof}

\subsubsection{Low regularity spacetimes} Following \cite{braunc1}, we can reduce the regularity of $\Rmet$ and $V$ from the previous subsection to local $\smash{\Cont^{1,1}}$-regularity, which we assume in this section --- otherwise, we retain the notation and the hypotheses therein. 

The \smash{Bakry--Émery--Ricci} tensor $\smash{\Ric^{N,V}}$ is still well-defined $\meas$-a.e., where again $\smash{\meas := \rme^{-V}\,\vol}$. Accordingly, given any $N>0$ we will say
\begin{align*}
\Ric^{N,V} \geq k\quad\textnormal{in all timelike directions}
\end{align*}
if for every smooth timelike vector field $X$ on $\mms$,  $\smash{\Ric^{N,V}(X,X) \geq k(x)\,\vert X\vert^2}$ holds for $\meas$-a.e.~$x\in\mms$. Following the sophisticated approximation argument from \cite[Lem.~2.8, Ch.~3]{braunc1} then leads to the following.

\begin{theorem}[Timelike convergence condition in low regularity]\label{Th:Low reg blub} Let $k\colon \mms\to \R$ be a continuous function, assume $g$ is of regularity $\smash{\Cont^{1,1}}$, and suppose $\smash{V \in \Cont^{1,1}(\mms)}$. Then for every $N \geq n$ the property
\begin{align*}
\Ric^{N,V} \geq k\quad\textnormal{\textit{in all timelike directions}} 
\end{align*}
implies the $\smash{\TCD_p^e(k,N)}$ condition for every $0<p<1$.
\end{theorem}

\begin{remark}[Converse implication] Even for constant $k$, it is still open whether the converse statement from \autoref{Th:Low reg blub} holds. In the positive signature  case, in this precise regularity this is only known under  additional hypotheses about the stability of optimal maps thus far \cite[Thm.~6.3, Rem.~6.4]{kunzober}.
\end{remark}

\begin{remark}[Stability]\label{Re:Stability} As   \autoref{Th:Low reg blub} is shown through approximation of  $\Rmet$ by smooth metrics, it can be seen as an instance of stability of our variable timelike curvature-dimension condition. 

We do not address the general question of stability, as a clear candidate for measured Lorentz--Gromov--Hausdorff convergence (e.g.~allowing for a precompactness theorem) is missing as of today. One proposal comes from \cite[Thm.~3.15]{cavalletti2020}, and further possibilites will be studied in \cite{braun+}. By using \autoref{Le:LSC sigma}, it is not hard to adapt the proof of \cite[Thm.~3.15]{cavalletti2020}, see also \cite[Thm.~3.29]{braun2023}, to show \autoref{Th:Stab}. 
\end{remark}

\begin{theorem}[Weak stability à la Cavalletti--Mondino]\label{Th:Stab} Let a sequence $(\scrM_n)_{n\in\N}$ converge to $\scrM_\infty$ in the pointed sense of \cite[Thm.~3.15]{cavalletti2020}, all spaces  satisfying \autoref{Ass:REG}. Let $(k_n)_{n\in\N}$ and $(N_n)_{n\in\N}$ be sequences of lower semicontinuous functions $k_n \colon \mms_n \to \R$ and numbers $1\leq N_n <\infty$ such that $\scrM_n$ satisfies $\smash{\TCD_p^e(k_n,N_n)}$, respectively. Lastly, let $k_\infty\colon \mms_\infty \to \R$ be lower semicontinuous and $1\leq N_\infty <\infty$ be such that
\begin{align}\label{Eq:APPR}
\begin{split}
\liminf_{n\to\infty} k_n &\geq k_\infty,\\
\limsup_{n\to\infty} N_n &\leq N_\infty.
\end{split}
\end{align}
Then $\scrM_\infty$ satisfies $\smash{\wTCD_p^e(k_\infty,N_\infty)}$.
\end{theorem}

Here the first inequality in \eqref{Eq:APPR} means for every $\varepsilon > 0$ there exists $n'\in\N$ such that for every $n\in\N$ with $n\geq n'$ and every $x\in\mms_n$,
\begin{align*}
k_n(x) \geq k\circ\iota_n(x) - \varepsilon,
\end{align*}
where $\iota_n$ is the isometric embedding defined on $\mms_n$ from \cite[Thm.~3.15]{cavalletti2020}.

If $\scrM_n$ is merely a timelike $p$-essentially nonbranching $\smash{\wTCD_p^e(k_n,N_n)}$ space for every $n\in\N$, by the previous result and \autoref{Th:Pathwise} we still obtain $\smash{\scrM_\infty}$ obeys the  $\smash{\wTCD_p^e(k_\infty,N_\infty)}$ condition.

A similar statement entails weak stability of the dimensionless variable timelike curvature-dimension conditon from \autoref{Def:TCD infty}.

If the local uniform lower boundedness in \eqref{Eq:APPR} is dropped, stability cannot be expected in general \cite[Ex.~33]{mccann+}.

\subsubsection{Finsler spacetimes} Next, we adapt the results from \cite{braunohta} for constant $k$ on Finsler spacetimes to variable $k$. We refer to \cite{braunohta, minguzzi2019} for details and references.

Let $\mms$ be a smooth, connected, noncompact topological manifold of dimension $\smash{n\in\N_{\geq 2}}$. A \emph{Lorentz--Finsler structure} on $\mms$ \cite[Def.~2.1]{braunohta} is a positively $2$-homogeneous function  $L\colon T\mms \to \R$, which outside the zero section is smooth and has non-degenerate Hessian with  signature $+,-,\dots,-$. Analogously to the metric tensor $\Rmet$ in \autoref{Sub:High}, one can use $L$ to build all elements of causality  theory to define \emph{Finsler spacetimes}, i.e.~Lorentz--Finsler structures with a time orientation \cite[Def.~2.3]{braunohta}. Let $\meas$ be a \emph{smooth}  measure on $\mms$, i.e.~$\meas$ is mutually absolutely continuous to the Lebesgue measure $\smash{\Leb^n}$ on each chart with smooth density. 

Uniqueness results for chronological $\smash{\ell_p}$-optimal couplings and $\smash{\ell_p}$-geodesics still hold  on Finsler spacetimes without curvature assumptions \cite[Thm.~4.17, Cor. 4.18]{braunohta}, in analogy to the results from \cite{mccann2020} mentioned in \autoref{Sub:High}.

As Finsler spacetimes come with several candidates for a  ``volume measure'' in general \cite[Sec. 9.1]{ohtacomp}, there is also no canonical Ricci tensor to consider. Given $\meas$, the correct choice is the weighted Ricci tensor $\smash{\Ric^N}$ it induces, as explained in \cite[Sec.~5.1]{braunohta}. We then say
\begin{align*}
\Ric^N \geq k\quad\textnormal{in all timelike directions}
\end{align*}
if $\smash{\Ric^N(v,v) \geq 2\,k(x)\,L(v)}$ for every $x\in\mms$ and every $v\in T_x\mms$ with $L(v)>0$.

The following theorem  echoes the analogous characterization for constant $k$ \cite[Thm.~5.9, Thm.~6.1]{braunohta}. The proof is performed along the same lines.
 
\begin{theorem}[Finslerian timelike convergence condition] Let $k\colon \mms\to \R$ be a continuous function. Then
\begin{align*}
\Ric^N \geq k\quad\textnormal{\textit{in all timelike directions}}
\end{align*}
if and only if $N\geq n$ and the $\smash{\wTCD_p^e(k,N)}$ condition holds for every $0<p<1$.

This equivalence still holds when $\smash{\wTCD_p^e(k,N)}$  is replaced by $\smash{\TCD_p^e(k,N)}$.
\end{theorem}

\subsection{Regular solutions to the Einstein equations} Akin  to the setting from \autoref{Sub:High}, let $\Rmet$ be a smooth, globally hyperbolic Lorentzian metric which solves the \emph{Einstein equations}
\begin{align}\label{Eq:Einstein equ}
\Ric - \frac{1}{2}\,\scal\,g + \Lambda\,g = 8\pi\,T.
\end{align}
Here $\Lambda\in \R$ is the cosmological constant, and $T$ is the energy-momentum tensor  from the least action principle for the Einstein--Hilbert action. In order to simplify the presentation, cf.~\eqref{Eq:RICN-2} below, we assume $\mms$ has dimension $\smash{n\in\N_{\geq 3}}$.

In the interpretation of \eqref{Eq:Inter TL Dir}, the  \emph{weak energy condition} asserts 
\begin{align*}
T\geq 0\quad\textnormal{in all timelike directions}.
\end{align*}
It clearly holds in vacuum, where $T=0$, in which case \eqref{Eq:Einstein equ} is equivalent to
\begin{align}\label{Eq:RICN-2}
\Ric = \frac{2\,\Lambda}{n-2}\,\Rmet
\end{align}
by \cite[Lem.~4.1]{mondinosuhr2022}. We thus land in the framework of timelike curvature-dimen\-sion conditions for constant $k$ again \cite{braun2022, cavalletti2020}. In the general non-vacuum case, assuming \eqref{Eq:Einstein equ} the  weak energy condition --- which is believed under all physically reasonable non-quantum circumstances \cite{wald1984} --- is equivalent to
\begin{align*}
\Ric \geq \frac{1}{2}\,\scal\,\Rmet - \Lambda\,\Rmet\quad\textnormal{in all timelike directions}.
\end{align*}

This motivates the consideration of the continuous function $\smash{k_\Rmet\colon\mms\to \R}$ with 
\begin{align*}
k_\Rmet(x) := \frac{1}{2}\,\scal\,\Rmet(x) - \Lambda.
\end{align*}
The proof of the following characterization is analogous to \autoref{Th:Strong energ}.

\begin{theorem}[Equivalence with the weak energy condition]  Assume $\scrM$ 
 smooth satisfies general relativity, i.e.~it comes from a solution to the Einstein equations \eqref{Eq:Einstein equ}. Then the weak energy condition is equivalent to $\smash{\TCD_p^e(k_\Rmet,n)}$ for every $0<p<1$.

This equivalence is still valid if $\smash{\TCD_p^e(k_\Rmet,n)}$ is replaced by $\smash{\wTCD_p^e(k_\Rmet,n)}$.
\end{theorem}

\subsection{Null convergence  condition} Unlike the positive signature case --- by the noncompactness of $\{\vert\cdot\vert^2  =1\}$  in $T_x\mms$, where $x\in\mms$ --- existence of a \emph{real-valued} function $k$ obeying \eqref{Eq:TL Ric} is an extra condition on $\mms$ and not automatically satisfied by every smooth Lorentzian spacetime. In fact, the existence of such a function $k$ is characterized by the null energy condition (or \emph{null convergence condition} if general relativity is not assumed). This has been observed in \cite{mccann+}; we briefly recapitulate its findings and translate the main result from \cite{mccann+} to our variable language.

Define the function $c\colon \mms \to \R\cup \{-\infty\}$ by
\begin{align}\label{Eq:c Ric def}
c(x) := \inf \vert v\vert^{-2}\,\Ric^{N,V}(v,v),
\end{align}
where the infimum is taken over all $v\in T_x\mms$ such that $\vert v\vert^2 > 0$. This function is only upper semicontinuous in general. Therefore, we have to consider its (possibly degenerate) lower semicontinuous envelope $k\colon \mms \to \R\cup \{-\infty\}$ defined by
\begin{align}\label{Eq:k Ric def}
k(x) := \inf \liminf_{n\to\infty} c(x_n),
\end{align}
where the infimum is taken over all sequences $(x_n)_{n\in\N}$ in $\mms$ converging to $x$.

In the sequel, we will say
\begin{align*}
\Ric^{N,V}\geq 0\quad\textnormal{in all null directions}
\end{align*}
if $\smash{\Ric^{N,V}(v,v)\geq 0}$ for every $x\in\mms$ and every $v\in T_x\mms$ with  $\vert v\vert^2=0$.

\begin{theorem}[Equivalence with the null convergence  condition] The condition
\begin{align*}
\Ric^{N,V} \geq 0\quad\textnormal{\textit{in all null directions}}
\end{align*}
holds if and only if the following are satisfied.
\begin{enumerate}[label=\textnormal{\textcolor{black}{(}\roman*\textcolor{black}{)}}]
\item\label{La:Leelu} The function $k$ from \eqref{Eq:k Ric def} does not attain the value $-\infty$.
\item We have $N\geq n$.
\item\label{La:Lelu} For every $0<p<1$, the $\smash{\wTCD_p^e(k,N)}$ condition holds.
\end{enumerate}

Moreover, the preceding equivalence still holds if $\smash{\wTCD_p^e(k,N)}$ is replaced by $\smash{\TCD_p^e(k,N)}$ in item  \ref{La:Lelu}.
\end{theorem}

\begin{proof} Again the last statement follows from  \autoref{Th:Pathwise} since in either case, \ref{La:Leelu} ensures $k$ is real-valued.

As concerns the rest, we first show the forward implication. The hypothesized weighted null energy condition implies the function $c$ from \eqref{Eq:c Ric def} is bounded from below on every compact subset of $\mms$ \cite[Thm.~26]{mccann+}. Thus, $k$ does not attain the value $-\infty$. The rest is then argued as in the proof of \autoref{Th:Strong energ}.

The key point of the converse implication is again \ref{La:Leelu}. Indeed, this hypothesis paired with lower semicontinuity of $k$ ensures $k$ is bounded from below on every compact subset of $\mms$. Thus for every compact subset $C\subset \mms$, there exists $K\in\R$ such that the metric measure spacetime structure on $J(C,C)$ induced by $\scrM$ obeys  $\smash{\wTCD_p^e(K,N)}$ for every $0<p<1$. According to \cite[Cor.~27]{mccann+} this yields the weighted null energy condition.
\end{proof}

\section{Geometric inequalities}\label{Sec:Geometric inequ}

We turn to standard geometric applications of timelike Ricci curvature bounds from below. In this chapter, from $\smash{\wTCD_p^e(k,N)}$ we deduce
\begin{itemize}
\item the timelike Brunn--Minkowski inequality, 
\item the timelike Bonnet--Myers diameter inequality, and
\item the timelike Bishop--Gromov volume growth inequality, along with several properties of $\meas$.
\end{itemize}
 The constants obtained in this chapter are not sharp (e.g.~\autoref{Cor:Hausdorffdim}), but this will
be rectified in \autoref{Ch:Applications} using the needle decompositions of \autoref{Ch:Localization}. 

In \autoref{Th:Schneider} we also state a related result known as \emph{Schneider's theorem} in Riemannian geometry \cite{schneider1972}, see also \cite{aubry2007,galloway1982}, considered in the setting of smooth spacetimes by Frankel--Galloway  \cite{fg}. It is transferred from \cite{ketterer2017} to our setting. In terms of general relativity, it predicts timelike geodesic incompleteness if,  in the far past (or future) of the point of reference,
the  Ricci curvature is positive and does not decay too quickly.

\subsection{Timelike Brunn--Minkowski inequality}\label{Sub:Brunn Minkowski} For any Borel sets $A_0,A_1\subset\mms$ and any $t\in (0,1)$, we define the set
\begin{align}\label{Eq:At}
A_t' := Z_t(A_0\times A_1),
\end{align}
where the latter is defined according to \eqref{Eq:Zt(C)}. 
Furthermore, let $G(A_0,A_1)$ denote the class of all $\gamma\in\TGeo(\mms)$ with $\gamma_0\in A_0$ and $\gamma_1 \in A_1$.

\begin{theorem}[Timelike Brunn--Minkowski inequality]\label{Th:Timelike BM} Assume $\smash{\TCD_p^e(k,N)}$. Let $A_0,A_1\subset\mms$ be Borel subsets of $\mms$  with positive and finite $\meas$-measure such that $\smash{(\mu_0,\mu_1)\in\scrP^\ac(\mms,\meas)^2}$ is  timelike $p$-dualizable, where
\begin{align}\label{Eq:Pair}
\begin{split}
\mu_0 &:= \meas[A_0]^{-1}\,\meas\mres A_0,\\
\mu_1 &:= \meas[A_1]^{-1}\,\meas\mres A_1.
\end{split}
\end{align}
Then for every $0\leq t\leq 1$,
\begin{align*}
\meas[A_t']^{1/N} &\geq \inf \sigma_{k_\gamma^-/N}^{(1-t)}(\vert\dot\gamma\vert)\,\meas[A_0]^{1/N} + \inf\sigma_{k_\gamma^+/N}^{(t)}(\vert\dot\gamma\vert)\,\meas[A_1]^{1/N},
\end{align*}
where both infima are taken over all $\smash{\gamma\in G(A_0,A_1)}$. Here, $\meas[A_t']$ is interpreted as the outer measure of $A_t'$ in case the latter is not $\meas$-measurable.
\end{theorem}

\begin{proof} By inner regularity of $\meas$, it suffices to prove the claim for compact $A_0$ and $A_1$. In this case, $A_t'$ is compact by \autoref{Le:Zt lemma}, and thus $\meas$-measurable.

By $\smash{\TCD_p^e(k,N)}$ the pair $(\mu_0,\mu_1)$ are connected by an $\smash{\ell_p}$-geodesic $(\mu_t)_{t\in [0,1]}$ from $\mu_0$ to $\mu_1$ such that for some $\bdpi\in \OptTGeo_p(\mu_0,\mu_1)$, we have
\begin{align*}
\scrU_N(\mu_t) \geq \sigma_{k_\bdpi^-/N}^{(1-t)}(\cost_\bdpi)\,\scrU_N(\mu_0) + \sigma_{k_\bdpi^+/N}^{(t)}(\cost_\bdpi)\,\scrU_N(\mu_1).
\end{align*}
To estimate the left-hand side from below, we use the  identities
\begin{align}\label{Eq:Identities}
\begin{split}
\scrU_N(\mu_0) &= \meas[A_0]^{1/N},\\
\scrU_N(\mu_1) &= \meas[A_1]^{1/N}
\end{split}
\end{align}
and that, by \autoref{Le:Properties} and Jensen's inequality, 
\begin{align*}
\sigma_{k_\bdpi^-/N}^{(1-t)}(\cost_\bdpi) \geq \int \sigma_{k_\gamma^-/N}^{(1-t)}(\vert\dot\gamma\vert)\d\bdpi(\gamma) \geq \inf \sigma_{k_\gamma^-/N}^{(1-t)}(\vert\dot\gamma\vert).
\end{align*}
where the infimum is taken over all $\gamma \in G(A_0,A_1)$; analogously, 
\begin{align*}
\sigma_{k_\bdpi^+/N}^{(t)}(\cost_\bdpi) \geq \inf \sigma_{k_\gamma^+/N}^{(t)}(\vert\dot\gamma\vert)
\end{align*}
with the infimum again running over all $\gamma\in \supp\bdpi$. 
Here we have implicitly used \eqref{Eq:Scaling prop sigma} twice. On the other hand, from \eqref{Eq:UN Jensen} and \autoref{Le:Intermediate pts geodesics} we obtain
\begin{align*}
\scrU_N(\mu_t) \leq \meas\big[\!\supp\mu_t\big]^{1/N} \leq \meas[A_t']^{1/N}.\tag*{\qedhere}
\end{align*}
\end{proof}

The same strategy, together with the lifting of chronological $\smash{\ell_p}$-optimal couplings from \autoref{Le:Geodesics plan}, lead to a Brunn--Minkowski-type inequality in the infinite-dimensional case.  Given any $0\leq t\leq 1$ we define $\smash{\ubar{k}_t \colon M^2\to \R}$ by
\begin{align*}
\ubar{k}_t(x,y) := \inf \int_0^1 \rmg(s,t)\,k(\gamma_s)\d s,
\end{align*}
where the infimum is taken over all $\gamma\in\TGeo(\mms)$ with $\gamma_0=x$ and $\gamma_1=y$, and $\rmg$ is from \eqref{Eq:1d Green's function}. This function should be compared to the function $\smash{\ubar{k}}$ from \cite{braun2021}, which has been extensively used to formulate and investigate  transport estimates for the heat flow on metric measure spaces with variable curvature bounds.

\begin{corollary}[Dimensionless timelike Brunn--Minkowski inequality]\label{Cor:BM inf dim} We assume $\smash{\TCD_p(k,\infty)}$ and retain the notation from  \autoref{Th:Timelike BM}. Then for every $0\leq t\leq 1$,
\begin{align*}
\log\meas[A_t'] &\geq (1-t)\log\meas[A_0] + t\log\meas[A_1]   - \inf \int_{\mms^2} \ubar{k}_t\,l^2\d\pi,
\end{align*}
where the infimum is taken over all $\smash{\ell_p}$-optimal couplings $\smash{\pi\in\Pi_\ll(\mu_0,\mu_1)}$. Again, $\meas[A_t']$ is interpreted as the outer measure of $A_t'$ if the latter is not $\meas$-measurable.
\end{corollary}

\begin{remark}[Variants]\label{Re:From geo to displ} Of course, \autoref{Th:Timelike BM} and \autoref{Cor:BM inf dim} hold analogously under $\smash{\wTCD_p^e(k,N)}$ and $\smash{\wTCD_p(k,\infty)}$. In both situations, however, one has to assume \emph{strong} timelike $p$-dualizability of the pair defined in \eqref{Eq:Pair}. 

In the case that $\scrM$ is a timelike $p$-essentially nonbranching $\smash{\wTCD_p^e(k,N)}$ space, better inequalities can be achieved if $(\mu_0,\mu_1)$ is strongly timelike $p$-dualizable. Indeed, by \autoref{Th:Uniqueness geos} and \autoref{Cor:Lifting nonbr} the unique $\smash{\ell_p}$-geodesic connecting $\mu_0$ to $\mu_1$ is in fact a displacement $\smash{\ell_p}$-geodesic, represented by $\smash{\bdpi\in\OptTGeo_p(\mu_0,\mu_1)}$. For every $0\leq t\leq 1$, arguing as for \autoref{Th:Timelike BM} then yields 
\begin{align*}
\meas\big[\!\supp(\eval_t)_\push\bdpi\big]^{1/N} \geq \inf \sigma_{k_\gamma^-/N}^{(1-t)}(\vert\dot\gamma\vert)\,\meas[A_0]^{1/N} + \inf\sigma_{k_\gamma^+/N}^{(t)}(\vert\dot\gamma\vert)\,\meas[A_1]^{1/N},
\end{align*}
where both infima are taken over all $\gamma\in \supp\bdpi$. Thus, considering the Suslin set
\begin{align}\label{Eq:At'}
A_t := \eval_t\big[G(A_0,A_1)\big],
\end{align}
which is no larger than $A_t'$ by \autoref{Le:Intermediate pts geodesics}, we have
\begin{align*}
\meas[A_t]^{1/N} \geq \inf \sigma_{k_\gamma^-/N}^{(1-t)}(\vert\dot\gamma\vert)\,\meas[A_0]^{1/N} + \inf\sigma_{k_\gamma^+/N}^{(t)}(\vert\dot\gamma\vert)\,\meas[A_1]^{1/N},
\end{align*}
where both infima are taken over all $\gamma\in G(A_0,A_1)$.

If $A_0\times A_1\Subset\{l>0\}$, the above assumption of timelike $p$-essential nonbranching can be dropped by \autoref{Le:Geodesics plan}.
\end{remark}

\subsection{Timelike Bonnet--Myers inequality} In this discussion, we will denote by $\smash{\OptTGeo_p^{\ac}(\scrP_\comp(\mms)^2,{ \meas})}$ the set of all $\smash{\bdpi\in\OptTGeo_p(\scrP_\comp(\mms)^2)}$ such that $(\eval_t)_\push\bdpi$ is $\meas$-absolutely continuous for every $0\leq t\leq 1$. We also recall the definition
\begin{align*}
\cost_\bdpi := \big\Vert l \circ (\eval_0,\eval_1)\big\Vert_{L^2(\TGeo(\mms),\bdpi))}
\end{align*}
of the $\Ell^2$-cost associated with $\bdpi$.

\begin{definition}[Effective diameter] The \emph{effective $l$-diameter} of $\mms$ with respect to $k$ and $N$ is defined by
\begin{align*}
\diam_{k/N}\mms &:= \sup  \cost_\bdpi,
\end{align*}
where the supremum is taken over all $\smash{\bdpi\in\OptTGeo_p^\ac(\scrP_\comp(\mms)^2,\meas)}$ such that for some, or equivalently by \autoref{Le:Properties}, \emph{every} $0\leq t\leq 1$, we have  $\smash{\sigma_{k_\bdpi^+/N}^{(t)}(\cost_\bdpi) < \infty}$.
\end{definition}

By definition of the effective $l$-diameter, we have
\begin{align}\label{Eq:Inequ diameters}
\diam_{k/N}\mms \leq \diam^l\supp\meas,
\end{align}
where the latter is the largest value $l$ can take on $\smash{(\supp\meas)^2}$.

\begin{theorem}[Timelike Bonnet--Myers inequality]\label{Th:BonnetMyers} If $\scrM$ satisfies $\smash{\wTCD_p^e(k,N)}$, then equality holds throughout \eqref{Eq:Inequ diameters}, i.e.
\begin{align*}
\diam_{k/N}\mms = \diam^l\supp\meas.
\end{align*}
\end{theorem}

\begin{proof} Assume the existence of $\varepsilon > 0$ and  $x,y\in\supp\meas$ such that
\begin{align*}
\diam_{k/N}\mms + 2\,\varepsilon\leq l(x,y), 
\end{align*}
whence $x\ll y$. By $l$-geodesy of $\mms$ and \autoref{Cor:Hausdorff}, there are precompact, open sets $A_0,A_1\subset\mms$ with
\begin{align}\label{Eq:diameps}
\diam_{k/N}\mms + \varepsilon \leq \inf\{l(x',y') : x'\in A_0,\,y'\in A_1\}.
\end{align}
Finally, in terms of $A_0$ and $A_1$ we define $\smash{\mu_0,\mu_1\in\scrP_\comp^\ac(\mms,\meas)}$ by \eqref{Eq:Pair}. Then the  pair $(\mu_0,\mu_1)$ is strongly timelike $p$-dualizable by \eqref{Eq:diameps} and \autoref{Ex:Str tl dual}. 

Every $\smash{\bdpi\in\OptTGeo_p(\mu_0,\mu_1)}$ and every $0< t< 1$ obeys
\begin{align}\label{Eq:sigma infty}
\sigma_{k_\bdpi^-/N}^{(1-t)}(\cost_\bdpi) =  \sigma_{k_\bdpi^+/N}^{(t)}(\cost_\bdpi) = \infty
\end{align}
by \autoref{Le:Properties} and our assumption. On the other hand, our curvature assumption provides us with an $\smash{\ell_p}$-geodesic $(\mu_t)_{t\in [0,1]}$ and $\smash{\bdpi\in \OptTGeo_p^\ac(\mu_0,\mu_1,\meas)}$ certifying the displacement semiconvexity inequality of the functional $\scrU_N$. Let $A_t'$ be the set from \eqref{Eq:At} in terms of $A_0$ and $A_1$. Combined with Jensen's inequality, \eqref{Eq:sigma infty}, and the identities \eqref{Eq:Identities} for $\scrU_N(\mu_0)$ and $\scrU_N(\mu_1)$, our discussion yields
\begin{align*}
\meas\big[J(\supp\mu_0,\supp\mu_1)\big]^{1/N}\geq \meas[A_t']^{1/N} \geq \scrU_N(\mu_t)\geq \infty.
\end{align*}
Since $\meas$ is finite on compact sets, this is however impossible.
\end{proof}

\subsection{Timelike Schneider inequality} 

\begin{theorem}[Past timelike Schneider inequality]\label{Th:Schneider} We assume $\smash{\wTCD_p^e(k,N)}$ and suppose that some $o\in\supp\meas$, $\beta> 0$, and $R>0$ satisfy
\begin{align*}
k\geq N\,\Big[\frac{1}{4}+\beta^2\Big]\,l(\cdot,o)^{-2}\quad\textnormal{\textit{on} }\{l(\cdot,o) > R\}.
\end{align*}
Then $\smash{\diam^l (I^-(o) \cap\supp\meas) \leq R\,\rme^{\pi/\beta}}$.
\end{theorem}

\begin{proof} The proof is comparable to the argument for  \autoref{Th:BonnetMyers}. For notational convenience, we set $x_1 := o$. Assume to the contrapositive that some $x_0 \in \supp\meas$ and some $0<\varepsilon <R$ satisfy the inequality
\begin{align}\label{Eq:l01 inequ}
l(x_0,x_1)> (R+4\,\varepsilon\,\rme^{\pi/\beta})\,\rme^{\pi/\beta}.
\end{align}
As in the proof of \autoref{Th:BonnetMyers}, we choose  precompact, open neighborhoods $A_0\subset\mms$ of $x_0$ and $A_1\subset\mms$ of $x_1$ such that for every $\smash{x\in \bar{A}_0}$ and every $\smash{y\in \bar{A}_1}$,
\begin{align}\label{Eq:l inequ}
\big\vert l(x,y) - l(x_0,x_1) \big\vert \leq \varepsilon\,\rme^{\pi/\beta}.
\end{align}
In particular, for every timelike maximizer $\smash{\bar\gamma\colon[0, l(x,x_1)]\to \mms}$ with proper-time parametrization connecting $x$ to $x_1$ we have, whenever $\smash{0<r< \varepsilon}$,
\begin{align*}
l(\bar\gamma_r, x_1) &= l(x,x_1)-r\\
&\geq l(x_0,x_1) - \varepsilon\,\rme^{\pi/\beta} -r \\
&> (R+ 4\,\varepsilon\,\rme^{\pi/\beta})\,\rme^{\pi/\beta} - \varepsilon\,\rme^{\pi/\beta} -r\\
&> R+\varepsilon - r\\
&> R.
\end{align*}
For such $r$, our hypothesis on $k$ together with \eqref{Eq:l inequ} thus yields
\begin{align*}
k(\bar\gamma_r) &\geq N\,\Big[\frac{1}{4}+\beta^2\Big]\,l(\bar\gamma_r,x_1)^{-2}\\
&= N\,\Big[\frac{1}{4}+\beta^2\Big]\,\big[l(x,x_1) - r\big]^{-2}\\
&\geq N\,\Big[\frac{1}{4}+\beta^2\Big]\,\big[\varepsilon\,\rme^{\pi/\beta} + l(x_0,x_1) - r\big]^{-2}.
\end{align*}

Next, define $\smash{\kappa\colon [0, \varepsilon] \to \R}$ through
\begin{align*}
\kappa(r) := \Big[\frac{1}{4}+\beta^2\Big]\,\big[\eta + l(x_0,x_1) -r\big]^{-2},
\end{align*}
where $\smash{\eta := \varepsilon\,\rme^{\pi/\beta}}$. In the notation of \eqref{Eq:kgamma defn}, the preceding paragraph  provides the bound $\smash{k_\gamma^+/N} \geq \kappa$ on $\smash{[0,\varepsilon]}$, which directly transfers to $\smash{\SIN_{k_\gamma^+/N} \geq \SIN_\kappa}$ on the same interval by \autoref{Re:Indep E}. We can explicitly express $\smash{\SIN_\kappa}$ as
\begin{align*}
\SIN_\kappa(r) = c\,\sqrt{\eta + l(x_0,x_1) - r}\,\sin\!\Big[\beta\log\!\Big[\frac{\eta  + l(x_0,x_1) - r}{\eta + l(x_0,x_1)}\Big] + \pi\Big]
\end{align*}
for a normalization constant $c>0$. Of course, the first zero of $\SIN_\kappa$ is $r_1 = 0$, and its second zero is $r_2 = (\eta + l(x_0,x_1))\,(1-\rme^{-\pi/\beta})$. Consequently, the  inequalities $\varepsilon<R$, \eqref{Eq:l01 inequ}, and \eqref{Eq:l inequ} yield, for every $\smash{x\in \bar{A}_0}$,
\begin{align}\label{Eq:r2 est}
\begin{split}
r_2 &= \varepsilon\,\rme^{\pi/\beta}\,(1-\rme^{-\pi/\beta}) + l(x_0,x_1) - l(x_0,x_1)\,\rme^{-\pi/\beta}\\
&< \varepsilon\,\rme^{\pi/\beta}\,(1-\rme^{-\pi/\beta}) + l(x,y) + \varepsilon\,\rme^{\pi/\beta} -R - 4\,\varepsilon\,\rme^{\pi/\beta}\\
&= l(x,y) - R - \varepsilon - 2\,\varepsilon\,\rme^{\pi/\beta}\\
&< l(x,y) - 2\,\varepsilon.
\end{split}
\end{align}

To wrap up the argument and to arrive at the desired contradiction, we falsify a weaker version of the timelike Brunn--Minkowski inequality from \autoref{Th:Timelike BM}. Consider the set of $l$-geodesics $G(A_0,A_1)$ connecting $A_0$ and $A_1$ as well as the associated $t$-intermediate point set $A_t'$, where $0\leq t\leq 1$, from \eqref{Eq:At}. Moreover, in terms of $A_0$ and $A_1$ we define $\smash{\mu_0,\mu_1\in\scrP_\comp^\ac(\mms,\meas)}$ as in \eqref{Eq:Pair}, respectively. By construction and \autoref{Ex:Str tl dual}, the pair $(\mu_0,\mu_1)$ is strongly timelike $p$-dualizable. Furthermore, if $0 <r<\varepsilon$, by  \eqref{Eq:r2 est} every $\gamma\in G(A_0,A_1)$ saturates
\begin{align*}
\vert\dot\gamma\vert - r = l(\gamma_0,\gamma_1) - r > r_2 + \varepsilon.
\end{align*}
In other words, for sufficiently small $r>0$, the argument $\vert \dot\gamma\vert -r$ lies strictly past the first zero of $\kappa$ for every $\gamma\in G(A_0,A_1)$. The nondegeneracy and monotonicity properties from \autoref{Le:Properties} and the definition of the distortion coefficients under consideration  imply for sufficiently small $t>0$ that
\begin{align*}
\meas[A_t']^{1/N} \geq \inf \sigma_{k_\gamma^-/N}^{(1-t)}(\vert\dot\gamma\vert)\,\meas[A_0]^{1/N} = \infty,
\end{align*}
where the infimum is taken over all $\gamma\in G(A_0, A_1)$. As in the proof of \autoref{Th:BonnetMyers}, this contradicts the Radon property of $\meas$.
\end{proof}

By \autoref{Re:Causal reversal}, \autoref{Th:Schneider} readily yields the following.

\begin{corollary}[Future timelike Schneider inequality]\label{Cor:FUTURE} Assume $\smash{\wTCD_p^e(k,N)}$, and suppose that some $o\in\supp\meas$, $\beta > 0$, and $R>0$ satisfy
\begin{align*}
k \geq N\,\Big[\frac{1}{4}+\beta^2\Big]\,l(o,\cdot)^{-2}\quad\textnormal{\textit{on} }\{l(o,\cdot)>R\}. 
\end{align*}
Then $\smash{\diam^l (I^+(o) \cap\supp\meas) \leq R\,\rme^{\pi/\beta}}$.
\end{corollary}

\subsection{Timelike Bishop--Gromov inequality}\label{Sub:Bishop Gromov}

Following \cite[Sec.~3.1]{cavalletti2020}, let
\begin{align*}
\sfB^l(x,r) := \{y\in I^+(x) : l(x,y) < r\}\cup\{x\}
\end{align*}
be the future ``ball''  anchored at $x\in \mms$ with radius $r>0$. Such a set typically has infinite volume, e.g.~it is a hyperboloid in flat Minkowski space. To take the $\meas$-measure of parts of such balls, one localizes them by star shaped sets. Here a set $E\subset I^+(x) \cup\{x\}$ is called \emph{$l$-star shaped} with respect to $x$ if every $\gamma\in\TGeo(\mms)$ with $\gamma_0=x$ and $\gamma_1 \in E$ does not leave $E$. Let $E$ be such a set which we assume to be compact in addition.  Using the Radon property of $\meas$, we set
\begin{align}
\label{Eq:volume}
\begin{split}
\sfv(r) &:= \meas\big[\bar{\sfB}^l(x,r) \cap E\big],\\
\sfs(r) &:= \limsup_{\delta \to 0} \delta^{-1}\,\meas\big[(\bar{\sfB}^l(x,r+\delta)\setminus \sfB^l(x,r)) \cap E\big]
\end{split}
\end{align}
where $\bar \sfB^l(x,r)$ denotes the closure of $\sfB^l(x,r)$.

The following construction of compact and $l$-star shaped sets will be useful in the proof of \autoref{Th:Bishop Gromov}. Here, analogously to \eqref{Eq:At}, $G(x,A)$ denotes the set of all $\gamma\in\TGeo(\mms)$ with $\gamma_0 = x$ and $\gamma_1\in A$, where $x\in\mms$ and $A\subset \mms$.

\begin{lemma}[Filled future cone]\label{Le:Star-shaped} Given any $x\in \mms$ and any compact $C\subset\mms$ with $\{x\}\times C \subset\{l>0\}$,  the set
\begin{align*}
E := \bigcup_{t\in [0,1]} \eval_t\big[G(x, C)\big]
\end{align*} 
is contained in $I^+(x) \cup\{x\}$, compact, and $l$-star shaped.
\end{lemma}

\begin{proof} The inclusion $E\subset I^+(x) \cup\{x\}$ is clear by \autoref{Cor:Causality}.

Now we show compactness of $E$. Since $E\subset J(x,C)$, by \autoref{Th:Compact Polish} it suffices to prove its sequential compactness. Therefore, let $(z_n)_{n\in\N}$ be a sequence in $E$ and write $\smash{z_n := \gamma^n_{t_n}}$ for some  sequences $\smash{(\gamma^n)_{n\in\N}}$ in $G(x,C)$ and $(t_n)_{n\in\N}$ in $[0,1]$, respectively. Since $\inf l(\{x\}\times C) > 0$,  \autoref{Cor:Cpt TGeo} entails the existence of a uniform limit $\gamma\in G(x,C)$ of a nonrelabeled subsequence of $(\gamma^n)_{n\in\N}$. By possibly passing to a further nonrelabeled subsequence, the sequence $(t_n)_{n\in\N}$ converges to some $0\leq t\leq 1$. In particular, $(z_n)_{n\in\N}$ converges to $\gamma_t \in E$.

We are left with $l$-star shapedness of $E$. Let $\gamma\in\TGeo(\mms)$ with $\gamma_0 = x$ and $\gamma_1\in E$. If $\gamma_1 \in C$ we directly get $\gamma_{[0,1]}\subset E$ by definition of $E$. Otherwise, there exist $\sigma\in G(x,C)$ and $0<t<1$   with $\gamma_1 = \sigma_t$. It follows that
\begin{align*}
l(\gamma_0,\gamma_1) = l(x,\sigma_t) = t\,l(x,\sigma_1).
\end{align*}
On the other hand, any $\alpha\in \TGeo(\mms)$ with $\alpha_0 = \sigma_t$ and $\alpha_1 = \sigma_1$ has length
\begin{align*}
l(\alpha_0,\alpha_1) = l(\sigma_t,\sigma_1) = (1-t)\,l(x,\sigma_1).
\end{align*}
As in the proof of \autoref{Le:Continuity geos}, an appropriate concatenation of $\gamma$ and $\alpha$  yields an element of $\TGeo(\mms)$ starting at $x$, terminating in $C$, and containing $\gamma_{[0,1]}$ as a sub\-seg\-ment. By the definition of $E$, we thus get $\smash{\gamma_{[0,1]}\subset E}$.
\end{proof}

For brevity, a \emph{future sphere} is any set of the form
\begin{align*}
\sfS^l(x,r) :=  \{l(x,\cdot)=r\}= (\bar{\sfB}^l(x,r) \setminus \sfB^l(x,r)) \setminus \{l(x,\cdot) = 0\},
\end{align*}
where $x\in \mms$ and $r>0$. \emph{Past spheres} are defined accordingly with respect to the causal reversal from \autoref{Ex:Causal reversal}. 
If the time orientation does not matter, we simply use the term \emph{sphere}.

Next, we foliate the function $k$ along spheres throughout $E$ as follows. Since the set $E$ is compact, so is $ E_r := \smash{\sfS^l(x,r)\cap E}$ for every $r> 0$ (and contained in $\{l(x,\cdot) = r\}$), and also the intersection of $\smash{\sfS^l(x,0) := \{l(x,\cdot)=0\}}$ with $E$. Hence $k$ assumes its minimum, say $k_x(r)$,  on $E_r$ if the latter is nonempty. Let 
\begin{align*}
R_x := \sup\{r \geq 0 : \sfS^l(x,r) \cap E \neq \emptyset\},
\end{align*}
and let $\smash{\ubar{k}_x\colon [0,R_x]\to \R}$ denote the lower semicontinuous envelope of $k_x$, i.e.
\begin{align}\label{Eq:k r function}
\ubar{k}_x(r) := \inf \liminf_{n\to\infty} k_x(r_n),
\end{align}
where the infimum is taken over all sequences $(r_n)_{n\in\N}$ in $[0,R_x]$ converging to $r$, where $[0,R_x] := [0,\infty)$ provided $R_x = \infty$. By a slight abuse of notation, we will also consider the function $\smash{\ubar{k}_x\colon J^+(x) \to \R}$ defined by
\begin{align}\label{Eq:BLU}
\ubar{k}_x(y) := \ubar{k}_x(l(x,y)).
\end{align}

\begin{theorem}[Timelike Bishop--Gromov inequality]\label{Th:Bishop Gromov} Assume $\smash{\wTCD_p^e(k,N)}$. Furthermore, let $E\subset \mms$ be compact and $l$-star shaped with respect to $x\in\supp\meas$. Then for every $\smash{0<r<R < R_x}$, the quantities from  \eqref{Eq:volume}  satisfy
\begin{align}\label{Eq:BG}
\begin{split}
\frac{\sfv(r)}{\sfv(R)} &\geq \frac{\displaystyle\int_0^r\SIN_{\ubar{k}_x/N}^N(t)\d t}{\displaystyle\int_0^R\SIN_{\ubar{k}_x/N}^N(t)\d t},\\
\frac{\sfs(r)}{\sfs(R)} &\geq \frac{\SIN_{\ubar{k}_x/N}^N(r)}{ \SIN_{\ubar{k}_x/N}^N(R)}.
\end{split}
\end{align}

In particular, $\meas$ neither gives mass to points in $\smash{I^\pm(x)}$, nor to spheres at $x$. 

Finally, if $\meas$ has an atom at $x$, in fact we have
\begin{align}\label{Eq:Future empty}
I^\pm(x) \cap \supp\meas = \emptyset.
\end{align}
\end{theorem}

\begin{proof} For simplicity, assume $\smash{\meas[E\cap U] > 0}$ for every open neigh\-borhood $U\subset\mms$ of $x$; we outline the necessary modifications in the more general case later.

We first prove the second estimate from \eqref{Eq:BG}. Without restriction, we may and will assume $\sfs(R) > 0$. Lastly, by applying the regularization procedure \eqref{Eq:kappan} to the function $\smash{\ubar{k}_x}$ from \eqref{Eq:k r function} relative to the compact space $[0,R+\delta\,R]$ endowed with the  Euclidean metric, we will also first suppose the induced function $\smash{\ubar{k}_x}$ from \eqref{Eq:BLU} to be bounded and continuous on $\smash{\bar{\sfB}^l(x,R + \delta\,R)\cap E}$ for sufficiently small $\delta> 0$.

Let $\mu_0 := \delta_{A_0}$ and, given $\delta > 0$, let $\smash{\mu_1 \in\scrP_\comp^\ac(\mms,\meas)}$ be the uniform distribution of $\smash{A_1}$ with respect to $\meas$, where
\begin{align*}
A_0 &:= \{x\},\\
A_1 &:= (\bar{\sfB}^l(x,R+ \delta\,R) \setminus \sfB^l(x,R))\cap E.
\end{align*}
When $R+\delta\,R < R_x$, we first claim the existence of $\smash{\bdpi\in\OptTGeo_p(\mu_0,\mu_1)}$ with
\begin{align}\label{Eq:Key est}
\begin{split}
&\meas\big[(\bar{\sfB}^l(x,r+\delta\,r)\setminus \sfB^l(x,r))\cap E\big]^{1/N}\\
&\qquad\qquad \geq \int \sigma_{\ubar{k}_{x,\gamma}^+/N}^{(t)}(\vert\dot\gamma\vert)\d\bdpi(\gamma)\,\meas\big[(\bar{\sfB}^l(x,R+\delta\,R)\setminus \sfB^l(x,R))\cap E\big]^{1/N}.
\end{split}
\end{align}
Since $\sfs(R)>0$, the measure $\mu_1$ is well-defined, and \eqref{Eq:Key est} is not void.

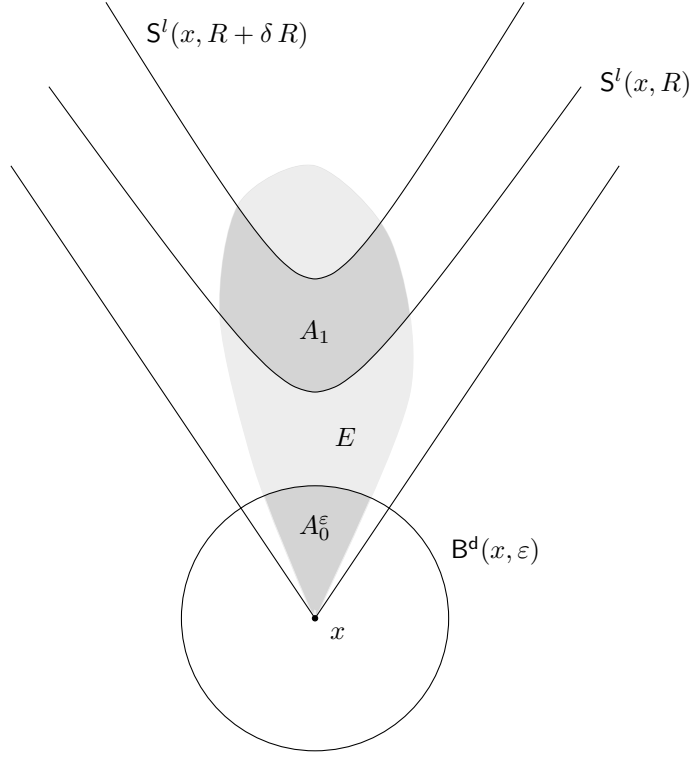
\begin{figure}
\centering
\begin{tikzpicture}
\def\sete{plot [smooth cycle] coordinates {(0,0)  (-0.005,0.1) (-0.1,0.2) (-0.8,2) (-1.25,4) (-1,5.5) (0,6) (1,5) (1.25,3) (0.1,0.2) (0.005,0.1) (0,0)}}
\def\circle{(0,0) circle(5em)}

    \fill[black!7] \sete; 
    \fill[white] \circle;

   \begin{scope}
        \clip \sete;
        \fill[black!17] \circle;
    \end{scope}
 
    \draw[black!7] \sete;
    \draw[black] \circle;

\begin{scope}
\clip[scale=1, domain=-3.5:3.5, smooth, variable=\x] plot ({\x}, {(2*\x*\x+1)^(1/2)+2});
\fill[black!17] plot [smooth cycle] coordinates {(0,0)  (-0.005,0.1) (-0.1,0.2) (-0.8,2) (-1.25,4) (-1,5.5) (0,6) (1,5) (1.25,3) (0.1,0.2) (0.005,0.1) (0,0)};
\end{scope}

\begin{scope}
  \clip[scale=1, domain=-2.5:2.5, smooth, variable=\x] plot ({\x}, {(2.75*\x*\x+1)^(1/2)+3.5});
  \fill[black!7] plot [smooth cycle] coordinates {(0,0)  (-0.005,0.1) (-0.1,0.2) (-0.8,2) (-1.25,4) (-1,5.5) (0,6) (1,5) (1.25,3) (0.1,0.2) (0.005,0.1) (0,0)};
\end{scope}
    
\draw[fill=black] (0,0) circle(.1em);
\node at (0.3,-0.2) {$x$};
\draw (0,0) -- (-4,6);
\draw (0,0) -- (4,6);
\draw[scale=1, domain=-3.5:3.5, smooth, variable=\x] plot ({\x}, {(2*\x*\x+1)^(1/2)+2});
\draw[scale=1, domain=-2.75:2.75, smooth, variable=\x] plot ({\x}, {(2.75*\x*\x+1)^(1/2)+3.5});
\node at (2.375,0.9) {$\sfB^\met(x,\varepsilon)$};
\node at (0.4,2.4) {$E$};
\node at (-1.15,7.75) {$\sfS^l(x, R+ \delta\,R)$};
\node at (4.35,7.1) {$\sfS^l(x, R)$};
\node at (0,1.2) {$A_0^\varepsilon$};
\node at (0,3.8) {$A_1$};
\end{tikzpicture}
\caption{An illustration of the proof of \autoref{Th:Bishop Gromov}.}\label{Fig:Bishop-Gromov}
\end{figure}

Let $t := r/R$, and let $\smash{\mu_0^\varepsilon\in\scrP_\comp^\ac(\mms,\meas)}$ be the uniform distribution of
\begin{align*}
A_0^\varepsilon := \sfB^\met(x,\varepsilon)  \cap E.
\end{align*}
Observe that unlike for $A_1$, the definition of $\smash{A_0^\varepsilon}$ contains a metric ball. Here we choose $\varepsilon > 0$ sufficiently small in such a way that $\smash{\bar{\sfB}^\met(x,\varepsilon)}$ is compact according to \autoref{Cor:Hausdorff} and \autoref{Cor:Strong causality},  and $\smash{A_0^\varepsilon \times A_1 \Subset\{l>0\}}$, as shown in \autoref{Fig:Bishop-Gromov}. Finally, define the set $\smash{A_t^\varepsilon\subset\mms}$ according to  \eqref{Eq:At'}, i.e.
\begin{align*}
A_t^\varepsilon := \eval_t\big[G(A_0^\varepsilon,A_1)\big].
\end{align*}
By \autoref{Ex:Str tl dual}, \autoref{Re:From geo to displ}, as well as  \eqref{Eq:BLU}, following the proof of \autoref{Th:Timelike BM} provides us with a plan $\smash{\bdpi^{\varepsilon,\delta}\in \OptTGeo_p(\mu_0^\varepsilon,\mu_1)}$ satisfying
\begin{align*}
\meas\big[A_t^\varepsilon\big]^{1/N} &\geq \sigma_{k_{\bdpi^\varepsilon}^+/N}^{(t)}(\cost_{\bdpi^\varepsilon})\,\scrU_N(\mu_1)\\ 
&\geq \int \sigma_{k_\gamma^+/N}^{(t)}(\vert\dot\gamma\vert)\d\bdpi^\varepsilon(\gamma)\,\meas[A_1]^{1/N}\\
&\geq \int \sigma_{\ubar{k}_{x,\gamma}^+/N}^{(t)}(\vert\dot\gamma\vert)\d\bdpi^\varepsilon(\gamma)\,\meas[A_1]^{1/N}.
\end{align*}
Note that $\smash{\mu_0 \otimes \mu_1}$ constitutes the only coupling of $\mu_0$ and $\smash{\mu_1}$; by construction, it is chronological. Tightness of timelike $\smash{\ell_p}$-optimal dynamical plans as asserted by  \autoref{Le:Geodesics plan}  gives the existence of a narrow limit $\smash{\bdpi\in\OptTGeo_p(\mu_0,\mu_1)}$ of a nonrelabeled subfamily of $(\bdpi^\varepsilon)_{\varepsilon > 0}$. By the lower semicontinuous dependence of the last integrand above on $\gamma\in\TGeo(\mms)$ provided by \autoref{Le:Properties},  \cite[Lem. 4.3]{villani2009}, and Levi's monotone convergence theorem we obtain
\begin{align*}
\meas[A_t]^{1/N} &= \liminf_{\varepsilon \to 0} \meas\big[A_t^\varepsilon\big]^{1/N} \geq \int \sigma_{\ubar{k}_{x,\gamma}^+/N}^{(t)}(\vert\dot\gamma\vert)\d \bdpi(\gamma)\,\meas[A_1]^{1/N}.
\end{align*}
The definition of $A_0$ and $\smash{A_1}$, the $l$-star shapedness of $E$, and the choice of $t$ easily imply $\smash{A_t \subset (\bar{\sfB}^l(x,r+\delta r)\setminus \sfB^l(x,r))\cap E}$. This terminates the proof of \eqref{Eq:BG}.

If $\meas[E\cap U] = 0$ for some open neighborhood $U\subset\mms$ of $x$, we still get \eqref{Eq:Key est} by the following modification of the previous argument. By using \autoref{Re:Causal reversal} and \autoref{Pr:TCD to TMCP} below, from $\smash{\wTCD_p^e(k,N)}$ we obtain the entropic timelike measure contraction property $\smash{\TMCP^e(k,N)}$ from \autoref{Def:TMCP} for the causal reversal $\smash{\scrM^\leftarrow}$ in \autoref{Ex:Causal reversal}. This suffices to run the above computation with $\smash{A_0}$ and $\smash{\mu_0}$ in place of $\smash{A_0^\varepsilon}$ and $\smash{\mu_0^\varepsilon}$, respectively.

From \eqref{Eq:Key est}, letting $\delta \to 0$, and following the proof of \cite[Thm.~5.9]{ketterer2017} we get
\begin{align*}
\frac{\sfs(r)}{\sfs(R)} \geq \sigma_{\ubar{k}_x/N}^{(t)}(R)^N = \frac{\SIN_{\ubar{k}_x/N}^N(r)}{\sin_{\ubar{k}_x/N}^N(R)},
\end{align*}
where the last equality is a consequence of  \autoref{Le:Properties}, \autoref{Th:BonnetMyers}, and our choice of $t$. In particular, as in the proof of \cite[Thm.~2.3]{sturm2006b} we obtain
\begin{align*}
\meas\big[(\bar{\sfB}^l(x,r)\setminus \sfB^l(x,r))\cap E\big] = 0.
\end{align*}
This implies the claim about $\meas$-negligibility of spheres in the future of $x$. Indeed, using Polishness of $\Top$ and \autoref{Cor:Strong causality} we construct countably many precompact, mutually disjoint sets $(A_i)_{i\in\N}$ covering $\smash{\{l(x,\cdot)=r\}}$ with $\smash{A_i\Subset I^+(x)}$ for every $i\in\N$. The claim easily follows by applying \autoref{Le:Star-shaped} and the previous discussion  to $\smash{C := \bar{A}_i}$, where $i\in\N$. The case of past spheres is dealt with analogously, recalling \autoref{Re:Causal reversal}.

To show $\smash{\meas[\{y\}] = 0}$ for every $y\in I^+(x)$, apply the previous argument to $r := l(x,y)$ and the compact, $l$-star shaped set $E$ from \autoref{Le:Star-shaped} defined for $C := \{y\}$; observe that $l(x,y) < R_x$ holds here by nontrivial chronology. By applying the preceding argument to the causal reversal from \autoref{Ex:Causal reversal} while taking \autoref{Re:Causal reversal} into account, we obtain the claim about the $\meas$-measure of truncated spheres in the past of $x$, points in $I^-(x)$, and truncated past light cones at $x$. 

Lastly, the same proof as for \cite[Thm.~5.9]{ketterer2017} allows us to remove the continuity hypothesis on $\smash{\ubar{k}_x}$ in order to finally obtain the second inequality in \eqref{Eq:BG}. In turn, this implies the first inequality from \eqref{Eq:BG} by \cite[Lem.~18.9]{villani2009}.

It remains to prove \eqref{Eq:Future empty} in the case $\meas$ gives mass to $x$. Suppose to the contrary that $y\in I^-(x) \cap \supp\meas$. Then applying the previous argumentation to $y$ replacing $x$ would entail the contradiction $\meas[\{x\}]=0$. In an analogous way, we falsify the assumption $y\in I^+(x) \cap \supp\meas$.
\end{proof}

By nontrivial chronology, the following is immediate from \autoref{Th:Bishop Gromov}.

\begin{corollary}[No atoms]\label{Cor:No atoms} The reference measure $\meas$ of a $\smash{\wTCD_p^e(k,N)}$ space has no atoms if it is fully supported.
\end{corollary}




The following is a consequence of \autoref{Th:Bishop Gromov} and argued as for \cite[Thm.~5.2]{mccannsaemann}; alternatively, it can be shown by using \autoref{Le:var to const} together with \autoref{Ex:Lipschitz} and \cite[Rem.~3.19, Thm.~3.35]{braun2022}.

\begin{corollary}[Dimension bound]\label{Cor:Hausdorffdim} Assume $\scrM$ is the metric measure spacetime induced by a globally hyperbolic Lipschitz spacetime, where $\meas$ is the top-dimensional  Lorentzian Hausdorff measure defined in \cite[Def.~2.3]{mccannsaemann}. If $\scrM$ satisfies $\smash{\wTCD_p^e(k,N)}$, then its  geometric Hausdorff dimension in the sense of \cite[Def.~3.1]{mccannsaemann} obeys
\begin{align}\label{nonsharp dimension bound}
\dim^l\mms \leq N+1.
\end{align}
\end{corollary}

\section{One-dimensional localization/needle decomposition}\label{Ch:Localization}

In this chapter, we extend the Lorentzian localization paradigm initiated in \cite{cavalletti2020}. We treat a more general class of potentials and cover spaces that might fail to be timelike nonbranching.    A related extension is  pursued for different purposes  in \cite{Akdemir24+,cm++}. Our plan is as follows.
\begin{itemize}
\item We introduce the transport sets with and without bad points, and describe  transport relations $\smash{\preceq_u}$, $\smash{\succeq_u}$, and  $\smash{\sim_u}$ thereon.
\item The relation $\smash{\sim_u}$ is an equivalence relation on the transport set without bad points, \autoref{Th:Equiv relation}.
\item In \autoref{Cor:Ray map}, we show every chain with respect to $\smash{\preceq_u}$ is parametrized by exactly one $l$-geodesic.
\item Under natural conditions --- which 
 can be inferred from curvature and time\-like essential nonbranching assumptions using the existence of optimal maps (\autoref{Th:Optimal maps}) and local-to-global property (\autoref{Th:Local to global}) --- we show the set of bad points thrown away to get \autoref{Th:Equiv relation} is $\meas$-negligible, \autoref{Th:All negligible}.
\item Assuming $\smash{\wTCD_p^e(k,N)}$ and timelike $p$-essential nonbranching, we show that on the transport set without bad points, $\meas$ admits a disintegration into ``absolutely continuous'' needles, \autoref{Th:Needle} and \autoref{Pr:Abs cont needles}.
\item Finally, \autoref{Th:Localization TCD} ensures   a.e.~needle constitutes a $\CD(k,N)$ metric measure space in the sense of Ketterer \cite{ketterer2017}.
\end{itemize}

We recall our assumption of $\Top$ being Polish. This  is implicitly used various times when we apply results about Suslin measurability \cite{srivastava1998}; recall \autoref{Sub:Suslin} for relevant basics of Suslin sets.

\subsection{Framework}\label{Sub:Framework} Throughout this chapter, we fix 
\begin{itemize}
\item a Borel subset $E\subset\mms$ which is \emph{$l$-geodesically convex}, i.e.~for every $\gamma\in \TGeo(\mms)$ such that $\gamma_0,\gamma_1\in E$, we have $\smash{\gamma_{[0,1]}\subset E}$, and 
\item a Borel function $u\colon E\to\R$  obeying the \emph{reverse $l$-Lipschitz condition}\footnote{In \cite{beran2023}, we refer to this property as \emph{$1$-steepness} of $u$ on $E$.}, i.e. every $x,y\in E$ satisfies
\begin{align}\label{Eq:Gleichung}
u(y) - u(x) \geq l(x,y).
\end{align} 
\end{itemize}
A concrete example from \cite{cavalletti2020} fitting in this setting is given in \autoref{Ex:CM} below. Observe that a subset $E\subset\mms$ is $l$-geodesically convex with respect to a metric spacetime if and only if it is  $\smash{l^\leftarrow}$-geodesically convex relative to its causal reversal from \autoref{Ex:Causal reversal}. In addition, a function $u\colon E\to \R$ obeys \eqref{Eq:Gleichung} if and only if $-u$ satisfies \eqref{Eq:Gleichung} with $l$ replaced by $\smash{l^\leftarrow}$.

In addition, we impose mild assumptions on  $\meas$, namely
\begin{itemize}
\item $\meas$ has full topological support, and
\item there exists a Borel weight function $\smash{w\colon \Tr_u^\End \to (0,\infty)}$ with $\inf w(C) > 0$ for every compact set $C\subset E$ and
\begin{align}\label{Eq:Weight}
\int_{\Tr_u^\End} w\d\meas = 1;
\end{align}
in particular $\smash{\meas[\Tr_u^\End]>0}$. Here $\smash{\Tr_u^\End\subset E}$ is the transport set with bad points from \autoref{Def:Tbr} below.
\end{itemize}
The hypothesis $\supp\meas=\mms$ could be dropped by appropriately restricting all definitions below to $\supp\meas$, but we refrain from that in favor of a simpler presentation. On the other hand, the existence of  a weight function $w$  holds in high generality, e.g.~on every proper Polish space \cite[Lem.~3.3]{cavallettinew} or simply if $\meas[E]$ is finite (provided $\smash{\meas[T^\End]>0}$ in both cases). It is used to normalize $\meas$ in order to apply the disintegration theorem in \autoref{Sub:Disintegration}.

\begin{example}[Compatibility with \cite{cavalletti2020}]\label{Ex:CM} To give a concrete example originating in \cite{cavalletti2020} from which e.g.~the synthetic Hawking singularity theorem was obtained, we first recall basic definitions \cite{cavalletti2020, galloway1986} for a set $V\subset\mms$ (assuming $\mms$ is proper). 

We call $V$ \emph{achronal} provided $x\not\ll y$ for every $x,y\in V$; in other words, $I^+(V)$ and $I^-(V)$ have empty intersection. This immediately implies well-definedness of the  \emph{signed time}  $\tsep_V \colon \mms \to [-\infty,\infty]$ 
 to $V$ given by
\begin{align*}
\tsep_V(x) := \sup_{y \in V} l_+(y,x) - \sup_{y\in V} l_+(x,y).
\end{align*}
Since $\smash{l_+}$ is continuous, $\tsep_V$ is lower semicontinuous  on $\smash{I^+(V) \cup V}$. 

Let $V$ be such an achronal Borel set which is in addition \emph{future timelike complete} \cite[p.~367]{galloway1986}, briefly FTC, i.e.~for every $y\in I^+(V)$ the intersection $J^-(y)\cap V$ has compact closure relative to $V$; observe that this covers especially the case when $V$ is a singleton by \autoref{Cor:Causality}. Then define 
\begin{align*}
E &:= I^+(V) \cup V,\\
u &:= \tsep_V\big\vert_{E^2}.
\end{align*}
By achronality of $V$ and \autoref{Le:Pushup openness}, for every $\gamma\in \smash{\TGeo(\mms)}$ with $\gamma_0,\gamma_1\in E$ we have $\gamma_t\in I^+(V)$ for every $t\in (0,1]$, which shows $l$-geodesic convexity of $E$. The inequality \eqref{Eq:Gleichung} on $\smash{E^2}$ is a consequence of the hypothesized FTC property and is shown with the same proof as for \cite[Lem.~4.1]{cavalletti2020}.

The existence of a function $w$ as above follows from \cite[Lems.~4.4, 4.10]{cavalletti2020}. In fact, under the natural condition $\meas[V]=0$ we have
\begin{align}\label{Eq:Nat cond}
\meas\mres I^+(V) = \meas\mres \Tr_u^\End.
\end{align}
\end{example}

\subsection{Transport relations and transport sets}\label{Sub:Transport relation} Now we construct 
\begin{itemize}
\item the transport relation $\smash{E_{\sim_u}^{2,\End}}$ and the transport set $\smash{\Tr_u^\End}$ relative to $E$ and $u$   \emph{with} bad points, and
\item their counterparts $\smash{E_{\sim_u}^2}$ and $\smash{\Tr_u}$ \emph{without} bad  points.
\end{itemize}
As shown in \autoref{Th:Equiv relation}, 
$\smash{E_{\sim_u}^2}$ will be an equivalence relation on $\Tr_u$. (In fact, $\Tr_u$ can be chosen slightly larger.) Compared to \cite[Prop. 4.5]{cavalletti2020}, timelike nonbranching is not available to us; our approach is instead inspired by the work \cite{cavalletti2014} for metric measure spaces, see also \cite{bianchini2013, caffarelli2001, cavalletti2017, klartag2017}. 

\subsubsection{Transport with bad points} The construction of the indicated sets $\smash{E_{\sim_u}^{2,\End}}$ and $\smash{\Tr_u^\End}$ follows \cite[Sec.~4.1]{cavalletti2020} and is briefly recapitulated here.


We define relations $\smash{\preceq_u}$ and $\smash{\succeq_u}$ on $E$ by declaring
\begin{align*}
x\preceq_u y \quad &:\Longleftrightarrow\quad x=y\quad\textnormal{or}\quad u(y) -u(x) = l(x,y) > 0,\\
x\succeq_u y \quad &:\Longleftrightarrow\quad y\preceq_ux.
\end{align*}
We consider the associated sets
\begin{align}\label{Eq:GAMMA}
\begin{split}
E_{\preceq_u}^2 &:= \{(x,y) \in E^2  : x \preceq_u y\},\\
E_{\succeq_u}^2 &:= \{(x,y) \in E^2 : x\succeq_u y\}.\\
\end{split}
\end{align}
Our assumptions on $E$ and $u$ guarantee $\smash{E_{\preceq_u}^2}$ and $\smash{E_{\succeq_u}^2}$ are Borel.

The following straightforward properties should give the reader a better feeling for the definitions of $\smash{\preceq_u}$ and $\smash{\succeq_u}$, respectively.

\begin{lemma}[Partial ordering]\label{Le:Partial ordering u} The relations $\smash{\preceq_u}$ and $\smash{\succeq_u}$ are partial orders on $E$.
\end{lemma}

\begin{proof} We only show the claim for $\smash{\preceq_u}$, the proof for $\smash{\succeq_u}$ is completely analogous. Let $x,y,z\in E$. By definition, we have $\smash{x\preceq_u x}$. If $\smash{x\preceq_u y}$ and $\smash{x\succeq_u y}$, then $x=y$ since otherwise $x\ll y$ or $y\ll x$, a contradiction to  the respective other half of the assumption. Lastly, assume $\smash{x\preceq_u y}$ and $\smash{y\preceq_u z}$. Without restriction, we assume $x$, $y$, and $z$ are distinct. Then $x\ll y\ll z$, and in particular $x\ll z$ by the push-up property from \autoref{Le:Pushup openness}. By the reverse $l$-Lipschitz property \eqref{Eq:Gleichung} of $u$ and the reverse triangle inequality for $l$, we easily obtain
\begin{align*}
u(z) - u(x) &\geq l(x,z)\\ 
&\geq l(x,y) + l(y,z)\\
&= \big[u(y) - u(x)\big] + \big[u(z) - u(y)\big]\\
&= u(z) - u(x),
\end{align*}
which forces equality to hold throughout.
\end{proof}

In a similar way, the following is verified \cite[Lem.~4.2]{cavalletti2020}.

\begin{lemma}[Cyclical monotonicity] $\smash{E_{\preceq_u}^2}$ is $l$-cyclically monotone, while $\smash{E_{\succeq_u}^2}$ is $l^\leftarrow$-cyclically monotone.
\end{lemma}

Next we define
\begin{align}
\label{G Ray}
G &:= \{\gamma\in\TGeo(\mms) : \gamma_s \preceq_u\gamma_t  \text{ for every }0 \leq s<t\leq 1\},
\\ \nonumber
G^\leftarrow &:= \{\gamma\in \TGeo^\leftarrow(\mms) : \gamma_s\succeq_u\gamma_t \textnormal{ for every }0\leq s<t\leq 1\}.
\end{align}
As every curve in $G$ and $\smash{G^\leftarrow}$ is timelike in the respective causal structure,  $G$ and $\smash{G^\leftarrow}$ do not contain constant curves by \autoref{Cor:Causality}. Moreover, note that if we would a priori replace $l$ by $\smash{l^\leftarrow}$, then $\smash{G^\leftarrow}$ precisely corresponds to the set $G$, and vice versa --- heuristically, $\smash{(G^\leftarrow)^\leftarrow = G}$.

The key property of $\smash{G}$ relative to $\smash{\preceq_u}$ is contained in the following lemma. An evident modification relative to $\smash{\succeq_u}$ holds for $\smash{G^\leftarrow}$; similarly for \autoref{Cor:Cty}.

\begin{lemma}[Propagation of the transport relation]\label{Le:G} Let $\gamma\in\TGeo(\mms)$ such that $\smash{\gamma_0\preceq_u\gamma_1}$. Then $\smash{\gamma_s\preceq_u\gamma_t}$ for every $0\leq s<t\leq 1$; in other words, $\smash{\gamma\in G}$.
\end{lemma}

\begin{proof} By \eqref{Eq:Gleichung}, the assumption $\smash{\gamma_0\preceq_u\gamma_1}$, and affine parametrization of $\gamma$,
\begin{align*}
u(\gamma_t) - u(\gamma_s) &= u(\gamma_1) - u(\gamma_0) - \big[u(\gamma_s)- u(\gamma_0)\big] - \big[u(\gamma_1) - u(\gamma_t)\big]\\
&\leq l(\gamma_0,\gamma_1) - l(\gamma_0,\gamma_s) - l(\gamma_t,\gamma_1)\\
&= l(\gamma_s,\gamma_t).
\end{align*}
Using \eqref{Eq:Gleichung} again, the previous inequality is in fact an identity.
\end{proof}

\begin{corollary}[Continuity of $u$ along rays]\label{Cor:Cty} For every $\smash{\gamma\in G}$, the function $u\circ \gamma$ is continuous on $[0,1]$.
\end{corollary}

These considerations motivate the definition of a further relation $\smash{\sim_u}$ on $E$ via
\begin{align*}
x \sim_u y\quad :\Longleftrightarrow\quad x\preceq_u y\quad\textnormal{or}\quad x\succeq_u y.
\end{align*}

\begin{definition}[Transport with bad points]\label{Def:Tbr} We define
\begin{enumerate}[label=\textnormal{\alph*.}]
\item the \emph{transport relation with bad points} by 
\begin{align*}
E_{\sim_u}^{2,\End} := E_{\preceq_u}^2 \cup E_{\succeq_u}^2,
\end{align*}
\item the \emph{transport set with bad points} by
\begin{align*}
\Tr_u^\End &:= \pr_1\big[E_{\sim_u}^{2,\End} \setminus \diag(\mms^2)\big].
\end{align*}
\end{enumerate}
\end{definition}

Clearly, $\smash{E_{\sim_u}^{2,\End}}$ is Borel; since the background topology $\Top$ is Polish, this implies $\smash{\Tr_u^\End}$ is Suslin. 

By definition of $\smash{\Tr_u^\End}$, for every $\smash{x\in \Tr_u^\End}$ there exists $y\in E \setminus \{x\}$ with $\smash{x\preceq_u y}$ or $\smash{x\succeq_u y}$, and in the respective case $x\ll y$ or $y\ll x$; in particular, $\smash{y\in \Tr_u^\End}$. Yet, this does \emph{not} mean $\smash{\Tr_u^\End}$ has no extremal points according to \autoref{Def:Extremal points}. Moreover, 
 $\smash{\sim_u}$ is  not generally transitive on the whole of  $\smash{\Tr_u^\End}$.

\subsubsection{Bad points}\label{Sub:Bad pointss} Let us now make precise what we mean by bad points through the following two definitions.

\begin{definition}[Timelike branching points]\label{Def:Branching points} We define 
\begin{enumerate}[label=\textnormal{\alph*.}]
\item the \emph{timelike forward branching set} by
\begin{align*}
B &:= \{x \in \Tr_u^\End : \textnormal{\textit{there exist} }y_1,y_2\in \Tr_u^\End \\
&\qquad\qquad \textnormal{\textit{such that} } x \preceq_u y_1 \textnormal{ \textit{and} } x\preceq_u y_2 \textnormal{ \textit{yet} } y_1\not\sim_u y_2\},
\end{align*}
\item the \emph{timelike backward branching set} by
\begin{align*}
B^\leftarrow &:= \{y \in \Tr_u^\End : \textnormal{\textit{there exist} }x_1,x_2\in \Tr_u^\End \\
&\qquad\qquad \textnormal{\textit{such that} } x_1 \preceq_u y \textnormal{ \textit{and} } x_2\preceq_u y  \textnormal{ \textit{yet} } x_1\not\sim_u x_2\}.
\end{align*}
\end{enumerate}

Any point in $B$ or $\smash{B^\leftarrow}$ is called \emph{timelike forward branching point} or \emph{timelike backward branching point}, respectively.
\end{definition}

\begin{remark}[About \autoref{Def:Branching points}]\label{Re:Triples} Let $\smash{x\in B}$ be fixed, and let   $y_1,y_2\in \smash{\Tr_u^\End}$ be as in the previous  definition of $B$. Then $x$, $y_1$, and $y_2$ are necessarily distinct; in particular, we have  $x\ll y_1$ and $x\ll y_2$. Moreover, if a point $\smash{x'\in T^\End}$  satisfies $\smash{x'\preceq_u x}$ in the previous situation, then $x'\in B$.

Analogous statements hold for $\smash{B^\leftarrow}$. 
\end{remark}

\begin{definition}[Extremal points]\label{Def:Extremal points} We define
\begin{enumerate}[label=\textnormal{\alph*.}]
\item the \emph{initial point set} by
\begin{align*}
a := \{x\in \Tr_u^\End : \textnormal{\textit{no} }x^-\in \Tr_u^\End\setminus \{x\} \textnormal{ \textit{satisfies} } x^-\preceq_u x\},
\end{align*}
\item the \emph{final point set} by
\begin{align*}
b := \{x\in \Tr_u^\End : \textnormal{\textit{no} }y^+\in  \Tr_u^\End\setminus \{y\} \textnormal{ \textit{satisfies} }y^+ \succeq_u y\}.
\end{align*}
\end{enumerate}

Any point in $a$ or $b$ is called \emph{initial point} or \emph{final point}, respectively.
\end{definition}

In general, there is no relation between the sets $B$, $\smash{B^\leftarrow}$, $a$, and $b$.

Any point in the union $\smash{B\cup B^{-1}\cup a \cup b}$ is called \emph{bad point}. This terminology is motivated by the following difficulties encountered in the sequel.
\begin{itemize}
\item Timelike branching points prevent $\smash{\sim_u}$ to be an equivalence relation on $\smash{\Tr_u^\End}$, cf.~\autoref{Th:Equiv relation}. 
\item On the other hand, the presence of extremal points  will have to be waived in order to prove absolute continuity of the conditional densities along a.e. needle in our disintegration \autoref{Th:Needle} later, cf.~\autoref{Pr:Abs cont needles}.
\end{itemize}
For these reasons, these sets will be excluded from $\smash{\Tr_u^\End}$ in the next subsection. It is highly nontrivial --- and in general false --- that $\smash{\Tr_u^\End\setminus (B\cup B^\leftarrow\cup a\cup b)}$ remains sufficiently large, yet its difference from  $\smash{\Tr_u^\End}$ is $\meas$-negligible under natural conditions, as demonstrated in \autoref{Sub:Negligibility}.

The sets $B$ and $\smash{B^\leftarrow}$ are Suslin; the tedious proof is similar to the argument for \autoref{Le:Measurable section} below (which does not create a circular reasoning). On the other hand, Suslin measurability of $a$ and $b$ is easy to infer from the identities
\begin{align*}
a &= \Tr_u^\End \setminus \pr_2\big[E_{\preceq_u}^2 \setminus \diag(\mms^2)\big],\\
b &= \Tr_u^\End \setminus \pr_1\big[E_{\preceq_u}^2\setminus \diag(\mms^2)\big].
\end{align*}

To get a grip on \autoref{Def:Branching points}, we prove the following lemma. Its statement is only needed again in \autoref{Le:Measurable section} below, but we will revisit its proof technique already in the proof of \autoref{Le:Single geodesic}. Furthermore, it highlights the connection of timelike branching points and branching phenomena for $l$-geodesics. A further, perhaps more illuminating evidence is deferred to \autoref{Ex:Cutlocus} below, where we relate $B$ and $\smash{B^\leftarrow}$ to certain cut loci for specific choices of $E$ and $u$.

We only concentrate on timelike forward branching. With evident changes, all statements hold in the timelike backward branching case as well.

\begin{lemma}[Branching points vs.~branching geodesics I]\label{Le:FBranching} Let $\smash{x\in \Tr_u^\End}$. Then $\smash{x\in B}$ if and only if there exist $\smash{\gamma^1,\gamma^2\in G}$ from \eqref{G Ray} with
\begin{align*}
\min_{t\in [0,1]} \met(\gamma_t^1,\gamma_t^2) >0
\end{align*}
and such that for every $0\leq t\leq 1$, we have $\smash{x\preceq_u \gamma_t^1}$ and $\smash{x\preceq_u\gamma_t^2}$ yet $\smash{\gamma_t^1\not\sim_u\gamma_t^2}$, and
\begin{align*}
u(\gamma_t^1) = u(\gamma_t^2).
\end{align*}
\end{lemma}

\begin{figure}
\centering
\begin{tikzpicture}
\draw[dashed] plot [smooth] coordinates {(0,0) (-0.1,0.5) (-1,2) (-0.5,3) (-0.4,4) (-1,5.5)};
\begin{scope}
    \clip (-1,2) rectangle (-0.3,5.5);
    \draw[thick] plot [smooth] coordinates {(0,0) (-0.1,0.5) (-1,2) (-0.5,3) (-0.4,4) (-1,5.5)};
  \end{scope}
\draw[dashed] plot [smooth] coordinates {(0,0) (0.2,1) (0.5,1.5) (0.1,2.5) (1,4) (0.5,5) (1.5, 6.5)};
\begin{scope}
    \clip (0,1.5) rectangle (1.5,5);
    \draw[thick] plot [smooth] coordinates {(0,0) (0.2,1) (0.5,1.5) (0.1,2.5) (1,4) (0.5,5) (1.5, 6.5)};
  \end{scope}
\draw[fill=black] (0,0) circle(.1em);
\draw[fill=black] (-1,5.5) circle(.1em);
\draw[fill=black] (1.5,6.5) circle(.1em);
\draw[fill=black] (0.5,5) circle(.1em);
\node at (0.32,0) {$x$};
\node at (-1.35,5.5) {$y_1$};
\node at (1.875,6.5) {$y_2$};
\node at (0.925,4.9) {$\smash{\eta_\alpha^2}$};
\node at (0.85,3) {$\gamma^2$};
\node at (-0.8,3.5) {$\gamma^1$};
\end{tikzpicture}
\caption{An illustration of the proof of \autoref{Le:FBranching}.}\label{Fig:Branching}
\end{figure}
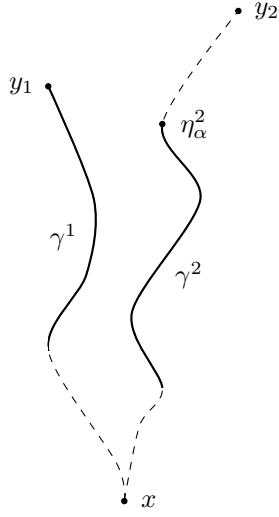

\begin{proof} The backward implication is clear by definition of $B$.

To prove the forward implication --- which is depicted in \autoref{Fig:Branching} --- assume $x\in B$, and let $\smash{y_1,y_2\in \Tr_u^\End}$ be as in  \autoref{Def:Branching points}. By symmetry of $\smash{\sim_u}$ we may and will assume $u(y_2) \geq u(y_1)$. By \autoref{Re:Triples} and \autoref{Le:G} there are $\smash{\eta^1, \eta^2 \in G}$ with $\smash{\eta_0^1=\eta_0^2=x}$, $\smash{\eta_1^1 = y_1}$, and $\smash{\eta_1^2 = y_2}$; in particular, $\smash{\eta_t^1,\eta_t^2\in \Tr_u^\End}$ for every $0\leq t\leq 1$. Since $\smash{x\preceq_u y_1}$, we obtain 
\begin{align*}
u(y_2) \geq u(y_1) = u(x) + l(x,y_1)> u(x).
\end{align*}
\autoref{Cor:Cty} and the intermediate value theorem imply
\begin{align}\label{Eq:Yet}
u(\eta_\alpha^2) = u(y_1)
\end{align}
for some $0<\alpha\leq 1$ which will be henceforth fixed. Then $\smash{\eta_\alpha^2 \neq y_1}$, since otherwise $\smash{y_1\preceq_u y_2}$ by \autoref{Le:G}, in conflict with our assumption $\smash{y_1\not\sim_u y_2}$. We even have $\smash{\eta_\alpha^2\not\sim_u y_1}$, since otherwise $\smash{\eta_\alpha^2 \ll y_1}$ or $\smash{y_1\ll \eta_\alpha^2}$ by the previous observation. In the first case, by \eqref{Eq:Yet} the observation $\smash{0 =  u(y_1) - u(\eta_\alpha^2) = l(\eta_\alpha^2,y_1) >0}$ leads to a contradiction; the case $\smash{y_1 \ll \eta_\alpha^2}$ is falsified analogously.

For every $0 \leq r\leq 1$, by employing the property $\smash{\eta^1,\eta^2\in G}$ and \eqref{Eq:Yet}, a direct computation shows
\begin{align*}
    -u(\eta_{\alpha r}^2) = -u(\eta_r^1).
\end{align*}
As $\smash{\eta_\alpha^2 \neq y_1 = \eta_1^1}$, continuity implies the existence of $\delta > 0$ such that $\smash{\eta_{1-s}^1\neq \eta_{\alpha (1-s)}^2}$ and, following the above argument, that $\smash{\eta_{1-s}^1\not\sim_u \eta_{\alpha (1-s)}^2}$ for every $0 \leq s\leq \delta$. A re\-sca\-ling argument then gives $\smash{\gamma^1,\gamma^2\in G}$ with the desired properties.
\end{proof}

\begin{remark}[Branching points vs.~branching geodesics II] By a slight abuse of terminology, let us call a point $\smash{x\in \mms}$ \emph{timelike forward branching} if there exist $\smash{\gamma^1,\gamma^2\in \TGeo(\mms)}$ and $0<\delta<1$ with $\smash{x\in \gamma^1\big\vert_{(0,\delta)} = \gamma^2\big\vert_{(0,\delta)}}$ yet $\smash{\gamma^1\neq\gamma^2}$. Then the proof of \autoref{Le:FBranching} shows that if $x\in B\setminus a$, then $x$ is such a timelike forward branching point. Conversely, a timelike forward branching point $\smash{x\in \Tr_u^\End}$ belongs to $B$ provided $\smash{x\preceq_u\gamma_t^1}$ and $\smash{x\preceq_u\gamma_t^2}$ yet $\smash{\gamma_t^1\not\sim_u\gamma_t^2}$ for some $\delta < t \leq 1$.
\end{remark}

With similar arguments, the reader may verify the following.

\begin{example}[Timelike cut loci]\label{Ex:Cutlocus} Let $o \in \mms$ be fixed.  Let $E$ and $u$ be defined as in \autoref{Ex:CM} in terms of $V := \{o\}$, i.e.
\begin{align*}
E &:= I^+(o) \cup\{o\},\\
u &:= l(o,\cdot).
\end{align*}
Note that in this case, we have $\smash{\Tr_u^\End = I^+(o)}$. Then 
the induced set $\smash{B^\leftarrow}$ coincides with the \emph{future timelike cut locus} $\smash{\TCut^+(o)}$ of $o$. The latter is the set of $y\in I^+(o)$ such that two distinct $\smash{\gamma^1,\gamma^2\in\TGeo(\mms)}$ connect $o$ and $y$.

Mutatis mutandis, considering instead
\begin{align*}
E &:= I^-(o) \cup\{o\},\\
u &:= -l(\cdot,o)
\end{align*}
we observe 
the  set $B$ coincides with the \emph{past timelike cut locus} $\smash{\TCut^-(o)}$, which is the set of $x\in I^-(o)$ such that two distinct $\smash{\gamma^1,\gamma^2\in\TGeo(\mms)}$ connect $x$ to $o$.
\end{example}


\subsubsection{Transport without bad points}

\begin{definition}[Transport without bad points]\label{Def:TV} We define
\begin{enumerate}[label=\textnormal{\alph*.}]
\item the \emph{transport set} by
\begin{align*}
\Tr_u := \Tr_u^\End\setminus (B\cup B^\leftarrow\cup a\cup b),
\end{align*}
\item the \emph{transport relation} by
\begin{align*}
E_{\sim_u}^2 := E_{\sim_U}^{2,\End} \cap \Tr_u^2.
\end{align*}
\end{enumerate}
\end{definition}


In other words, $\smash{E_{\sim_u}^2}$ corresponds to the restriction of the relation $\smash{\sim_u}$ to $\Tr_u$.

By the discussions from \autoref{Sub:Bad pointss}, the sets $\Tr_u$ and $\smash{E_{\sim_u}^2}$ are both Suslin.

\begin{theorem}[Equivalence relation]\label{Th:Equiv relation} The relation $\smash{\sim_u}$ constitutes an equivalence relation on $\smash{\Tr_u^\End \setminus (B \cup B^\leftarrow)}$, a fortiori on its subset $\Tr_u$.
\end{theorem}

\begin{proof} We abbreviate $\smash{\Tr_u' := \Tr_u^\End \setminus (B \cup B^\leftarrow)}$. By definition of $\smash{\sim_u}$, we have $\smash{x\sim_u x}$ for every $\smash{x\in \Tr_u'}$. Moreover, $\smash{\sim_u}$ is symmetric by construction. To show transitivity, let $\smash{x,y,z\in \Tr_u'}$ be mutually distinct points  with $\smash{x\sim_u y}$ and $\smash{y\sim_u z}$. By \autoref{Le:Partial ordering u}, we only have to consider the following two cases.

\textbf{Case 1.} Assume $\smash{x\preceq_u y}$ and $\smash{y\succeq_u z}$, hence $\smash{z\preceq_u y}$. This immediately implies $\smash{x\sim_u z}$; indeed, otherwise $\smash{x\not\sim_u z}$ contradicts the assumption $\smash{y\notin B^\leftarrow}$.

\textbf{Case 2.} The case $\smash{x\succeq_u y}$ and $\smash{y\preceq_u z}$ is treated similarly to Case 3 by contradicting the assumption $\smash{y\notin B}$.
\end{proof}

The equivalence class of $x\in \Tr_u$ with respect to $\sim_u$ is denoted 
 $\tilde x^u$.


\subsection{Ray map and quotient map} In order to disintegrate $\meas$ with respect to $u$ later, we now have to build an $\meas$-measurable quotient map of $\smash{E_{\sim_u}^2}$. 

First, in \autoref{Le:Single geodesic} we prove that every equivalence class of $\smash{\sim_u}$ is crossed by exactly one $l$-geodesic. This will make the so-called ray map $\ray$ from  \autoref{Cor:Ray map} well-defined. In turn, $\ray$ will be used to construct the above mentioned quotient map in \autoref{Pr:Quotient map}.

\begin{lemma}[Nonbranching rays]\label{Le:Single geodesic} Let $x,y_1,y_2\in \Tr_u$ satisfy $\smash{x\sim_u y_1}$ and $\smash{x\sim_u y_2}$. Then there exists $\smash{\gamma\in G}$ with $\{x,y_1,y_2\} \subset \gamma_{[0,1]}$. 

Moreover, if any $\smash{\gamma'\in G}$ satisfies the previous conclusion too, there exists $\smash{\gamma''\in G}$ such that $\smash{\gamma_{[0,1]}\cup \gamma_{[0,1]}' \subset \gamma_{[0,1]}''}$.
\end{lemma}

\begin{proof} To show the first statement, without restriction we assume $x$, $y_1$, and $y_2$ are distinct,  otherwise we just apply \autoref{Le:G}. We  have to distinguish four cases.

\textbf{Case 1.} Assume $\smash{y_1\preceq_u x}$ and $\smash{x\preceq_u y_2}$. By \autoref{Re:Triples}, we have $y_1\ll x \ll y_2$, and by transitivity of $\smash{\preceq_u}$ we have $\smash{y_1\preceq_u y_2}$. Therefore,
\begin{align*}
    l(y_1,y_2) &= u(y_2)- u(y_1)\\ 
    &= \big[u(y_2) - u(x)\big] + \big[u(x) - u(y_1)\big]\\
    &= l(y_1,x) + l(x,y_2).
\end{align*}
As in the proof of \autoref{Le:Continuity geos}, the concatenation of two appropriate elements in $\smash{\TGeo(\mms)}$ --- one going from $y_1$ to $x$,  one going from $x$ to $y_1$ --- yields an element $\gamma\in \smash{\TGeo(\mms)}$ from $y_1$ to $y_2$. By \autoref{Le:G}, we get $\gamma\in G$ with $\smash{\{y_1,x,y_2\}\subset\gamma_{[0,1]}}$ and hence the desired curve.

\textbf{Case 2.} Analogously to Case 1, the assumption $\smash{y_2\preceq_u x}$ and $\smash{x\preceq_u y_1}$ yields an element $\smash{\gamma\in G}$ with $\smash{\{y_2,x,y_1\}\subset \gamma_{[0,1]}}$.

\textbf{Case 3.} Assume $\smash{x\preceq_u y_1}$ and $\smash{x\preceq_u y_1}$.  Since $\smash{\sim_u}$ is  an equivalence relation on $\Tr_u$ by \autoref{Th:Equiv relation}, and since $y_1\neq y_2$, either  $y_1\ll y_2$ or $y_2\ll y_1$ by \autoref{Re:Triples}. Without restriction, we assume 
\begin{align*}
u(y_2) > u(y_1) > u(x).
\end{align*}
Indeed, equality in the second place never occurs under our standing hypotheses, and a possible equality in the first place  can be easily contradicted as in the proof of \autoref{Le:FBranching}. Now, assume to the contrapositive that no $\smash{\gamma\in G}$ from $x$ to $y_2$ passes through $y_1$. Then by geodesy, \autoref{Le:G}, and a continuity argument there exist $\smash{\gamma\in G}$ and $0<t<1$ with $\gamma_0=x$, $\gamma_1=y_2$, $u(\gamma_t) = u(y_1)$, $\smash{x\preceq_u\gamma_t}$, and $\gamma_t\neq y_1$. Arguing again as in the proof of \autoref{Le:FBranching}, we get $\smash{\gamma_t\not\sim_u y_1}$, but this contradicts the assumption $\smash{x\notin B}$.

\textbf{Case 4.} The case $\smash{y_1\preceq_u x}$ and $\smash{y_2\preceq_u x}$ is treated as in Case 3 by falsifying the assumption $\smash{x\in B^\leftarrow}$.

The last statement is argued similarly.
\end{proof}

Now we define the set 
\begin{align*}
    D_{\ray} &= \{(\alpha,t,y) \in \Tr_u\times \R \times \Tr_u : \alpha \sim_u y,\, l_+(\alpha,y)-l_+(y,\alpha)=t\}    
\end{align*}

and the multi-valued map $\smash{\ray \colon \Dom(\ray) \to 2^{\Tr_u}}$ in terms of its graph
\begin{align*}
    \graph\ray := (\Dom(\ray) \times \Tr_u) \cap D_\ray,
\end{align*}
where $\Dom(\ray) \subset \Tr_u\times \R$ is given by
\begin{align*}
    \Dom(\ray) := \{(\alpha,t) \in \Tr_u \times \R : D_\ray \cap (\{(\alpha,t)\} \times \Tr_u) \neq \emptyset\}.
\end{align*}
In other words, to a given pair $(\alpha,t)\in\Dom(\ray)$ the multi-valued map $\ray$ assigns the nonempty set of points $y\in \Tr_u$ with ``distance'' to $\alpha$ being equal to $t$. It is clear that $(\alpha,0)\in \Dom(\ray)$ for every $\alpha\in \Tr_u$.

\begin{corollary}[Ray map]\label{Cor:Ray map} The map $\ray$ defined above is single-valued on $\Dom(\ray)$, hence may and will be identified with a nonrelabeled Borel map $\ray\colon \Dom(\ray) \to \Tr_u$.  

This map obeys the following properties for every $\alpha\in \pr_1(\Dom(\ray))$. 
\begin{enumerate}[label=\textnormal{\textcolor{black}{(}\roman*\textcolor{black}{)}}]
    \item\label{La:I} \textnormal{\textbf{Intervals.}} The set 
    \begin{align*}
        \Dom(\ray)(\alpha) := \pr_2\big[\Dom(\ray) \cap (\{\alpha\}\times \R)\big]
    \end{align*}
    is a nonempty convex subset of $\R$, i.e.~an interval.
    \item\label{La:II} \textnormal{\textbf{Totality.}} The map $F\colon \Dom(\ray)(\alpha) \to \Tr_u$ defined by $F(t) := \ray(\alpha,t)$ is injective. Moreover, it maps surjectively onto the ray $\smash{\tilde{\alpha}^u}$.
    \item\label{La:III} \textnormal{\textbf{Order isometry.}}  For every $s,t\in \Dom(\ray)(\alpha)$ with the property $s <t$, we have $\smash{F(s)\preceq_u F(t)}$ and
    \begin{align*}
        l(F(s), F(t)) = t-s.
    \end{align*}
\end{enumerate}
\end{corollary}

\begin{proof} We first show $\ray$ is single-valued on  $\Dom(\ray)$. By definition of the latter, it suffices to prove $\#\ray(\alpha,t) \leq 1$ for every $(\alpha,t) \in \Dom(\ray)$. We only treat the case $t>0$; the case $t<0$ is addressed similarly, while the case $t=0$ is clear from \autoref{Cor:Causality}. Thus, let $y_1,y_2\in \ray(\alpha,t)$, and note that $y_1$ and $y_2$ are distinct from $\alpha$. The point $\alpha$ is no timelike forward branching point by definition of $\Tr_u$, thus $\smash{y_1\sim_u y_2}$. Then \autoref{Le:Single geodesic} implies that $x$, $y_1$, and $y_2$ lie on a single  $l$-geodesic segment. Since  $l(\alpha,y_1)=l(\alpha,y_2)=t>0$, this forces $y_1=y_2$.

Borel measurability of $\ray$ follows from analyticity of $\graph \ray$, compare with the proof of \autoref{Le:Measurable section} below.

For item \ref{La:I}, observe that $\ray(\alpha,0)=\alpha$ implies $0\in \Dom(\ray)(\alpha)$. Next, we claim $[0,t]\subset \Dom(\ray)(\alpha)$ for every $t \in \Dom(\ray)(\alpha) \cap (0,\infty)$; the case $t\in \Dom(\ray)(\alpha)\cap (-\infty,0)$ follows by an analogous argument. By definition of $\ray$, we have $\smash{\alpha\preceq_u  F(t)}$ with $l(\alpha,F(t))=t$. Using \autoref{Le:G}, let $\smash{\gamma\in G}$ connect $\alpha$ to $F(t)$. In particular, for $0 < s<t$ we have $\smash{\alpha\preceq_u\gamma_{s/t}}$ and $l(\alpha,\gamma_{s/t})=s$, whence $F(s) = \gamma_{s/t}$ as above.

Item \ref{La:II} follows from the definition of $\ray$.

Item \ref{La:III} in the cases $0 \leq s < t$ or $s< t\leq 0$ follows  as in \ref{La:I} by identifying intermediate points of $F$ and of $l$-geodesics. If $s < 0 < t$, by definition of $\ray$ and transitivity of $\sim_u$ we obtain
\begin{align*}
    l(F(s),F(t)) &= u(F(t)) - u(F(s))\\
    &= u(F(t)) - u(\alpha) - \big[u(F(s)) - u(\alpha)\big]\\
    &= l(\alpha,F(t)) - l(\alpha,F(s))\\ 
    &= t-s.\qedhere
\end{align*}
\end{proof}

\begin{remark}[Openness] In fact, since $\Tr_u$ has empty  intersection with the set $a\cup b$ of extremal points, $\Dom(\ray)(\alpha)$ is an open  subset of $\R$ for every $\alpha\in\pr_1(\Dom(\ray))$. 
\end{remark}

For every $\alpha\in \pr_1(\Dom(\ray))$, the map $\ray(\alpha,\cdot)$ follows the ``$l$-geodesic ray'' $t\mapsto F(t)$ on $\Dom(\ray)(\alpha)$. We thus call $\ray$  the \emph{ray map}. 

Given this map $\ray$, the proof of the following result is now analogous to  \cite[Lem. 3.9]{cavalletti2017} --- with $\met$ replaced by $l$ --- hence omitted. Observe that the similar argument from \cite[Prop.~4.9]{cavalletti2020} does not apply directly  since $u$ may attain negative values. Moreover, the $\sigma$-compactness in the proof of \cite[Lem.~3.9]{cavalletti2017} is unclear in our case and has to be replaced by mere Suslin measurability; this is why the map $\Quot$ below is stated as $\meas$-measurable, and not Borel measurable as in \cite[Lem.~3.9]{cavalletti2017}.

Recall that a map $\Quot\colon \Tr_u\to \Tr_u$ is called \emph{quotient map} for the equivalence relation $\smash{\sim_u}$ on $\Tr_u$ if $\graph\Quot \subset  \smash{E_{\sim_u}^2}$, and $\Quot(x)=\Quot(y)$ whenever $x,y\in\Tr_u$ obey $\smash{x\sim_u y}$.

\begin{proposition}[Quotient map]\label{Pr:Quotient map} There exists an $\meas$-measurable quotient map $\Quot\colon \Tr_u\to \Tr_u$ for the equivalence relation $\smash{\sim_u}$ on $\Tr_u$.

 Moreover, the image of $\Quot$ can be written as a countable union of level sets of  $u$. More precisely, there exist rational numbers $(a_i)_{i\in\N}$ and a sequence of mutually disjoint Suslin subsets $(Q_i)_{i\in\N}$ of $\Tr_u$ such that $Q_i \subset u^{-1}(a_i)$ for every $i\in \N$ and
\begin{align*}
    \sfQ(\Tr_u) = \bigcup_{i\in\N} Q_i.
\end{align*}
\end{proposition}

\begin{definition}[Quotient set and transport rays]\label{Def:Quot set}  The set $Q := \Quot(\Tr_u)$ will be termed  the \emph{quotient set}. 

Any set of the form $\mms_\alpha := \sfQ^{-1}(\alpha)$, $\alpha\in Q$, will be called  a \emph{transport ray}.
\end{definition}

The following readily follows from basic properties of $\ray$ and $\Quot$.

\begin{corollary}[Right-inverse of $\ray$]
    The map $\sfh\colon \Tr_u \to \Tr_u \times \R$ given by
    \begin{align*}
        \sfh(x) := \begin{cases}
            (x,0) & \textnormal{if } \sfQ(x) = x,\\
            (\sfQ(x), l(x,\sfQ(x))) & \textnormal{if }\sfQ(x) \neq x \textnormal{ yet } x\preceq_u \sfQ(x),\\
            (\sfQ(x), -l(\sfQ(x),x)) & \textnormal{if }\sfQ(x) \neq x \textnormal{ yet } \sfQ(x)\preceq_u x
        \end{cases}
    \end{align*}
    is right-inverse to $\ray$ on $\Tr_u$.
\end{corollary}

\subsection{Negligibility of bad points}\label{Sub:Negligibility} As highlighted in \autoref{Sub:Bad pointss}, in general the set $\smash{B\cup B^\leftarrow\cup a\cup b}$ removed from $\smash{\Tr_u^\End}$ to define $\Tr_u$ can be quite large. The goal of this section is to provide a natural condition
(\autoref{Th:All negligible} below),
which combines with Theorems \ref{Th:Optimal maps} and \ref{Th:Local to global} to ensure 
$\meas$-negligibility of the set of bad points relative to $\smash{\Tr_u^\End}$.

The hypotheses in all relevant results proven along the way decrease in their generality.  Ultimately, they all hold under curvature hypotheses, cf.~\autoref{Th:Needle}.

In the sequel, we call $\mu\in\Prob(\mms)$ an \emph{$\Ell^\infty$-measure} if it is $\meas$-absolutely continuous with $\meas$-essentially bounded Radon--Nikodým density.

\subsubsection{Negligibility of timelike branching points}\label{Subsub:Exclusion branching points}  The goal now is to prove the following \autoref{Th:Vanishing B}. For metric measure spaces, an analogous conclusion has already been observed to be implied by a similar hypothesis \cite{cavalletti2014,cavallettigigli,cavalletti2021,
cavalletti2017} sometimes called ``good transport behavior'' \cite{galaz2018}.  Note that in our case, the condition on $\mu_0$ and $\mu_1$ below might be empty, e.g.~if they do not admit any causal coupling. However, by \autoref{Ex:Str tl dual} and \autoref{Th:Optimal maps} this holds for sufficiently many marginals.

\begin{theorem}[Negligibility of timelike forward branching points]\label{Th:Vanishing B} Assume the existence of $0<p<1$ with the following property. For every $\Ell^\infty$-measure $\smash{\mu_0\in \scrP_\comp(\mms)}$ and every $\mu_1\in\scrP_\comp(\mms)$ such that
\begin{align*}
\supp\mu_0 \times \supp \mu_1 \subset \{l>0\},
\end{align*}
every $\smash{\ell_p}$-optimal coupling of $\mu_0$ and $\mu_1$ is induced by a map. Then
\begin{align*}
\meas[B]=0.
\end{align*}
\end{theorem}

Analogous conditions for the $\meas$-negligibility of $\smash{B^\leftarrow}$ are given in \autoref{Cor:Vanishing B backward}.

The proof of \autoref{Th:Vanishing B} requires some preparations. First, we prove the correspondence of \autoref{Le:FBranching} between elements of $\smash{B}$ and suitable  pairs of elements in $\smash{G}$ to hold in a measurable way. To this aim, we need  the subsequent version of the von Neumann selection theorem \cite[Thm.~5.5.2]{srivastava1998}.

Recall that for any sets $X$ and $Y$, a map $\Sec\colon \pr_1(F) \to Y$ is called \emph{section} of a given set $F\subset X\times Y$ if $\graph \Sec \subset F$.

\begin{theorem}[Von Neumann selection theorem]\label{Th:Selection} Let $X$ and $Y$ be Polish spaces, $F\subset X\times Y$ be Suslin, and $\scrA$ be the $\sigma$-algebra generated by all Suslin subsets of $X$. Then there exists an $\scrA$-measurable section $\Sec\colon \pr_1(F)\to Y$ of $F$.
\end{theorem}

In the sequel, let $X$ be the completion of $\smash{\TGeo(\mms) \cup \TGeo^{\leftarrow}(\mms)}$ with respect to the uniform distance $\met_\infty$. By a continuity argument, $X$ is Polish with respect to $\met_\infty$. Furthermore, the Borel $\sigma$-algebra of $X$ is equal to   the Borel $\sigma$-al\-gebra of $\smash{\Cont([0,1];\mms)}$ relative to $X$.

\begin{lemma}[Measurable selection of branching geodesics]\label{Le:Measurable section} There exists an $\meas$-measurable map $\smash{\Sec\colon B\to X^2}$ such that for every $\smash{(x,\gamma^1,\gamma^2)\in \graph \Sec}$ and every $0\leq t\leq 1$, the triple $(x,\gamma_t^1,\gamma_t^2)$ obeys the conclusions from \autoref{Le:FBranching}. That is, we have $\gamma^1,\gamma^2\in G$ such that
\begin{align*}
    \min_{t\in[0,1]} \met(\gamma_t^1,\gamma_t^2) >0,
\end{align*}
furthermore $\smash{x\preceq_u \gamma^1_t}$ and $\smash{x\preceq_u \gamma_t^2}$ yet $\smash{\gamma_t^1\not\sim_u\gamma_t^2}$, and
\begin{align*}
u(\gamma_t^1) = u(\gamma_t^2).
\end{align*}
\end{lemma}

\begin{proof} In the notation of \autoref{Th:Selection}, let $\scrA$ denote the $\sigma$-algebra generated by all Suslin sets of $\mms$. For later use, we claim that
\begin{align*}
F' := F_1 \cap F_2 \cap F_3 \cap F_4 \cap F_5 
\end{align*}
is a Suslin set, where
\begin{align*}
F_1 &:= \Tr_u^\End \times G^2,\\
F_2 &:= \{(x,\gamma^1,\gamma^2) \in \mms \times X^2 : x\preceq_u \gamma_0^1\},\\
F_3 &:= \{(x,\gamma^1,\gamma^2) \in \mms \times X^2 : x\preceq_u \gamma_0^2\},\\
F_4 &:= \mms \times \{(\gamma^1,\gamma^2)\in X^2 : \met(\gamma_0^1,\gamma_0^2)>0\},\\
F_5 &:= \mms \times \{(\gamma^1,\gamma^2) \in X^2 : u(\gamma_0^1) = u(\gamma_0^2) \textnormal{ and } u(\gamma_1^1) = u(\gamma_1^2)\}.
\end{align*}

Clearly, $F_4$ is Borel measurable. To prove that $F_1$ is Suslin, it suffices to prove that  $\smash{G}$ is. By \autoref{Le:G},
\begin{align*}
G &= \{\gamma \in \TGeo(\mms) : u(\gamma_1) - u(\gamma_0) = l(\gamma_0,\gamma_1)\}\\
&= \TGeo(\mms) \cap \{\gamma\in X : u(\gamma_1) = u(\gamma_0) + l(\gamma_0,\gamma_1)\}.
\end{align*}
Arguing as in \autoref{Subsub:Various}, we prove $\sigma$-compactness of $\smash{\TGeo(\mms)}$ in $X$, hence $\TGeo(\mms)$ is a Borel subset of $X$. The quantities $u(\gamma_1)$ and $u(\gamma_0) + l(\gamma_0,\gamma_1)$ depend Borel measu\-rably on $\gamma\in X$; the claim follows. Similarly, we establish that $F_5$ is a Suslin set. Finally, by definition of $\smash{\preceq_u}$ we obtain the decomposition
\begin{align*}F_2 &= A_1 \cup \big[A_2 \cap A_3 \cap A_4\big],
\end{align*}
where
\begin{align*}
    A_1 &:= \{(x,\gamma^1,\gamma^2)\in \mms\times X^2 : x = \gamma_0^1 \},\\
    A_2 &:= U \times \big[\eval_0^{-1}(U) \cap X\big] \times X,\\
    A_3 &:= \{(x,\gamma^1)\in \mms\times X : u(\gamma_0^1) = u(x) + l(x,\gamma_0^1)\}\times X,\\
    A_4 &:= \{(x,\gamma^1) \in \mms \times X : l(x,\gamma_0)>0\}.
\end{align*}
As above, it is then easily seen that these  are all Suslin sets, and so is $F_2$; similarly, we verify that $F_3$ is Suslin.

Now we construct the selection map $\Sec$. To exclude possible intersections along the curves (including the endpoints), we consider the Suslin set
\begin{align*}
F := F'\cap \big[\mms \times \{(\gamma^1,\gamma^2) \in X^2 : \min_{t\in[0,1]} \met(\gamma_t^1,\gamma_t^2) >0\}\big].
\end{align*}
By \autoref{Le:FBranching}, for every $x\in B$ we have $\smash{F \cap [\{x\}\times X^2] \neq \emptyset}$, in fact $\smash{B} \subset \pr_1(F)$. The inclusion $\smash{B \supset \pr_1(F)}$ is satisfied as well, which is derived as in the proof of \autoref{Le:FBranching}. Therefore, \autoref{Th:Selection} provides us with an $\scrA$-measurable selection map $\Sec\colon B \to X^2$. By construction, we have $\smash{\Sec(B) \subset G^2}$, and the other properties are again argued as in the proof of \autoref{Le:FBranching}.
\end{proof}

For the proof of \autoref{Th:Vanishing B}, we also need the following result about cyclical  monotonicity. The proof is completely analogous to the one of \cite[Prop.~4.12]{cavalletti2020}.

\begin{lemma}[A criterion for cyclical monotonicity]\label{Le:Cyclical mon} Let $0<p<1$ be arbitrary, and let $\smash{\Delta\subset  E^2_{\preceq_u}}$ satisfy
\begin{align*}
\big[u(x_1) - u(x_0)\big]\,\big[u(y_1) - u(y_0)\big] \geq 0
\end{align*}
for every $(x_0,y_0),(x_1,y_1)\in \Delta$. Then $\Delta$ is $l^p$-cyclically monotone.
\end{lemma}

\begin{proof}[Proof of \autoref{Th:Vanishing B}] Assume to the contrary that $\smash{\meas[B] > 0}$. In this case, we construct an $\Ell^\infty$-measure $\mu_0\in\Prob_\comp(\mms)$ and a measure $\mu_1\in \scrP_\comp^\ac(\mms,\meas)$ with
\begin{align*}
\supp\mu_0\times\supp\mu_1\subset\{l>0\}
\end{align*}
and a chronological $\smash{\ell_p}$-optimal coupling $\smash{\pi\in\Pi(\mu_0,\mu_1)}$ which is not induced by a map, which contradicts our assumption.

We start with some preparations. By inner regularity of $\meas$, we may and will fix a compact set $\smash{C'\subset B}$ with $0<\meas[C']<\infty$. Let $\Sec$ be the map given by  \autoref{Le:Measurable section}, and write $\smash{\gamma^j(x) := \pr_j \circ \Sec(x)}$, where $j=1,2$. An iterated application of Lusin's theorem gives the existence of an $\meas$-measurable set $C''\subset C'$ with $\meas[C'']>0$ such that $\smash{\Sec\big\vert_{C''}}$ and, for every $i=0,1$ and every $j=1,2$, the restriction of the function $\smash{x\mapsto u(\gamma^j(x)_i)}$ to $C''$ is continuous.

Given $x\in C''$, set $\smash{\alpha_x := u(\gamma^1(x)_0)}$ and $\smash{\beta_x := u(\gamma^1(x)_1)}$; note that $\alpha_x < \beta_x$ since $\smash{\pr_1\circ\Sec(C'') \subset G}$. Using that $\meas$ has full support, by a continuity argument we find an $\meas$-measurable set $C\subset C''$ with $\meas[C]>0$, henceforth fixed, such that
\begin{align*}
\sup_{x\in C}\alpha_x < \inf_{x\in C} \beta_x.
\end{align*}

Now fix $c > 0$ with
\begin{align}\label{Eq:c value}
\sup_{x\in C} \alpha_x < c  < \inf_{x\in C} \beta_x.
\end{align}
Since $\smash{u(\gamma^1(x)_t) = u(\gamma^2(x)_t)}$ for every $x\in C$ and every $0\leq t\leq 1$ by \autoref{Le:Measurable section}, there exists a continuous map $\sft\colon C \to (0,1)$ such that for every $x\in C$,
\begin{align}\label{Eq:sft gleich c}
u(\gamma^1(x)_{\sft(x)}) = u(\gamma^2(x)_{\sft(x)}) = c.
\end{align}
For $j =1,2$, define the continuous map $\smash{T^j\colon C \to \mms}$ by
\begin{align*}
T^j(x) := \gamma^j(x)_{\sft(x)}.
\end{align*}

Finally, define $\pi\in\scrP(\mms^2)$ by
\begin{align*}
2\,\pi := (\Id, T^1)_\push\big[\meas[C]^{-1}\,\meas\mres C\big] + (\Id,T^2)_\push\big[\meas[C]^{-1}\,\meas\mres C\big].
\end{align*}
Observe that $\smash{\mu_0 := (\pr_1)_\push\pi \in \scrP_\comp^\ac(\mms,\meas)}$ and $\smash{\mu_1 := (\pr_2)_\push\pi \in \scrP_\comp(\mms)}$. The relations \eqref{Eq:c value} imply that $\pi$ splits mass, i.e.~cannot be induced by a map. Furthermore, by construction of $\smash{T^1}$ and $\smash{T^2}$, cf.~\autoref{Le:Measurable section}, $\pi$ is chronological; in particular, up to restricting $\pi$ to a suitable product set we may and will assume
\begin{align*}
\supp\mu_0\times\supp\mu_1\subset\{l>0\}.
\end{align*}
To thus arrive at the desired contradiction, we claim $\smash{\ell_p}$-optimality of $\pi$. By construction, $\pi$ is concentrated on the set
\begin{align*}
\Delta := \{(x,\gamma^1(x)_{\sft(x)}) : x\in C\} \cup \{(x,\gamma^2(x)_{\sft(x)}) : x\in C\}.
\end{align*}
By \eqref{Eq:sft gleich c}, for every $(x_0,y_0),(x_1,y_1)\in \Delta$ we easily get
\begin{align*}
\big[u(x_1) - u(x_0)\big]\,\big[u(y_1) - u(y_0)\big] = 0.
\end{align*}
\autoref{Le:Cyclical mon} entails $l^p$-cyclical monotonicity of $\Delta$; since $\pi$ is chronological, its $\smash{\ell_p}$-optimality follows from \cite[Prop.~2.8]{cavalletti2020}.
\end{proof}

Mutatis mutandis, the following \autoref{Cor:Vanishing B backward} is shown. 

\begin{corollary}[Negligibility of timelike backward branching points]\label{Cor:Vanishing B backward} Assume the existence of  $0<p<1$ with the following property. For every  $\mu_0\in\smash{\scrP_\comp(\mms)}$ and every $\Ell^\infty$-measure $\smash{\mu_1\in\scrP_\comp(\mms)}$ such that
\begin{align*}
\supp\mu_0\times\supp\mu_1 \subset\{l>0\},
\end{align*}
every $\smash{\ell_p}$-optimal coupling of $\mu_0$ and $\mu_1$ is induced by a map. Then
\begin{align*}
\meas[B^\leftarrow] =0.
\end{align*}
\end{corollary}

Hence, if both assumptions from \autoref{Th:Vanishing B} and \autoref{Cor:Vanishing B backward} hold, then
\begin{align*}
\meas[B\cup B^\leftarrow]=0.
\end{align*}
In fact, thanks to \autoref{Le:Cyclical mon} the transport exponents in the two previous results do not need to coincide in order to get the same conclusion.

\subsubsection{Negligibility of extremal points}\label{Subsub:Exclusion endpoints} Now we address the $\meas$-negligibility of the sets $a$ and $b$ from \autoref{Def:Extremal points}. We essentially follow \cite[Sec.~4.4]{cavalletti2020} up to formulating the results in slightly larger generality as we neither assume curvature bounds nor timelike nonbranching, yet the proofs remain similar. 

Given a Suslin set $\smash{A\subset \Tr_u^\End}$ and $s\geq 0$, we define the Suslin set $\smash{A_s\subset \Tr_u^\End}$ by 
\begin{align*}
A_s := \pr_2\big[\{(\alpha,y) \in (A\times \Tr_u^\End) \cap {E_{\preceq_u}^2}  : l(\alpha,y) = s\}\big]
\end{align*}
In other words, $A_s$ is the set of points ``translated forward'' by the $l$-distance $t$ according to the relation $\preceq_u$. In particular, the set of all $s\geq 0$ such that $\meas[A_s]>0$ is Lebesgue measurable, which allows us to evaluate its $\smash{\Leb^1}$-measure. Up to bad points, the set $A_s$ is closely related to the ray map $\ray$ from \autoref{Cor:Ray map}.

\begin{lemma}[Size of sets moved along rays]\label{Le:Size} Assume the existence of $0<p<1$ such that for every $\Ell^\infty$-measure $\smash{\mu_0\in \Prob_\comp(\mms)}$ and every $\mu_1\in \Prob_\comp(\mms)$ with 
\begin{align*}
\supp\mu_0\times\supp\mu_1\subset\{l>0\},
\end{align*}
every element $\smash{\bdpi\in \OptTGeo_p(\mu_0,\mu_1)}$ is induced by a map and the measure   $(\eval_t)_\push\bdpi$ is absolutely continuous with respect to $\meas$ for every $0\leq t < 1$. Then for every Suslin set $A\subset \smash{\Tr_u^\End \setminus b}$ with $\meas[A]>0$, there exist $r > 0$ and a compact set $C \subset A$ with
\begin{align*}
\bigcup_{s\in[0,r]} C_s \Subset \mms,
\end{align*} 
and for every $0\leq s<r$ we have $\smash{C_s \subset \Tr_u^\End\setminus b}$ and $\meas[C_s] > 0$.

In particular, we obtain
\begin{align}\label{Eq:Lebesgue measure}
\Leb^1\big[\{0 \leq  s < r : \meas[A_s]>0 \}\big]>0.
\end{align}
\end{lemma}

Before entering the proof, we  relate  \autoref{Le:Size} to \autoref{Th:Vanishing B}. This becomes relevant in the proof of \autoref{Cor:a negligible}.

\begin{remark}[About the hypothesis of \autoref{Le:Size}]\label{Re:Abt hypothesis} The assumption of the above lemma is stronger than the one of \autoref{Th:Vanishing B}. Indeed, in the notation of both results, by \autoref{Le:Geodesics plan} every chronological $\smash{\ell_p}$-optimal coupling of $\mu_0$ and $\mu_1$ can be written as the push-forward of some $\smash{\bdpi\in\OptTGeo_p(\mu_0,\mu_1)}$ by $(\eval_0,\eval_1)$. But if the latter is induced by a map, so is the former.

Analogously, if the causal reversal $\smash{\scrM^{\leftarrow}}$ satisfies the assumption of \autoref{Le:Size}, then the hypothesis of \autoref{Cor:Vanishing B backward} holds.
\end{remark}

\begin{proof}[Proof of \autoref{Le:Size}] Arguing as for \cite[Prop.~4.12]{cavalletti2020} by using \autoref{Le:Cyclical mon}, for every $r\geq 0$ we deduce the Suslin set
\begin{align*}
\Lambda_r :=\{(\alpha,y)\in  (A\times \Tr_u^\End)\cap  E_{\preceq_u}^2: l(\alpha,y)=r\}
\end{align*}
is $l^p$-cyclically monotone. We claim
\begin{align}\label{Eq:bigcup E}
\bigcup_{r>0} \pr_1(\Lambda_r)=A.
\end{align}
Indeed, the inclusion ``$\subset$'' is clear since $\pr_1(\Lambda_r)\subset A$ for every $r>0$. On the other hand, since $A$ does not contain final points, for every $q\in A$ there exist $r>0$ and $\smash{y\in \Tr_u^\End}$ such that $(\alpha,y)\in \Lambda_r$, which establishes ``$\supset$''.

Now observe that if $0 \leq r' \leq r$ then  $\smash{\pr_1(\Lambda_r) \subset \pr_1(\Lambda_{r'})}$ thanks to \autoref{Le:G}. Together with \eqref{Eq:bigcup E}, Levi's monotone convergence theorem, and our assumption $\meas[A]>0$, we get $\meas[\pr_1(\Lambda_r)]>0$ for a sufficiently small parameter $r>0$ which is  henceforth fixed. Owing to inner regularity of $\meas$, let $C\subset \pr_1(\Lambda_s)$ be compact with $\meas[C]>0$.  
\autoref{Th:Selection} (von Neumann's selection theorem) yields the existence of a Suslin measurable map $\smash{T\colon C\to \Tr_u^\End}$ with $\graph T\subset \Lambda_r$. By Lusin's theorem and again inner regularity of $\meas$, up to shrinking $C$ we may and will assume $T$ to be continuous --- in particular, $T(C)$ is a compact subset of $\smash{\Tr_u^\End}$.

Now consider the measures $\smash{\mu_0\in \Prob_\comp^\ac(\mms,\meas)}$ and $\mu_1\in \Prob_\comp(\mms)$ defined by
\begin{align*}
\mu_0 &:= \meas[C]^{-1}\,\meas\mres C,\\
\mu_1 &:= T_\push\mu_0.
\end{align*}
By construction, the coupling $\pi := (\Id,T)_\push\mu_0 \in \Pi(\mu_0,\mu_1)$ is concentrated on $\Lambda_r$. This yields   both its $l^p$-cyclical monotonicity and its chronology; in particular, $\pi$ is $\smash{\ell_p}$-optimal by \cite[Prop.~2.8]{cavalletti2020}. Since every appropriate restriction of $\pi$ is $\smash{\ell_p}$-optimal as well \cite[Lem.~2.10]{cavalletti2020}, we may and will assume 
\begin{align*}
\supp\mu_0 \times \supp\mu_1 \subset \{l>0\}.
\end{align*}
The lifting procedure from 
\autoref{Le:Geodesics plan} and our assumption thus yield the existence of $\smash{\bdpi\in \OptTGeo_p(\mu_0,\mu_1)}$ with $(\eval_t)_\push\bdpi$ $\meas$-absolutely continuous for every $0\leq t <1$. \autoref{Le:Intermediate pts geodesics} yields $(\eval_t)_\push\bdpi \subset J(\supp\mu_0,\supp\mu_1)$ for every $0\leq t\leq 1$. In addition, arguing as in \autoref{Re:Abt hypothesis} we deduce $\pi = (\eval_0,\eval_1)_\push\bdpi$. As $T$ moves mass from $\supp\mu_0$ to $\supp\mu_1$ with constant $l$-length $r$, \autoref{Le:Intermediate pts geodesics} again implies $(\eval_t)_\push\bdpi$ is concentrated on $C_{rt}$ for every $0\leq t < 1$; $\meas$-absolute continuity of $(\eval_t)_\push\bdpi$  yields
\begin{align*}
\meas[C_{rt}] > 0.
\end{align*}
Collecting everything leads to the desired properties of $C_s$ for every $0\leq s<r$.

The inequality \eqref{Eq:Lebesgue measure} follows straightforwardly from the previous inequality and the inclusion $C_s \subset A_s$ for every $0\leq s < r$. This terminates the proof.
\end{proof}

\begin{theorem}[Negligibility of extremal points]\label{Cor:a negligible} Let  $\scrM$ and its causal reversal $\smash{\scrM^{\leftarrow}}$ satisfy the assumption of  \autoref{Le:Size}. Then
\begin{align*}
\meas[a]= \meas[b]=0.
\end{align*}
\end{theorem}

\begin{proof}  We only show $\meas[a]=0$, the proof of the identity $\meas[b]=0$ is analogous.

By \autoref{Re:Abt hypothesis} and \autoref{Cor:Vanishing B backward}, it suffices to prove
\begin{align*}
\meas[a\setminus B^\leftarrow] = 0.
\end{align*}
Assume to the contrary  $\meas[a\setminus B^\leftarrow] > 0$. Let $C\subset a\setminus B^\leftarrow$ be a compact subset as provided by applying \autoref{Le:Size} to $A:= a\setminus B^\leftarrow$. Retaining the notation of this lemma, we first claim $C_s \cap C_{s'}=\emptyset$ for every $0<s<s'<r$. Indeed, otherwise let $y\in C_s \cap C_{s'}$. By definition, there exist $\alpha,\alpha'\in C$ such that $l(\alpha,y) = s$ and $\alpha \preceq_u y$ as well as $l(\alpha',y)=s'$ and $\alpha' \preceq_u y$. As in Case 1 in the proof of \autoref{Th:Equiv relation} --- here we use that $a\setminus B^\leftarrow$ does not contain timelike backward branching points --- this leads to a contradiction.

Again by \autoref{Le:Size}, we get $\meas[C_s \setminus B^\leftarrow] = \meas[C_s] > 0$ for uncountably many $0\leq s<r$. On the other hand, the set $\smash{\bigcup_{s\in[0,r)} (C_s \setminus B^\leftarrow)}$ has compact closure, yet is intersected by uncountably many mutually disjoint Suslin sets of positive $\meas$-measure. This contradicts the finiteness of $\meas$ on compact subsets of $\mms$.
\end{proof}

Combining \autoref{Cor:a negligible} with \autoref{Re:Abt hypothesis}, \autoref{Th:Vanishing B}, and \autoref{Cor:Vanishing B backward} leads to the following.

\begin{corollary}[Negligibility of bad points]\label{Th:All negligible} Under the same hypotheses as stated in \autoref{Cor:a negligible}, the Suslin set of bad points is $\meas$-negligible.
\end{corollary}

\begin{remark}[Comparison with {\cite{cavalletti2020}}] Retain the notation of \autoref{Ex:CM}. In our case, the transport set $\smash{\mathcal{T}_V}$ constructed in \cite[Sec.~4.1]{cavalletti2020} coincides with $T \cup ((B \cup B^\leftarrow)\setminus (a\cup b))$. In particular, our transport set $\Tr_u$ differs from the previous set $\smash{\mathcal{T}_V}$ only  by an $\meas$-negligible set by \autoref{Cor:a negligible} (or \cite[Cor.~4.15, Thm.~4.17]{cavalletti2020}) in the framework considered in \cite{cavalletti2020}.
\end{remark}

\subsection{Disintegration of $\meas$ on the transport set}\label{Sub:Disintegration} After these preparations, we are in a position to state and prove our disintegration theorem.  This is the place where the weight function $w$ from \autoref{Sub:Framework} enters through the proof.

We recall the quotient map $\Quot\colon \Tr_u\to\Tr_u$ from \autoref{Pr:Quotient map}, the quotient set $Q := \Quot(\Tr_u)$,  and the transport rays $\mms_\alpha := \Quot^{-1}(\alpha)$, $\alpha\in Q$, from \autoref{Def:Quot set}. 

\begin{theorem}[Needle decomposition]\label{Th:Needle} Suppose $\scrM$ is a timelike $p$-essentially non\-branching $\smash{\TCD_p^e(k,N)}$ space. Then there exist
\begin{itemize}
\item a probability measure $\q \in \Prob(Q)$ which is mutually absolutely continuous to $\smash{\Quot_\push\big[\meas\mres\Tr_u\big]}$, and
\item a map $\meas_\cdot\colon Q \to \Meas(\mms)$
\end{itemize}
satisfying the disintegration formula
\begin{align*}
\meas\mres \Tr_u^\End = \meas\mres \Tr_u = \int_Q \meas_\alpha\d\q(\alpha)
\end{align*}
as well as the subsequent properties.
\begin{enumerate}[label=\textnormal{(\roman*)}]
\item \textnormal{\textbf{Measurability.}} For every $\meas$-measurable set $B\subset E$, the map $\alpha\mapsto\meas_\alpha[B]$ is $\q$-measurable.
\item \textnormal{\textbf{Strong consistency.}} For $\q$-a.e.~$\alpha\in Q$, the measure $\meas_\alpha$ is concentrated on the transport ray $\mms_\alpha$.
\item \textnormal{\textbf{Disintegration.}} For every $\meas$-measurable set $B\subset E$ and every $\q$-measur\-able set $A\subset Q$, 
\begin{align*}
\meas\big[B\cap \Quot^{-1}(A)\big] = \int_A \meas_\alpha[B] \d\q(\alpha).
\end{align*}
\item \textnormal{\textbf{Local finiteness.}} For every compact set $C\subset E$, 
\begin{align*}
\q\textnormal{-}\!\esssup_{\alpha\in Q} \meas_\alpha[C] < \infty.
\end{align*}
\end{enumerate}

Furthermore, given any $\q\in \Prob(Q)$ as above the disintegration $\meas_\cdot$ is $\q$-essentially unique, in the sense that for every map $\smash{\meas_\cdot'\colon Q\to \Meas(\mms)}$ with the same properties, 
\begin{align*}
\meas_\cdot = \meas_\cdot'\quad\q\textnormal{-a.e.}
\end{align*}
\end{theorem}

\begin{proof} \autoref{Th:Uniqueness geos} combined with \autoref{Th:All negligible} yields 
\begin{align*}
\meas\mres\Tr_u^\End = \meas\mres \Tr_u.
\end{align*}
It thus suffices to find the desired disintegration of $\meas\mres \Tr_u$.

Define $\mmeas,\q\in \Prob(\mms)$ by $\mmeas := w\,\meas\mres \Tr_u$, where $w$ is the weight function normalizing $\meas\mres E$  from \autoref{Sub:Framework},  and $\q := \Quot_\push\mmeas$. By construction, $\q$ is concentrated on $Q$. We claim $\q$ is mutually absolutely continuous to $\smash{\Quot_\push\big[\meas\mres\Tr_u\big]}$. Indeed, let $A\subset Q$ be a Suslin set with $\q[A]=0$. This means
\begin{align}\label{Eq:IntNull}
\int_{\Quot^{-1}(A)} w\d\meas\mres \Tr_u = 0,
\end{align}
and therefore $\smash{\meas\mres\Tr_u\big[\Quot^{-1}(A)\big]=0}$ since $w$ is positive on $E$. On the other hand, the identity  $\smash{\meas\mres\Tr_u\big[\Quot^{-1}(A)\big]=0}$ clearly implies \eqref{Eq:IntNull}, hence $\q[A]=0$.

The disintegration theorem for probability measures \cite[Thm.~452I]{fremlin} applied to $\mmeas$ and $\q$ yields the existence of a map $\mmeas_\cdot\colon Q\to \Prob(\mms)$ with the following properties.
\begin{itemize}
\item For every $\mmeas$-measurable set $B\subset E$, the map $\alpha\mapsto\mmeas_\alpha[B]$ is $\q$-measurable.
\item For $\q$-a.e.~$\alpha\in Q$, the measure $\mmeas_\alpha$ is concentrated on the transport ray $\mms_\alpha$.
\item For every $\mmeas$-measurable set $B\subset E$ and every $\q$-measurable set $A\subset Q$,
\begin{align*}
\mmeas\big[B\cap \Quot^{-1}(A)\big] = \int_A\mmeas_\alpha[B]\d\q(\alpha).
\end{align*}
\end{itemize}
Moreover, every map $\smash{\mmeas_\cdot'\colon Q\to \Prob(\mms)}$ with these properties satisfies
\begin{align*}
\mmeas_\cdot = \mmeas_\cdot'\quad\q\textnormal{-a.e.}
\end{align*}

It is now straightforward to verify that setting 
\begin{align*}
\meas_\cdot := w^{-1}\,\mmeas_\cdot,
\end{align*}
leads to all desired properties; see  \cite[Sec.~4.2]{cavalletti2020} for details. We only comment on local finiteness, which does not  follow from the above disintegration theorem. For a compact subset $C\subset\mms$, since $c := \inf w(C) > 0$ we obtain
\begin{align*}
\q\textnormal{-}\!\esssup_{\alpha\in Q} \meas_\alpha[C]  \leq c^{-1}\,\q\textnormal{-}\!\esssup_{\alpha\in Q}\mmeas_\alpha[C] \leq c^{-1}.\tag*{\qedhere}
\end{align*}
\end{proof}



\begin{remark}[Relation of $\meas\mres E$ and $\meas\mres \Tr_u$]  If needed for applications, the identity $\meas\mres E = \meas \mres \Tr_u$ has to be checked separately. What we only know a priori thanks to the inclusion $\Tr_u\subset E$ is the inequality ``$\geq$''.
\end{remark}

Lastly, we get to the additional important regularity property of absolute continuity of the conditional measures from \autoref{Th:Needle}. To this aim, recall the ray map $\ray\colon \Dom(\ray) \to \Tr_u$ from \autoref{Cor:Ray map}, where $\Dom(\ray)\subset \Tr_u\times \R$. For every $\alpha\in Q$, the map $\ray(\alpha,\cdot)\colon \Dom(\ray)(\alpha)\to \Tr_u$, where $\Dom(\ray)(\alpha)$ is an open interval in $\R$ containing~$0$, parametrizes the transport ray $\smash{\tilde{\alpha}^u}$ passing through $\alpha = \ray(\alpha,0)$. Since $\ray(\alpha,\cdot)$ is an order isometry, we obtain
\begin{align*}
\ray(\alpha,\cdot)_\push\big[\Leb^1\mres \Dom(\ray)(\alpha)\big] = \scrH^1\mres \ray(\alpha,\Dom(\ray)(\alpha))
\end{align*}
as an equality between measures on $\mms_\alpha$. To ease the notation, we write
\begin{align*}
\Leb^1_\alpha := \ray(\alpha,\cdot)_\push\big[\Leb^1\mres \Dom(\ray)(\alpha)\big].
\end{align*}

The proof of the following result is completely analogous to  \cite[Prop.~4.16]{cavalletti2020} in the special case of \autoref{Ex:CM} by using \autoref{Le:Size} and thus omitted.

\begin{proposition}[Absolute continuity of conditional measures]\label{Pr:Abs cont needles} Assume $\scrM$ is a timelike $p$-essentially nonbranching  $\smash{\TCD_p^e(k,N)}$ space, and let $\q$ be a disintegration of $\meas\mres \Tr_u$ as provided by \autoref{Th:Needle}. Then for $\q$-a.e.~$\alpha\in Q$, the conditional measure $\meas_\alpha$ is absolutely continuous with respect to $\smash{\Leb^1_\alpha}$. 
\end{proposition}

In the following, for $\alpha\in Q$ we write $h_\alpha$ for the ($\smash{\Leb^1_\alpha}$-uniquely determined) Radon--Nikodým derivative of $\meas_\alpha$ with respect to $\smash{\Leb^1_\alpha}$.

\begin{remark}[Identification]\label{Re:Ident halpha} It will be useful to identify the transport ray $\mms_\alpha\subset\mms$, where $\alpha\in Q$, with its parametrization $\Dom(\ray)(\alpha) \subset\R$ by the homeomorphism $\ray(\alpha,\cdot)$; recall \autoref{Cor:Ray map}. We employ barred notations for the  real counterparts  of all objects associated with $\mms_\alpha$ as follows:
\begin{align}\label{Real bar}
\begin{split}
\bar{\mms}_\alpha &:= \Dom(\ray)(\alpha),
\\\bar{\meas}_\alpha &:= \ray(\alpha,\cdot)^{-1}_\push\meas_\alpha,
\\\bar{\Leb}_\alpha^1 &:= \Leb^1\mres \bar{\mms}_\alpha,
\\\bar{h}_\alpha &:= h_\alpha \circ \ray(\alpha,\cdot).
\end{split}
\end{align}
From \autoref{Pr:Abs cont needles}, we deduce $\smash{\bar{\meas}_\alpha}$ is $\smash{\Leb^1}$-absolutely continuous, and
\begin{align*}
\bar{\meas}_\alpha = \bar{h}_\alpha\,\bar{\Leb}^1_\alpha.
\end{align*}
\end{remark}

Arguing as for \autoref{Th:Needle} and \autoref{Pr:Abs cont needles}, we deduce the following.

\begin{theorem}[Smooth needle decomposition]\label{Th:Smooth needle} Let $\scrM$ form the  metric measure spacetime induced by a smooth, globally hyperbolic Finsler space\-time. Moreover, let $E$ and $u$ be  given as in \autoref{Sub:Framework}. Then the conclusions of \autoref{Th:Needle} and \autoref{Pr:Abs cont needles} hold.
\end{theorem}

This echoes Klartag's Riemannian result \cite[Thm.~3.23]{klartag2017} in the absence of curvature bounds; general uniqueness results of chronological $\smash{\ell_p}$-optimal couplings and $\smash{\ell_p}$-geodesics (see \cite[Thm.~5.8, Cor.~5.9]{mccann2020} for the Lorentzian case and \cite[Thm.~4.17, Cor.~4.18]{braunohta} for general Finsler spacetimes) explain why \autoref{Th:Smooth needle} does not  require curvature hypotheses.

\subsection{Localization of $\smash{\TCD_p^e(k,N)}$}\label{Sub:Local} The crucial property of our needle decomposition is that the $\smash{\TCD_p^e(k,N)}$ property passes over to $\q$-a.e.~needle and makes them one-dimensional  variable $\CD$ metric measure spaces, \autoref{Re:Connection CD}. Its proof follows Cavalletti--Mondino's counterpart for metric measure spaces \cite[Thm.~4.2]{cavalletti2017},  which for constant $k$ they extended to Lorentzian length spaces in \cite[Thm.~3.2]{cm++}.

For the following discussion, let us recall the identifications \eqref{Real bar} from \autoref{Re:Ident halpha}. 
Given any $\alpha\in Q$ we define $\smash{k_\alpha\colon \bar{\mms}_\alpha \to \R}$ by
\begin{align*}
k_\alpha  := k \circ \ray(\alpha,\cdot).
\end{align*}
Thanks to \autoref{Cor:Ray map}, $\smash{k_{\alpha}}$ is  lower semicontinuous. In particular, the induced distortion coefficients $\smash{\sigma_{k_\alpha^\pm/N}^{(t)}(\theta)}$  are defined for every $0\leq \theta<\diam^l\mms_\alpha$ and every $0\leq t\leq 1$  analogously to \autoref{Def:dist coeff general k}.

Given any $K\in \R$, $\theta > 0$, and $0\leq t\leq 1$, recalling our simplifying assumption $N>1$ we need to consider the geometrically averaged distortion coefficients
\begin{align}\label{Eq:Tau inter}
\tau_{K,N}^{(t)}(\theta) := t^{1/N}\,\sigma_{K/(N-1)}^{(t)}(\theta)^{1-1/N}
\end{align}
further detailed in \autoref{Def:tau dist coeff}. Although  \eqref{Eq:tau sigma inequ} below implies
\begin{align*}
\sigma_{K/N}^{(t)}(\theta)\leq \tau_{K,N}^{(t)}(\theta)
\end{align*}
the following lemma \cite[Lem.~2.1]{dengsturm} ``reverses'' this inequality locally at the cost of a slightly worse real  constant $K' < K$.

\begin{lemma}[Local comparison  of distortion coefficients]\label{Le:Local comp}  For every $K,K'\in\R$ with $K' < K$, there exists $\theta > 0$ such that for every $0 < \theta' < \theta$ and every $0\leq t\leq 1$,
\begin{align*}
\sigma_{K/N}^{(t)}(\theta')\geq \tau_{K',N}^{(t)}(\theta').
\end{align*}
\end{lemma}

\begin{theorem}[Disintegration into $\smash{\CD(k,N)}$ needles]\label{Th:Localization TCD} Assume $\scrM$ is a timelike $p$-essentially nonbranching $\smash{\TCD_p^e(k,N)}$ space, and let $\q$ be a disintegration according to \autoref{Th:Needle}. Then  $\smash{\bar{\meas}_\alpha}$ is absolutely continuous with respect to the one-dimensional Lebesgue measure 
for $\q$-a.e.~$\alpha\in Q$, and its $\smash{\bar{\Leb}^1_\alpha}$-density $\smash{\bar{h}_\alpha}$ 
has a nonrelabeled $\smash{\bar{\Leb}^1_\alpha}$-version which is locally Lipschitz continuous and strictly positive on the open real interval $\smash{\bar{\mms}_\alpha}$, continuous on its closure, and satisfies the following inequality for every $t_0,t_1\in \bar{\mms}_\alpha$ with $t_0<t_1$ and every $0\leq s\leq 1$:
\begin{align*}
\bar{h}_\alpha((1-s)\,t_0 + s\,t_1)^{1/(N-1)} &\geq \sigma_{k_\alpha^-/(N-1)}^{(1-s)}(t_1-t_0)\,\bar{h}_\alpha(t_0)^{1/(N-1)}\\
&\qquad\qquad + \sigma_{k_\alpha^+/(N-1)}^{(s)}(t_1-t_0)\,\bar{h}_\alpha(t_1)^{1/(N-1)}.
\end{align*}
\end{theorem}

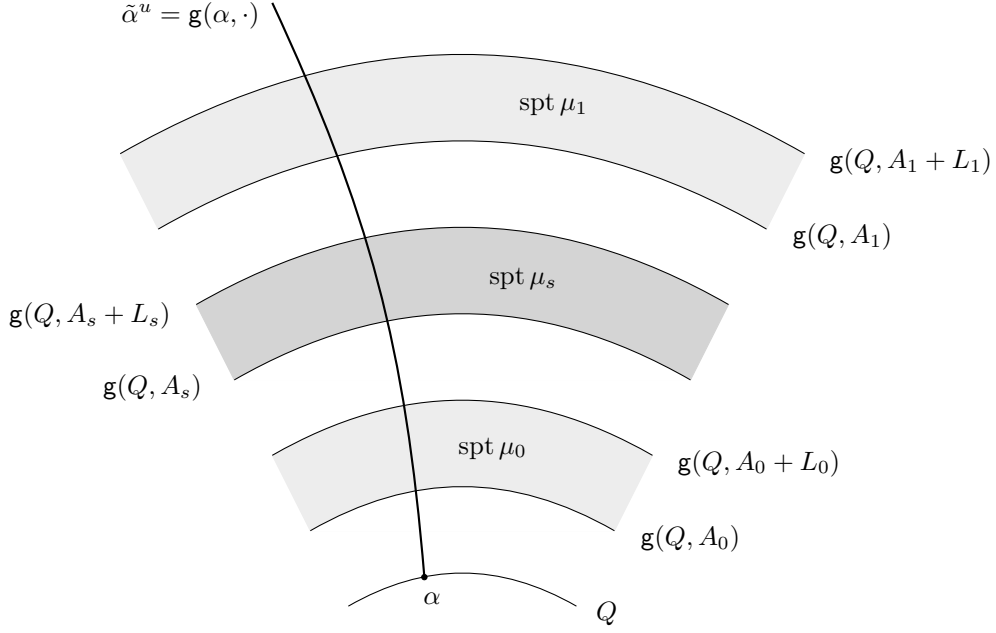
\begin{figure}
\centering
\begin{tikzpicture}
\draw (0,0) to[bend left] (3,0);
\draw[black!7, fill=black!7] plot coordinates {(-1,2) (4,2)  (3.5,1) (-0.5,1) (-1,2)};
\draw[fill=black!7] (-1,2) to[bend left] (4,2);
\draw[white, thick] (-0.5,1) -- (3.5,1);
\draw[fill=white] (-0.5,1) to[bend left] (3.5,1);
\draw[black!7, fill=black!7] plot coordinates {(-3,6) (6,6) (5.5,5) (-2.5,5) (-3,6)};
\draw[white, thick] (5.5,5) -- (-2.5,5);
\draw[fill=black!7] (-3,6) to[bend left] (6,6);
\draw[fill=white] (-2.5,5) to[bend left] (5.5,5);
\draw[black!17, fill=black!17] plot coordinates {(-2,4) (5,4)  (4.5,3) (-1.5,3) (-2,4)};
\draw[white, thick] (-1.5,3) -- (4.5,3);
\draw[fill=black!17] (-2,4) to[bend left] (5,4);
\draw[fill=white] (-1.5,3) to[bend left] (4.5,3);
\node at (3.4,-0.1) {$Q$};
\node at (4.5,0.9) {$\ray(Q,A_0)$};
\node at (5.4,1.9) {$\ray(Q,A_0 + L_0)$};
\node at (-2.575,2.9) {$\ray(Q,A_s)$};
\node at (-3.4,3.85) {$\ray(Q,A_s+ L_s)$};
\node at (6.5,4.9) {$\ray(Q,A_1)$};
\node at (7.4,5.9) {$\ray(Q,A_1+L_1)$};
\node at (1.9,2.1) {$\supp\mu_0$};
\node at (2.3,4.35) {$\supp\mu_s$};
\node at (2.7, 6.65) {$\supp\mu_1$};
\draw[thick] (1,0.385) to[bend right=10] (-1,8);
\draw[fill=black] (1,0.385) circle(.1em);
\node at (1.1,0.1) {$\alpha$};
\node at (-2.075,7.75) {$\smash{\tilde{\alpha}^u = \ray(\alpha,\cdot)}$};
\end{tikzpicture}
\caption{An illustration of the proof of \autoref{Th:Localization TCD}.}\label{Fig:Localization}
\end{figure}

\begin{proof} Roughly speaking, we proceed as follows. We cover $E$ with precompact, $l$-geodesically convex sets. On each of these sets, we will displacement interpolate mass distributions from $u^{-1}([A_0,A_0+L_0])$ to $u^{-1}([A_1,A_1+L_1])$ for suitable real numbers $A_0$, $A_1$, $L_0$, and $L_1$ along the $l$-geodesic needles. The $\smash{\TCD_p^e(k,N)}$ property along this transport reduces to a pathwise inequality for the respective densities by \autoref{Th:Pathwise}, and some clever choices of the involved constants leads to the desired estimate. The latter both globalizes along the entire transport along the needles and throughout the cover of $E$; see also \autoref{Re:Connection CD}. 

We first suppose $E$ is precompact (and understand all objects relative to $E$). By \eqref{Eq:kappan}, we may and will then first assume $k$ is continuous. In turn, up to covering the set $E$ with small  chronological diamonds, given any $\varepsilon > 0$ we may and will thus assume the existence of $K\in \R$ with $K-\varepsilon\leq k\leq K+\varepsilon$ on $E$. 

By \autoref{Pr:Quotient map}, $Q$ is the disjoint union of countably many Suslin subsets $(Q_i)_{i\in\N}$ of $\Tr_u$ such that  $\smash{Q_i\subset u^{-1}(a_i)}$ for every $i\in\N$, where $(a_i)_{i\in\N}$ are rational numbers. Evidently, it suffices to show the claim for $\q$-a.e.~$\alpha\in Q_i$, where $i\in\N$ is arbitrary. In the sequel, we identify $Q_i$ with $\ray(Q_i,0)$. Since $\Dom(\ray)(\alpha)$ is an open convex interval in $\R$, up to partitioning $Q_i$ into countably many pieces and arguing separately we  assume the  existence of $b_0,b_1\in\R$ with $\smash{(b_0,b_1)\subset \bar{\mms}_\alpha}$ for every $\alpha\in Q_i$.  To ease notation,  let $a_i = 0$, $b_0 < 0 < b_1$, $\q[Q_i] > 0$, and $Q_i = Q$.

Let $a_0 < A_0 < A_1 < a_1$ and $L_0,L_1 > 0$ satisfy $A_0 + L_0 < A_1$ and $A_1 + L_1 < a_1$. Define $\smash{\mu_0,\mu_1\in\Prob(\mms)}$ by
\begin{align*}
\mu_0 &= \int_Q \ray(\alpha,\cdot)_\push\big[L_0^{-1}\,\Leb^1\mres [A_0, A_0 + L_0]\big]\d\q(\alpha),\\
\mu_1 &= \int_Q \ray(\alpha,\cdot)_\push\big[L_1^{-1}\,\Leb^1\mres [A_1, A_1 + L_1]\big]\d\q(\alpha).
\end{align*}
By \autoref{Pr:Abs cont needles}, $\mu_0$ and $\mu_1$ are $\meas$-absolutely continuous, and for $\q$-a.e.~$\alpha\in Q$  their respective $\meas$-densities satisfy the formulas
\begin{align}\label{Eq:Formula Density 1}
\begin{split}
\rho_0^{-1}\circ\ray(\alpha,\cdot) &= L_0\,\bar{h}_\alpha \quad  \textnormal{on } [A_0,A_0+L_0],\\
\rho_1^{-1}\circ\ray(\alpha,\cdot) &= L_1\,\bar{h}_\alpha \quad  \textnormal{on }[A_1,A_1+L_1].
\end{split}
\end{align}
Our hypothesis on $E$ makes $\mu_0$ and $\mu_1$ compactly supported. By construction and since  $\ray(\alpha,\cdot)$ parametrizes $l$-geodesics passing through $\alpha = \ray(\alpha,\cdot)$,
\begin{align*}
\supp\mu_0\times\supp\mu_1\subset\{l\geq A_1 - A_0 - L_0\} \subset \{l>0\}.
\end{align*}

Given any $0\leq s\leq 1$, we define $\mu_s\in \Prob(\mms)$ by
\begin{align*}
\mu_s &:= \int_Q\ray(\alpha,\cdot)_\push\big[L_s^{-1}\,\Leb^1\mres[A_s, A_s+L_s]\big]\d\q(\alpha),
\end{align*}
where
\begin{align*}
A_s &:= (1-s)\,A_0 + s\,A_1,\\
L_s &:= (1-s)\,L_0 + s\,L_1.
\end{align*}
As for $\mu_0$ and $\mu_1$, $\mu_s$ has compact support and is $\meas$-absolutely continuous with $\meas$-density given for $\q$-a.e.~$\alpha\in Q$ by the formula
\begin{align}\label{Eq:Formula Density 2}
\rho_s^{-1}\circ \ray(\alpha,\cdot) = L_s\,\bar{h}_\alpha\quad\textnormal{on }[A_s,A_s+L_s].
\end{align}
Using \autoref{Le:Cyclical mon}, it is not hard to verify $(\mu_s)_{s\in[0,1]}$ is an $\smash{\ell_p}$-geodesic from $\mu_0$ to $\mu_1$. By \autoref{Cor:Lifting nonbr}, it necessarily coincides with the unique displacement $\smash{\ell_p}$-geodesic from $\mu_0$ to $\mu_1$ --- which is also the one provided by \autoref{Th:Pathwise}. In particular, by \autoref{Re:LIPSCHITZ} we have local Lipschitz continuity of  $\smash{\bar{h}_\alpha}$  for $\q$-a.e.~$\alpha\in Q$.

From the latter pathwise version of $\smash{\TCD_p^e(k,N)}$ and our assumption on $k$ on $E$, \eqref{Eq:Formula Density 1}, and \eqref{Eq:Formula Density 2},  we thus get the following estimate for $\q$-a.e.~$\alpha\in Q$. For every $A_0 \leq t_0 \leq A_0 + L_0$ and every $A_1 \leq t_1\leq A_1+L_1$ represented by $t_0 = A_0 + \lambda\,L_0$ and $t_1 = A_1 + \lambda\,L_1$ for some $0\leq \lambda\leq 1$, and for every $0\leq s\leq 1$,
\begin{align*}
L_s^{1/N}\,\bar{h}_\alpha((1-s)\,t_0 + s\,t_1)^{1/N} &\geq \sigma_{k_\alpha^-/N}^{(1-s)}(t_1-t_0)\, L_0^{1/N}\,\bar{h}_\alpha(t_0)^{1/N}\\
&\qquad\qquad + \sigma_{k_\alpha^+/N}^{(s)}(t_1-t_0)\,L_1^{1/N}\,\bar{h}_\alpha(t_1)^{1/N}\\
&\geq \sigma_{(K-\varepsilon)/N}^{(1-s)}(t_1-t_0)\,L_0^{1/N}\,\bar{h}_\alpha(t_0)^{1/N}\\
&\qquad\qquad + \sigma_{(K-\varepsilon)/N}^{(s)}(t_1-t_0)\,L_1^{1/N}\,\bar{h}_\alpha(t_1)^{1/N}.
\end{align*}
Setting $s=1/2$ in the previous inequality yields
\begin{align*}
&L_{1/2}^{1/N}\,\bar{h}_\alpha(A_{1/2} + \lambda\,L_{1/2})^{1/N}\\
&\qquad\qquad \geq \sigma_{(K-\varepsilon)/N}^{(1/2)}(A_1-A_0 + \lambda\,\vert L_1 - L_0\vert)\,\big[L_0^{1/N}\,\bar{h}_\alpha(A_0 + \lambda\,L_0)^{1/N}\\
&\qquad\qquad\qquad\qquad  + L_1^{1/N}\,\bar{h}_\alpha(A_1+\lambda\,L_1)^{1/N}\big].
\end{align*}
By the continuity of $\smash{\bar{h}_\alpha}$, we can set $\lambda = 0$ and obtain
\begin{align}\label{Eq:INEQU!}
\begin{split}
L_{1/2}^{1/N}\,\bar{h}_{\alpha}(t_{1/2})^{1/N} &\geq \sigma_{(K-\varepsilon)/N}^{(1/2)}(t_1-t_0)\,\big[L_0^{1/N}\,\bar{h}_\alpha(t_0)^{1/N}  + L_1^{1/N}\,\bar{h}_\alpha(t_1)^{1/N}\big],
\end{split}
\end{align}
where $2\,t_{1/2} := t_0 + t_1$. The $\q$-conegligible set satisfying the previous inequality depends on $t_0$, $t_1$, $L_0$, and $L_1$. However, again by continuity of $\smash{\bar{h}_\alpha}$ for $\q$-a.e.~$\alpha\in Q$, we deduce the following. There exists a $\q$-conegligible subset of $Q$ on which \eqref{Eq:INEQU!} holds for every $a_0 < t_0 < t_1 < a_1$ and every $L_0,L_1 > 0$ with $t_0 + L_0 < t_1$ and $t_1+L_1 < t_1$. From now on, we fix $\alpha \in Q$ which belongs to this set. Without loss of generality, we suppose $\smash{\bar{h}_\alpha(t_0)> 0}$ or $\smash{\bar{h}_\alpha(t_1)>0}$.

By \autoref{Le:Local comp}, there exists $\theta > 0$ such that  if $t_1 - t_0 < \theta$, then
\begin{align*}
L_{1/2}^{1/N}\,\bar{h}_\alpha(t_{1/2})^{1/N} &\geq \tau_{K-2\varepsilon,N}^{(1/2)}(t_1-t_0)\,\big[L_0^{1/N}\,\bar{h}_\alpha(t_0)^{1/N}  + L_1^{1/N}\,\bar{h}_\alpha(t_1)^{1/N}\big].
\end{align*}
By the definition \eqref{Eq:Tau inter}, this is equivalent to
\begin{align*}
(L_0 + L_1)^{1/N}\,\bar{h}_\alpha(t_{1/2})^{1/N} &\geq \sigma_{(K-2\varepsilon)/(N-1)}^{(1/2)}(t_1-t_0)^{1-1/N}\,\big[L_0^{1/N}\,\bar{h}_\alpha(t_0)^{1/N} \\
&\qquad\qquad + L_1^{1/N}\,\bar{h}_\alpha(t_1)^{1/N}\big].
\end{align*}
For sufficiently small $L>0$, the optimal choices
\begin{align*}
L_0 &= L\,\bar{h}_\alpha(t_0)^{1/(N-1)}\,\big[\bar{h}_\alpha(t_0)^{1/(N-1)} + \bar{h}_\alpha(t_1)^{1/(N-1)}\big]^{-1},\\
L_1 &= L\,\bar{h}_\alpha(t_1)^{1/(N-1)}\,\big[\bar{h}_\alpha(t_0)^{1/(N-1)} + \bar{h}_\alpha(t_1)^{1/(N-1)}\big]^{-1}
\end{align*}
easily lead to
\begin{align*}
\bar{h}_\alpha(t_{1/2})^{1/(N-1)}\geq \sigma_{(K-2\varepsilon)/(N-1)}^{(1/2)}(t_1-t_0)\,\big[\bar{h}_\alpha(t_0)^{1/(N-1)} + \bar{h}_\alpha(t_1)^{1/(N-1)}\big].
\end{align*}
Recall this inequality holds if $t_1 - t_0 < \theta$, where $\theta > 0$ depends on $K$ and $\varepsilon$ through \autoref{Le:Local comp}. Yet, this midpoint inequality globalizes and holds along all interior points of the geodesic which connects fixed yet arbitrary endpoints \cite[Lem.~5.1, Thm.~5.2]{cavstu}; consequently, for every $a_0 < t_0 < t_1 < a_1$ and every $0\leq s\leq 1$,
\begin{align*}
\bar{h}_\alpha((1-s)\,t_0 + s\,t_1)^{1/(N-1)} &\geq \sigma_{(K-2\varepsilon)/(N-1)}^{(1-s)}(t_1-t_0)\,\bar{h}_\alpha(t_0)^{1/(N-1)}\\
&\qquad\qquad + \sigma_{(K-2\varepsilon)/(N-1)}^{(s)}(t_1-t_0)\,\bar{h}_\alpha(t_1)^{1/(N-1)}\\
&\geq \sigma_{(k_\alpha -3\varepsilon)^-/(N-1)}^{(1-s)}(t_1-t_0)\,\bar{h}_\alpha(t_0)^{1/(N-1)}\\
&\qquad\qquad \sigma_{(k_\alpha - 3\varepsilon)^+/(N-1)}^{(s)}(t_1-t_0)\,\bar{h}_\alpha(t_1)^{1/(N-1)}.
\end{align*}
By the arbitrariness of $\varepsilon$ and \autoref{Le:Properties}, we conclude
\begin{align*}
\bar{h}_\alpha((1-s)\,t_0 + s\,t_1)^{1/(N-1)} &\geq \sigma_{k_\alpha^-/(N-1)}^{(1-s)}(t_1-t_0)\,\bar{h}_\alpha(t_0)^{1/(N-1)}\\
&\qquad\qquad + \sigma_{k_\alpha^+/(N-1)}^{(s)}(t_1-t_0)\,\bar{h}_\alpha(t_1)^{1/(N-1)}.
\end{align*}
By \autoref{Le:Properties} again, we remove the continuity hypothesis on $k$ on $E$.

Finally, we outline how to drop the initial assumption of precompactness of $E$. Owing to Polishness of $\Top$, we cover $E$ with countably many small chronological diamonds. The intersection of each of these sets with $E$ is precompact and $l$-geodesically convex. By the above argumentation, the claimed estimate holds on each of these small sets. We  then apply the local-to-global property of $\smash{\CD(k_\alpha,N)}$ densities --- after an evident variable adaptation of \cite[Lem.~A.3]{cavalletti2021}, cf.~\cite{braunc2} plus \autoref{Re:Connection CD} below --- to extend  the statement to all of $E$.
\end{proof}

\begin{remark}[The case $N=1$]\label{Re:Limit constant} In the limit case $N=1$, one can follow a similar argument as above and show that for $\q$-a.e.~$\alpha\in Q$, the $\smash{\bar{\Leb}_1^\alpha}$-density $\smash{\bar{h}_\alpha}$ is constant. Compare with the proof of \cite[Thm.~4.2]{cavalletti2017}.
\end{remark}

We  summarize \autoref{Th:Localization TCD} by saying $\q$ is a  \textit{$\CD(k,N)$ disintegration}. This is justified by the following remark.

\begin{remark}[Connection to $\smash{\CD(k_\alpha,N)}$ metric measure spaces]\label{Re:Connection CD} The conclusion of the previous theorem is equivalent to say $\smash{\bar{h}_\alpha}$ is a \emph{$\smash{\CD(k_\alpha,N)}$ density} on $\smash{\bar{\mms}_\alpha}$ for $\q$-a.e.~$\alpha\in Q$. In particular, the triple $\smash{(\bar{\mms}_\alpha, \vert\cdot-\cdot\vert, \bar{\meas}_\alpha)}$ forms a $\smash{\CD(k_\alpha,N)}$ metric measure space after Ketterer \cite[Def.~4.4]{ketterer2017}. Both statements follow from evident variable adaptations of \cite[Def.~A.1, Thm.~A.2]{cavalletti2021} further discussed in \cite{braunc2}.
\end{remark}

The argument for the following consequence of \autoref{Th:Localization TCD} is standard. 
Recalling \eqref{Real bar}, 
given any $\alpha\in Q$ and any $\smash{x_0,x_1\in\bar{\mms}_\alpha}$ with $x_0< x_1$, we consider the ($\alpha$-dependent) functions $\smash{\kappa_{x_0,x_1}^\pm\colon [0, \vert x_1-x_0\vert]\to \R}$ defined by
\begin{align}\label{Eq:BRRRP}
\begin{split}
\kappa_{x_0,x_1}^+(s\,\vert x_1-x_0\vert) &:= \bar{k}_\alpha((1-s)\,x_0 + s\,x_1),\\
\kappa_{x_0,x_1}^-(s\,\vert x_1-x_0\vert) &:= \bar{k}_\alpha((1-s)\,x_1 + s\,x_0).
\end{split}
\end{align}

\begin{corollary}[Comparison inequality]\label{Cor:Comparison} Let $\q$ be a $\CD(k,N)$ disintegration. Let $\alpha\in Q$ such that $\smash{\bar{h}_\alpha}$ satisfies the conclusion of \autoref{Th:Localization TCD}. If $a < t_0 < t_1 < b$ all belong to $\smash{\bar{\mms}_\alpha}$, we have
\begin{align*}
\Big[\frac{\sin_{\kappa_{t_0,b}^-}(b-t_1)}{\sin_{\kappa_{t_0,b}^-}(b-t_0)}\Big]^{N-1} \leq \frac{\bar{h}_\alpha(t_1)}{\bar{h}_\alpha(t_0)} \leq \Big[\frac{\sin_{\kappa_{a,t_1}^+}(t_1-a)}{\sin_{\kappa_{a,t_1}^+}(t_0-a)}\Big]^{N-1}.
\end{align*}
\end{corollary}

\section{Applications}\label{Ch:Applications}

Now we apply the localization machinery from the previous chapter to sharpen the geometric inequalities obtained in \autoref{Sec:Geometric inequ} in the following setting. 

\begin{assumption}[Third standing assumption]\label{Ass:TCD} $\scrM$ is a timelike $p$-essentially nonbranching $\smash{\TCD_p^e(k,N)}$ metric measure spacetime   according to \autoref{Def:Metric measure spacetime} with $\supp\meas=\mms$.
\end{assumption}

We  need \emph{not} assume the existence of a weight function $w$ according to \autoref{Sub:Framework} a priori. For an appropriate choice of $E$ and $u$, such a weight is always provided by default thanks to (the proof of) \autoref{Th:Mean zero}. In particular, the Polish background topology $\Top$ is not assumed to be proper.

\subsection{Localization by optimal transport} We apply the disintegration technique developed in \autoref{Ch:Localization} to the localization with respect to  a mean zero function $f$ (roughly speaking, by localizing with respect to a Kantorovich potential $u$ for the $\ell_1$-optimal transport problem from $f_+$ to $f_-$). The additional benefit is that the mean zero property of $f$ passes over to a.e.~transport ray; see also \cite{Akdemir24+}.
Inter alia, this will be used to deduce the sharp Brunn--Min\-kowski inequality in \autoref{Th:Sharp BMink} from its one-dimensional counterpart. 
 A related result for constant $k$ was recently obtained in \cite[Thm.~5.2]{cm++}.

\begin{theorem}[Localization by mean zero functions]\label{Th:Mean zero} Assume \autoref{Ass:TCD}. Let a nontrivial function $f\in\Ell^1(\mms,\meas)$ vanish outside a compact subset of $\mms$, suppose
\begin{align*}
\int_\mms f\d\meas =0,
\end{align*}
and assume the probability measures $\smash{\mu_0,\mu_1\in \Prob_\comp^\ac(\mms,\meas)}$ defined by
\begin{align*}
\mu_0 &:= \big\Vert f_+\big\Vert^{-1}_{\Ell^1(\mms,\meas)}\,f_+\,\meas,\\
\mu_1 &:= \big\Vert f_-\big\Vert^{-1}_{\Ell^1(\mms,\meas)}\,f_-\,\meas
\end{align*}
satisfy the chronology condition
\begin{align}\label{Eq:Chron cond}
\supp\mu_0\times\supp\mu_1 \subset\{l>0\}.
\end{align}
Then there exists a Suslin set $\Tr_u\subset\mms$ such that
\begin{enumerate}[label=\textnormal{\textcolor{black}{(}\roman*\textcolor{black}{)}}]
\item\label{Label:0} $\meas[\Tr_u]>0$,
\item\label{Label:1} $\meas\mres \Tr_u$ admits a $\CD(k,N)$ disintegration $\q$ according to \autoref{Sub:Local},
\item\label{Label:2}  $f$ vanishes identically $\meas$-a.e.~on $\smash{\Tr_u^\sfc :=\mms\setminus \Tr_u}$, and
\item\label{Label:3} $f$ has mean zero with respect to $\meas_\alpha$ for  $\q$-a.e.~$\alpha\in Q$, i.e.
\begin{align*}
\int_{\mms_\alpha} f\d\meas_\alpha =0.
\end{align*}
\end{enumerate}
\end{theorem}

\begin{proof} In order to localize we construct  $E\subset\mms$, a function $u\colon E \to \R$, and a weight function $\smash{w\colon \Tr_u^\End\to (0,\infty)}$ which match the framework of \autoref{Sub:Framework}. The hypothesized chronology condition implies the pair $(\mu_0,\mu_1)$ obeys  strong $l_+$-Kantorovich duality according to \autoref{Def:Strong Kantorovich} by \autoref{Ex:Strong Kantorovich}.  The  Kantorovich duality formula for $\smash{\ell_1}$ from \autoref{Th:Rubinstein}  implies the existence of 
\begin{itemize}
\item a compact, causally convex set $E\subset \mms$ containing $\supp\mu_0\cup\supp\mu_1$, and 
\item a Borel Kantorovich potential $u\colon E \to \R$ for the $\ell_1$-optimal transport from $\mu_0$ to $\mu_1$ satisfying the reverse $l$-Lipschitz condition \eqref{Eq:Gleichung}. 
\end{itemize}
Clearly, $E$ is  $l$-geodesically convex. It remains to prove the existence of a weight function according to \eqref{Eq:Weight}. We claim $\smash{\meas[\Tr_u^\End]>0}$ ---  since $\smash{\Tr_u^\End\subset E}$ is precompact,  by the Radon property of $\meas$ one can then simply take 
\begin{align*}
w := \meas[\Tr_u^\End]^{-1}\,\One_{\Tr_u^\End}.
\end{align*}
To this aim, fix  an $\ell_1$-optimal coupling $\pi\in\Pi(\mu_0,\mu_1)$  of $\mu_0$ and $\mu_1$, cf.~\autoref{Le:Existence}. By the chronology condition \eqref{Eq:Chron cond} and \autoref{Th:Rubinstein}, $\pi$ is concentrated on 
\begin{align*}
 E_{\prec_u}^2:= E_{\preceq_u}^2 \setminus \diag(\mms^2),
\end{align*}
where $\smash{E_{\preceq_u}^2}$ is from \eqref{Eq:GAMMA}. In particular $\smash{\pi[(\Tr_u^\End)^2]=1}$, and therefore 
\begin{align}\label{Eq:mu0mu1End}
\mu_0[\Tr_u^\End] = \mu_1[\Tr_u^\End]=1.
\end{align}
Since $\mu_0$ is $\meas$-absolutely continuous, this forces $\smash{\meas[\Tr_u^\End]>0}$. 

By \autoref{Th:All negligible}, we thus obtain $\meas[\Tr_u]>0$.  This proves \ref{Label:0}.

Thus, the localization procedure from \autoref{Ch:Localization} relative to $E$ and $u$ applies and provides us with an associated $\CD(k,N)$ disintegration $\q$ of $\meas\mres \Tr_u$, yielding  \ref{Label:1}.

Next, we show \ref{Label:2}. Assume to the contrary the existence of an $\meas$-measurable set $\smash{B \subset \Tr_u^\sfc}$ with $\meas[B]>0$ and $f\neq 0$ on $B$. We may and will assume $\mu_0[B]>0$; the proof in the case $\mu_1[B]>0$ is analogous. Moreover, by \autoref{Th:All negligible}, up to removing an $\meas$-negligible set we may and will assume $\smash{B\subset (\Tr_u^\End)^\sfc}$. Yet, by using \eqref{Eq:mu0mu1End} we arrive at the contradiction 
\begin{align*}
1=\mu_0[\mms] \geq \mu_0[\Tr_u^\End] + \mu_0[B]>1.
\end{align*}

Finally, we turn to \ref{Label:3}. We first claim
\begin{align*}
\Quot_\push\mu_0[A] = \Quot_\push\mu_1[A]
\end{align*}
for every $\q$-measurable set $A\subset Q$. Indeed, for a fixed coupling $\pi$ as in the first paragraph, by using \autoref{Th:All negligible} and 
\begin{align*}
\mu_0[\Tr_u]=\mu_1[\Tr_u]=1,
\end{align*}
implied by \eqref{Eq:mu0mu1End} and $\meas$-absolute continuity of $\mu_0$ and $\mu_1$, as above we obtain
\begin{align*}
\pi[E_{\prec_u}^2 \cap \Tr_u^2] = 1.
\end{align*}
Moreover, if $x,y\in\Tr_u$ satisfy $\smash{x \preceq_uy}$, by \autoref{Th:Equiv relation} they necessarily belong to the same transport ray, symbolically $\Quot(x) = \Quot(y)$. This easily yields
\begin{align*}
(\Quot^{-1}(A) \times E) \cap  E_{\prec_u}^2\cap \Tr_u^2 = (E \times \Quot^{-1}(A))\cap E_{\prec_u}^2\cap \Tr_u^2.
\end{align*}
These observations entail
\begin{align*}
\mu_0\big[\Quot^{-1}(A)\big] &= \pi\big[(\Quot^{-1}(A)\times E) \cap E_{\prec_u}^2 \cap\Tr_u^2\big]\\
&= \pi\big[(E\times\Quot^{-1}(A))\cap  E_{\prec_u}^2 \cap \Tr_u^2\big]\\
&=\mu_1\big[\Quot^{-1}(A)\big],
\end{align*}
which is the claimed identity. From this, we deduce
\begin{align*}
\int_A\int_{\mms_\alpha} f_+(x)\d\meas_\alpha(x)\d\q(\alpha) &= \big\Vert f_+\big\Vert_{\Ell^1(\mms,\meas)} \,\Quot_\push\mu_0[A]\\
&= \big\Vert f_-\big\Vert_{\Ell^1(\mms,\meas)} \,\Quot_\push\mu_1[A]\\
&= \int_A\int_{\mms_\alpha}f_-(x)\d\meas_\alpha(x)\d\q(\alpha).
\end{align*}
The arbitrariness of $A$ terminates the proof of \ref{Label:3}.
\end{proof}

\begin{remark}[Smooth localization by mean zero functions] In the smooth setting of \autoref{Th:Smooth needle}, the conclusions \ref{Label:0}, \ref{Label:2}, and \ref{Label:3} from \autoref{Th:Mean zero} remain valid irrespective of curvature assumptions.
\end{remark}

\subsection{Sharp geometric inequalities}\label{Sub:Sharp geom inequ} Using \autoref{Th:Mean zero}, we now sharpen various results obtained in \autoref{Sec:Geometric inequ}. We consider the following variable generalization of the geometrically averaged distortion coefficients from \eqref{Eq:Tau inter}, cf.~\cite[Def.~4.1]{ketterer2017}.

\begin{definition}[Averaged distortion coefficients along $l$-geodesics]\label{Def:tau dist coeff} Given any $\gamma\in\TGeo(\mms)$, $0\leq \theta\leq \vert\dot\gamma\vert$, and $0\leq t\leq 1$, we set
\begin{align*}
\tau_{k_\gamma^\pm,N}^{(t)}(\theta) := t^{1/N}\,\sigma_{k_\gamma^\pm/(N-1)}^{(t)}(\theta)^{1-1/N}.
\end{align*}
\end{definition}

\begin{remark}[The case $N=1$] The interpretation of the previous definition in the case $N=1$ is as follows. If $\smash{k_\gamma^\pm}$ is nonpositive on $[0,1]$, then $\smash{\tau_{k_\gamma^\pm,1}^{(t)}(\theta) = t}$ for every $0\leq \theta\leq \vert\dot\gamma\vert$ and every $0\leq t\leq 1$; otherwise $\smash{\tau_{k_\gamma^\pm,1}^{(t)}(\theta) < \infty}$ if and only if $\theta =0$.
\end{remark}

By  \cite[Cor.~4.2]{ketterer2017}, in the context of \autoref{Def:tau dist coeff} we have
\begin{align}\label{Eq:tau sigma inequ}
\tau_{k_\gamma^\pm,N}^{(t)}(\theta) \geq \sigma_{k_\gamma^\pm/N}^{(t)}(\theta).
\end{align}

For the next result, we recall the definition
\begin{align*}
A_t := \eval_t\big[G(A_0,A_1)\big]
\end{align*}
of $t$-intermediate points of $l$-geodesics starting in $A_0$ and ending in $A_1$ from \eqref{Eq:At'}.

\begin{theorem}[Sharp timelike Brunn--Minkowski inequality I]\label{Th:Sharp BMink} Let \autoref{Ass:TCD} hold. Let $A_0,A_1\subset\mms$ be Borel sets with positive and finite $\meas$-measure such that the pair $\smash{(\mu_0,\mu_1)\in\Prob^\ac(\mms,\meas)^2}$ is timelike $p$-dualizable, where
\begin{align*}
\mu_0 &:= \meas[A_0]^{-1}\,\meas\mres A_0,\\
\mu_1 &:= \meas[A_1]^{-1}\,\meas\mres A_1.
\end{align*}
Let $\bdpi$ denote the only element of $\smash{\OptTGeo_p(\mu_0,\mu_1)}$. Then for every $0\leq t\leq 1$,
\begin{align*}
 \meas[A_t]^{1/N} 
&\geq\meas\big[\!\supp(\eval_t)_\push\bdpi\big]^{1/N} 
\\ &\geq \inf\tau_{k_\gamma^-,N}^{(1-t)}(\vert\dot\gamma\vert)\,\meas[A_0]^{1/N} + \inf\tau_{k_\gamma^+,N}^{(t)}(\vert\dot\gamma\vert)\,\meas[A_1]^{1/N},
\end{align*}
where both infima are taken over all $\gamma\in \supp\bdpi$.

 The inequality holds a fortiori if 
both infima are taken over all $\gamma\in G(A_0,A_1)$.
\end{theorem}

\begin{proof} By inner regularity of $\meas$, we may and will assume $A_0$ and $A_1$ are compact. By replacing $\mms$ with the compact emerald $J(A_0,A_1)$ and by the Radon property of $\meas$, up to normalization we may and will assume 
  $\meas$  is a probability measure. Moreover, we will first consider the reinforced chronology condition $\smash{A_0\times A_1\subset\{l>0\}}$, cf.~\autoref{Ex:Str tl dual}. Then the function $f\in \Ell^1(\mms,\meas)$ defined by
\begin{align*}
f := \meas[A_0]^{-1}\,\One_{A_0} - \meas[A_1]^{-1}\,\One_{A_1}
\end{align*}
satisfies the hypotheses of \autoref{Th:Mean zero}. Let $\q$ denote a disintegration of $\meas\mres \Tr_u$ provided by this result; in particular,
\begin{align*}
\meas\mres \Tr_u = \int_Q\meas_\alpha\d\q(\alpha).
\end{align*}
We trivially disintegrate the restriction of $\meas$ to $\smash{\Tr_u^\sfc}$ according to
\begin{align*}
\meas\mres \Tr_u^\sfc = \int_{\Tr_u^\sfc} \delta_\alpha\d\meas(\alpha).
\end{align*}
Since $\smash{Q\cap\Tr_u^\sfc=\emptyset}$, the definitions $\smash{\q' := \q + \meas\mres \Tr_u^\sfc}$ and $\smash{\meas'_\cdot := \One_Q\,\meas_\cdot + \One_{\Tr_u^\sfc}\,\delta_\cdot}$ lead to the well-defined disintegration formula
\begin{align*}
\meas = \int_{Q \cup \Tr_u^\sfc} \meas'_\alpha\d\q'(\alpha).
\end{align*}

Given any $\smash{\alpha\in Q}$, we use the notations  $A_{0,\alpha} := A_0 \cap \mms_\alpha$ and $A_{1,\alpha} := A_1\cap \mms_\alpha$. Let $A_{t,\alpha}$ be the set of $t$-intermediate points for the one-dimensional 
 nondecreasing (i.e. $W_2$-optimal) 
transport from $\smash{\meas[A_{0,\alpha}]^{-1}\,\meas_\alpha\mres A_{0,\alpha}}$ to $\smash{\meas[A_{1,\alpha}]^{-1}\,\meas_\alpha\mres A_{1,\alpha}}$ (regarded as a subset of $\mms_\alpha$ by order isometry). Note that $\bdpi$ is necessarily concentrated on affinely parametrized subsegments of the considered transport rays by \autoref{Th:Uniqueness geos}. 
As $\q$ is a $\CD(k,N)$ disintegration and hence obeys the conclusion of \autoref{Th:Localization TCD}, by an evident variable adaptation of \cite[Lem.~3.2]{cmgeometric2017} $\smash{\q'}$-a.e.~$\alpha\in Q$ thus satisfies
\begin{align*}
\meas_\alpha[A_{t,\alpha}] \geq \Big[\!\inf\tau_{k_\gamma^-,N}^{(1-t)}(\vert\dot\gamma\vert)\,\meas_\alpha[A_{0,\alpha}]^{1/N} + \inf\tau_{k_\gamma^+,N}^{(t)}(\vert\dot\gamma\vert)\,\meas_\alpha[A_{1,\alpha}]^{1/N}\Big]^N,
\end{align*}
where both infima  are taken over all $\gamma\in \supp\bdpi$, which makes these coefficients notably independent of $\alpha$ ---  we agree on the latter interpretation of all appearing infima for the rest of the proof. Finally, observe that the mean zero property of $f$ with respect to $\meas_\alpha$ for $\q'$-a.e.~$\alpha\in Q$ translates into
\begin{align*}
\frac{\meas_\alpha[A_{0,\alpha}]}{\meas[A_0]} = \frac{\meas_\alpha[A_{1,\alpha}]}{\meas[A_1]}.
\end{align*}
From this, we deduce 
\begin{align}\label{Eq:Main claim}
\meas_\alpha[A_{t,\alpha}] \geq \frac{\meas_\alpha[A_{0,\alpha}]}{\meas[A_0]}\,\Big[\!\inf\tau_{k_\gamma^-,N}^{(1-t)}(\vert\dot\gamma\vert)\,\meas[A_0]^{1/N} + \inf\tau_{k_\gamma^+,N}^{(t)}(\vert\dot\gamma\vert)\,\meas[A_1]^{1/N}\Big]^N.
\end{align}

Next, we claim \eqref{Eq:Main claim} to hold for $\smash{\q'}$-a.e.~$\smash{\alpha\in \Tr_u^\sfc}$, properly interpreted with $\meas_\alpha$ replaced by $\delta_\alpha$, and $A_{0,\alpha}$, $A_{1,\alpha}$, and $A_{t,\alpha}$ replaced by $A_0$, $A_1$, and $A_t$, respectively. This  is trivially satisfied if $\smash{\q'[\Tr_u^\sfc] = \meas[\Tr_u^\sfc]=0}$, thus we suppose  $\smash{\meas[\Tr_u^\sfc]>0}$. Since $\meas[A_0]^{-1}\,\One_{A_0} + \meas[A_1]^{-1}\,\One_{A_1} = f = 0$ $\meas$-a.e.~on $\smash{\Tr_u^\sfc}$ by \autoref{Th:Mean zero}, we get
\begin{align*}
\meas[\Tr_u^\sfc \cap (A_0\,\triangle\, A_1)]=0.
\end{align*}
Since $A_0\cap A_1 = \emptyset$ by our chronology assumption and \autoref{Cor:Causality}, 
\begin{align*}
\q'\mres \Tr_u^\sfc = \meas\mres \Tr_u^\sfc = \meas\mres (\Tr_u^\sfc \cap A_0^\sfc \cap A_1^\sfc).
\end{align*}
This implies $\delta_\alpha[A_0] = 0$ for $\smash{\q'}$-a.e.~$\smash{\alpha\in \Tr_u^\sfc}$, which trivially verifies \eqref{Eq:Main claim}.

Taking these two observations together, we  obtain
\begin{align*}
\meas\big[\!\supp(\eval_t)_\push\bdpi\big] &= \int_{Q\cup\Tr_u^\sfc} \meas'_\alpha\big[\!\supp(\eval_t)_\push\bdpi \cap \mms_\alpha\big]\d\q'(\alpha)\\
&\geq \int_{Q\cup\Tr_u^\sfc} \meas'_\alpha[A_{t,\alpha}]\d\q'(\alpha)\\
&\geq \Big[\!\inf\tau_{k_\gamma^-,N}^{(1-t)}(\vert\dot\gamma\vert)\,\meas[A_0]^{1/N}\\
&\qquad\qquad  + \inf\tau_{k_\gamma^+,N}^{(t)}(\vert\dot\gamma\vert)\,\meas[A_1]^{1/N}\Big]^N\int_{Q\cup\Tr_u^\sfc} \frac{\meas_\alpha'[A_{0,\alpha}]}{\meas[A_0]}\d\q'(\alpha)\\
&= \Big[\!\inf\tau_{k_\gamma^-,N}^{(1-t)}(\vert\dot\gamma\vert)\,\meas[A_0]^{1/N} + \inf\tau_{k_\gamma^+,N}^{(t)}(\vert\dot\gamma\vert)\,\meas[A_1]^{1/N}\Big]^N.
\end{align*}
This establishes the claim under the hypothesis $A_0\times A_1\subset\{l>0\}$.

If $\mu_0$ and $\mu_1$ are merely timelike $p$-dualizable, we proceed as outlined in the last part of the proof of \autoref{Th:Pathwise} by reducing the claim to the previous case by appropriately decomposing $\{l>0\}$ relative to a timelike $p$-dualizing coupling of $\mu_0$ and $\mu_1$. The respective $t$-intermediate point sets have $\meas$-negligible intersection by \autoref{Le:Mutually singular}. The timelike Brunn--Minkowski inequality for $A_0$ and $A_1$ then follows by concavity of the function $r \mapsto r^{1/N}$ on $[0,\infty)$. We leave out the details. 
\end{proof}

Since the proofs of \autoref{Th:BonnetMyers} and  \autoref{Th:Schneider} are both conducted by contradicting the timelike Brunn--Minkowski inequality, from \autoref{Th:Sharp BMink} we obtain the following improved versions under \autoref{Ass:TCD}. In particular, we point out our  simplifying assumptions $N>1$ and $\supp\meas = \mms$.

\begin{corollary}[Sharp timelike Bonnet--Myers inequality]\label{Cor:111} We have
\begin{align*}
\diam_{k/(N-1)}\mms = \diam^l\mms.
\end{align*}
\end{corollary}

\begin{corollary}[Sharp timelike Schneider inequality]\label{Cor:Sharp Schneider} Assume that some $o\in \mms$, $\beta > 0$, and $R>0$ satisfy
\begin{align*}
k\geq (N-1)\,\Big[\frac{1}{4}+\beta^2\Big]\,l(\cdot,o)^{-2}\quad\textnormal{\textit{on }}\{l(\cdot,o) > R\}.
\end{align*}
Then $\smash{\diam^l I^-(o) \leq R\,\rme^{\pi/\beta}}$.

Analogously, if instead the asymptotic condition
\begin{align}\label{Eq:k lb}
k\geq (N-1)\,\Big[\frac{1}{4}+\beta^2\Big]\,l(o,\cdot)^{-2}\quad\textnormal{\textit{on }}\{l(o,\cdot) > R\}
\end{align}
holds, then $\smash{\diam^l I^+(o) \leq R\,\rme^{\pi/\beta}}$.
\end{corollary}

\begin{example}[Sharpness of \autoref{Th:Sharp BMink} and  \autoref{Cor:Sharp Schneider}]\label{Ex:WP} Here, we outline an example of a smooth metric measure spacetime satisfying \eqref{Eq:k lb} for $\beta = 0$ with infinite $l$-diameter. In particular, this shows that the factor $\beta^2$ in the hypotheses of  \autoref{Cor:Sharp Schneider} cannot be replaced by a nonpositive number.

Fix $\smash{N\in \N_{\geq 2}}$. Let $I := (0,\infty)$, and define $f\colon I \to (0,\infty)$ by $\smash{f(r) := \sqrt{r}}$. Let $h$ denote the standard Euclidean metric on $\smash{\R^{N-1}}$. Consider the warped product $\smash{\mms := I\times_f \R^{N-1}}$, i.e.~we endow $\smash{I\times \R^{N-1}}$ with the Lorentzian metric
\begin{align*}
g := \rmd r^2 - f(r)^2\,h = \rmd r^2 - r\,h.
\end{align*}
By \cite[Thm.~3.66]{beem1996}, $(\mms,g)$ is a smooth, globally hyperbolic Lorentzian spacetime. Moreover, in any  direction $v\in T_x\mms$ of the form $v = \partial_r + w$, where $w$ is a smooth vector field on $\smash{\R^{N-1}}$, the induced Ricci tensor is computed in a standard way \cite[p.~114]{beem1996}, yielding
\begin{align*}
\Ric(v,v) &= - (N-1)\,\frac{f''(r)}{f(r)}\,g(v,v) + (N-2)\,\big[\big\vert f'(r)\big\vert^2 - f(r)\,f''(r)\big]\,h(w,w)\\
&\geq \frac{N-1}{4}\,r^{-2}\,g(v,v).
\end{align*}
By \autoref{Th:Strong energ}, the induced metric measure spacetime $\scrM$ obeys $\smash{\TCD_p(k,N)}$ for a continuous function $k\colon \mms\to (0,\infty)$ such that for every $o\in\mms$, it behaves like $(N-1)\,l(o,\cdot)^{-2}/4$ for $l(o,\cdot)$ sufficiently large. Thus \eqref{Eq:k lb} with $\beta = 0$ is saturated for sufficiently large $R > 0$, yet $\smash{\diam^l I^+(o) = \infty}$; of course, in this example, by the precise asymptotics of the Ricci tensor, \eqref{Eq:k lb} cannot hold for any $\beta > 0$.

In particular, this example shows the sharpness of \autoref{Th:Sharp BMink}.
\end{example}

Moreover, using \autoref{Th:Sharp BMink} we sharpen  \autoref{Th:Bishop Gromov} under \autoref{Ass:TCD} as follows. We retain the inherent notation fixed in \autoref{Sub:Bishop Gromov}. The sharpness of the subsequent theorem is seen as in \autoref{Ex:WP}, see \cite[Rem.~5.11]{cavalletti2020} for constant $k$.

\begin{corollary}[Sharp timelike Bishop--Gromov inequality]\label{Cor:222} Let $E\subset \mms$ be compact and $l$-star shaped with respect to $x\in\mms$. Then for every $0<r<R<R_x$, we have
\begin{align*}
\frac{\sfv(r)}{\sfv(R)} &\geq \frac{\displaystyle\int_0^r \SIN_{\ubar{k}_x/(N-1)}^{N-1}(t)\d t}{\displaystyle\int_0^R \SIN_{\ubar{k}_x/(N-1)}^{N-1}(t)\d t},\\
\frac{\sfs(r)}{\sfs(R)} &\geq \frac{\SIN_{\ubar{k}_x/(N-1)}^{N-1}(r)}{\SIN_{\ubar{k}_x/(N-1)}^{N-1}(R)}.
\end{align*}
\end{corollary}

\begin{corollary}[Sharp dimension bound]\label{Cor:Hausdorffdim+} Retain the assumptions and notation from \autoref{Cor:Hausdorffdim}. The same argument for this corollary now yields the following improved bound on geometric dimension \eqref{nonsharp dimension bound}:
\begin{align*}
\dim^l\mms \leq N.
\end{align*}

This is shown to be sharp by Minkowski spacetime \cite[Thm.~4.8]{mccannsaemann}.
\end{corollary}

\subsection{Hawking-type singularity theorem and integral curvature bounds} In order to formulate our main result of  this section, \autoref{Th:Hawking}, we introduce some machinery, following \cite[Sec.~5.1]{cavalletti2020} and \cite[Ch.~2]{kontou}.

\subsubsection{Coarea formula and mean curvature bounds}\label{Sub:1} We retain \autoref{Ass:TCD}. Moreover, we assume we are situated in the setting of \autoref{Ex:CM}. This means $\mms$ is proper, $u$ is the future distance function $\tau_V$ from a given achronal FTC Borel set $V\subset\mms$ with $\meas[V]=0$, and $\smash{E := I^+(V)\cup V}$. Given any $t\geq 0$, we consider the map $f_t$ assigning to any $\alpha\in Q$ the unique point $f_t(\alpha) \in (\cl\,\mms_\alpha \cap \{u = t\})\setminus b$ provided the latter set is nonempty, symbolically $\alpha\in \Dom(f_t)$. By definition, $f_t$ is injective and $f_0(\alpha)\in V$ for every $\alpha\in\Dom(f_0)$. The Borel measurable map $\smash{f_0\circ\Quot\colon\Tr_u \to V}$ should be interpreted as a ``footpoint projection''. We define $\sfs\colon V\to [0,\infty]$ by
\begin{align*}
\sfs(x) := \sup\{t \geq 0 : \Quot(x)\in \Dom(f_t)\}.
\end{align*}
Lastly, for every $\meas$-measurable set $\smash{A\subset\Tr_u}$ with $\meas[A] < \infty$,  \autoref{Th:Needle} implies
\begin{align}\label{Eq:Coarea}
\meas[A] = \int_0^\infty \hh_t[A]\d t = \int_0^\infty \hh_t\big[A \cap \{u = t\}\big]\d t,
\end{align}
where for every $t\geq 0$,
\begin{align*}
\hh_t := (f_t)_\push\big[\bar{h}_\cdot(t)\,\q\big]
\end{align*}
can be interpreted as the surface measure of the level set $\{u=t\}$;  here $\smash{\bar h_\alpha}$, where $\alpha\in Q$, is from \eqref{Real bar}. Compare with \cite[Prop.~5.1]{cavalletti2020}.

The following notion of upper mean curvature bounds for $V$  has been proposed in \cite[Def.~5.2]{cavalletti2020} following \cite{kheintze}. It is consistent with the smooth case \cite[Rem.~5.4]{cavalletti2020},  and is motivated by the classical realization of mean curvature as the first variation of surface area.  Analogously, lower mean curvature bounds for $V$ can be synthesized. 

\begin{definition}[Upper mean curvature bounds]\label{Def:Mean curv} We say $V$ has forward mean curvature bounded from above by $-\beta$, where $\beta\in\R$, if
\begin{enumerate}[label=\textnormal{\alph*.}]
\item $\hh_0$ is a Radon measure satisfying $(f_0)_\push\q \ll \hh_0$, and
\item for every compactly supported, bounded Borel function $\phi\colon V\to [0,\infty)$ and every induced normal variation 
\begin{align}\label{Eq:Normal var}
V_{t,\phi} := \{u \leq t\,\phi\circ f_0 \circ \Quot\}\cap\Tr_u,
\end{align}
where $t>0$, we have
\begin{align*}
-\beta\int_V\phi^2\d\hh_0 \geq \liminf_{t\to 0} \frac{2}{t^2}\,\Big[\meas\big[V_{t,\phi}\big] - t\int_V \phi\d\hh_0\Big].
\end{align*}
\end{enumerate}
\end{definition}

\subsubsection{Segment inequality}\label{Sub:2} Next, given any Borel set $A\subset V$, following \cite[Def.~2.9]{kontou} the \emph{future development} of $A$ up to time $T>0$ is defined by
\begin{align*}
A^{0,T} := \{f_t\circ \Quot(x) \in\Tr_u : x\in A,\, 0\leq t\leq \min\{T,\sfs(x)\}\}.
\end{align*}
In other words, $\smash{A^{0,T}}$ is the set of all points lying on a transport ray which starts in $A$ and has $l$-distance at most $T$ to $V$, including earlier termination of the respective needle. The set $\smash{\Reg_L} 
:= \{x \in V : \sfs(x) \ge L\}$ of points of reach $L$ denotes the set of all $x\in V$ such that the unique transport ray $\smash{\mms_\alpha}$, where $\alpha\in Q$,  whose closure contains $x$ has length at least $L>0$. 

The next result is a version of the Riemannian segment inequality by Cheeger--Colding \cite{cc} (see also \cite{sprouse}),  adapted to the synthetic setting from \cite[Prop.~3.9]{kontou}. For a Borel function $\psi\colon \mms\to [0,\infty)$ and $T>0$, define $F_{\psi,T}\colon V \to [0,\infty]$ by
\begin{align*}
F_{\psi,T}(x) := \int_0^{\min\{T,\sfs(x)\}} \psi\circ f_t \circ \Quot(x)\d t.
\end{align*}
Additionally, given any $\alpha\in Q$ and any $x_0 < x_1$ both belonging to $\smash{\bar{\mms}_\alpha}$, we consider the nonnegative ($\alpha$-dependent) function $\smash{\varkappa_{x_0,x_1}^-}\colon [0,\vert x_1-x_0\vert]\to [0,\infty)$  according to \eqref{Eq:BRRRP} yet defined in terms of $k_-$ instead of $k$, i.e.
\begin{align}\label{Eq:varkappa def}
\varkappa_{x_0,x_1}^-(s\,\vert x_1- x_0\vert) :=  (\bar{k}_-)_\alpha((1-s)\,x_1 + s\,x_0).
\end{align}

\begin{proposition}[Segment inequality]\label{Pr:Segment} Retain the previous notation. Given any $\delta > 0$, let $A\subset \Reg_{T+\delta}$ be a Borel set with $0 < \hh_0[A] < \infty$. Then
\begin{align*}
\hh_0\textnormal{-}\!\essinf_{x\in A} F_{\psi,T}(x) \leq c'^{-1}\,\hh_0[A]^{-1}\int_{A^{0,T}} \psi\d\meas,
\end{align*}
where
\begin{align*}
c' := \q\textnormal{-}\!\!\essinf_{\alpha\in \sfQ(A)} \Big[\frac{\SIN_{-\varkappa_{0,T+\delta}^-/(N-1)}(\delta)}{\SIN_{-\varkappa_{0,T+\delta}^-/(N-1)}(T+\delta)}\Big]^{N-1}.
\end{align*}

In particular, if $k\geq K$ on $\smash{I^+(A)\cup A}$ for some constant $K < 0$,
\begin{align*}
\hh_0\textnormal{-}\!\essinf_{x\in A} F_{\psi,T}(x) \leq \Big[\frac{\sinh(\sqrt{-K/(N-1)}\,(T+\delta))}{\sinh(\sqrt{-K/(N-1)}\,\delta)}\Big]^{N-1}\,\hh_0[A]^{-1}\int_{A^{0,T}} \psi\d\meas.
\end{align*}
\end{proposition}

\begin{proof} 

The second statement follows from \autoref{Ex:Constant pot}  and the monotonicity property stated in \autoref{Le:Properties}. For the first statement, we may and will assume $c'>0$. For simplicity, we also suppose $k$ is continuous. The lower semicontinuous case can be handled by approximation along $\q$-a.e.~needle, recalling \eqref{Eq:kappan}.

For $\q$-a.e.~$\alpha\in \Quot(A)$, combining \autoref{Cor:Comparison} with \autoref{Le:Properties} and the trivial inequality $k \geq -k_-$ on $\mms$ implies the following inequality for every $t_0 < t_1 < b$ belonging to $\bar{\mms}_\alpha$:
\begin{align*}
\bar{h}_\alpha(t_1) \geq \Big[\frac{\SIN_{- \varkappa_{t_0,b}^-/(N-1)}(b-t_1)}{\SIN_{- \varkappa_{t_0,b}^-/(N-1)}(b-t_0)}\Big]^{N-1}\,\bar{h}_\alpha(t_0)
\end{align*}
Replacing $b$ by $T+\delta$, $t_1$ by any number $0 < t < T$, and letting $t_0 \to 0$ yields
\begin{align*}
\bar{h}_\alpha(t) \geq \Big[\frac{\SIN_{-\varkappa_{0,T+\delta}^- /(N-1)}(T+\delta-t)}{\SIN_{- \varkappa_{0,T+\delta}^-/(N-1)}(T+\delta)}\Big]^{N-1}\,\bar{h}_\alpha(0)\geq c'\,\bar{h}_\alpha(0).
\end{align*}
In the last step, we have used nondecreasingness of $\smash{\SIN_{-\varkappa_{0,T+\delta}^-/(N-1)}}$.

The coarea formula \eqref{Eq:Coarea} implies
\begin{align*}
\int_{A^{0,T}} \psi\d\meas &= \int_0^T\!\!\int_{f_t\circ\Quot(A)} \psi \d\hh_t\d t\\
&=\int_0^T\!\!\int_{\Quot(A)} \psi\circ f_t(\alpha) \,\bar{h}_\alpha(t)\d\q(\alpha)\d t\\
&\geq c'\int_0^T\!\!\int_{\Quot(A)} \psi\circ f_t(\alpha)\,\bar{h}_\alpha(0)\d\q(\alpha)\d t\\
&=  c'\int_0^T\!\!\int_{A} \psi\circ f_t\circ \Quot\d\hh_0\d t\\
&\geq c' \,\hh_0[A]\,\hh_0\textnormal{-}\!\essinf_{x\in A} F_{\psi,T}(x);
\end{align*}
here we have used the identities $\alpha = \Quot\circ f_0(\alpha)$ for every $\alpha\in Q$ and $f_0\circ \Quot(A) = A$ in the second last identity, and the last equality follows by Fubini's theorem.
\end{proof}

\subsubsection{Main result} We are ready to state our main result of this section. It is a version of Hawking's singularity theorem \cite{hellis}, assuming an integral curvature bound instead of a ``pointwise'' inequality. It is proven along similar lines as the genuine synthetic Hawking singularity theorem  \cite[Thm.~5.6]{cavalletti2020}. For smooth Lorentzian spacetimes, it has been shown in \cite[Thm.~4.1]{kontou} assuming a uniform lower bound on the timelike Ricci curvature; we point out that our result is new even in the weighted smooth case. A past version holds with evident changes.

For $T>0$, $\eta \geq 0$, and $\alpha\in Q$ we recall the definition \eqref{Eq:varkappa def} of the nonpositive ($\alpha$-dependent) potential $\smash{\varkappa_{0,T+\eta}^-\colon[0,T+\eta]\to [0,\infty)}$ defined by
\begin{align*}
\varkappa_{0,T+\eta}^-(s\,(T+\eta)) := (\bar{k}_-)_\alpha((1-s)\,(T+\eta)).
\end{align*}

\begin{theorem}[Synthetic Hawking-type singularity theorem I]\label{Th:Hawking} Let \autoref{Ass:TCD} hold, and assume $\mms$ is proper. Let $V\subset\mms$ be an achronal FTC Borel set such that $\meas[V]=0$ and $\hh_0[V]<\infty$. We assume the  forward mean curvature of $V$ is bounded from above by $-\beta$, where $\beta > 0$. Let $T>0$ and $\delta > 0$  obey
\begin{align}\label{Eq:HYL}
\hh_0[V]^{-1}\int_{V^{0,T}} k_-\d\meas < c\,\Big[\beta - \frac{N-1}{T}\Big],
\end{align}
where
\begin{align*}
c := \q\textnormal{-}\!\!\essinf_{\alpha\in \sfQ(V)} \Big[\frac{\SIN_{-\varkappa_{0,T+\delta}^-/(N-1)}(\delta)}{\SIN_{-\varkappa_{0,T+\delta}^-/(N-1)}(T+\delta)}\Big]^{N-1}.
\end{align*}
Then we have
\begin{align*}
\hh_0\big[V\setminus \Reg_{T+\delta}\big] > 0.
\end{align*}

In other words, there exists a transport ray starting in $V$ whose maximal domain of definition in $\smash{I^+(V)\cup V}$ is strictly contained in $[0,T+\delta)$.
\end{theorem}

It is implicit in the \autoref{Def:Mean curv} of upper mean curvature boundedness that $\hh_0[V]>0$, so that the statement makes sense. Moreover, if $\hh_0[V]=\infty$, an evident variant of this theorem still holds when replacing $V$ by an appropriate subset $V'$ thereof such that $0 < \hh_0[V'] < \infty$.

\begin{proof}[Proof of \autoref{Th:Hawking}] As for  \autoref{Pr:Segment}, we assume continuity of $k$.

Our hypotheses imply $\smash{c>0}$ and $\smash{\beta > T^{-1}}$. We abbreviate
\begin{align}\label{Eq:LL DEFF}
L := c^{-1}\,\hh_0[V]^{-1}\int_{V^{0,T}} k_-\d\meas.
\end{align}
We assume to the contrary that 
\begin{align*}
\hh_0\big[V\setminus \Reg_{T+\delta}\big] = 0.
\end{align*}
Therefore \autoref{Pr:Segment} is applicable to $A := \smash{\Reg_{T+\delta} \cap V}$, up to removing an $\hh_0$-negligible set; moreover, $\smash{\hh_0[A] = \hh_0[V]}$. Combining this with \eqref{Eq:HYL},
\begin{align*}
\hh_0\textnormal{-}\!\essinf_{x\in A} F_{k_-,T}(x)  \leq c^{-1}\,\hh_0[A]^{-1}\int_{A^{0,T}} k_-\d\meas\leq L.
\end{align*}
Let $\varepsilon >0$ be a small number to be chosen later. Then there is a Borel set $A'\subset A$ with $\hh_0[A'] > 0$ such that 
\begin{align}\label{Eq:Lepsilon}
\hh_0\textnormal{-}\!\esssup_{x\in A'}F_{k_-,T}(x) \leq L+\varepsilon.
\end{align}

Let $\phi\colon V\to [0,\infty)$ be a nontrivial bounded Borel function with compact support in $A'$; by scaling, we may and will henceforth assume $\sup \phi(V) < T$. By definition of $A'$, $\q$-a.e.~transport ray $\mms_\alpha$ has length at least $T+\delta$. Recalling the definition \eqref{Eq:Normal var} of the  normal variation $V_{t,\phi}$ induced by $\phi$, for every $0<t<1$ we obtain
\begin{align}\label{Eq:Heree}
\begin{split}
\meas\big[V_{t,\phi}\big] - t\int_V\phi\d\hh_0 &= \int_Q\int_0^t \bar{h}_\alpha(s\,\phi\circ f_0(\alpha))\,\phi\circ f_0(\alpha)\d s\d\q(\alpha)\\
&\qquad\qquad - t\int_Q \phi\circ f_0(\alpha)\,\bar{h}_\alpha(0)\d\q(\alpha)\\
&= \int_Q\int_0^t\big[\bar{h}_\alpha(s\,\phi\circ f_0(\alpha)) - \bar{h}_\alpha(0)\big]\d s\,\phi\circ f_0(\alpha)\d\q(\alpha).
\end{split}
\end{align}
Let the nonnegative function $\smash{\varkappa_{0,T}^-\colon [0,T]\to [0,\infty)}$ be defined as in \eqref{Eq:varkappa def}. As for \autoref{Pr:Segment}, we prove that for $\q$-a.e.~$\alpha\in Q$, every $0<s<t$ certifies
\begin{align*}
\bar{h}_\alpha(s\,\phi\circ f_0(\alpha)) \geq \Big[\frac{\SIN_{-\varkappa_{0,T}^-/(N-1)}(T-s\,\phi\circ f_0(\alpha))}{\SIN_{-\varkappa_{0,T}^-/(N-1)}(T)}\Big]^{N-1}\,\bar{h}_\alpha(0).
\end{align*}
Consider the twice continuously differentiable function $v\colon [0,T]\to \R$ with
\begin{align*}
v(\theta) := \SIN_{-\varkappa_{0,T}^-/(N-1)}(\theta).
\end{align*}
Note that $v$ still depends on $\alpha\in Q$; this  will not affect our estimates, since they are performed along $\q$-a.e.~transport ray. By boundedness of $-k_-$ along $\q$-a.e.~needle starting at $A'$, a Taylor expansion thus gives
\begin{align*}
&\meas\big[V_{t,\phi}\big] - t\int_V\phi\d\hh_0\\
&\qquad\qquad \geq \int_Q\int_0^t\Big[\Big[\frac{v(T-s\,\phi\circ f_0(\alpha))}{v(T)}\Big]^{N-1}-1\Big]\d s\,\bar{h}_\alpha(0)\,\phi\circ f_0(\alpha)\d \q(\alpha)\\
&\qquad\qquad = \int_Q\int_0^t \Big[\!- (N-1)\,\frac{v'(T)}{v(T)}\,s\,\phi\circ f_0(\alpha) + \rmo(s)\Big]\d s\,\bar{h}_\alpha(0)\,\phi\circ f_0(\alpha)\d\q(\alpha).
\end{align*}
On the other hand, along $\q$-a.e.~ray starting in $A'$ we have
\begin{align*}
(N-1)\,v'(T) &=  N-1 + \int_0^T \varkappa_{0,T}^-(r)\,v(r)\d r\\
&\leq N-1 +  \Big[\!\int_0^T (\bar{k}_-)_\alpha(T-r)\d r\Big]\,v(T)\\
&\leq  N-1 +  (L+\varepsilon)\,v(T);
\end{align*}
here we have successively used the ODE solved by $v$, the nonnegativity of $\smash{\varkappa_{0,T}^-}$, the nondecreasingness of $v$,  and  \eqref{Eq:Lepsilon}. Moreover, Sturm's comparison theorem implies $v(T) \geq T$, cf.~e.g.~\cite[Thm.~3.1]{ketterer2017}.  Combining these observations yields
\begin{align*}
\meas\big[V_{t,\phi}\big] - t\int_V\phi\d\hh_0 &\geq -\int_Q\int_0^t \Big[\frac{N-1}{v(T)} + L+\varepsilon \Big]\,s  
\d s\,\bar{h}_\alpha(0)\,\phi^2\circ f_0(\alpha)\d\q(\alpha)\\
&\qquad\qquad + \int_Q\int_0^t \rmo(s) \d s\,\bar{h}_\alpha(0)\,\phi\circ f_0(\alpha)\d\q(\alpha)\\
 &\geq -\Big[\frac{N-1}{T} + L+\varepsilon \Big]\,\frac{t^2}{2}\int_Q \bar{h}_\alpha(0)\,\phi^2\circ f_0(\alpha)\d\q(\alpha)\\
 &\qquad\qquad + \int_Q \rmo(t^2)\,\bar{h}_\alpha(0)\,\phi\circ f_0(\alpha)\d\q(\alpha). 
\end{align*}
The upper boundedness of the forward mean curvature of $V$ according to \autoref{Def:Mean curv} by $-\beta$ and Fatou's lemma  yield
\begin{align*}
-\beta\int_V \phi^2\d\hh_0 &\geq \liminf_{t\to 0} \frac{2}{t^2}\,\Big[\meas\big[V_{t,\phi}\big] - t\int_V\phi\d\hh_0\Big] \geq -\Big[\frac{N-1}{T} + L+\varepsilon \Big]\int_V\phi^2\d\hh_0. 
\end{align*}
Canceling the integrals on both sides of this inequality and recalling \eqref{Eq:LL DEFF} yields
\begin{align*}
c^{-1}\,\hh_0[V]^{-1}\int_{V^{0,T}} k_-\d\meas \geq \beta - \frac{N-1}{T} -\varepsilon.
\end{align*}
Choosing $\varepsilon$ small enough contradicts our hypothesis \eqref{Eq:HYL}.
\end{proof}

\begin{remark}[The case $N=1$] If $N=1$, achronal FTC Borel sets with vanishing $\meas$-measure cannot have a negative forward upper mean curvature bound. Indeed, assume such a set $V\subset\mms$ exists. Since $N=1$, by \autoref{Re:Limit constant} $\smash{\bar{h}_\alpha}$ is constant for $\q$-a.e.~$\alpha\in Q$. On the other hand, letting $t\to 0$ in \eqref{Eq:Heree} --- with an appropriate function $\phi$ ---  and using the upper boundedness of the forward mean curvature of $V$ by a negative number easily leads to a contradiction.
\end{remark}

\addtocontents{toc}{\protect\setcounter{tocdepth}{1}}

\appendix

\section{Variable timelike measure contraction property}\label{Ch:TMCP}

\subsection{Definition and basic properties} Now we discuss a weaker variant of the entropic timelike curvature-dimension condition from \autoref{Def:TCDe}.

\begin{definition}[Variable timelike measure contraction property]\label{Def:TMCP} We term $\scrM$ to satisfy the \emph{future entropic timelike measure contraction property} $\smash{\TMCP^{e,+}(k,N)}$ if for every $\smash{\mu_0 \in \scrP_\comp^\ac(\mms,\meas)}$ and every $\smash{x_1\in \bigcap_{x \in \supp \mu_0} I^+(x)\cap\supp\meas}$, there exist
\begin{itemize}
\item an $\smash{\ell_{1/2}}$-geodesic $\smash{(\mu_t)_{t\in[0,1]}}$\footnote{By the assumptions on $\mu_0$ and $\mu_1$ as well as \autoref{Cor:Equiv notions lp geo}, we could equivalently start with rough or displacement $\smash{\ell_{1/2}}$-geodesics here.} connecting $\mu_0$ to $\smash{\mu_1 := \delta_{x_1}}$, and
\item a plan $\smash{\bdpi\in\OptTGeo_{1/2}(\mu_0,\mu_1)}$
\end{itemize}
such that for every $0\leq t\leq 1$,
\begin{align*}
\scrU_N(\mu_t) \geq \sigma_{k_\bdpi^-/N}^{(1-t)}(\cost_\bdpi)\,\scrU_N(\mu_0).
\end{align*}

Moreover, we say that $\scrM$ satisfies the \emph{past entropic timelike measure contraction property} $\smash{\TMCP^{e,-}(k,N)}$ if its causal reversal $\smash{\scrM^\leftarrow}$ obeys $\smash{\TMCP^{e,+}(k,N)}$.

Lastly, if $\scrM$ certifies $\smash{\TMCP^{e,+}(k,N)}$ and $\smash{\TMCP^{e,-}(k,N)}$ simultaneously, it is termed to have the \emph{entropic timelike measure contraction property} $\smash{\TMCP^e(k,N)}$.
\end{definition}

In other words, the past version $\smash{\TMCP^{e,-}(k,N)}$ asks for convexity properties of $\smash{\scrU_N}$ along the transport from a point mass to an $\meas$-absolutely continuous, compactly supported mass distribution in its ``chronological'' \emph{future}. Compared to  \cite{cavalletti2020}, we have decided to include both time directions in our timelike measure contraction property to make it invariant under causal reversal (recall  \autoref{Re:Causal reversal}).

\begin{remark}[Independence of the transport exponent]\label{Re:Indep exp} This remark should be compared to   {\cite[Rem.~2.4]{cavalletti2022}}. The two objects $(\mu_t)_{t\in[0,1]}$ and $\bdpi$ from  \autoref{Def:TMCP} do not depend on the chosen exponent $1/2$. Indeed, for  $\mu_0$ to $\mu_1$ as above, for every $0<p<1$ we have
\begin{align}\label{Eq:OptOpt}
\OptTGeo_p(\mu_0,\mu_1) = \OptTGeo_{1/2}(\mu_0,\mu_1)
\end{align}
since $\smash{\Pi_{\ll}(\mu_0,\mu_1) = \Pi(\mu_0,\mu_1)}$ is a singleton. Furthermore, $(\mu_t)_{t\in[0,1]}$ forms in fact  a displacement $\smash{\ell_p}$-geodesic.  Indeed, by \autoref{Cor:Equiv notions lp geo} it is represented by a plan $\bdpi\in\smash{\OptTGeo_{1/2}(\mu_0,\mu_1)}$. Given any $0\leq s<t\leq 1$, from \eqref{Eq:OptOpt} we thus get
\begin{align*}
\ell_p(\mu_s,\mu_t) &\geq \big\Vert l \circ (\eval_s,\eval_t)\big\Vert_{\Ell^p(\TGeo(\mms),\bdpi)}\\
&= (t-s)\,\big\Vert l\circ(\eval_0,\eval_1)\big\Vert_{\Ell^p(\TGeo(\mms),\bdpi)}\\
&= (t-s)\,\ell_p(\mu_0,\mu_1).
\end{align*}
The claim follows as in \autoref{Re:From geo to displ}.
\end{remark}

Several basic statements valid under  $\smash{\wTCD_p^e(k,N)}$ have evident counterparts for $\smash{\TMCP^e(k,N)}$. In particular, this pertains to
\begin{itemize}
\item consistency in $k$ and $N$  (\autoref{Pr:Consistency}), 
\item behavior under measure perturbations (\autoref{Pr:Potential}),  
\item  maps and uniqueness for chronological $\smash{\ell_p}$-optimal couplings (\autoref{Th:Optimal maps}), and  
\item maps, uniqueness, and regularity of displacement $\smash{\ell_p}$-geode\-sics (\autoref{Th:Uniqueness geos}) together with its \autoref{Cor:Lifting nonbr}. 
\end{itemize}
Thanks to \autoref{Re:Indep exp}, these results automatically hold for every $0<p<1$.  We omit the details, but point out that in some of their proofs, the following implication has been implicitly used.  No circular reasoning occurs, though.

\begin{proposition}[Consistency with \autoref{Def:TCDe}]\label{Pr:TCD to TMCP} The condition $\smash{\wTCD_p^e(k,N)}$ implies $\smash{\TMCP^e(k,N)}$.
\end{proposition}

\begin{proof} Let $\smash{\mu_0}$ and $\smash{\mu_1}$ be as in \autoref{Def:TMCP}. Since
\begin{align*}
\supp\mu_0\times\{x_1\}\subset\{l>0\},
\end{align*}
for sufficiently large $n\in\N$ \autoref{Cor:Strong causality} implies
\begin{align}\label{Eq:Suppvarepsilon}
\supp\mu_0 \times \bar{\sfB}^\met(x_1,2^{-n})\subset\{l>0\}.
\end{align}
As $x_1\in\supp\meas$, the uniform distribution $\smash{\mu_1^n\in\scrP^\ac(\mms,\meas)}$ of $\smash{\sfB^\met(x_1,2^{-n})}$ with respect to $\meas$ is well-defined. Moreover, by  \autoref{Cor:Hausdorff} and shifting $n$ if necessary we may and will assume it to have compact support. In particular, as everything will take place only on the  compact causal diamond between $\supp\mu_0$ and  $\smash{\bar{\sfB}^\met(x_1,1/2)}$, by \eqref{Eq:kappan} we may and will assume $k$ to be continuous on this set. The general case will easily follow by Levi's monotone convergence theorem and \autoref{Le:Properties}.

For every $n\in\N$, by \eqref{Eq:Suppvarepsilon} and \autoref{Ex:Str tl dual}, $\smash{\wTCD_p^e(k,N)}$ endows us with
\begin{itemize}
\item an $\smash{\ell_p}$-geodesic $\smash{(\mu_t^n)_{t\in[0,1]}}$ from $\smash{\mu_0^n := \mu_0}$ to $\smash{\mu_t^n}$, and
\item a plan $\smash{\bdpi^n\in\OptTGeo_p(\mu_0,\mu_1^n)}$
\end{itemize}
such that for every $0\leq t\leq 1$,
\begin{align}\label{Eq:INEQU UN}
\begin{split}
\scrU_N(\mu_t^n) &\geq \sigma_{k^-_{\bdpi^n}/N}^{(1-t)}(\cost_{\bdpi^n})\,\scrU_N(\mu_0) + \sigma_{k^+_{\bdpi^n}/N}^{(t)}(\cost_{\bdpi^n})\,\scrU_N(\mu_1^n)\\
&\geq \sigma_{k^-_{\bdpi^n}/N}^{(1-t)}(\cost_{\bdpi^n})\,\scrU_N(\mu_0).
\end{split}
\end{align}
Now we pass to the limit as $n\to\infty$. Clearly, $\smash{(\mu_1^n)_{n\in\N}}$ converges narrowly to $\mu_1$ and is uniformly compactly supported according to \autoref{Def:Unif cpt supp}. Combining \eqref{Eq:Suppvarepsilon}, \autoref{Cor:Cpt TGeo}, and Prokhorov's theorem as in the proof of \autoref{Le:Geodesics plan}, the sequence $\smash{(\bdpi^n)_{n\in\N}}$ has a narrow  limit $\smash{\bdpi\in\scrP(\Cont([0,1];\mms))}$ concentrated on  the set $G_r\subset\TGeo(\mms)$ from \autoref{Cor:Cpt TGeo} along a nonrelabeled subsequence,  where $r := \inf l(\supp\mu_0\times \{x_1\}) > 0$. Since $\smash{(\eval_0,\eval_1)_\push\bdpi}$ constitutes  the unique (chronological) coupling of $\mu_0$ and $\mu_1$, we obtain $\smash{\bdpi\in\OptTGeo_p(\mu_0,\mu_1)}$. This yields
\begin{align}\label{Eq:LSC dist coeff in proof}
\liminf_{n\to \infty} \sigma_{k_{\bdpi^n}^-/N}^{(1-t)}(\cost_{\bdpi^n}) \geq \sigma_{k_{\bdpi}^-/N}^{(1-t)}(\cost_\bdpi) 
\end{align}
for every $0\leq t\leq 1$ by \autoref{Le:LSC sigma}, as well as
\begin{align*}
\ell_p(\mu_0,\mu_1) &= \big\Vert l\circ(\eval_0,\eval_1)\big\Vert_{\Ell^p(\TGeo(\mms),\bdpi)}\\
&\leq \liminf_{n\to\infty} \big\Vert l\circ (\eval_0,\eval_1)\big\Vert_{\Ell^p(\TGeo(\mms),\bdpi^n)}\\
&= \liminf_{n\to\infty} \ell_p(\mu_0,\mu_1^n).
\end{align*}
The inequality follows from continuity of $\smash{l_+}$ and \cite[Lem.~4.3]{villani2009}. In turn, the latter inequality combined with \autoref{Pr:Compactness lp geos} entails the existence of a rough $\smash{\ell_p}$-geodesic $(\mu_t)_{t\in[0,1]}$ from $\mu_0$ to $\mu_1$ such that, up to a further nonrelabeled subsequence, $\smash{(\mu_t^n)_{n\in\N}}$ converges narrowly to $\mu_t$ every rational $0\leq t\leq 1$. For such $t$, invoking \autoref{Le:Zt lemma}, upper semicontinuity of $\scrU_N$, \eqref{Eq:INEQU UN}, and \eqref{Eq:LSC dist coeff in proof} yields
\begin{align*}
\scrU_N(\mu_t) \geq \limsup_{n\to\infty} \scrU_N(\mu_t^n) \geq \sigma_{k_\bdpi^-/N}^{(1-t)}(\cost_\bdpi)\,\scrU_N(\mu_0).
\end{align*}
But by \autoref{Cor:Equiv notions lp geo}, $\smash{(\mu_t)_{t\in [0,1]}}$ constitutes in fact a displacement $\smash{\ell_p}$-geodesic and in particular narrowly continuous. Employing the upper semicontinuity properties of $\scrU_N$ and time continuity of the involved distortion coefficients, ensured by continuity of $k$, the previous estimate extends to \emph{every} $0\leq t\leq 1$ too. Recalling \autoref{Re:Indep exp}, the proof is terminated.
\end{proof}

\subsection{Geometric consequences} With appropriate modifications, all geometric inequalities from \autoref{Sec:Geometric inequ} hold under our variable timelike measure contraction property. Instead of stating them in their nonsharp form, in the sequel we use localization to directly show their sharp versions. This is analogous to Chapters \ref{Ch:Localization} and \ref{Ch:Applications}, whence we only sketch the arguments and leave their details to the reader. Similarly to \autoref{Ass:TCD}, in this section we work in the following setting.

\begin{assumption}[Modification of \autoref{Ass:TCD}]\label{Ass:Mod} $\scrM$ is  a timelike $p$-essentially nonbranching $\TMCP^e(k,N)$ metric measure spacetime according to \autoref{Def:Metric measure spacetime} with $\supp\meas = \mms$.
\end{assumption}

The key is the following analog of \autoref{Th:Localization TCD}. Recall the definitions \eqref{Real bar} of the barred quantities appearing therein.

\begin{theorem}[Disintegration into $\MCP(k,N)$ needles] 
Assume \autoref{Ass:Mod}. Let $\q$ be a disintegration according to \autoref{Th:Needle}. Then the conditional measure  $\smash{\bar{\meas}_\alpha}$ is absolutely continuous with respect to the one-dimensional Lebesgue measure 
for $\q$-a.e.~$\alpha\in Q$, and its $\smash{\bar{\Leb}^1_\alpha}$-density $\smash{\bar{h}_\alpha}$ 
has a nonrelabeled $\smash{\bar{\Leb}^1_\alpha}$-version which is locally Lipschitz continuous and strictly positive on the open real interval $\smash{\bar{\mms}_\alpha}$, continuous on its closure, and satisfies the following inequality for every $t_0,t_1\in \bar{\mms}_\alpha$ with $t_0<t_1$ and every $0\leq s\leq 1$:
\begin{align}\label{EStated}
\bar{h}_\alpha((1-s)\,t_0 + s\,t_1)^{1/(N-1)} &\geq \sigma_{k_\alpha^-/(N-1)}^{(1-s)}(t_1-t_0)\,\bar{h}_\alpha(t_0)^{1/(N-1)}.
\end{align}
\end{theorem}

\begin{proof} As \autoref{Th:Needle} assumes the stronger  $\smash{\TCD_p^e(k,N)}$ condition, we first have to prove the existence of  a disintegration $\q$ satisfying its conclusions. However, the curvature assumption only enters its proof only through \autoref{Th:All negligible} (negligibility of bad points). In turn, the latter result only assumes the setting of \autoref{Le:Size}, which is in place under \autoref{Ass:Mod} as noted in the previous section.

The proof of the claimed absolute continuity of the conditional measures is then analogous to \cite[Prop.~4.16]{cavalletti2020}. The inequality \eqref{EStated} is shown along the same lines as for \autoref{Th:Localization TCD}; see  \cite[Thm.~4.18]{cavalletti2020} for the argument for  constant $k$.
\end{proof}

Following the lines of \autoref{Th:Sharp BMink}, the subsequent result is then shown. Recall \eqref{Eq:At'} for the definition of the set 
\begin{align*}
    A_t := \eval_t\big[G(A_0,A_1)\big]
\end{align*}
of $t$-intermediate points of $l$-geodesics from  $A_0$ to $A_1$, and \autoref{Def:tau dist coeff} for the definition of the averaged distortion coefficients appearing below.

\begin{theorem}[Sharp timelike Brunn--Minkowski inequality II]\label{Th:BMINKTMCP} Let \autoref{Ass:Mod} hold. Let $A_0\subset\mms$ be a Borel set with positive and finite $\meas$-measure. Assume $x \ll x_1$ for every $x$ in the closure of $A_0$. Define $\mu_0,\mu_1\in\Prob(\mms)$ by
\begin{align*}
\mu_0 &:= \meas[A_0]^{-1}\,\meas\mres A_0,\\
\mu_1 &:= \delta_{x_1}.
\end{align*}
Let $\bdpi$ denote the only element of $\smash{\OptTGeo_p(\mu_0,\mu_1)}$. Then for every $0\leq t\leq 1$,
\begin{align*}
\meas[A_t]^{1/N}\geq \meas\big[\!\supp(\eval_t)_\push\bdpi\big]^{1/N} \geq \inf\tau_{k_\gamma^-/N}^{(1-t)}(\vert\dot\gamma\vert)\,\meas[A_0]^{1/N},
\end{align*}
where the infimum is taken over all $\gamma\in\supp\bdpi$.

The inequality holds a fortiori if the infimum is taken over all $\gamma\in G(A_0,A_1)$.
\end{theorem}

In turn, the conclusions of the
\begin{itemize}
\item sharp timelike Bonnet--Myers inequality (\autoref{Cor:111}),
\item sharp timelike Schneider inequality (\autoref{Cor:Sharp Schneider}),
\item sharp timelike Bishop--Gromov inequality (\autoref{Cor:222}), and
\item sharp dimension bound (\autoref{Cor:Hausdorffdim+}) 
\end{itemize}
remain unchanged under the weaker \autoref{Ass:Mod}: their proofs all employ the timelike Brunn--Minkowski inequality in its version from \autoref{Th:BMINKTMCP}.

Lastly, \autoref{Th:Hawking} translates into the following analog.  Indeed, the estimate from \autoref{Cor:Comparison} used in its proof remains valid under \eqref{EStated}.

\begin{theorem}[Synthetic Hawking-type singularity theorem II]  Let \autoref{Ass:Mod} hold, and assume $\mms$ is proper. Let $V\subset\mms$ be an achronal FTC Borel set such that $\meas[V]=0$ and $\hh_0[V]<\infty$. We assume the  forward mean curvature of $V$ is bounded from above by $-\beta$, where $\beta > 0$. Let $T>0$ and $\delta > 0$  obey
\begin{align*}
\hh_0[V]^{-1}\int_{V^{0,T}} k_-\d\meas < c\,\Big[\beta - \frac{N-1}{T}\Big],
\end{align*}
where
\begin{align*}
c := \q\textnormal{-}\!\!\essinf_{\alpha\in \sfQ(V)} \Big[\frac{\SIN_{-\varkappa_{0,T+\delta}^-/(N-1)}(\delta)}{\SIN_{-\varkappa_{0,T+\delta}^-/(N-1)}(T+\delta)}\Big]^{N-1}.
\end{align*}
Then we have
\begin{align*}
\hh_0\big[V\setminus\Reg_{T+\delta}\big] > 0.
\end{align*}

In other words, there exists a transport ray starting in $V$ whose maximal domain of definition in $\smash{I^+(V)\cup V}$ is strictly contained in $[0,T+\delta)$.
\end{theorem}

\section{Limit curves; lifting causality from events to measures}
\label{Ch:Loose}

In this appendix, we prove \autoref{Th:Compact Polish} (enhanced topological properties of $\mms$) and the limit curve \autoref{Th:Limit curve theorem} for general strictly causal curves. Both proofs rely on linking our approach from \autoref{Ch:Metric} to the one of \emph{bounded Lorentzian metric spaces} from Minguzzi--Suhr \cite{minguzzi2022}, to which we refer for precise definitions as we will only borrow several results from \cite{minguzzi2022}.  Moreover, inspired by \cite{eckstein2017} we study causality --- most notably, global hyperbolicity --- for the space $\scrP(\mms)$.

In the following, only our first standing \autoref{Ass:GHLLS} is needed. The second-countability of $\Top$ will not be required.

\subsection{From metric spacetimes to bounded Lorentzian metric spaces} The bounded Lorentzian metric spaces we construct are 
quotients \cite[Sec.~1.3]{minguzzi2023} of causal diamonds in $\mms$. To do so, fix $a,b\in \mms$ with $a\ll b$. Then
\begin{itemize}
\item $X := J(a,b)$ is a space with relative topology inherited from $\Top$, and 
\item the restriction $\smash{\tsep := l_+\big\vert_{X^2}}$ to $X^2$ is continuous, finite by \autoref{Cor:Finiteness l}, has compact superlevel sets, and satisfies the reverse triangle inequality \eqref{Eq:Reverse tau}.
\end{itemize}
Hence, we are in the setting of \cite[Sec.~1.3]{minguzzi2022}. 
Following their procedure, we define an equivalence relation $\sim$ on $X$ by $x\sim y$ if $\tsep(x,z) = \tsep(y,z)$ and $\tsep(z,x) = \tsep(z,y)$ for every $z\in X$. Let $\smash{\tilde{X}}$ denote the quotient space with quotient topology $\smash{\tilde\Top}$, and let $\smash{\pi\colon X\to \tilde{X}}$ denote the canonical projection given by $\pi(x) := \tilde{x}$. The function $\smash{\tilde{\tsep}\colon \tilde{X}^2\to [0,\infty)}$ given by $\smash{\tilde{\tsep}(\tilde{x}, \tilde{y}) := \tsep(x,y)}$ is well-defined and continuous.

The first part of the following lemma is then easy to deduce. The crucial point for us lies in the second statement.

\begin{lemma}[Quotienting out undistinguished points]
\label{Le:BLMS} The tuple $\smash{(\tilde{X}, \tilde{\tsep})}$ constitutes a bounded Lorentzian metric space in the sense of \cite[Def.~1.1]{minguzzi2022}. 
Moreover, the restriction of $\pi$ to $I(a,b)$ is a homeomorphism onto $\pi(I(a,b))$.
\end{lemma}

\begin{proof} The first statement follows from \cite[Prop.~1.18]{minguzzi2022}.

We now show the second claim. Clearly, $\pi$ is surjective and continuous on $X$, hence its restriction to $I(a,b)$ is. 

We now claim injectivity of $\pi$ on $I(a,b)$. Let $x,y\in I(a,b)$ with $\tilde{x} = \tilde{y}$. Owing to the nontrivial chronology in \autoref{Def:LLS}, let $\smash{\gamma\colon [0,1]\to \mms}$ be a timelike curve with midpoint $x$. Let $z\in \mms$ [sic] obey $x\ll z$. Since $I(a,b)$ is open and the curve $\smash{\gamma}$ is continuous, we have $\smash{\gamma_t\in X}$ for $0\leq t\leq 1$ sufficiently close to $1/2$. Up to moving $t$ closer to $1/2$, by continuity of $\smash{l_+}$ we have $l(\gamma_t, z)>0$. Then the reverse triangle inequality \eqref{Eq:Reverse tau} and our assumption $\tilde{x}=\tilde{y}$ imply
\begin{align*}
l(y,z) \geq l(y,\gamma_t) + l(\gamma_t, z) = l(x,\gamma_t) + l(\gamma_t, z) >0,
\end{align*}
and thus $y\ll z$. Reversing the roles of $x$ and $y$ and working with a timelike curve passing through $y$ thus leads to $\smash{I^+(x) = I^+(y)}$, where both chronological futures are understood globally with respect to $l$.  \autoref{Cor:Distinction} then implies $x=y$, hence the claimed injectivity of the restriction of $\pi$ to $I(a,b)$.

Finally, as $X$ is compact, $\smash{\tilde{X}}$ is Hausdorff as a Polish space, and $\pi$ is continuous, the latter is a closed map. Thus the restriction of $\pi$ to $\smash{\pi^{-1}(\pi(I(a,b))) = I(a,b)}$ is relatively closed, which implies homeomorphy.
\end{proof}

\subsection{Proof of \autoref{Th:Compact Polish}} Before entering the proof of the mentioned result, we restate it for convenience.

\begin{theorem}[Polish compact sets] \autoref{Ass:GHLLS} implies the relative topology of $\Top$ on every compact subset $C\subset\mms$ is Polish.
\end{theorem}

\begin{proof} We first claim that the relative topology on $I(a,b)$ is metrizable for every $a,b\in\mms$. Let $\smash{(\tilde{X},\tilde{\tsep})}$ as in \autoref{Le:BLMS}. By \cite[Thm.~1.9]{minguzzi2022}, the latter is a Polish space, which immediately implies metrizability of $I$ by \autoref{Le:BLMS} (e.g.~by taking the pullback metric of $\pi$ under a complete metric $\smash{\tilde{\met}}$ inducing the target topology).

Nontrivial chronology shows $\mms$ is locally metrizable --- every point in $\mms$ has a neighborhood which is metrizable with respect to the subspace topology. Since the latter is Hausdorff by \autoref{Cor:Hausdorff}, and $C$ is clearly paracompact, Smirnov's metrization theorem \cite[Thm.~4.4.19]{engelking} implies metrizability of the relative topology of $\Top$ on the set $C$.

Now the argument is as for \cite[Thm.~1.9]{minguzzi2022}. Since $C$ is compact, it is Lindelöf. Together with metrizability, this implies second-countability of $C$. As every second countable space which is also Hausdorff and locally compact is Polish, the claim follows from \autoref{Cor:Hausdorff} again.
\end{proof}

As for \autoref{Cor:Finiteness l}, we stress that the previous proof only requires $\smash{l_+}$ to be continuous near the diagonal of $\smash{\mms^2}$.

\subsection{Proof of \autoref{Th:Limit curve theorem}} Now we turn to our general limit curve theorem. For convenience, we restate it here as well.

\begin{theorem}[Limit curve theorem]  Assume \autoref{Ass:GHLLS}. Let $C_0,C_1\subset\mms$ be compact subsets, and let $\smash{(\gamma^n)_{n\in\N}}$ be a sequence of strictly causal curves on $[0,1]$ such that $\smash{\gamma_0^n \in C_0}$ and $\smash{\gamma_1^n\in C_1}$ for every $n\in\N$. Then $\smash{(\gamma^n)_{n\in\N}}$ has a subsequence which --- after reparametrization --- converges uniformly, in the sense of \autoref{Def:Unif cvg}, 
\begin{enumerate}[label=\textnormal{(\roman*)}]
\item to a strictly causal curve defined on a compact interval $[\alpha,\beta] \subset \R$ of positive length unless the endpoints of the extracted subsequence converge to the same limit, and 
\item to a constant causal curve defined on $[0,1]$ otherwise.
\end{enumerate}
\end{theorem}

\begin{proof} By compactness of $C_0$ and $C_1$, the sequences $\smash{(\gamma_0^n)_{n\in\N}}$ and $\smash{(\gamma_1^n)_{n\in\N}}$ have limits 
$a'' \in C_0$ and $b'' \in C_1$
along subsequences that we do not relabel.  By nontrivial chronology in \autoref{Def:LLS}, let $a,a',b,b'\in \mms$ with $a\ll a' \ll a''$ and $b'' \ll b' \ll b$. Recall that $a$, $a'$, and $a''$ are mutually distinct by \autoref{Cor:Causality}, and the same holds for $b$, $b'$, and $b''$. Furthermore, by \autoref{Le:Closure},  global hyperbolicity, and continuity of $\smash{l_+}$, $J(a',b')$ is compactly contained in $I(a,b)$. Lastly, up to removing finitely many members of either sequence, we assume $\smash{\gamma_0^n \in I^+(a')}$ and $\smash{\gamma_1^n \in I^-(b')}$ for every $n\in\N$; in other words, by causality and \autoref{Le:Pushup openness} the image of $\smash{\gamma^n}$ is entirely contained in $I(a',b')$ for every $a,b\in\mms$.  See \autoref{Fig:Enlargement}.

\begin{figure}
\centering
\begin{tikzpicture}
\draw[fill=black!7] (0.4,-3) -- (4.25,2) -- (0.4, 7) -- (-3.5, 2) -- (0.4,-3); 
\draw[fill=black!17] (0.5,-2) -- (3.55,2) -- (0.5,6) -- (-2.55,2) -- (0.5,-2);
\draw[fill=black!27] (0,0) -- (1.5,2) -- (0,4) -- (-1.5,2) -- (0,0);
\draw[fill=black] (0.4,-3) circle(.1em);
\draw[fill=black] (0.4,7) circle(.1em);
\node at (0.65,-3.2) {$a$};
\node at (0.675, 7.2) {$b$};
\draw[dashed] (0.4,-3) -- (0.5,-2);
\draw[dashed] (0.5,6) -- (0.4,7);
\draw [thick] plot [smooth] coordinates {(0,0) (-0.2,1) (0.3,2) (-0.1,2.8) (0,4)};
\draw[fill=black] (0,0) circle(.1em);
\node at (-0.3,-0.2) {$a''$};
\draw[fill=black] (0,4) circle(.1em);
\node at (-0.25,4.2) {$b''$};
\node at (-0.3,2.45) {$\sigma$};
\draw[fill=black] (1,4.5) circle(.1em);
\draw[fill=black] (0.8,-0.5) circle(.1em);
\draw [thick] plot [smooth] coordinates {(0.8,-0.5) (1.5,0.5) (0.5,1.5) (1.8,2.7) (0.8, 3.75) (1,4.5)};
\node at (0.75, -0.95) {$\smash{\gamma^n_0}$};
\node at (0.8,4.75) {$\smash{\gamma^n_1}$};
\node at (2.1,2.4) {$\sigma^n$};
\draw[dashed] plot [smooth] (0.5,-2) -- (0,0);
\draw[dashed] plot [smooth] (0,4) -- (0.5,6);
\draw[fill=black] (0.5,-2) circle(.1em);
\node at (0.25,-2.2) {$a'$};
\draw[fill=black] (0.5,6) circle(.1em);
\node at (0.2,6.1) {$b'$};
\draw[black, fill=white] (4.25,2) circle(.25em);
\draw[black, fill=white] (-3.5,2) circle(.25em);
\end{tikzpicture}
\caption{An illustration of the proof of \autoref{Th:Limit curve theorem}.}\label{Fig:Enlargement}
\end{figure}
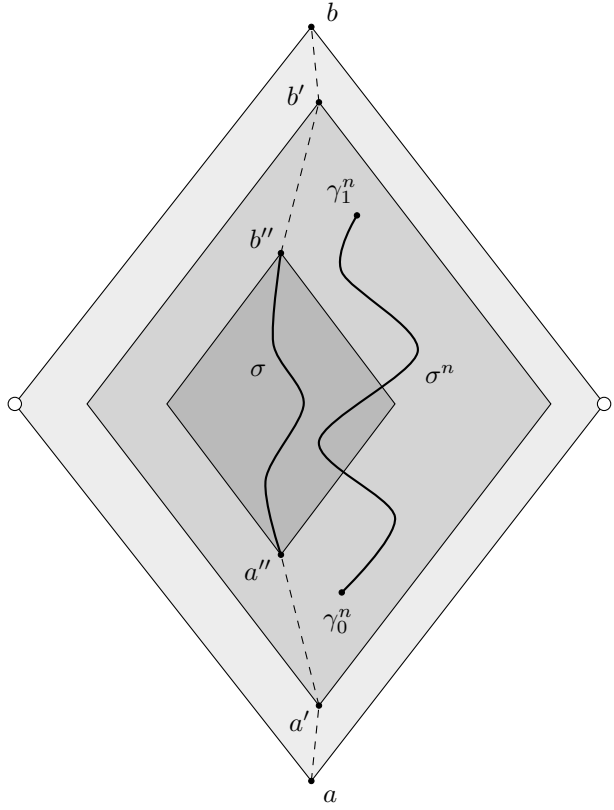

We now consider the bounded Lorentzian metric space $\smash{(\tilde{X}, \tilde{\tsep})}$ from \autoref{Le:BLMS} constructed from $X := J(a,b)$. We retain all inherent  notation from above, and we will show uniform convergence with respect to a metric $\met$ as in \autoref{Cor:Polish}. Without restriction, to make $\smash{\tilde{X}}$ compact we add its spacelike boundary to it \cite[Rem.~1.2, Cor.~1.6, Prop.~1.7]{minguzzi2022}. Lastly, by \cite[Thm.~5.7]{minguzzi2022}, $\smash{\tilde{X}}$ admits a bounded and uniformly continuous time function $\smash{\sft\colon \tilde{X}\to \R}$ 
in the terminology of \cite[Def.~5.6]{minguzzi2022}. 

By using the reverse triangle inequality \eqref{Eq:Reverse tau} and injectivity of $\pi$ implied by \autoref{Le:BLMS}, it is easy to see that $\smash{\tilde{\gamma}^n\colon [0,1] \to \tilde{X}}$ is isocausal in the sense of \cite[Def.~5.1, Sec.~5.2]{minguzzi2022}\footnote{This means strictly causal in our terminology.} for every $n\in\N$, where 
\begin{align*}
\tilde{\gamma}^n := \pi\circ\gamma^n;
\end{align*}
since $\smash{\tilde{\gamma}^n_{[0,1]}}\subset I(a,b)$, none of these  curves passes through the spacelike boundary. 

First assume the endpoint limits $a''$ and $b''$ extracted above to be distinct. As $\smash{\tilde{\tsep}([a''], [b'']) = \tsep(a'',b'') = l(a'',b'') \geq 0}$ by definition of the quotient time separation $\smash{\tilde{\tsep}}$, we obtain $\sft([b'']) > \sft([a''])$, and by continuity of $\sft$ we may and will assume the difference $\smash{\sft(\tilde\gamma_1^n) - \sft(\tilde\gamma_0^n)}$ to be bounded away from zero uniformly in $n\in\N$. Thanks to \cite[Prop.~5.8]{minguzzi2022}, the curve $\smash{\tilde{\gamma}^n}$ can be reparametrized with respect to $\sft$ for every $n\in\N$. More precisely, an explicit possibility to do so is to consider the re\-para\-metrization $\smash{\tilde{\sigma}^n\colon[\alpha_n,\beta_n] \to \tilde{X}}$ defined by
\begin{align*}
\tilde{\sigma}^n &:= \tilde{\gamma}^n \circ (\sft\circ \tilde{\gamma}^n)^{-1},\\
\alpha_n &:= \sft(\tilde\gamma_0^n),\\
\beta_n &:= \sft(\tilde\gamma_1^n).
\end{align*}
Then $(\alpha_n)_{n\in\N}$ and $(\beta_n)_{n\in\N}$ have no common accumulation point. Therefore, the limit curve theorem \cite[Thm.~5.12]{minguzzi2022}  entails uniform convergence of  $\smash{(\smash{\tilde{\sigma}^n)_{n\in\N}}}$ to a strictly causal curve $\smash{\tilde\sigma\colon [\alpha,\beta]\to\tilde{X}}$, where $\alpha,\beta\in \R$ with $\alpha < \beta$; more precisely, 
\begin{align*}
\alpha &:= \sft(a''),\\
\beta &:= \sft(b'').
\end{align*}
Since $J(a',b')$ is compactly contained in $I(a,b)$, by \autoref{Le:BLMS} the map $\smash{\pi}$ is a homeomorphism of $J(a',b')$ onto its compact image $\smash{\pi(J(a',b'))}$. In particular, the inverse of this map is uniformly continuous. As  such maps preserve uniform convergence, we get uniform convergence of $\smash{(\pi^{-1}\circ \tilde{\sigma}^n)_{n\in\N}}$ to $\smash{\pi^{-1}\circ \tilde\sigma}$. The former are reparametrizations of $\smash{(\gamma^n)_{n\in\N}}$, and we obtain our claim. Note again that the limit curve (or any reparametrization thereof) thus obtained is strictly causal.

Lastly, we investigate the case $a'' = b''$.  By \autoref{Cor:Strong causality}, there exists a neighborhood $\smash{V\subset\sfB^\met(a'',\varepsilon)}$ such that every strictly causal curve whose endpoints lie in $V$ does not leave $\smash{\sfB^\met(a'',\varepsilon)}$. For every sufficiently large $n\in\N$, our hypothesis yields $\smash{\gamma_0^n,\gamma_1^n\in V}$. Consequently,
\begin{align*}
\limsup_{n\to\infty}\sup_{t\in[0,1]} \met(a'',\gamma_t^n)  \leq \varepsilon.
\end{align*}
Since $\varepsilon$ was arbitrary, the claim is established.
\end{proof}

\subsection{Causality for probability measures}\label{Sub:Causality pm} In \cite{eckstein2017} a natural causal structure of the space $\scrP(\mms)$ was introduced and studied for smooth spacetimes $\mms$. As outlined here, the definition easily extends to our nonsmooth setting. In fact, we succeed in showing global hyperbolicity of $\mms$ (according to \autoref{Th:GHy}) ``lifts'' to $\scrP(\mms)$.

Define a relation $\preceq$ on $\scrP(\mms)$ by $\mu\preceq\nu$ provided $\smash{\Pi_\leq(\mu,\nu) \neq \emptyset}$. This can be restated by saying $\smash{\ell_p(\mu,\nu)\geq 0}$ for some --- equivalently, every --- $0<p\leq 1$. As for the causal relation $\leq$ on $\mms$, given $\scrX,\scrY\subset\scrP(\mms)$ we then define
\begin{itemize}
\item the \emph{causal future} of $\scrX$ by $\scrJ^+(\scrX) := \{\ell_p(\mu,\cdot) \geq 0\}$,
\item the \emph{causal past} of $\scrY$ by $\scrJ^-(\scrY) := \{\ell_p(\cdot,\nu)\geq 0\}$, and 
\item the \emph{causal emerald} between $\scrX$ and $\scrY$ by $\scrJ(\scrX,\scrY) := \scrJ^+(\scrX) \cap \scrJ^-(\scrY)$.
\end{itemize}

\begin{proposition}[Order property of $\preceq$]\label{Pr:Order} The relation $\preceq$ is an  order, i.e.~it is reflexive, anti\-symmetric, and transitive.
\end{proposition}

\begin{proof} Transitivity follows from the gluing lemma. Nonemptiness of $\smash{\Pi_\leq(\mu,\mu)}$ for every $\mu\in\scrP(\mms)$ follows from causality of the diagonal coupling of $\mu$. In order to prove antisymmetry, we first claim this coupling is the \emph{only} element in $\Pi_\leq(\mu,\mu)$. This, however, follows from \cite[Thm.~12]{eckstein2017} whose proof extends with no changes to our setting by using the generalized time function from \autoref{Th:Ex time function}.

Now let $\mu,\nu\in\scrP(\mms)$ obey $\mu\preceq\nu\preceq\mu$. By the previous part, gluing of elements of $\Pi_\leq(\mu,\nu)$ and $\Pi_\leq(\nu,\mu)$ induces a measure $\omega\in\scrP(\mms^3)$ with marginals $\mu$, $\nu$, and $\mu$. By the previous paragraph, $(\pr_1,\pr_3)_\push\omega$ is the diagonal coupling of $\mu$, and hence $\omega$-a.e.~$(x,y,z)\in\mms^3$ satisfies $x\leq y \leq z = x$, thus $x=y$ by causality. Consequently, $(\pr_1,\pr_2)_\push\omega$ is concentrated on the diagonal of $\smash{\mms^2}$, and hence $\mu=\nu$.
\end{proof}

A nice consequence of independent interest is the following. 

\begin{theorem}[Global hyperbolicity]\label{Th:GHy prob meas} \autoref{Ass:GHLLS} implies the space $\scrP(\mms)$ is globally hyperbolic in the following sense.
\begin{enumerate}[label=\textnormal{\textcolor{black}{(}\roman*\textcolor{black}{)}}]
\item\label{La:n1} The relation $\preceq$ is a narrowly closed order.
\item\label{La:n2} If $\scrC\subset\scrP(\mms)$ is narrowly compact, so is $\scrJ(\scrC,\scrC)$.
\end{enumerate}
\end{theorem}

\begin{proof} To show \ref{La:n1}, by \autoref{Pr:Order} it suffices to prove narrow closedness. For sequences $(\mu_n)_{n\in\N}$ and $(\nu_n)_{n\in\N}$ in $\scrP(\mms)$ satisfying $\mu_n\preceq \nu_n$ for every $n\in\N$ and converging to $\mu,\nu\in\scrP(\mms)$, respectively, let $(\pi_n)_{n\in\N}$ a sequence of elements $\pi_n\in\Pi_\leq(\mu_n,\nu_n)$. Since the marginal sequences are tight, so is $(\pi_n)_{n\in\N}$, and a nonrelabeled subsequence thus converges narrowly to some $\pi\in\Pi(\mu,\nu)$ by Prokhorov  (\autoref{Th:Prokhorov}). Alexandrov (\autoref{Th:Alexandrovs theorem}) shows $\pi$ is causal.

We turn to the proof of \ref{La:n2}. We first argue that if $\mu,\nu\in\scrP(\mms)$ satisfy $\mu\preceq\nu$, then every closed set $C\subset\mms$ satisfies
\begin{align}\label{Eq:Mon mu nu}
\begin{split}
\mu\big[J^+(C)\big] &\leq \nu\big[J^+(C)\big],\\
\nu\big[J^-(C)\big] &\leq  \mu\big[J^-(C)\big];
\end{split}
\end{align}
compare with \cite[Thm.~8]{eckstein2017}. This statement makes sense as $J^\pm(C)$ is closed, hence Borel, by closedness of $C$. Now, by the push-up property the indicator function $\One_{J^+(C)}$ is a causal function on $\mms$, in the sense that if $x,y\in\mms$ satisfy $x\leq y$, then $\One_{J^+(C)}(x) \leq \One_{J^+(C)}(y)$. Given $\pi\in\Pi_\leq(\mu,\nu)$, this easily yields
\begin{align*}
\mu\big[J^+(C)\big] = \int_{\mms^2} \One_{J^+(C)} \circ\pr_1\d\pi \leq \int_{\mms^2} \One_{J^+(C)} \circ\pr_2\d\pi = \nu\big[J^+(C)\big].
\end{align*}
The second inequality in \eqref{Eq:Mon mu nu} follows analogously. Finally, let $\scrC$ as hypothesized. Given  $\varepsilon > 0$, Prokhorov's theorem implies the existence of a compact set $C\subset\mms$ such that every $\mu\in\scrC$  satisfies
\begin{align*}
\mu\big[C^\sfc\big] &\leq \varepsilon.
\end{align*}
Given any $\sigma \in\scrJ(\scrC,\scrC)$ and $\mu,\nu\in\scrC$ with $\mu\preceq\sigma\preceq\nu$, \eqref{Eq:Mon mu nu} then implies
\begin{align*}
\sigma\big[J(C,C)^\sfc\big] &\leq \sigma\big[J^+(C)^\sfc\big] + \sigma\big[J^-(C)^\sfc\big]\\
&= 2 - \sigma\big[J^+(C)\big] - \sigma\big[J^-(C)\big]\\
&\leq 2-\mu\big[J^+(C)\big] - \nu\big[J^-(C)\big]\\
&\leq 2-\mu[C] - \nu[C]\\
&\leq 2\,\varepsilon.
\end{align*}
Together with closedness of $\preceq$, this shows the desired compactness by Prokhorov's theorem. The proof is terminated.
\end{proof}

Note that in general, the conclusion of \autoref{Th:GHy prob meas} ``projects down'' to $\mms$, in the sense that it implies $\leq$ is a closed order, and the induced causal emeralds are compact. This can be seen by considering Dirac masses, respectively.

Similarly to \autoref{Le:Zt lemma}, we can now shortly discuss (pre)compactness of intermediate point sets of  measures. The corresponding \autoref{Le:Zt lemma II} is only used in the proof of \autoref{Th:Good}, but might be of independent interest.

For $\mu,\nu\in\scrP(\mms)$ and $0\leq t\leq 1$, we define the (possibly empty, $p$-dependent) $t$-intermediate point set
\begin{align*}
\scrZ_t^p(\mu,\nu) := \{\sigma\in \scrJ(\mu,\nu) : \ell_p(\mu,\sigma) = t\,\ell_p(\mu,\nu), \,\ell_p(\sigma,\nu) = (1-t)\,\ell_p(\mu,\nu)\}.
\end{align*}
Moreover, for $\scrC\subset\Prob(\mms)^2$ we define
\begin{align*}
\scrZ_t^p(\scrC) &:= \bigcup_{(\mu,\nu)\in\scrC} \scrZ_t^p(\mu,\nu),\\
\scrZ^p(\scrC) &:= \bigcup_{t\in[0,1]} \scrZ_t^p(\scrC).
\end{align*}

The following result is not straightforward  due to the anticipated lack of semicontinuity of $\smash{\ell_p}$ \cite[Rem.~2.17]{cavalletti2020}. Still, the identities defining intermediate point sets and the reverse triangle inequality \eqref{Eq:Reverse lp} for $\ell_p$ guarantee enough rigidity.

\begin{lemma}[Intermediate point sets inherit compactness II]\label{Le:Zt lemma II} Let $0\leq t\leq 1$. If $\scrC\subset\scrP(\mms)^2$ is precompact, so are $\smash{\scrZ_t^p(\scrC)}$ and $\smash{\scrZ^p(\scrC)}$. Moreover, if $\scrC = \{(\mu,\nu)\}$ for $\smash{\mu,\nu\in\Prob_\comp(\mms)}$, then $\smash{\scrZ_t^p(\scrC)}$ and $\smash{\scrZ^p(\scrC)}$ are compact.
\end{lemma}

\begin{proof} The first claim follows from \autoref{Th:GHy prob meas} as for \autoref{Le:Zt lemma}.

The last statement is trivial if $\mu\not\preceq\nu$, hence we assume $\smash{\ell_p(\mu,\nu)\geq 0}$. It suffices to show $\smash{\scrZ^p(\mu,\nu)}$ is narrowly closed under the given assumption $\scrC = \{(\mu,\nu)\}$. Then $C=J(\supp \mu,\supp \nu)$ is compact.
Let $(\sigma_n)_{n\in\N}$ be a sequence in $\smash{\scrZ^p(\scrC)}$ narrowly converging to $\sigma\in\Prob(\mms)$. We claim $\smash{\sigma\in\scrZ^p(\scrC)}$. Indeed, let $(\mu_n,\nu_n)_{n\in\N}$ be a  sequence in $\scrZ(\scrC)$ and let $(t_n)_{n\in\N}$ be a sequence in $[0,1]$ --- without loss of generality converging to $0\leq t\leq 1$ --- such that, for every $n\in\N$,
\begin{align*}
\ell_p(\mu,\sigma_n) &= t_n\,\ell_p(\mu,\nu),\\
\ell_p(\sigma_n,\nu) &= (1-t_n)\,\ell_p(\mu,\nu).
\end{align*}
By \cite[Prop.~2.9]{mccann2020}, we have $\supp\sigma_n\subset J(\supp\mu,\supp\nu)$ for every $n\in\N$. \autoref{Le:Existence} implies the existence of $\smash{\ell_p}$-optimal couplings $\smash{\pi_1^n\in\Pi_\leq(\mu,\sigma_n)}$ and $\smash{\pi_2^n\in\Pi_\leq(\sigma_n,\nu)}$. By Prokhorov's theorem, Alexandrov's theorem, and \cite[Lem.~4.3]{villani2009} there exist couplings $\smash{\pi_1\in\Pi_\leq(\mu,\sigma)}$ and $\smash{\pi_2\in\Pi_\leq(\sigma,\nu)}$ such that
\begin{align*}
\ell_p(\mu,\nu) &\geq \ell_p(\mu,\sigma) + \ell_p(\sigma,\nu)\\
&\geq \Vert l\Vert_{\Ell^p(\mms^2,\pi_1)} + \Vert l\Vert_{\Ell^p(\mms^2,\pi_2)}\\
&\geq \limsup_{n\to\infty} \Vert l \Vert_{\Ell^p(\mms^2,\pi_1^n)} + \limsup_{n\to\infty} \Vert l \Vert_{\Ell^p(\mms^2,\pi_2^n)}\\
&= \limsup_{n\to\infty} t_n\,\ell_p(\mu,\nu) + \limsup_{n\to\infty}\, (1-t_n)\,\ell_p(\mu,\nu)\\
&=\ell_p(\mu,\nu).
\end{align*}
This forces the implicit inequalities
\begin{align*}
\ell_p(\mu,\sigma) &\geq t\,\ell_p(\mu,\nu),\\
\ell_p(\sigma,\nu) &\geq (1-t)\,\ell_p(\mu,\nu)
\end{align*}
to be equalities. The desired inclusion $\smash{\sigma\in\scrZ_t^p(\scrC)}$ is shown.
\end{proof}



\section{Kantorovich duality formula for $\smash{\ell_1}$}\label{Ch:Rubinstein}

The disintegration result from \autoref{Th:Mean zero} requires the existence of appropriate data $E$ and $u$ according to \autoref{Sub:Framework} in order to perform the disintegration. In this chapter, we describe this procedure by using basic Kantorovich duality theory for $\smash{\ell_1}$.  As  in \cite{Akdemir24+,cm++}, we obtain an analog of the Kantorovich duality formula \cite{kant42} which mirrors the famous duality formula for the $1$-Wasserstein distance on a metric space in terms of $1$-Lipschitz functions \cite[p.~72]{villani2009}; cf.~\autoref{Th:Rubinstein}.

Following \cite[Sec.~2.4]{cavalletti2020}, we start with basic definitions from Kantorovich duality theory for $\ell_1$ (see also \cite[Ch.~4]{mccann2020}). Given two sets $A,B\subset\mms$, a function $\phi\colon A\to \R$ is \emph{$l$-concave} relative to $(A,B)$ if there exists a  function $\psi\colon B\to \R$ with
\begin{align*}
\phi(x) = \inf_{y\in B} \big[\psi(y) -l(x,y)\big].
\end{align*}
The \emph{$l$-transform} of $\phi$ is the function $\smash{\phi^{(l)}\colon B\to \R\cup\{-\infty,\infty\}}$ given by
\begin{align}\label{Eq:Formula}
\phi^{(l)}(y) := \sup_{x\in A} \big[\phi(x) + l(x,y)\big].
\end{align}

Analogous notions are set up by replacing $l$ with $\smash{l_+}$.

\begin{definition}[Strong Kantorovich duality {\cite[Def.~2.23]{cavalletti2020}}]\label{Def:Strong Kantorovich} A pair $\smash{(\mu,\nu)\in\Prob(\mms)^2}$ is said to satisfy \emph{strong $l$-Kantorovich duality} if
\begin{enumerate}[label=\textnormal{\alph*.}]
\item $0<\ell_1(\mu,\nu) < \infty$, and
\item there exist closed subsets $A,B\subset\mms$ with $\mu[A]=\nu[B]=1$, a $\mu$-integrable upper semicontinuous  function $\phi\colon A\to \R$ which is $l$-concave relative to $(A,B)$ such that its $\smash{l}$-transform $\smash{\phi^{(l)}\colon B\to \R}$ is $\nu$-integrable and obeys
\begin{align}\label{Eq:Original Kantorovich}
\ell_1(\mu,\nu) = \int_\mms\phi^{(l)}\d\nu - \int_\mms\phi\d\mu.
\end{align}
\end{enumerate}

The notion of \emph{strong $\smash{l_+}$-Kantorovich duality} of $(\mu,\nu)$ is defined analogously by replacing $l$ with $\smash{l_+}$.
\end{definition}

\begin{example}[Criterion for strong Kantorovich duality]\label{Ex:Strong Kantorovich}  By the proofs of \cite[Thm. 2.26, Cor.~2.29]{cavalletti2020}, every pair $\smash{(\mu,\nu)\in\Prob_\comp(\mms)^2}$ satisfying
\begin{align*}
\supp\mu\times\supp\nu \subset \{l>0\}
\end{align*}
satisfies strong $\smash{l_+}$-Kantorovich duality (hence strong $l$-Kantorovich duality).
\end{example}

\begin{theorem}[Kantorovich duality for $\ell_1$]\label{Th:Rubinstein} Assume that  the pair $(\mu,\nu)\in \smash{\Prob_\comp(\mms)^2}$ satisfies strong $l$-Kantorovich duality. Then the following  hold.
\begin{enumerate}[label=\textnormal{(\roman*)}]
\item \textnormal{\textbf{Kantorovich duality formula.}} We have
\begin{align*}
\ell_1(\mu,\nu) = \inf \Big[\!\int_\mms u\d\nu - \int_\mms u\d\mu\Big].
\end{align*}
The infimum is taken over all Borel functions $u\colon E\to \R$ defined on some Borel set $E\subset\mms$ containing $\supp\mu\cup\supp\nu$ with $u\in \Ell^1(\mms,\mu)\cap\Ell^1(\mms,\nu)$ and for every $x,y\in E$, we have
\begin{align}\label{Eq:Reverse LIP}
u(y) - u(x) \geq l(x,y).
\end{align} 
\item \textnormal{\textbf{Existence of minimizers.}} The infimum is attained at a function $u\colon E\to \R$ with the additional property that $E$ constitutes a compact, causally convex set containing $\supp\mu\cup\supp\nu$.
\item \textnormal{\textbf{Optimal couplings.}} For every minimizer $u$ in the above Kantorovich--Ru\-bin\-stein formula and every $\ell_1$-optimal coupling $\pi\in\Pi(\mu,\nu)$,
\begin{align*}
u\circ\pr_2 - u\circ\pr_1 = l\quad\pi\textnormal{-a.e.}
\end{align*}
\end{enumerate}
\end{theorem}

\begin{proof} To some extent, the argument is standard, cf.~e.g.~\cite[p.~68]{villani2009}; let us only outline the key steps. Retain the notation from \autoref{Def:Strong Kantorovich}.  Without loss of generality, we assume $A$ and $B$ are compact. 
On $E := J(A,B)$, define $u\colon E\to \R$ by  the formula \eqref{Eq:Formula}, that is,
\begin{align*}
u(y) := \sup_{x\in A} \big[\phi(x) + l(x,y)\big].
\end{align*}
$E$ is causally convex and, by global hyperbolicity, compact.  Our specific choice of $E$ ensures $u$ does not attain the value $-\infty$. It is straightforward to verify  \eqref{Eq:Reverse LIP} for every $x,y\in E$. Furthermore, we have $\smash{u = \phi^{(l)}}$ on $B$, in particular $u\in \Ell^1(\mms,\nu)$. On the other hand, we have $u \geq \phi$ on $A$, while $u$ is bounded from above on the compact set $A$ by upper semicontinuity of $\phi$ and $l$, whence $u\in\Ell^1(\mms,\mu)$.

Since $0<\ell_1(\mu,\nu)< \infty$ by the hypothesized strong $l$-Kantorovich duality, given any $\varepsilon > 0$ there is  an $\varepsilon$-almost maximizer  $\smash{\pi\in\Pi_\leq(\mu,\nu)}$ in the definition of $\ell_1(\mu,\nu)$. By the previous part, this entails
\begin{align*}
\ell_1(\mu,\nu) &\leq \int_{\mms^2} l\d\pi +\varepsilon\\
&\leq \int_\mms u\d\nu - \int_\mms u\d\mu +\varepsilon\\
&\leq \int_\mms \phi^{(l)}\d\nu - \int_\mms \phi\d\mu + \varepsilon. 
\end{align*}
Combining the arbitrariness of $\varepsilon$ and \eqref{Eq:Original Kantorovich} proves $u$ attains the infimum in the  Kantorovich duality formula for $\smash{\ell_1}$. This  yields the first two claims. 

The last follows from the simple observation that the Kantorovich duality formula phrased in terms of an $\smash{\ell_1}$-optimal coupling $\pi\in\Pi(\mu,\nu)$ forces equality in \eqref{Eq:Reverse LIP} for $\pi$-a.e.~$\smash{(x,y)\in \mms^2}$.
\end{proof}

\end{document}